\documentclass{itpthesis}

\usepackage{mtunicode}
\usepackage{qittools}

\usepackage{graphicx}

\usepackage{makeidx}
\makeindex

\usepackage{subfigure}
\usepackage{epsdice}

\begin{document}

%............................ 

\ethdissnumber{20213}

%\title{A Framework for Non-Asymptotic Quantum Information Theory and Entropic 
%Uncertainty Relations}
\title{A Framework for Non-Asymptotic Quantum Information Theory}

\author{Marco Tomamichel}
\previousdegree{Dipl.\ Ing.\ ETH}
\authorinfo{born March 13, 1981, in St.~Gallen\\ citizen of Bosco Gurin, TI, 
Switzerland}

\referees{Prof.\ Dr.\ Renato Renner, examiner\\ Prof.\ Dr.\ Andreas Winter, co-examiner\\ Prof.\ Dr.\ Amos Lapidoth, co-examiner}
\degreeyear{2012}

\maketitle

%............................ 

%\begin{dedication}
%\end{dedication}

%............................ 

\begin{acknowledgements}
  \vspace{-0.3cm}
First and foremost, I want to thank my advisor, Renato Renner. I am grateful for his willingness to supervise me and for his help with both conceptual and technical questions. Most importantly, I appreciated his unsurmountable optimism in the face of scientific challenges. For every problem I encountered, he readily proposed at least one potential solution and discussed it with me in detail.

A special thanks also goes to Roger Colbeck and Fré\-dé\-ric Du\-puis, who shared office with me as post-docs. They generously set aside time to answer my questions and I learned a lot through these interactions.

On top of that, I am greatly indebted to all my coauthors for giving me the opportunity to collaborate with and learn from them. This includes Mario Berta, Roger Colbeck, Frédéric Dupuis, Torsten Franz, Fabian Furrer, Nicolas Gisin, Esther Hänggi, Stefan Hengl, Char\-les Lim Ci Wen, Renato Renner, Christian Schaf\-f\-ner, 
Volkher Scholz, Adam Smith, Oleg Szehr, Reinhard Werner and Severin Winkler.

I want to thank my co-examiners, Andreas Winter and Amos Lapidoth, for their constructive comments that helped improve the final version of this thesis. Andreas Winter's suggestions led to a significant improvement of the asymptotic equipartition property. Matthias Christandl proposed to consider chained uncertainty relations. Frédéric Du\-puis and Johan {\AA}berg were bothered by me with various technical questions during the writeup, and always had an answer. Normand Beaudry, Mario Berta, Dejan Dukaric, Philip\-pe Faist, Lídia del Rio, Cyril Stark, and Severin Winkler all gave me helpful comments on drafts of this thesis. 

Finally, I want to extend my gratitude to all members of the Quantum Information Science group, past and present, for making my time at the Institute for Theoretical Physics very enjoyable. It goes without saying that this work would have been impossible without the continuos support of my family and friends. 
They kept me grounded in this intriguing world of ours, which fortunately distinguishes itself from the abstract quantum world of this thesis through its reluctance to be described by any meaningful mathematical model.

The present version contains minor corrections that are due to very helpful comments by Mark M.~Wilde.

\end{acknowledgements}

%............................ 

% abstract - en

\begin{abstract}
  This thesis consolidates, improves and extends the smooth entropy framework for non-asymptotic information theory and cryptography.

We investigate the conditional min- and max-entropy for quantum states, generalizations of classical Rényi entropies. We introduce the purified distance, a novel metric for unnormalized quantum states and use it to define smooth entropies as optimizations of the min- and max-entropies over a ball of close states. We explore various properties of these entropies, including data-processing inequalities, chain rules and their classical limits. The most important property is an entropic formulation of the asymptotic equipartition property, which implies that the smooth entropies converge to the von Neumann entropy in the limit of many independent copies.
The smooth entropies also satisfy duality and entropic uncertainty relations that provide limits on the power of two different observers to predict the outcome of a measurement on a 
quantum system.

Finally, we discuss three example applications of the smooth entropy framework. We show a strong converse statement for source coding with quantum side information, characterize randomness extraction against quantum side information and prove information theoretic security of quantum key distribution using an intuitive argument based on the entropic uncertainty relation.
\end{abstract}

% abstract - de

%\renewcommand{\abstractname}{Zusammenfassung}
%
%
%\begin{abstract}
%  \input{abstractde}
%\end{abstract}

%............................ 
\newpage
\startnumbering

\addtocontents{toc}{~\hfill\textbf{Page}\par}
\addtocontents{lof}{~\hfill\textbf{Page}\par}
\addtocontents{lot}{~\hfill\textbf{Page}\par}

\cleardoublepage
\phantomsection
\addcontentsline{toc}{chapter}{Contents}

\tableofcontents

%............................ 

%\chapter*{}
%\vspace{-3.5cm}

\cleardoublepage
\phantomsection
\addcontentsline{toc}{chapter}{Abbreviations and Notation}

\chapter*{Abbreviations and Notation}

\begin{table}[!ht]
  \begin{tabular}{c|p{0.70\textwidth}}
    \textbf{Abbreviation} & \textbf{Description} \\
    \hline
    \vspace{-9pt} & \\ % increases spacing after the horizontal line
    CPM & 
      Completely positive map ☼*{CPM}\\
    TP-CPM & 
      Trace-preserving completely positive map ☼*{TP-CPM}\\
    SDP & 
      Semi-definite program ☼*{SDP}\\
    POVM & 
      Positive operator-valued measurement ☼*{POVM}\\
    QKD & 
      Quantum key distribution ☼*{QKD} \\
    i.i.d. & 
      Independent and identically distributed ☼*{i.i.d.} \\
    UCR & 
      Uncertainty relation ☼*{UCR|see{uncertainty}} ☼*{uncertainty!relation} \\
    CQ & 
      Classical-quantum (e.g.\ CQ-states) ☼*{state!CQ} \\
    lhs.\ & 
      Left-hand side (of an equation) ☼*{lhs.}\\
    rhs.\ & 
      Right-hand side (of an equation) ☼*{rhs.}
  \end{tabular}
  \caption{List of Abbreviations.}
  \label{tb:abbr}
\end{table}

\begin{table}[!ht]
  \begin{tabular}{c|p{0.70\textwidth}}
    \textbf{Symbol} & \textbf{Description} \\
    \hline
    \vspace{-9pt} & \\ % increases spacing after the horizontal line
    $A, B', AD$ & 
      Typical physical systems and joint systems \\
    $X, \hat{Y}, Z $ & 
      Typical registers (random variables) \\
    $ℋ¬{AB}, ℋ¬X$ & 
      Hilbert spaces corresponding to a joint quantum and to a classical system \\
    $d¬A$ &
      Dimension of the system $A$, $d¬A = \dim{ℋ¬A}$ \\
    $\trace¬A \equiv \trace¬{ℋ¬A}$ & 
      Partial trace over subsystem $A$ \\
    $\onorm{ℋ}$ &
      Normalized quantum states on $ℋ$ \\
    $\osub{ℋ}$ &
       Sub-normalized states on $ℋ$ \\
    $ρ, \tauh, \omegaA^i, \sigmab¬{BC}$ & 
      Typical (possibly sub-normalized) quantum states \\
    $\ket{ψ}, φ¬{ABC}$ &
      Typical pure states, i.e.\ rank $1$ quantum states \\
    $π¬A, γ¬{AA'}$ &
      Completely mixed state on $A$ and maximally entangled state between $A$ and $A'$\\
    $D(ρ, τ)$ & Generalized trace distance between $ρ$ and $τ$ \\
    $F(ρ, τ)$ & Generalized fidelity between $ρ$ and $τ$ \\
    $P(ρ, τ)$ & Purified distance between $ρ$ and $τ$
  \end{tabular}
   ☼*{notation} 
  \caption{Notational Conventions for Quantum Mechanics.}
  \label{tb:notation}
\end{table}

\begin{table}[!ht]
  \begin{tabular}{c|p{0.70\textwidth}}
    \textbf{Symbol} & \textbf{Description} \\
    \hline
    \vspace{-9pt} & \\ % increases spacing after the horizontal line
    $\log \equiv \log_2$ & 
      Binary logarithm, i.e.\ logarithm to the basis $2$ \\
    $\ln$ & 
      Natural logarithm \\
    $ℝ, ℂ$ &
      Real and complex numbers \\
    $ℋ, ℋ'$ &
      Typical Hilbert spaces \\
    $\ket{φ}, \ket{ν}, \ket{θ_i}$ & 
      Typical elements of a Hilbert space, kets \\
    $\bra{ψ}, \bra{\phi}$ & 
      Typical functionals on a Hilbert space, bras \\
    $\trace, \trace¬{ℋ}$ & 
      Trace and partial trace over $ℋ$ \\
    $\olin{ℋ}, \olin{ℋ,ℋ'}$ & 
      Linear operators on $ℋ$ and from $ℋ$ to $ℋ'$\\
    $\oherm{ℋ}, \opos{ℋ}$ & 
      Hermitian and positive semi-definite operators on $ℋ$ \\
    $A ≥ B$ & 
      Equivalent to $A - B ∈ \opos{ℋ}$ \\
    $X†$ & 
      The adjoint operator of $X$ \\
    $Xᵀ$ & 
      The transpose operator of $X$, defined with regards to a 
      basis that needs to be specified \\
    $X\inv$ & 
      The generalized inverse operator of $X$ \\
    $\abs{X}$ & 
      Modulus, $\abs{X} = \sqrt{X†X}$ \\
    $\norm{X}[1]$ &
      Trace norm, $\norm{X}[1] = \trace\abs{X}$ \\
    $\norm{X}[∞]$ &
      Largest singular value of $X$, operator norm \\
    $ⅈ$ &
      Identity operator \\
    $E, X, L$&
      Typical linear operators on a Hilbert space \\
    $M, N$&
      Typical positive semi-definte operators \\
    $U, V, W$ & 
      Typical unitary operators or isometries \\
    $P, Π$ &
      Typical projectors \\
    $\cE, \cF, \cI$ &
      Typical trace-preserving completely positive maps \\
  \end{tabular}
   ☼*{notation} 
  \caption{Notational Conventions for Mathematical Expressions.}
  \label{tb:operators}
\end{table}

%............................ 

\chapter{Introduction}
\label{ch:intro}

This chapter starts with a rather philosophical introduction into quantum mechanics that does not 
assume any prior knowledge. Here, we attempt to explain and justify the information theoretic 
approach this thesis will take on the topic. This will lead into a short overview of some 
relevant aspects of information and quantum information theory. Then, we discuss the 
importance of non-asymptotic quantum information theory to characterize elementary information 
processing tasks. Finally, the introduction ends with a detailed outline of the thesis.

\section{Quantum Mechanics}
\label{se:intro/qm}

The laws of ☼[quantum mechanics]{quantum!mechanics} govern the behavior of
microscopic physical systems and are verified daily in experiments conducted in
physics laboratories worldwide. 
Here, we restrict our attention to non-relativistic 
quantum mechanics and take a static perspective on quantum theory, where the state
of a physical system is the central object of interest.

Non-relativistic quantum mechanics is the theoretical basis of 
today's semiconductor industry.
As we further miniaturize physical devices used for information processing,
the impact of quantum mechanics will become more and more relevant. Hence,
a thorough understanding of quantum physics will be pivotal for successfully 
engineering the next generation of information processing devices. 
In quantum cryptography~\cite{bb84,ekert91}, some of the more peculiar effects 
of quantum mechanics are already exploited today in order to ensure secrecy
of the communication between two distant parties.

While an understanding of the laws of quantum mechanics is thus necessary in order to fully 
comprehend the physical world surrounding us, these laws are nevertheless in stark contrast with
our intuition about the causal structure of the universe. In our everyday
experience, any observation about a physical system can be predicted perfectly
given a complete objective description of the state of said system. For example,
position, momentum and spin of a football are part of its objective state and
every ☼{observer} can verify their value independently, given appropriate 
measurement equipment. This perspective is deeply ingrained in our language: We
talk about the position, momentum and spin of a football, directly
associating our observations (the outcomes of position, momentum and spin
measurements) with the football and, thus, implying
that there exists an objective reality\,|\,the state of the football\,|\,beyond
our observations.

However, quantum mechanics does not allow for an objective description of
the state of a physical system that deterministically predicts all
observations about said system. For example, it is impossible to write
down a quantum mechanical state of an electron such that the
outcomes of both position and momentum measurements can be predicted with
arbitrary precision. This is known as the Heisenberg
☼[uncertainty principle]{uncertainty!principle}~\cite{heisenberg27}
and has mystified physicists since the early days of quantum mechanics. In what
sense can position and momentum then be considered real?
Is it even permissible to speak of the position and momentum of a quantum
mechanical object like an electron?

Furthermore\,|\,as if the loss of determinism was not enough to confuse our
human minds\,|\,quantum mechanics generally does not even allow for an
objective description of the state of a physical system that provides the
☼[probabilities]{probability} with which different observations about it are made. These
probabilities, even if the description of the state is complete within quantum
mechanics, are in general subjective to the observer. It is crucial to
note here that classical theory allows for describing subjective knowledge about
the state of a system, for example through conditional probability distributions. 
(We will encounter an example model of subjective classical information in the next section.) 
However, such a description cannot be considered
complete within classical theory, since, in principle, every observer may hold a copy of 
the full objective state of the system, e.g.\ position, momentum and spin 
of the football.

To convince yourself that quantum mechanics is incompatible with this notion of objectivity, consider two particles that are in an ☼[entangled]{entanglement} state. More specifically, this could be two electrons, $A$ and $B$, that are in a spin
singlet state. In this setting, quantum mechanics postulates that an observer
controlling electron $B$ can perfectly predict the outcomes of all
possible spin measurements on electron $A$.\footnote{The observer simply measures the spin of $B$ in the same direction as the spin of $A$ is measured to get a fully anti-correlated result.}
We call this observer an ☼[omniscient]{observer!omniscient} observer of the electron $A$.
At the same time, all other observers (i.e.\ observers who do not have access to $B$) are ☼[ignorant]{observer!ignorant} about $A$, which means
that they will see all possible measurement outcomes on $A$ with equal probability.
Quantum mechanics, in contrast to classical theory, does not allow to copy 
the quantum information encoded in the state of electron $B$ and share it with other observers.
On an intuitive level, this is often explained using the concepts of 
☼{no-cloning}~\cite{wootters82} or the ☼[monogamy of entanglement]{entanglement}.
These statements
can be formulated quantitatively and we will see in Chapter~\ref{ch:entropies}
that the fact that one observer is close to omniscient about a quantum system
implies that all other observers are necessarily close to ignorant about it.

This apparent deficiency of quantum mechanics led Einstein, Podolsky and 
Rosen to ask ``Can
quantum-mechanical description of physical reality be considered
complete?''~\cite{epr35}. The search for a theory, consistent with
quantum mechanics, that would assign objective descriptions\footnote{Such
objective descriptions of a system are usually called ☼[(local) hidden
variables]{hidden variable} in this context, since they are not accessible from
within quantum mechanics.} to physical systems that probabilistically predict
observations about said systems, was abandoned when Bell~\cite{bell64} as well as 
Kochen and Specker~\cite{kochen68} proved their pivotal
theorems. They show that such objective descriptions
cannot fully predict the probabilities of certain observations that
are consistent with quantum mechanics and can be verified
experimentally.\footnote{In fact, it was recently shown that if such objective
descriptions have any predictive power exceeding quantum
mechanics, they can be falsified experimentally under minimal assumptions.~\cite{colbeck08,colbeck10}}

Furthermore, the existence of preferred observers renders ☼[quantum cryptography]{quantum!cryptography} possible. In quantum key distribution, we consider two players, traditionally called
Alice and Bob, who want to share a secret over a public channel.
Once the two players can establish that Bob is almost ☼[omniscient]{observer!omniscient} 
about a quantum system Alice controls, a shared secret can be produced by an arbitrary
measurement of Alice's system: while Bob can predict Alice's observed
measurement result, by the laws of quantum mechanics, any eavesdropper is
guaranteed to have almost no knowledge about it.

%So, how can we successfully investigate the quantum mechanical world if our
%intuition cannot be trusted and our language is evidently not suitable for the
%task? In theoretical physics, a promising remedy is to apply the formal 
%language of mathematics to the problem.
%The fundamental laws of quantum mechanics can be phrased as
%postulates in such a language, and the logical structure of the language
%then allows us to infer non-trivial implications from these postulates. These
%implications, hopefully, will in turn enhance our understanding of the
%physical world described by quantum mechanics.
%In this thesis, we will mostly use 
%concepts from linear algebra and information theory.

\section{Information Theory}
\label{se:intro/it}

Luckily, some of the counter-intuitive effects of quantum
mechanics\,|\,in particular, the appearance of non-determinism and subjective
knowledge as explained above\,|\,have also been considered in a completely
different context, in probability and ☼{information theory}.
The latter was founded by Nyquist~\cite{nyquist28} and
Shannon~\cite{shannon48} in the early days of telegraph communication in
order to investigate information sources and the capacities of channels to
transmit digital information. 

Following Shannon, 
consider a source that outputs English text. Clearly, the
different letters in the Latin alphabet will not be produced with equal
probability by such a source. This is used, for example, in the Morse alphabet,
where the most frequent letters are given the simplest codes.\footnote{The most
frequent letters `e' and `t' are given the codes `$\,\cdot\,$' and `$-$',
respectively, whereas all other letters have codes with at least two
symbols.} Going one step further, note that certain combinations of letters are
more likely than others. For example, the letter `q' is almost always followed
by a `u'. Hence, given that the last letter this source produced was a `q', the
information content, or ☼{uncertainty}, of the next letter is very low. 

This implies that the probability with which letters occur is
relevant in order to characterize the information content of the source\,|\,a
source that produces every letter with equal probability produces more
information than a source that is biased towards certain letters.
To quantify these considerations, Shannon used ☼[entropies]{entropy}.
Very generally, he considered ☼[events]{event} that are known to
occur with a certain probability $p$ and assigned to 
them the value $-\log p$, called
☼{surprisal}, which measures how surprising the occurrence of the event is. 
(The logarithm is taken to the binary basis throughout this thesis.) If an
event occurs with certainty its surprisal is zero and as the event gets less
likely its surprisal can grow arbitrarily large. In the above example, the
surprisal of the letters `e' and `t' is smaller than the surprisal of less common
letters. Moreover, the surprisal of the letter `u' given that the last
letter was `q' is close to zero. 

The ☼[Shannon entropy]{entropy!Shannon} of a
source is the average surprisal of the events it produces and is introduced 
more formally in the following.
For the purpose of this thesis, a ☼[probability distribution]{probability} on a ☼{random variable} $X$ that takes values from a discrete 
set $\cX$ is a function $P¬X$ from $\cX$ to non-negative real numbers with the property 
that $∑_{x ∈ \cX} P¬X(x) = 1$.
The Shannon entropy of a random variable $X$ with distribution $P¬X$ is then defined as
\begin{align}
  \hvn{X}[\tinyP] := - ∑_{x ∈ \cX} P¬X(x) \log P¬X(x) ¶[shannon].
\end{align}
Here, $-\log P¬X(x)$ is the surprisal of the event ``$X = x$'' and has the properties
described above. On one hand, Shannon~\cite{shannon48} and later Rényi~\cite{renyi61} derived the mathematical 
form of the entropy, Eq.~§[shannon], from intuitive axioms that measures of the average 
surprisal should satisfy. On the other hand, 
the Shannon entropy is intimately related to physics and was inspired directly by
the Gibbs entropy of thermodynamics~\cite{gibbs1876, boltzmann1872}.

Next, we consider the case where we have ☼{side information} about $X$, modeled as another 
random variable $Y$.
The two random variables, $X$ and $Y$, have a joint probability distribution, $P¬{XY}(x, y)$. 
The ☼{marginal} probability distributions of the individual random variables 
$X$ and $Y$ are thus given 
as $P¬X(x) = ∑_y P¬{XY}(x,y)$ and $P¬Y(y) = ∑_x P¬{XY}(x,y)$, respectively. 
We employ the ☼[conditional Shannon entropy]{entropy!Shannon} of $X$
conditioned on $Y$, which is defined as $\hvn{X|Y}[\tinyP] := \hvn{XY}[\tinyP] - \hvn{Y}[\tinyP]$.
This definition, aside from its operational interpretation that we will discuss
below, has very natural properties that we would expect from a conditional entropy, i.e.\
a measure about the average ☼{surprisal} or uncertainty of $X$ given 
☼{side information} $Y$. For example,
we have $\hvn{X|Y}[\tinyP] ≤ \hvn{X}[\tinyP]$, namely, the uncertainty about $X$ increases if one ignores the side information. Equality holds if the side information is independent of $X$, i.e.\ if $P¬{XY}(x,y) = P¬X(x) P_Y(y)$.\footnote{More generally, and in accordance with our intuition, any function applied to the random variable 
$Y$ can at most increase the uncertainty about $X$, i.e.\ $\hvn{X|Y}[\tinyP] ≤ 
\hvn{X|Z}[\tinyP]$ if $Z =  f(Y)$. The entropy $\hvn{X|Z}[\tinyP]$ is evaluated for 
the probability distribution induced 
by $f$, i.e.\ the distribution 
$P¬{XZ}(x, z) = ∑_{y: f(y) = z} P¬{XY}(x, y)$.}
Relations of this type are called
☼{data processing} inequalities and will be discussed in 
Chapter~\ref{ch:smooth}.

Conditional entropies can be employed to model subjective classical information.
To see this, we return to Shannon's text sources and treat them quantitatively. 
Let us first consider a source that outputs
all $26$ letters of the English alphabet with equal probability. Clearly, the entropy 
of the output $U$ of this source is $\hvn{U} = \log 26 \approx 4.7$.
A source that outputs letters of English text, $X$, has lower entropy\,|\,$\hvn{X} 
\approx 4.14$~\cite{shannon51}\,|\,due to the non-uniform distribution of the different 
letters. Moreover, if we consider two consecutive letters of
English text, denoted $X_1$ and $X_2$, we find that the conditional 
entropy\,|\,$\hvn{X_2|X_1} \approx 3.56$~\cite{shannon51}\,|\,is even smaller 
due to correlations between the
probabilities of adjacent letters.

For another example, consider an unbiased die that is thrown secretly in a cup. The state of the die before the result is revealed is modeled as a random variable, $X$, on the set $\{ \epsdice{1}, 
\epsdice{2}, \epsdice{3}, \epsdice{4}, \epsdice{5}, \epsdice{6} \}$ of faces with uniform probability distribution $P¬X(\cdot) = \frac{1}{6}$.
The Shannon entropy or uncertainty about the outcome is simply $\hvn{X}[\tinyP] = \log 6 \approx 2.58$.
However, a cheating ☼{observer} may peek under the cup and take note of a face on the side of the die. This face, let us denote it by $Y$, is a random variable itself and we can describe the correlations between the result of the throw and $Y$ using the joint probability distribution $P¬{XY}$. For example, if the cheater sees a `$\epsdice{1}$', he can rule out the throws `$\epsdice{1}$' and `$\epsdice{6}$' whereas all other throws are equally likely. This results in the probabilities $P¬{XY}(\epsdice{1}, \epsdice{1}) = P¬{XY}(\epsdice{6}, \epsdice{1}) = 0$ and $P¬{XY}(\epsdice{2}, \epsdice{1}) = P¬{XY}(\epsdice{3}, \epsdice{1}) = P¬{XY}(\epsdice{4}, \epsdice{1}) = P¬{XY}(\epsdice{5}, \epsdice{1}) = \frac{1}{6} \cdot \frac{1}{4} = \frac{1}{24}$.  The uncertainty the cheating observer has about the outcome is thus reduced
to $\hvn{X|Y}[\tinyP] = 2$. Another observer may have seen the result already and stored it in a random variable $Z$; hence, $\hvn{X|Z}[\tinyP] = 0$.
The subjective uncertainty an observer has about the outcome of the throw thus critically depends on the available ☼{side information}, $Y$ or $Z$.%
%Note that this situation resembles the quantum spin measurement we discussed in the previous 
%section. However, while multiple observers may own a device as powerful as $Z$, the laws of 
%quantum mechanics preclude the existence of more than one perfect observer about a quantum 
%system.

The Shannon entropy has found a vast number of applications in information theory. Shannon's source coding theorem~\cite{shannon48} concerns itself with the question of how much we can compress the output of a source. If a source produces a long stream of ☼[independent and identically-distributed]{i.i.d.} (i.i.d.) symbols $X$, it states that any attempt to compress the output to less than $\hvn{X}[\tinyP]$ logical bits per symbol will almost certainly lead to information loss. Here, $P¬X$ is the probability distribution of the source symbols.
On the other hand, it is possible to compress
the stream to arbitrarily close to $\hvn{X}[\tinyP]$ bits per symbol with insignificant probability of information loss. 
More generally, if some additional side information $Y$ about each $X$ is available, the quantity $\hvn{X}[\tinyP]$ can be replaced by $\hvn{X|Y}[\tinyP]$ in the above statements~\cite{slepian73}, where $P¬{XY}$ is the joint probability distribution of source and side information. The Shannon entropy thus quantifies the amount of memory (in bits) needed 
to store the output of a source reliably. 
Memory can be considered a ☼{resource} in source coding and, more generally, the Shannon entropy
is often used to analyze the resource usage of a task in the limit of many independent and identical repetitions. In the following, we call this the 
☼{i.i.d.\ limit}

As a further important example, the capacity of a channel to transmit information (in the i.i.d.\ limit of many independent uses of a memoryless channel) can be expressed in terms of Shannon 
entropies~\cite{shannon48}.

\section{Quantum Information Theory}
\label{se:intro/qit}

An observation about a physical system, in the sense described in Section~\ref{se:intro/qm}, constitutes an ☼{event}. 
Moreover, a ☼{measurement}, i.e.\ a complete set of
mutually exclusive observations about a physical system, can be seen as an
information source. It is thus very natural to try to apply the entropy
formalism of the previous section to quantum measurements. In addition to that, many information theoretic tasks 
(we will discuss the example of source compression in more detail) can be generalized to
the quantum setting and thus the question arises whether the resource usage of these
tasks can be characterized using an analogue of the Shannon entropy as well.
We will see that this is possible indeed.

In the quantum formalism, to be comprehensively introduced in Chapter~\ref{ch:pre}, a random variable is modeled as a (classical) ☼{register}, $X$ and the probability distribution, $P¬X$, is represented as a diagonal matrix with the probabilities $P¬X(x)$ as eigenvalues. More generally, the state of a quantum system, $A$, is modeled as a positive semi-definite operator with unit trace, $ρ_A$, called the density operator or ☼{quantum state}. Thus, a random variable or register is a special case of a ☼[quantum system]{system} and the quantum formalism encompasses classical probability theory. Furthermore, a measurement is simply a map from a quantum system to a register.

Schumacher and Wootters~\cite{schumacher95} introduced the term ☼{qubit} to denote the smallest unit of quantum information, which can be represented as a two-by-two matrix. The spin degree of freedom of an electron, as we have seen before, constitutes a physical example of such a qubit system.
Holevo's bound~\cite{holevo73} implies that a qubit can only store one bit of 
classical information. This indicates that the power of quantum information processing over classical information processing lies in the possible correlations between multiple quantum systems and not necessarily in individual quantum systems themselves.

To see a first application of entropies in quantum information theory, we consider a simple ☼{uncertainty relation}. Let $A$ be qubit, e.g.\ an electron spin in a state $ρ¬A$, and let $X$ and $Y$ be registers containing the outcomes of two ☼[complementary measurements]{measurement!complementary} on $A$, for example spin measurements in the $x$ or $y$ direction. Then, an entropic formulation of Heisenberg's uncertainty principle tells us that $\hvn{X}[ρ] + \hvn{Y}[ρ] ≥ 1$.~\cite{maassen88}.
In other words, independently of how the initial state $ρ¬A$ is prepared (as long as there is no quantum side information present), there will be uncertainty about at least one of the outcomes.
Such uncertainty relations, expressed in terms of entropies, are the topic of 
Chapter~\ref{ch:uc}.

The ☼[von Neumann entropy]{entropy!von Neumann}~\cite{vonneumann32} of a quantum system
is given as
\begin{align}
  \hvn{A}[ρ] := - \log \tr{ρ¬A \log ρ¬A} ¶,
\end{align}
where $\trace$ denotes the trace. For a register this expression reduces to the
☼[Shannon entropy]{entropy!Shannon}.
In order to investigate the effect of quantum correlations, we introduce the conditional von Neumann entropy. For a bipartite quantum state $ρ¬{AB}$, this is given as 
$\hvn{A|B}[ρ] := \hvn{AB}[ρ] - \hvn{B}[ρ]$. If the system $A$ is a classical register, this entropy is nonnegative and can safely be interpreted as a measure of uncertainty given quantum side information.

To see this, we reconsider the example of Section~\ref{se:intro/qm}, where two electrons, $A$ and $B$, are in a spin singlet state.
As we have noted before, quantum mechanics predicts that a preferred observer controlling the electron $B$ can predict any outcome of a spin measurement on
the $A$ electron with certainty. Thus, the ☼{surprisal} of any observation about electron $A$
is zero for this observer. More formally, we denote by $X$ the random variable that
stores the outcome of an arbitrary spin measurement on the electron $A$.
Then, the conditional ☼[von Neumann entropy]{entropy!von Neumann} evaluates to 
$\hvn{X|B}[ρ] = 0$. Here, the entropy is evaluated for the post-measurement state $ρ¬{XB}$ that results from measuring $X$ on the joint ☼{quantum state} $ρ¬{AB}$.
On the other hand, all other observers will see a uniformly random
measurement outcome and thus have maximum surprisal. We denote by $C$ such an 
☼[ignorant]{observer!ignorant} observer. And indeed, for any tripartite quantum state 
$ρ¬{ABC}$ that is compatible with the marginal state $ρ¬{AB}$, we find $\hvn{X|C}[ρ] = 1$.
Again, the entropy is evaluated for the post-measurement state $ρ¬{XC}$ that results from measuring $X$ on the joint state $ρ¬{AC}$.

However, the conditional von Neumann entropy can also be evaluated for
the quantum state $ρ¬{AB}$ of the electron pair before measurement. This entropy evaluates to
$\hvn{A|B}[ρ] = -1$ in the example above. Generally, the conditional entropy can be negative 
in the presence of ☼[entanglement]{entanglement} and its interpretation as a measure of uncertainty has been controversial. (See, e.g.,~\cite{horodecki05,berta10,delrio11} for recent
work elucidating the issue.)

Generalizing the source coding theorem with side information of the previous section, we ask how much we can compress a long stream of ☼{i.i.d.} random variables $X$ such that an observer with
quantum side information $B$ about each $X$ can reconstruct the original 
from the compressed stream. Such side information can be modeled in a joint quantum 
state $ρ¬{XB}$. The answer was given by Devetak and Winter~\cite{devetak03}, 
who propose a quantum generalization of the Slepian-Wolf theorem. 
They show that, consistent with the classical result, the stream can be compressed to 
$\hvn{X|B}[ρ]$ bits per symbol in the ☼{i.i.d.\ limit} of long streams.

In another generalization of source coding, we ask how much we can
compress a long stream of independent and identical quantum systems $A$ that are in the state 
$ρ¬A$. The answer, that such a string can be compressed to $\hvn{A}[ρ]$ qubits, was given by 
Schumacher~\cite{schumacher95} in his pioneering work on quantum information theory.

We can go one step further into the quantum world and consider 
☼{state merging}~\cite{horodecki06}. Here, we start with a
joint quantum state $ρ¬{AB}$ shared between two parties, $A$ and $B$.
The task is to recreate the state (including its correlations with
the environment) at $B$ using free classical communication between $A$ and $B$ and 
☼{entanglement} between $A$ and $B$ as a ☼{resource}. The amount of this resource that 
needs to be utilized (in the ☼{i.i.d.\ limit} of many copies of the state) is then quantified 
by the conditional von Neumann entropy, $\hvn{A|B}[ρ]$. The conditional entropy can be consistently interpreted as the amount of entanglement needed to complete the task. In particular, if $\hvn{A|B}[ρ]$ is negative (which may happen only in the presence of entanglement in $ρ¬{AB}$), it is possible to
extract $-\hvn{A|B}[ρ]$ units of entanglement from the state while recreating it at $B$.

%.........................

\section{Non-Asymptotic Information Theory}

So far, we have considered tasks in the ☼{i.i.d.\ limit} of many repetitions and found that
the resource usage is characterized by expressions involving the Shannon and von 
Neumann entropies. This leads to the question of what happens to the resource usage
when we consider a finite number of trials and allow a small probability of failure.
In many contexts, in particular in cryptography, we are even interested in the 
amount of resource needed to perform a task just once, with high probability of success. 
We call this the ☼{one-shot} setting in the following.
These questions, which we are going to tackle in the quantum setting, have also been the 
topic of recent research in classical information theory (see, e.g.,~\cite{polyanskiy10,polyanskiythesis10,wang09}).

It turns out that the Shannon entropy and von Neumann entropies, which have been invaluable tools in the previous sections, are insufficient to characterize the 
required ☼[resources]{resource} in the one-shot setting. 
Nevertheless, there are other entropies that can be used instead to quantify uncertainty, 
as we will see in the following.

We will now introduce some aspects of non-asymptotic information theory on the example of ☼{source compression}. For this purpose, let us consider a source $X$ with probability distribution 
$P¬X$. The number of bits of memory needed to store the output of this source so that it can be recovered with certainty is given 
by $\lceil \hh{0}{X}[\tinyP] \rceil$, where $\hh{0}{X}[\tinyP]$ denotes the ☼[Hartley entropy]{entropy!Hartley}~\cite{hartley28} of the distribution $P¬X$, defined as
\begin{align}
  \hh{0}{X}[\tinyP] := \log \abs<b>{ \{ x : P¬X(x) > 0 \} } ¶,
\end{align}
The Hartely entropy corresponds to the ☼[Rényi entropy]{entropy!Rényi} of order 
$0$~\cite{renyi61} and simply measures the size of the support of $X$. As an example, we consider again a source that outputs characters of the English alphabet. If we want to
store a single character produced by this source such that it can be recovered with 
certainty, we clearly need $\lceil \log 26 \rceil = 5$ bits of memory as a ☼{resource}. This first result is rather unsatisfactory since the resource usage does not depend on the actual probability distribution of the letters but only on the size of the alphabet.

More interestingly, we may ask how much memory we need to store the output 
if we allow a small probability of failure, $ε$.
One way to tackle such problems is by investigating encoders that assign code words (i.e.\ binary strings) of a fixed length $m$ (in bits) to the events a source produces. These code words are then stored in $m$ bits of memory and a decoder is later used to compute an estimate 
of $X$ from that memory.
For a source $X$ with probability distribution $P¬X$, we are thus interested in the minimum code length, $m^{ε}(X)¬P$, for which there exists an encoder and a decoder that achieve a probability of failure not exceeding $ε$.
Gallager~\cite{gallager79} showed that a random assignment of source events to code words on average leads to a failure probability of at most $ε$ if the code length is 
sufficiently long. 
His results imply that the minimal code length satisfies
\begin{align}
  m^{ε}(X)¬P ≤ \hmax{X}[\tinyP] + \log \frac{1}{ε} + 1 ¶[Gallager].
\end{align}
Here, $\hmax{X}[\tinyP]$ denotes the ☼[max-entropy]{entropy!max-entropy}, which corresponds to 
the ☼[Rényi entropy]{entropy!Rényi} of order $\frac{1}{2}$ and is defined as
\begin{align}
  \hmax{X}[\tinyP] := \log \Big( \sum_{x ∈ \cX} √{P¬X(x)} \Big)^2 ¶.
\end{align}
Upper bounds of the type~§[Gallager] are called ☼[direct bounds]{direct bound} and show that there exist protocols using a certain amount of resource that do not exceed a fixed probability of failure. 

The above analysis can be further refined by ☼{smoothing} the max-entropy. For this purpose, let us consider probability distributions $Q¬X$ that are close to $P¬X$ and have max-entropy $\hmax{X}[\tinyQ]$ smaller than $\hmax{X}[\tinyP]$. In principle, we could design the encoding and decoding scheme for a source with distribution $Q¬X$ instead of $P¬X$. Clearly, this reduces the upper bound in~§[Gallager].
In fact, the upper bound can be expressed in terms of a ☼[smooth max-entropy]{entropy!smooth},
$\hmax[ε_1]{X}[\tinyP] := \inf_{Q \approx P} \hmax{X}[\tinyQ]$, which minimizes the max-entropy over probability distributions $Q¬X$ that are within ☼{statistical distance} $ε_1$ 
of $P¬X$.%
\footnote{Note that we use the statistical distance as a metric here for convenience of exposition. In Chapter~\ref{ch:pd}, we will argue that the ☼{purified distance} should be used instead to define the smooth entropies in the quantum setting. All smooth entropies used in this thesis, except in this introductory exposition, are thus based on the purified distance as a metric.}
More precisely, the statistical distance is defined as $D(P¬X, Q_X) := \frac{1}{2} ∑_x \abs<b>{P¬X(x) 
- Q¬X(x)}$ and we require that $D(P¬X, Q_X)≤ ε_1$.

However, if this scheme is applied to the original source, we incur 
an additional error that grows with the statistical distance between $P¬X$ and $Q¬X$.
The total probability of failure, $ε ≤ ε_1 + ε_2$, is thus split into two contributions: the statistical distance of the distributions, $ε_1$, and the contribution from the 
Gallager bound, which we denote $ε_2$. Hence, we get a family of direct bounds on the minimal code length:
\begin{align}
  m^{ε}(X)¬{P} ≤ \hmax[ε_1]{X}[\tinyP] + \log \frac{1}{ε_2} + 1, \quad \forall\, ε_1, ε_2 \tn{with} ε_1 + ε_2 = ε  ¶[renesdirect].
\end{align}

We also consider ☼[converse bounds]{converse bound} that give a lower bound on the resources
required to achieve a certain probability of success. In fact, it can be shown that
\begin{align}
  m^{ε}(X)¬{P} ≥ \hmax*{X}[\tinyP] ¶[renesconverse].
\end{align}
Hence, both the lower and upper bound on the quantity $m^{ε}(X)¬P$ can be expressed in terms of a
smooth max-entropy. We thus say that the required memory for ☼{one-shot} source compression is
☼[characterized]{characterize} by the smooth max-entropy.
(Note that a more detailed analysis of one-shot source compression with quantum side information can be found in Chapter~\ref{ch:app} and~\cite{renesrenner10}.)

To see why the Shannon entropy does not suffice to characterize the ☼{one-shot} version of 
source compression, consider a source that produces the symbol `a' with probability
$\frac{1}{2}$ and $k$ other symbols each with probability $\frac{1}{2 k}$. On the
one hand, for any fixed failure probability $\eps \ll 1$, the converse bound 
in~§[renesconverse] evaluates to approximately $\log k$ for large enough $k$. 
This implies that we cannot compress this source much beyond the ☼[Hartley entropy]{entropy!Hartley}. On the other hand, the Shannon entropy of this distribution is $\frac{1}{2}(\log k + 2)$ and underestimates the required resources by a factor of two.
 
In this thesis, we will mainly encounter two entropic quantities that are defined for quantum states in Chapters~\ref{ch:entropies} and~\ref{ch:smooth}. 
Surprisingly, it turns out that these two entropies suffice to characterize the resource usage of many information theoretic tasks in the ☼{one-shot} setting in a manner similar to source compression. The first quantity, the ☼[smooth min-entropy]{entropy!smooth}, $H_{\min}^{ε}$, is a generalization of the min-entropy or Rényi entropy of order $∞$, which evaluates the 
minimum ☼{surprisal} of a random variable $X$. Namely,
\begin{align}
  \hmin{X}[\tinyP] := \min_{x ∈ \cX} - \log P¬X(x) ¶,
\end{align}
One of the major applications of the smooth min-entropy is in ☼{randomness extraction}, where it characterizes the amount of uniform randomness that can be extracted from a biased source~\cite{impagliazzo89}.
The second quantity, the ☼[smooth max-entropy]{entropy!smooth}, $H_{\max}^{ε}$, is used to characterize source compression, as we have seen above.

%More specifically, when we say that a quantity $Q^{ε}(X)$ is characterized by a smooth entropy, 
%e.g.\ $H_{\min}^{ε}(X)$, we mean that there exists upper and lower bounds on $Q^{ε}$ of the type 
%of Eq.~§[renesbound], i.e.\
%$H_{\min}^{f(ε)}(X) \pm \log g(ε)$, where $g(ε)$ is an arbitrary polynomial in $ε$, and $f(ε) → 0$ 
%when $ε → 0$.

The smooth min- and max-entropies were first generalized to the 
quantum settings by Renner and König~\cite{rennerkoenig05,renner05}. The smooth min-entropy we use in this thesis is a refined version of the smooth min-entropy proposed in~\cite{renner05}, whereas the smooth max-entropy is based on later work by Renner, König and Schaffner~\cite{koenig08}.
The significance of these extensions to the quantum setting stems from their operational meaning. 
For example, the quantum
generalization of the min-entropy, $\hmin*{X|B}[ρ]$, characterizes randomness extraction against
quantum side information $B$, i.e.\ it characterizes the amount of uniform randomness,
independent of the side information $B$, that can be extracted from $X$~\cite{renner05,tomamichel10}. The quantum generalization of the max-entropy, 
$\hmax*{X|B}[ρ]$, characterizes source compression with quantum side 
information $B$ in the ☼{one-shot} setting~\cite{renesrenner10}. Moreover, the amount of entanglement needed in one-shot ☼{state merging} is characterized by $\hmax*{A|B}[ρ]$~\cite{berta08}. We will discuss source compression with quantum side information as well as randomness extraction against quantum side information in Chapter~\ref{ch:app}.

In addition to their operational meaning, the smooth entropies exhibit many useful properties, including ☼{data processing} inequalities and ☼[chain rules]{chain rule}.
Furthermore, the smooth min- and max-entropies converge to the von Neumann entropy in the ☼{i.i.d.\ limit}. For any $0 < ε < 1$,
\begin{align}
  \lim_{n → ∞} \frac{1}{n} \hmin*{X^n|B^n}[ρ] = \lim_{n → ∞} \frac{1}{n} \hmax*{X^n|B^n}[ρ] = \hvn{X|B} ¶.
\end{align}
We call this the entropic form of the ☼[asymptotic equipartition property]{AEP} and 
it is the topic of Chapter~\ref{ch:aep}.

This means that if the resource usage is characterized by a smooth entropy in the ☼{one-shot} setting, the resource usage in the ☼{i.i.d.\ limit} is given by the von Neumann entropy.
In fact, a simple analysis shows the upper and lower bounds 
in~§[renesdirect] and~§[renesconverse] on the code rate, $r = \frac{1}{n} m^{ε}(X^n)$,
converge to the von Neumann limit when we let $n$ go to infinity for any fixed $0 < ε < 1$. This
also shows the advantage of the smoothed direct bound~§[renesdirect] over the Gallager bound~§[Gallager], for
which such an asymptotic convergence can only be shown using additional techniques.

Moreover, the smooth entropies satisfy a ☼{duality} relation. For any tripartite quantum state $ρ¬{ABC}$ and any $0 ≤ ε < 1$, we find~\cite{koenig08,tomamichel09} 
\begin{align}
  \hmin*{A|B}[ρ] ≥ -\hmax*{A|C}[ρ] ¶[introdual]
\end{align}
and equality holds if the joint state $ρ¬{ABC}$ is pure.\footnote{Pure states offer the most complete description of a joint quantum system.}
This relation provides a connection between the min- and max-entropy, and thus the tasks characterized by them, that does not exist in classical information theory. 

It also allows us to close the circle to the discussion of quantum mechanics at the beginning of this chapter. For this purpose, let $B$ and $C$ be two observers of a quantum system $A$. Then, the min-entropy, $\hmin{A|B}$, can be viewed (cf.~Chapter~\ref{ch:entropies}) as the distance of $B$ to an ☼[omniscient observer]{observer!omniscient} of the quantum system $A$. Furthermore, the max-entropy, $-\hmax{A|C}$, can be viewed as the distance of $C$ to an 
☼[ignorant observer]{observer!ignorant} of $A$. The duality relation, Eq.~§[introdual], thus states that if $B$ is close to an omniscient observer of the quantum system $A$, then $C$ is at least as close to an ignorant observer of $A$. This can be seen as
a manifestation of the subjective knowledge of observers quantum mechanics imposes on the world.

%.........................

\section{Goal and Outline}
\label{se:intro/outline}

The goal of this thesis is to consolidate the smooth entropy framework for non-asymptotic information theory and to introduce important additions to the framework, including the entropic
asymptotic equipartition property and various uncertainty relations.
This work should be a useful reference for researchers interested in the smooth entropy framework for non-asymptotic quantum information theory. The focus of this work is thus mainly on the properties of the smooth entropies and not on their applications. 

The remainder of this thesis is organized as follows. 

In Chapter~\ref{ch:pre}, the notation and mathematical foundations of Hilbert space quantum 
mechanics are introduced. Relevant results of linear algebra are summarized in Section~\ref
{se:pre/lin} and the axioms of quantum mechanics are introduced in~\ref{se:pre/qm}. Moreover, 
Section~\ref{se:pre/math} covers operator convex functions and semi-definite programming, 
completing a mathematical toolkit that will be used extensively in this thesis.

In Chapter~\ref{ch:pd}, we introduce a novel measure of distance between (potentially incomplete) 
quantum states, the purified distance. In particular, we explore its properties and argue why they 
are relevant for the definition of the smooth quantum entropies.

In Chapter~\ref{ch:entropies}, we formally define the min- and max-entropies for quantum states, give
a collection of different expressions for the entropies and explore 
some of their properties. 
In particular, we explore the relation between the min- and the max-entropy and 
their relation to classical Rényi entropies. We also show
that both entropies are continuous functions of the quantum state.

In Chapter~\ref{ch:smooth}, we introduce smoothing and define the smooth min- and max-entropy. Various properties of the smoothing operation are discussed.
We establish relations between the smooth min- and max-entropy and 
investigate data processing inequalities. The special case where one or more systems
are classical is considered in detail. We also give a list of chain rules that have recently
been shown for the smooth min- and max-entropies.

In Chapter~\ref{ch:aep}, we show that the smooth entropies converge to the von Neumann entropy
when we consider a sequence of independent and identically distributed quantum systems. This result is an entropic version of the asymptotic equipartition property and confirms the
fundamental role of the von Neumann entropy in quantum information theory. In addition to the asymptotic result, the chapter also provides bounds for finite block lengths.

In Chapter~\ref{ch:uc}, we introduce a variety of entropic uncertainty relations that give bounds, analogous to Heisenberg's uncertainty principle, on the uncertainty of the outcomes
of two incompatible measurements on a quantum system. These uncertainties are expressed in
terms of smooth entropies as well as von Neumann entropies, which makes them directly 
applicable to problems in quantum cryptography.

In Chapter~\ref{ch:app}, we combine results from Chapters~\ref{ch:pd}-\ref{ch:aep} to investigate two information theoretic tasks, source compression with quantum side information and randomness extraction against quantum side information. In particular, we show a strong converse statement for source compression which implies that any attempt to compress to less than the Shannon limit will fail with high probability.
These results can then be used in another example application, quantum key distribution. There, 
we show how the entropic uncertainty relations of Chapter~\ref{ch:uc} can be employed in order to prove security of the original {BB84} quantum key distribution protocol in a concise and 
intuitive way.

The thesis ends with Chapter~\ref{ch:outlook} in a short conclusion and outlook. 
Some open problems are discussed.

\chapter{Preliminaries}
\label{ch:pre}

The preliminaries consist of three sections. The first two sections, which
cover the mathematical foundations of Hilbert space quantum mechanics in finite
dimensions, can be skipped entirely by readers already familiar with the concepts of 
linear algebra used in quantum information theory. The notation used throughout this 
thesis is introduced in these sections and is summarized in
Tables~\ref{tb:notation} and~\ref{tb:operators}.
The last section then introduces some mathematical tools that
are needed to derive the results of the following chapters, including
semi-definite programming and operator monotone functions.

\section{Linear Algebra on Hilbert Spaces}
\label{se:pre/lin}

☼*{linear algebra} ☼*{algebra|see{linear algebra}}

This section is based on many introductory text books, mostly on the two books
\emph{Matrix Analysis} and \emph{Positive Definite Matrices} by
Bhatia~\cite{bhatia97, bhatia07}. Moreover, John Watrous's lecture
notes~\cite{watrous-ln08} as well as Nielsen and Chuang's \emph{Quantum Computation
and Quantum Information}~\cite{nielsen00} were an invaluable resource.

\subsection{Hilbert Spaces}

\subsubsection{Bras and Kets}

Let $ℋ$ be a finite-dimensional ☼{vector space} over the complex numbers
equipped with an ☼{inner product} $\ip{·,·} : ℋ \times ℋ → ℂ$.
In the following, we will call $ℋ$ a ☼{Hilbert space}. The dual space of
$ℋ$ is the Hilbert space of (linear) functionals from $ℋ$ to $ℂ$. We use Dirac's
bra-ket notation to denote elements of $ℋ$ and its dual space, $ℋ^*$. Every
ket, $\ket{ψ} ∈ ℋ$, is in one-to-one correspondence with its dual
bra, $\bra{ψ} ∈ ℋ^*$. The bra is defined in terms of the ket via the Hilbert 
space's inner product as
\begin{align}
  \bra{ψ} : \ket{φ} ↦ \braket{ψ|φ} := \ip<b>{\ket{ψ}, \ket{φ}} ¶[pre/lin/bra]\,.
\end{align}

The right-hand side of the above equation gives a natural expression for the
inner product in terms of a bra-ket product, and we will use this notation
frequently. The bra-ket product $\braket{·|·}$ in~§ has the following
properties, which follow
directly from properties of the underlying inner product.
\begin{itemize}
  ☛ Conjugate symmetry: $\braket{ψ|φ} = \overline{\braket{φ|ψ}}$.
  ☛ Sesquilinearity: Let $\ket{θ_1}, \ket{θ_2} ∈ ℋ$ and let $α_1, α_2 ∈ ℂ$ be such
    that $\ket{φ} = α_1 \ket{θ_1} + α_2 \ket{θ_2}$. Then, $\braket{ψ|φ} = α_1
    \braket{ψ|θ_1} + α_2 \braket{ψ|θ_2}$ and, due to conjugate symmetry,
    $\braket{φ|ψ} = \bar{α}_1 \braket{θ_1|ψ} + \bar{α}_2
    \braket{θ_2|ψ}$.
  ☛ Positive-definiteness: $\braket{ψ|ψ} ≥ 0$ with equality if and
    only if $\ket{ψ} = \mathbf{0}$, where $\mathbf{0}$ is the zero
    element of the vector space.
\end{itemize}

\subsubsection{Norms and Metrics}

The usefulness of Hilbert spaces in physics stems in part from the fact that
they have natural measures of angle and distance through the inner product. In
fact, the inner product of a Hilbert space $ℋ$ induces a ☼{norm}, $\norm{·}
: ℋ → ℝ$, which in
turn induces a ☼{metric}, $D(·,·) : ℋ \times ℋ → ℝ$. They are
given by the expressions
\begin{align}
  \norm<b>{\ket{ψ}} := √{\braket{ψ|ψ}} \tn*{and}
  D \big( \ket{ψ},\ket{φ} \big) := \norm<b>{\ket{ψ} - \ket{φ}} ¶.
\end{align}

More generally\,|\,and for later reference\,|\,a metric on an arbitrary set is 
defined as follows.
\begin{definition}[Metric]
  \label{df:metric}
  Let $\cX$ be a set and let $a, b, c ∈ \cX$. Then, the
  functional $D: \cX \times \cX → ℝ$ is a metric on $\cX$ if it satisfies
  \begin{enum}
    ☛ Positive-definiteness: $D(a, b) ≥ 0$ with equality
      if and only if $a = b$.
    ☛ Symmetry: $D(a, b) = D(b, a)$.
    ☛ Triangle inequality: $D(a, c) ≤ D(a, b) + D(b, c)$.
  \end{enum}
\end{definition}
\noindent It is easy to verify that the induced metric introduced above fulfills
these conditions.

\subsubsection{Orthonormal Bases}
Let $\sB$ be a set of kets $\ket{θ_1}, \ket{θ_2}, \ldots, \ket{θ_n} ∈ ℋ$.
The ☼[linear span]{span} of $\sB$ is the subspace of $ℋ$ containing all linear
combinations of these kets,
\begin{align}
  \vecspan{\sB} = \vecspan{\ket{θ_1}, \ket{θ_2}, \ldots, \ket{θ_n} }
    := \Big\{ \sum_{i=1}^n α_i \ket{θ_i} : α_i ∈ ℂ \Big\} ¶.
\end{align}
We call $\sB$ a ☼{basis} of $ℋ$ if $\vecspan{\sB} = ℋ$. Furthermore,
a basis $\{ \ket{e_i} \}_{i=1}^d := \{ \ket{e_1}, \ket{e_2}, \ldots, \ket{e_d} \}$ 
is called
☼[orthonormal]{basis!orthonormal} ☼*{orthonormal|see{basis}} if its elements are 
mutually orthonormal. Formally, this means that $\braket{e_i|e_j} = δ_{ij}$, where 
the ☼{Kronecker delta} or ☼{indicator} function is given by $δ_{i,j} = 1$ if $i = j$ and $0$ otherwise.
Every orthonormal basis has exactly $d$ elements, where $d := \dim{ℋ}$ is the
Hilbert ☼{dimension} of $ℋ$. A ket $\ket{ψ} ∈ ℋ$ has a unique
☼{decomposition} into any orthonormal basis $\{ \ket{e_i} \}$, where it can be
represented as a $d \times 1$ column vector.
\begin{align}
  \ket{ψ} = \sum_{i = 1}^d  \braket{e_i|ψ} \, \ket{e_i} \sim
    \mat{c}{\braket{e_1|ψ} \\ \braket{e_2|ψ} \\ \vdots \\ \braket{e_d|ψ}} ¶.
\end{align}
Similarly, a bra $\bra{ψ} ∈ ℋ^*$ can be represented as a $1 \times d$ row vector.
\begin{align}
  \bra{ψ} = \sum_{i = 1}^d \braket{ψ|e_i} \bra{e_i} \sim
    \mat{cccc}{ \braket{ψ|e_1} & \braket{ψ|e_2} & \ldots & \braket{ψ|e_d} } ¶.
\end{align}

\subsubsection{Direct Sum Spaces}

Given two Hilbert spaces $ℋ$ and $ℋ'$, we introduce the ☼[direct
sum]{direct sum} Hilbert space of $ℋ$ and $ℋ'$, denoted $ℋ ⨁
ℋ'$. This space consists of linear combinations of tuples in $ℋ \times ℋ'$, which we denote by $\ket{ψ} ⨁ \ket{ψ'}$, where $\ket{ψ} ∈ ℋ$ and $\ket{ψ'} ∈ ℋ'$.
The direct sum space is motivated via its inner product, 
which we want to be a sesquiliniar extension of the relation
\begin{align}
  \ip<B>{ \ket{ψ} ⨁ \ket{ψ'},  \ket{φ} ⨁ \ket{φ'} } = \braket{ψ|φ} + \braket{ψ'|φ'} ¶
\end{align}
for any two tuples $\ket{ψ} ⨁ \ket{ψ'}, \ket{φ} ⨁ \ket{φ'} ∈ ℋ ⨁ ℋ'$.
This inner product is positive-definite if and only if
$α \big(\ket{ψ} ⨁ \ket{ψ'}\big) +\ket{φ} ⨁ \ket{φ'}
   = \big (α \ket{ψ} + \ket{φ} \big) ⨁ \big( α \ket{ψ'} + \ket{φ'} \big)$ 
   for any $α ∈ ℂ$.
This constitutes the rule for multiplication by a scalar and addition of elements in $ℋ ⨁ ℋ'$ and we may thus write
\begin{align}
  ℋ ⨁ ℋ' := \vecspan<b>{ \ket{ψ} ⨁ \ket{ψ'} : \ket{ψ} ∈ ℋ ,  
    \ket{ψ'} ∈ ℋ' } ¶\,.
\end{align}
An orthonormal basis of this space is given by $\{ \ket{e_i} ⨁ \mathbf{0} \} \cup \{ \mathbf{0} 
⨁ \ket{e_j'} \}$, where $\{ \ket{e_i} \}$ is an orthonormal basis of $ℋ$ and 
$\{ \ket {e_j'} \}$ is an orthonormal basis of $ℋ'$. Hence, $\dim{ℋ ⨁ ℋ'} 
= d + d'$, where $d$ and $d'$ are the Hilbert dimensions of $ℋ$ and $ℋ'$, 
respectively.

\subsubsection{Tensor Product Spaces}
\label{sc:pre/lin/tensor}

Given two Hilbert spaces $ℋ$ and $ℋ'$, we define the ☼[tensor
product]{tensor product} Hilbert space of $ℋ$ and $ℋ'$, denoted $ℋ ⨂
ℋ'$. The space consists of linear combinations of tuples (pure tensors) in $ℋ \times ℋ'$, denoted $\ket{ψ} ⨂ \ket{ψ'}$, where $\ket{ψ} ∈ ℋ$ and $\ket{ψ'} ∈ ℋ'$.
Again, we motivate the tensor product space via its inner product, which
we want to be a sesquiliniar extension of the relation
\begin{align}
  \ip<B>{ \bra{ψ} ⨂ \bra{ψ'} , \ket{φ} ⨂ \ket{φ'} } 
  = \braket{ψ|φ} \, \braket{ψ'|φ'} ¶\,.
\end{align}
for any two tuples $\ket{ψ} ⨂ \ket{ψ'}, \ket{φ} ⨂ \ket{φ'} ∈ ℋ ⨂ ℋ'$.
This inner product is positive-definite if and only if 
the following relations are satisfied. For any $α ∈ ℂ$, we need
\begin{align}
  &α \big(\ket{ψ} ⨂ \ket{ψ'} \big) =
    (α \ket{ψ} ) ⨂ \ket{ψ'} = \ket{ψ} ⨂ ( α \ket{ψ'} )\,, ¶\\
  &\ket{ψ} ⨂ \ket{ψ'} + \ket{ψ} ⨂ \ket{φ'} = \ket{ψ} ⨂ 
    \big( \ket{ψ'} + \ket{φ'} \big) \tn*{and}¶\\
  &\ket{ψ} ⨂ \ket{ψ'} + \ket{φ} ⨂ \ket{ψ'} 
    = \big( \ket{ψ} + \ket{φ} \big) ⨂ \ket{ψ'} ¶\,.
\end{align}
The tensor product Hilbert space is then defined as the
vector space built from linear combinations of all pure tensors modulo 
the above equivalence relations.
Moreover, if $ℋ$ has an orthonormal basis $\{ \ket{e_i} \}_i$ and $ℋ'$ has
an orthonormal basis $\{ \ket{e_j'} \}_j$, then the pure tensors $\ket{e_i} ⨂
\ket{e_j'} ∈ ℋ ⨂ ℋ'$ form an orthonormal basis of $ℋ ⨂ ℋ'$ and
$\dim{ℋ ⨂ ℋ'} = d · d'$.

\subsection{Operators on Hilbert Spaces}
☼*{operator}

\subsubsection{Linear Operators}

We denote the set of ☼[linear operators]{operator!linear} from $ℋ$ to
$ℋ'$ (vector space ☼[homomorphisms]{homomorphism}) by $\olin{ℋ, ℋ'}$. 
Every operator $L ∈ \olin{ℋ, ℋ'}$ has a unique decomposition ☼*{decomposition} 
into
any pair of orthonormal bases, $\{ \ket{e_i} \}$ of $ℋ$ and $\{ \ket{e_j'} \}$
of $ℋ'$.
The operator can be represented as a $d' \times d$ matrix in these bases.
\begin{align}
  L = \sum_{i,j} \braket{e_j'|L|e_i} \, \ket{e_j'}\!\bra{e_i} \sim
    \mat{ccccc}{
    [L]_{11} & [L]_{12} & [L]_{13} & \cdots & [L]_{1d} \\
    {}[L]_{21} & [L]_{22} & [L]_{23} & \cdot & \vdots  \\
    \vdots & \cdot & \cdot & \cdot  & \vdots \\
    {}[L]_{d'1} & \cdots & \cdots & \cdots & [L]_{d'd} } ¶[pre/lin/mat],
\end{align}
where $[L]_{ji} := \braket{e_j'|L|e_i}$.
Note that representations of kets, bras and operators in a particular basis are
only used as illustrations throughout this thesis.

For every $L ∈ \olin{ℋ, ℋ'}$, we define its ☼{adjoint} operator $L† ∈ \olin{ℋ',
ℋ}$ as the unique operator that satisfies
\begin{align}
  \braket{ψ|L|φ} = \overline{\braket{φ|L†|ψ}} \tn*{for all}
    \ket{φ} ∈ ℋ \tn{and} \ket{ψ} ∈ ℋ' ¶\,.
\end{align}
(This is equivalent to the condition $\ip<b>{\ket{ψ}, L \ket{φ}} = \ip<b>{L†
\ket{ψ}, \ket{φ}}$ expressed in terms of the inner product.) In particular,
this definition, together with~§[pre/lin/bra], implies that a ket $L \ket{φ}$
has the dual bra $\bra{φ} L†$. 

Moreover, we will also need the ☼{transpose} $Lᵀ$ of $L$ with respect to a
pair of bases,  $\{ \ket{e_i} \}$ of $ℋ$ and $\{ \ket{e_j'} \}$ of $ℋ'$.
This is defined as
\begin{align}
  Lᵀ := \sum_{i,j} \braket{e_j'|L|e_i} \ket{e_i}\!\bra{e_j'} ¶\,.
\end{align}

The ☼{kernel} of an operator $L ∈ \olin{ℋ, ℋ'}$ is the subspace of $ℋ$ spanned
by all kets that are mapped to zero by $L$, that is
\begin{align}
  \ker{L} := \vecspan{ \ket{φ} :  \ket{φ} ∈ ℋ \tn{and} L \ket{φ} = 0 } ¶\,.
\end{align}
In contrast, the ☼{support} of $L$ is the subspace orthogonal to
the kernel of $L$, namely
\begin{align}
  \supp{L} := \vecspan{ \ket{φ} : \forall \ket{ψ} \in \ker{L} \tn{we have} \braket{ψ|φ} = 0 } ¶\,.
\end{align}The ☼{rank} of $L$ is the dimension of its support, $\rank{L} := \dim{\supp{L}}$. 
Finally, the ☼{image} of $L$ is the subspace of $ℋ'$ spanned by
$L$, that is
\begin{align}
  \imag{L} := \vecspan{ L \ket{φ} : \ket{φ} ∈ ℋ } ¶\,.
\end{align}

\subsubsection{Projectors, Identity and Inverse}

The linear operators from $ℋ$ onto itself (vector space ☼[endomorphisms]{endomorphism}) are
denoted $\olin{ℋ} := \olin{ℋ, ℋ}$.
A ☼{projector} into a subspace $ℋ'$ of $ℋ$ is an operator $P ∈ \olin{ℋ}$ with 
$\supp{P} = \imag{P} = ℋ'$ that acts as an ☼{identity} on all $\ket{φ'} ∈ ℋ'$, i.e.\ 
$P \ket{φ'} = \ket{φ'}$.
Thus, $P^2 = P$ and $P = P†$.
We use the symbol $ⅈ$ to denote the ☼[identity operator]{operator!identity} on $ℋ$, 
which is the projector from $ℋ$ onto $ℋ$.
The identity operator can be decomposed in any orthonormal basis $\{ \ket{e_i}
\}$ of $ℋ$ as $ⅈ = \sum_i \proj{e_i}$. For any linear operator $L ∈ \olin{ℋ}$,
we denote the projector onto its support by $Π^L$. 

For operators $L ∈ \olin{ℋ}$, we define their ☼{inverse} (if it exists), $L\inv ∈ 
\olin{ℋ}$, as the unique operator satisfying $L\inv L = L L\inv = ⅈ$. We will often 
use a ☼[generalized inverse]{inverse!generalized}, which is defined for every 
operator $L ∈ \olin{ℋ, ℋ'}$ and is the inverse of $L$ on its support. This means, 
we define the generalized inverse $L\inv$ as the unique operator satisfying $L\inv L 
= Π^L$ and $\supp{L\inv} = \imag{L}$.

\subsubsection{Isomorphisms, Unitaries and Isometries}

An ☼{isomorphism} is a linear bijective map associating elements of 
two structured sets in 
a way that preserves that structure. In particular, an isomorphism $u\!: ℋ → ℋ'$ 
between Hilbert spaces $ℋ$ and $ℋ'$ preserves the inner product and, thus, 
satisfies 
\begin{align}
  \ip<b>{\ket{ψ},\ket{φ}} = \ip<b>{u(\ket{ψ}), u(\ket{φ})} ¶[pre/lin/iso]
\end{align}
for all $\ket{ψ}, \ket{φ} ∈ ℋ$. We call two Hilbert spaces $ℋ$ and $ℋ'$ 
isomorphic, denoted $ℋ \iso ℋ'$, if there exists an isomorphism between $ℋ$ and 
$ℋ'$. Two Hilbert spaces are isomorphic if and only if they have the same Hilbert 
dimension. (To see this, note that an isomorphism can always be seen as a bijective 
map between elements of two orthonormal bases.)

A ☼[unitary operator]{operator!unitary} is an operator $U ∈ \olin{ℋ}$ that is an 
isomorphism between $ℋ$ and itself. 
%(an ☼{automorphism}). 
Since an isomorphism has 
to satisfy~§[pre/lin/iso], we have $\braket{ψ|φ} = \braket{ψ|U†U|φ}$ or, 
equivalently, $U†U = ⅈ$. Hence, the inverse of a unitary operator $U$ is well 
defined and $U\inv = U†$.

A ☼[partial isometry]{isometry!partial} is an operator $V ∈ 
\olin{ℋ, ℋ'}$  that is an isomorphism between $\supp{V} \subseteq ℋ$ and $\imag
{V} \subseteq ℋ'$.
A partial isometry satisfies $V†V = Π^V$ and, thus, $V\inv = V†$ is its generalized 
inverse (and a partial isometry too). An ☼{isometry} is a partial isometry with 
full support on $ℋ$.

In the following, we denote the set of unitary operators on a Hilbert space $ℋ$
as $\ouni{ℋ}$ and the set of partial isometries from $ℋ$ to $ℋ'$ as $\ouni{ℋ,
ℋ'}$. Isometries can always be understood as ☼[embeddings]{embedding}. 
Let $ℋ$ and $ℋ'$ be two Hilbert spaces such that 
$\dim{ℋ} < \dim{ℋ'}$ and let $V ∈ \ouni{ℋ, ℋ'}$ be an
isometry that embeds $ℋ$ into $ℋ'$, i.e.\ it satisfies $V†V = ⅈ$ on $ℋ$. Then,
for every operator $L ∈ \olin{ℋ}$, we implicitly define its embedding $L' := V L
V† ∈ \olin{ℋ'}$.

\subsubsection{Trace}

For any Hilbert space $ℋ$, the ☼{trace} of an operator in $\olin{ℋ}$ is the
(linear) functional $\trace : \olin{ℋ} → ℂ$ with the defining properties
\begin{align}
  \tr{AB} = \tr{BA} \tn*{and} \tr{ⅈ} = d ¶[pre/lin/tr]
\end{align}
for all linear operators $A ∈ \olin{ℋ, ℋ'}$ and $B ∈ \olin{ℋ', ℋ}$. 
In particular the trace is invariant under unitary conjugation, 
$\tr{U† L\, U} = \tr{L}$ for any $U ∈ \ouni{ℋ}$ and $L ∈ \olin{ℋ}$. This implies
that there exists a representation of the trace as a functional on matrix
representations that is independent of the choice of basis used for the
representation. 

If we choose an orthonormal basis $\{ \ket{e_i} \}$ for $ℋ$ and $\{ \ket{e_j'}
\}$ for $ℋ'$, the operators $AB$ and $BA$ can be represented as matrices with
entries
\begin{align}
  [AB]_{ji} = ∑_k \braket{e_j'|A|e_k} \braket{e_k|B|e_i'} \tn*{and} 
  [BA]_{ji} = ∑_k \braket{e_k'|A|e_i} \braket{e_j|B|e_k'} ¶
\end{align}
and it is easy to verify that the only functional that satisfies~§[pre/lin/tr]
is the sum of the diagonal elements, $∑_i [AB]_{ii} = ∑_i [BA]_{ii}$. Hence, for
any $L ∈ \olin{ℋ}$ and any basis $\{ \ket{e_i} \}$ of $ℋ$, we have $\tr{L} = ∑_i
[L]_{ii} = ∑_i \braket{e_i|L|e_i}$, which is how the trace of a matrix is
commonly defined.

\subsubsection{Singular Values}

For any operator $L ∈ \olin{ℋ, ℋ'}$, there always exists a pair of bases $\{ \ket
{e_i} \}$ of $ℋ$ and $\{ \ket{e_i'} \}$ of $ℋ'$ such that $L$ can be decomposed as
\begin{align}
  L = ∑_i s_i\, \ket{e_i'}\!\bra{e_i} 
    \tn*{and} s_1 ≥ s_2 ≥ \ldots ≥ s_m > 0 ¶[pre/lin/svd]\,,
\end{align}
where $m$ is the rank of $L$. This is called the ☼[singular value decomposition]
{decomposition!singular value} and the unique positive $s_i = s_i(L)$ are called 
\emph{singular values}. The singular values are invariant under unitary rotations, 
as these operations can be absorbed into the basis. Thus, for any $U ∈ \ouni{ℋ}$ 
and $V' ∈ \ouni{ℋ'}$,
\begin{align}
  s_i(V' L\, U) = s_i(L) ¶[pre/lin/svd-uni]\,.
\end{align}

The support and image of $L$ can be expressed in terms of the two bases, 
that is~$\supp{L} = \vecspan{ \ket{e_i} }$ and $\imag{L} = \vecspan{ \ket
{e_i'} }$.

Note that the operator $L†L = \sum_i {s_i}^2\, \proj{e_i} ∈ \olin{ℋ}$ has a
unique positive square root, namely the ☼{modulus} of $L$,
\begin{align}
  \abs{L} := √{L†L} = \sum_i s_i \proj{e_i} ¶.
\end{align}
Comparing this with the singular value decomposition~§[pre/lin/svd], we find
the ☼[polar decomposition]{decomposition!polar}, $L = W \abs{L}$,
where $W: \ket{e_i} ↦ \ket{e_i'}$ is a partial isometry from $ℋ$ to $ℋ'$
defined through the bases of the singular value decomposition of $L$.

\subsubsection{Hilbert-Schmidt Inner Product and Schmidt Decomposition}
☼*{inner product!Hilbert-Schmidt}

Given two Hilbert spaces $ℋ$ with orthonormal basis $\sB = \{ \ket{e_i} \}$ and
$ℋ'$, we associate operators in $\olin{ℋ, ℋ'}$ with tensors in $ℋ ⨂ ℋ'$ through
the ☼{isomorphism} $\veciso_{\sB} \!: L ↦ ∑_i \ket{e_i} ⨂ L
\ket{e_i}$.
Using the decomposition of $L$ in~§[pre/lin/mat], we find that $\veciso(L)$ can
be written as 
\begin{align}
  \veciso_{\sB}(L) = \sum_{i,j}\, [L]_{j,i}\, \ket{e_i} ⨂ \ket{e_j'} ¶\,,
\end{align}
which simply corresponds to the rearrangement of the matrix entries of $L$ (in
the given bases) as a column vector. 

This isomorphism is useful because it
induces an inner product on linear operators. We define the ☼[Hilbert-Schmidt
inner product]{inner product!Hilbert-Schmidt} on the complex ☼{Hilbert space}
$\olin{ℋ, ℋ'}$ as
\begin{align}
  \ip{A, B} := \ip<b>{\veciso_{\sB}(A),\veciso_{\sB}(B)} = 
      \sum_{i,j} \braket{e_j|e_i} \braket{e_j|A†B|e_i} = 
      \sum_i \braket{e_i|A†B|e_i} ¶,
\end{align}
for any $A, B ∈ \olin{ℋ, ℋ'}$.
This expression is equal to the trace and, thus, independent of the basis $\sB$ 
chosen for the isomorphism:
\begin{align}
 \ip{A, B} = \tr{A†B} ¶[pre/lin/hilbert-schmidt].
\end{align}

The tensor $\veciso_{\sB}(ⅈ) = ∑_i \ket{e_i} ⨂ \ket{e_i}$ takes a special role
in the analysis of quantum systems. In particular, we will need the following
property.
\begin{lemma}[Mirror Lemma]
  \label{lm:mirror}
  Let $ℋ$ be a Hilbert space with orthonormal basis 
  $\sB$ and $L ∈ \olin{ℋ}$, then 
  $(ⅈ ⨂ L) \veciso_{\sB}(ⅈ) = (Lᵀ ⨂ ⅈ) \veciso_{\sB}(ⅈ)$, where the
  transpose is taken with regards to $\sB$.
\end{lemma}
\begin{proof}
  By inspection.
\end{proof}
This, together with the above isomorphism, can be used to prove the existence
of the ☼[Schmidt decomposition]{decomposition!Schmidt}.
\begin{lemma}[Schmidt Decomposition]
  \label{lm:schmidt-decomp}
  Let $ℋ$, $ℋ'$ be Hilbert spaces and let $\ket{\theta} ∈ ℋ ⨂ ℋ'$. Then,
  there exist orthonormal bases $\{ \ket{θ_i} \}$ of $ℋ$ and 
  $\{ \ket{θ_i'} \}$ of $ℋ'$ as well as non-negative numbers $s_i$ such that
  \begin{align}
    \ket{\theta} = ∑_i s_i\, \ket{θ_i} ⨂ \ket{θ_i'} ¶\,.
  \end{align}
\end{lemma}
\begin{proof}
  Given a linear operator $L ∈ \olin{ℋ, ℋ'}$ and an orthonormal basis $\sB = \{
\ket{e_i} \}$, we use the singular value decomposition of $L$ to get
  \begin{align}
    \veciso_{\sB}(L) = \veciso_{\sB}(V† S\, U) 
    = ∑_i s_i(L)\ Uᵀ \ket{e_i} ⨂ V† \ket{e_i} ¶[pre/lin/schmidt-decomp]\,,  
  \end{align}
  where $S = ∑_i s_i(L)\, \proj{e_i}$ is diagonal in $\sB$. The isometries $U ∈ 
\ouni{ℋ}$ and $V ∈ \ouni{ℋ', ℋ}$ map the bases of the  singular value
decomposition to $\sB$. Hence, every tensor can be written in the
form~§[pre/lin/schmidt-decomp], which concludes the proof.
%Note also that the numbers $s_i$ are called \emph{Schmidt coefficients} and
%correspond to the singular values of the linear operator in corresponding to
%the tensor $\ket{\theta}$.
\end{proof}

\subsubsection{Operator Norms}
☼*{norm!operator}

The singular value decomposition allows the definition of various norms on $\olin
{ℋ, ℋ'}$. In particular, we will often use the ☼[Schatten norms]{norm!
Schatten}. For $p ≥ 1$, they are defined as (cf.~e.g.~\cite{bhatia97,watrous-ln08})
\begin{align}
  \norm{L}[p] := \bigg( ∑_{i=1}^m s_i(L)^p \bigg)^{\frac{1}{p}} ¶.
\end{align}
In particular, we will often use the infinity norm, $\norm{·}[∞]$, which is equal to the induced norm of the underlying Hilbert space norm, that is
\begin{align}
  \norm{L}[∞] = s_1(L) = \sup_{ \ket{ψ} ∈ ℋ } \
    \frac{ \norm{ L \ket{ψ} } }{ \norm{\ket{ψ}} } ¶\,.
\end{align}
The $1$-norm is equal to the trace norm,
\begin{align}
  \norm{L}[1] = ∑_{i=1}^m s_i(L) = \trace{ \abs{L} } ¶\,.
\end{align}
And, finally, the $2$-norm is equal to the induced norm of the 
Hilbert-Schmidt inner product,
\begin{align}
  \norm{L}[2] = \bigg( ∑_{i=1}^m s_i(L)^2 \bigg)^{½} = √{\ip{L, L}} ¶\,.
\end{align}

These three norms satisfy $\norm{L}[∞] ≤ \norm{L}[2] ≤ \norm{L}[1]$. And, thanks to §[pre/lin/svd-uni], they are ☼[invariant under unitary rotations]{unitarily invariant}. Hence, $\norm*{L} = \norm*{V' L\, U}$, where $\norm*{·}$ denotes any of the Schatten norms introduced above. For any three operators $L, M, N ∈ \olin{ℋ}$, these norms satisfy
\begin{align}
  \norm*{MLN} ≤ \norm{M}[∞]\, \norm*{L}\, \norm{N}[∞] ¶[pre/lin/norm-mult]\,.
\end{align}
Moreover, these norms are ☼{sub-multiplicative}, i.e.\ $\norm*{MN} ≤ \norm*{M} \norm*{N}$.
%Or, more generally, a ☼{Hölder inequality} holds.
%For any $p, q, r ≥ 0$, we have~\cite{bhatia97}
%\begin{align}
%  \norm*<b>{ \abs{MN}^p }^{1/p} ≤ \norm*<b>{ \abs{M}^q }^{1/q}\, 
%    \norm*<b>{ \abs{N}^r }^{1/r}
%      \tn*{whenever} \frac{1}{p} = \frac{1}{q} + \frac{1}{r} \,.
%      ¶[pre/lin/hoelder]
%\end{align}

\subsection{Positive Semi-Definite Operators}

\subsubsection{Hermitian Operators}

An operator $M ∈ \olin{ℋ}$ is called ☼{self-adjoint} or
☼[Hermitian]{operator!Hermitian} ☼*{Hermitian|see{operator}}
if $M† = M$. The set of Hermitian operators on $ℋ$ is denoted $\oherm{ℋ} := \{
M ∈ \olin{ℋ} : M† = M \}$. Note that the real vector space $\oherm{ℋ}$ together
with the Hilbert-Schmidt inner product in~§[pre/lin/hilbert-schmidt] form a ☼[real
Hilbert space]{Hilbert space}.

Every $M ∈ \oherm{ℋ}$ has an ☼[eigenvalue
decomposition]{decomposition!eigenvalue}, namely 
\begin{align}
  M = ∑_i λ_i\, \proj{e_i} ¶\,, \tn*{where} λ_i ∈ ℝ \tn{and}
    \abs{λ_1} ≥ \abs{λ_2} ≥ \ldots ≥ \abs{λ_m} > 0 ¶.
\end{align}
The coefficients $λ_i = λ_i(M)$ are called ☼[eigenvalues]{eigenvalue} of $M$ and
$\ket{e_i}$ are ☼[eigenvectors]{eigenvector}. Together, the eigenvectors form an
☼{eigenbasis} $\{ \ket{e_i} \}$, which is an orthonormal basis of $\supp{M} =
\imag{M}$. This eigenbasis is unique if and only if all the eigenvalues are
mutually different. The singular values can be expressed in terms of the eigenvalues,
that is $s_i(M) = \abs{λ_i(M)}$.\footnote{The existence of the eigenvalue
decomposition follows, for example, from the fact that the singular value
decompositions of $M†M = ∑_i {s_i}^2 \proj{e_i}$ and $MM† = ∑_i {s_i}^2
\proj{e_i'}$ have to be the same. Thus, if all singular values are mutually different, we
have $\ket{e_i'} = e^{ι φ} \ket{e_i}$ and the eigenvalue decomposition is of the
form $M = ∑_i ± s_i \proj{e_i}$. More generally, an eigenvalue decomposition
exists for all ☼{normal} $L ∈ \olin{ℋ}$, 
where normal means that $L†L = LL†$.}

The eigenvalue decomposition is particularly useful to define the action of functions on
operators. Note, for example, that $M^2 = MM = ∑_i {λ_i}^2\, \proj{e_i}$,
$\abs{M} = ∑_i \abs{λ_i}\, \proj{e_i}$ and $M\inv = ∑_i \frac{1}{λ_i}
\proj{e_i}$. We generalize this and define the action of arbitrary functions $f:
ℝ / \{ 0 \} → ℝ$ on Hermitian operators as
\begin{align}
  f: M = ∑_i λ_i\, \proj{e_i}\ ↦\ ∑_i f(λ_i)\, \proj{e_i} ¶.
\end{align}
Note that the function $f$ only acts on the support of the Hermitian operator,
leaving its kernel intact.

For any $M ∈ \oherm{ℋ}$, we denote by $\{ M \}_+$ and $\{ M \}_-$ the projections of $M$ onto its positive and negative eigenspace, respectively, that is 

\begin{align}
  \{ M \}_+ = \sum_{i : λ_i > 0} λ_i\, \proj{e_i} \tn*{and}
  \{ M \}_- = \sum_{i : λ_i < 0} λ_i\, \proj{e_i} ¶\,.
\end{align}

\subsubsection{Positive Semi-Definite Operators}

The set of ☼[positive semi-definite]{operator!positive semi-definite}
operators, $\opos{ℋ} \subset \oherm{ℋ}$ is the set of operators that have
positive inner products with all vectors in $ℋ$, namely
\begin{align}
  \opos{ℋ} :=&\, \big\{ M ∈ \oherm{ℋ} : \braket{ψ|M|ψ} ≥ 0\ 
    \tn{for all} \ket{ψ} ∈ ℋ \big\} ¶\\
    =&\, \{ M ∈ \oherm{ℋ} : λ_i(M) > 0\ \forall i \} ¶\,.
\end{align}
We write $M ≥ 0$ if and only if $M ∈ \opos{ℋ}$. In the following, we simply call
these operators ☼[positive]{operator!positive}.
Moreover, given two Hermitian operators $M, N ∈ \oherm{ℋ}$, we write $M ≥ N$ if and only if $M - N ≥ 0$ and we write $M ≤ N$ if and only if $N - M ≥ 0$. Note that this relation constitutes a partial order on $\oherm{ℋ}$.

\subsection{Operators on Tensor Spaces}

We will often use linear operators on ☼{tensor product} spaces, for example $\olin{ℋ ⨂ ℋ', \bar{ℋ} ⨂ \bar{ℋ}'} \iso \olin{ℋ, \bar{ℋ}} ⨂ \olin{ℋ', \bar{ℋ}'}$, where the second isomorphic space is constructed from the Hilbert spaces $\olin{ℋ, \bar{ℋ}}$ and $\olin{ℋ', \bar{ℋ}'}$ as in Section~\ref{sc:pre/lin/tensor}. Hence, we can write every element $L ∈ \olin{ℋ ⨂ ℋ', \bar{ℋ} ⨂ \bar{ℋ}'}$ as a linear combination of the form
\begin{align}
  L = ∑_{α,β} χ_{αβ}\, σ_α ⨂ σ'_β \,, \tn*{where} χ_{αβ} ∈ ℂ ¶[pre/lin/op-decomp]\,,
\end{align}
and $\{ σ_α \}$ and $\{ σ'_β \}$ are bases of $\olin{ℋ, \bar{ℋ}}$ and $\olin{ℋ', \bar{ℋ}'}$, respectively. If $L$ is Hermitian, then the $χ_{αβ}$ can be chosen real and the bases can be chosen Hermitian. (This follows directly from the fact that $\oherm{ℋ}$ is a real Hilbert space.) Furthermore, this allows us to introduce a trivial extension of operators $L ∈ \olin{ℋ, \bar{ℋ}}$ to operators $L ⨂ ⅈ$ in $\olin{ℋ ⨂ ℋ', \bar{ℋ} ⨂ ℋ'}$. In the following, this extension is assumed implicitly whenever an operator defined on a subspace is applied to a tensor space.

\subsection{Completely Positive Maps}

\subsubsection{Super-Operators}

A ☼{super-operator} is a linear map from linear operators on one Hilbert space to
linear operators on another Hilbert space. For example, if the two Hilbert
spaces are $ℋ$ and $ℋ'$, we denote the set of super-operators from $\olin{ℋ}$ to
$\olin{ℋ'}$ by $\olin{\olin{ℋ}, \olin{ℋ'}}$. The super-operators form a vector
space and, since every operator in $\olin{ℋ}$ can be decomposed according to~§[pre/lin/op-decomp], the action of a super-operator on an operator in a (larger) tensor space is well-defined via linearity. We use the symbol $\circ$ to denote concatenations of super-operators, for example $(\sE \circ \sF) [·] = \sE [\sF [ · ]]$.

For any super-operator $ℰ ∈ \olin{\olin{ℋ}, \olin{ℋ'}}$, its adjoint
super-operator, $ℰ† ∈ \olin{\olin{ℋ'}, \olin{ℋ}}$, is defined via the
Hilbert-Schmidt inner product as the unique operator satisfying
\begin{align}
  \ip<b>{ ℰ[A], B } = \ip<b>{ A, ℰ†[B] } \tn*{for all} A ∈ \olin{ℋ}, B ∈
\olin{ℋ'} ¶\,.
\end{align}

\subsubsection{Completely Positive Maps}

Super-operators that (consistently) map positive operators onto positive
operators are called ☼[completely positive maps]{CPM}
(CPMs).
\begin{definition}[Completely Positive Map]
  Let $ℋ$ and $ℋ'$ be Hilbert spaces. A super-operator $ℰ ∈
\olin{\olin{ℋ}, \olin{ℋ'}}$ is called completely  positive, if, for any
auxiliary Hilbert space $ℋ''$, it holds that
	\begin{align}
	  ℰ[M] ≥ 0  \tn*{for all} M ∈ \opos{ℋ ⨂ ℋ''}  ¶\,.
	\end{align}
\end{definition}

An example of such a completely positive super-operator is the ☼{conjugation} with an
operator $L ∈ \olin{ℋ, ℋ'}$, that is the map $\sL: M ↦ L M L†$.
We will often use the following basic property of completely positive maps. Let
$ℰ ∈ \olin{\olin{ℋ}, \olin{ℋ'}}$ be completely positive, then
\begin{align}
  A ≥ B ⇒ ℰ[A] ≥ ℰ[B] \tn*{for all} A, B ∈ \oherm{ℋ} ¶[pre/lin/basic-step]\,.
\end{align}

Furthermore, we take note of the following property of positive semi-definite operators. 
For any $M, N ∈ \opos{ℋ}$, we have 
\begin{align}
  \tr{MN} = \tr{√{M} N √{M}} ≥ 0 ¶\,, 
\end{align}
where the last inequality follows from the fact that the conjugation with $√{M}$ is a completely 
positive map. In particular, if $X, Y ∈ \oherm{ℋ}$ satisfy $X ≥ Y$, we find $\tr{M X} 
≥ \tr{M Y}$.

A completely positive map $ℰ$ is called ☼[trace preserving]{trace!preserving} if
$\tr{ℰ[M]} = \tr{M}$ for all $M ∈ \olin{ℋ}$. We denote the set of all trace preserving
completely positive maps (☼[TP-CPMs]{TP-CPM}) from $\olin{ℋ}$ to $\olin{ℋ'}$ by
$\tpcpm{ℋ, ℋ'}$. If $ℰ$ is a TP-CPM, then its adjoint $ℰ†$ is completely positive and 
☼{unital}. Since
\begin{align}
  \tr{M} = \tr{ℰ[M]} = \tr<b>{ℰ†[ⅈ] M} \tn*{for all}\! M ∈ \olin{ℋ} ¶\,,
\end{align}
unital maps are defined by the property that they map the identity onto an identity, i.e.\
$ℰ†[ⅈ] = ⅈ$.

Finally, a completely positive map $ℰ$ is called ☼[trace
non-increasing]{trace!non-increasing} if $\tr{ℰ[M]} ≤ \tr{M}$ for all $M ∈ \opos{ℋ}$.
It is easy to verify (using the same argument as above) that its adjoint is completely positive and ☼{sub-unital}, i.e.\ it satisfies $ℰ†[ⅈ] ≤ ⅈ$.

\subsubsection{Partial Trace}

Given a bipartite Hilbert space $ℋ ⨂ ℋ'$, we are interested in the ☼{partial
trace} super-operator, denoted $\trace¬{ℋ'} ∈ \tpcpm{ℋ ⨂ ℋ', ℋ}$.
The partial trace is defined as the adjoint super-operator to $L ↦ L ⨂ ⅈ$, which
maps operators $L ∈ \olin{ℋ}$ to $\olin{ℋ ⨂ ℋ'}$. It is easy to verify,
using~§[pre/lin/hilbert-schmidt], that the
☼{trace} is the special case of the partial trace where $ℋ \iso ℂ$.

To justify this definition, let us investigate the action of the partial 
trace $\trace¬{ℋ'}$ of a product state 
$K ⨂ L$, where $K ∈ \olin{ℋ}$ and $L ∈ \olin{ℋ'}$, on an arbitrary state $Q ∈ \olin{ℋ}$.
We have
\begin{align}
  \ip{Q, \tr[ℋ']{K ⨂ L}} = \ip{K ⨂ L, Q ⨂ ⅈ¬{ℋ'}} = \ip{K, Q} \tr{L} ¶\,.
\end{align}
Since this holds for all $Q$, we have $\tr[ℋ']{K ⨂ L} = \tr{L} K$. 

The partial trace is cyclic
in operators on the same Hilbert space and commutes with operators on other Hilbert spaces.
Specifically, let $L ∈ \olin{ℋ ⨂ ℋ', ℋ ⨂ ℋ''}$ and $K ∈ \olin{ℋ'', ℋ'}$. 
We use the decomposition $L = ∑_{α,β} χ_{αβ}\, σ_α ⨂ σ'_β$ of~§[pre/lin/op-decomp]. Since
the partial trace is linear and $K$ only acts on $ℋ'$, we have, using the cyclicity of the
trace,
\begin{align}
  \tr[ℋ']{K L} = ∑_{α,β} \tr{K σ'_β}\, σ_α = ∑_{α,β} \tr{σ'_β K}\, σ_α = \tr[ℋ'']{L K} ¶\,.
\end{align}
Moreover, using this decomposition, it is easy to verify that 
$\tr[ℋ]{K L} = K \tr[ℋ]{L}$ and $\tr[ℋ]{L K} = \tr[ℋ]{L} K$.

\subsubsection{Choi-Jamiolkowski Isomorphism}
\label{sc:choi-jami}

In the same spirit as the $\veciso$-isomorphism between
operators $\olin{ℋ, ℋ'}$ and tensors $ℋ ⨂ ℋ'$, we now define the
☼[Choi-Jamiolkowski isomorphism]{isomorphism!Choi-Jamiolkowski}~\cite{jamiliolkowski72}, $Ω$, 
between super-operators in $\olin{\olin{ℋ}, \olin{ℋ'}}$ and operators in 
$\olin{ℋ' ⨂ ℋ}$.
\begin{align}
  Ω : ℰ ↦ ω^{ℰ} = ℰ \big[ \proj{Γ} \big] , \tn{where} 
    \ket{Γ} = ∑_i \ket{e_i} ⨂ \ket{e_i} ∈ \opos{ℋ ⨂ ℋ} ¶%[choi-jami]
\end{align}
and $\{ \ket{e_i} \}$ is an orthonormal basis of $ℋ$. The state $ω^{ℰ}$ is called the Choi-Jamiolkowski state of $ℰ$. The inverse operation $Ω\inv$ maps Choi-Jamiolkowski states $ω^{ℰ} ∈ \olin{ℋ' ⨂ ℋ}$ to super-operators
\begin{align}
  ℰ : \olin{ℋ} → \olin{ℋ'},\ X ↦ \tr<b>[ℋ]{ω^{ℰ} (ⅈ¬{ℋ'} ⨂ Xᵀ)} ¶%[choi-jami-inv],
\end{align}
where the transpose is taken with regards to the basis $\{ \ket{e_i} \}$.

There are various relations between properties of super-operators and properties of the corresponding Choi-Jamiolkowski states. The following can be verified by inspection.
\begin{compact}
  ☛ $ℰ$ is completely positive $\iff$ $ω^{ℰ}$ is positive semi-definite.
  ☛ $ℰ$ is trace-preserving $\iff$ $\tr[ℋ']{ω^{ε}} = ⅈ¬{ℋ}$.  
  ☛ $ℰ$ is unital $\iff$ $\tr[ℋ]{ω^{ℰ}} = ⅈ¬{ℋ'}$.
\end{compact}

\subsubsection{Kraus Operators and Stinespring Dilation}

The following two lemmas are of crucial importance in quantum information theory. They describe two alternative representations of completely positive\footnote{An extension of this to arbitrary super-operators is possible (see, e.g.~\cite{watrous-ln08}).} maps, especially trace non-increasing and trace preserving maps. 

Every completely positive super-operator can be represented as a sum of conjugations of the input with ☼[Kraus operators]{operator!Kraus}.~\cite{kraus69,kraus70}.
☼*{representation!Kraus}

\begin{lemma}[Kraus Representation]
  \label{lm:kraus}
  A super-operator $ℰ ∈ \olin{\olin{ℋ}, \olin{ℋ'}}$ is completely positive if and only if
  there exists a finite set of linear operators $\{ E_k \}$, $E_k ∈ \olin{ℋ, ℋ'}$ such that
  \begin{align}
    ℰ[A] = \sum_k E_k A\, E_k†\ \tn*{for all} A ∈ \olin{ℋ} ¶\,.
  \end{align}
  Furthermore, a completely positive $ℰ$ is trace non-increasing if 
  and only if $\sum_k E_k† E_k ≤ ⅈ$ and it is trace preserving if and 
  only if $\sum_k E_k† E_k = ⅈ$.
\end{lemma}
The operators $\{ E_k \}$ are called Kraus operators. Note that the adjoint $ℰ†$ of
$ℰ$ is completely positive and has Kraus operators $\{ E_k† \}$ since
\begin{align}
  \tr{ℰ†[B] A} = \tr{B\, ℰ[A]} = \tr<B>{ ∑_k E_k† B E_k\, A} \tn*{for all} A ∈ \olin{ℋ}¶\,.
\end{align}

Moreover, every CPM can be decomposed into its ☼{Stinespring dilation} form~\cite{stinespring54} as follows.

\begin{lemma}[Stinespring Dilation]
  \label{lm:stinespring}
  A super-operator $ℰ ∈ \olin{\olin{ℋ}, \olin{ℋ'}}$ is completely positive if and only if
  there exists a Hilbert space $ℋ''$ and an operator $L ∈ \olin{ℋ, ℋ' ⨂ ℋ''}$ such that
  \begin{align}
    ℰ(A) = \tr[ℋ'']{L A L†} \tn*{for all} A ∈ \olin{ℋ} ¶\,.
  \end{align}
  Moreover, if $ℰ$ is trace preserving then $L$ is an 
  isometry. If $ℰ$ is trace non-increasing, then $L$ is an isometry followed by a 
  projection in $\opos{ℋ''}$.
\end{lemma}

Proofs of these lemmas can be found in most quantum information textbooks (see, e.g.~\cite{nielsen00}). We will only prove the result for trace non-increasing CPMs, which is a bit less standard.

\begin{proof}
  Let $ℰ$ be a trace non-increasing CPM. First, note that its Kraus representation has to satisfy,
  for all $A ∈ \opos{ℋ}$,
  \begin{align}
    \tr{ℰ(A)} = \tr<B>{\sum_k E_k A\, E_k†} = \tr<B>{\sum_k E_k† E_k A} 
      ≤ \tr{A}  ¶\,.
  \end{align}
  This is equivalent to the condition $\sum_{k=1}^n E_k† E_k ≤ ⅈ$. Hence, the map can be extended 
  to a trace preserving CPM by adding another Kraus operator into the sum, e.g.
  \begin{align}
    E_{n+1} =  \Big( ⅈ - \sum_{k=1}^n E_k† E_k \Big)^{\frac{1}{2}} ¶.
  \end{align}
  We now construct a Stinespring dilation for this extended operation. In fact, a possible dilation is $L = \sum_{k=1}^{n+1} E_k ⨂ \ket{e_k}$, where $\{ \ket{e_k} \}$ is an orthonormal basis of $ℋ''$. This operation is an isometry and a Stinespring dilation of $ℰ$ can be recovered by projecting onto $\sum_{k=1}^n \proj{e_k}$ after applying $L$.
\end{proof}

\section{Quantum Mechanics}
\label{se:pre/qm}

☼*{quantum!mechanics} 

In this thesis, we will use a mathematical model for ☼[quantum
mechanics]{quantum!mechanics}\,|\,the density operator formalism on finite-dimensional
Hilbert spaces\,|\, that is restricted to physical systems with a finite 
dimensional configuration 
space. This means that continuous
observables such as position and momentum need to be considered
discretized and bounded. While it is unknown whether such a framework is
sufficient to describe
all possible correlations between observations of physical systems~\cite{scholz08}\footnote{More precisely, it is unclear whether all correlations can be approximated to arbitrary precision using the density operator formalism.}, 
it offers the opportunity to
focus on the main physical concepts without getting lost in delicate
mathematical arguments. Moreover, the main results of this thesis are
independent of the actual ☼{dimension} of the physical
systems under consideration and some of the results have already been re-derived
in a more general model of quantum mechanics that is based on infinite dimensional 
Hilbert spaces~\cite{furrer10} or (von Neumann) operator algebras of 
observables~\cite{furrer11}.

This section is partly inspired by Nielsen and Chuang~\cite{nielsen00} as well as Hardy~\cite{hardy01} and introduces quantum mechanics from a strictly information theoretic
perspective.\footnote{Quantum mechanics is often formulated in a way that highlights its relation to classical mechanics. In quantum information science, due
to its interdisciplinary nature between physics, computer science and information theory, a different approach to quantum mechanics has proven fruitful.}
The most important deviations from the standard treatment of quantum mechanics
are pointed out.

\subsection{Systems and States}

We use a very abstract notion of ☼[physical systems]{system}, describing them as general
purpose information carriers without specifying their actual physical
realization. In fact, the only system parameter we consider is its dimension, 
which, as we will see, corresponds to the dimension of the ☼{Hilbert space} used 
to describe the system.

The ☼{dimension} of an isolated system is given by the number
of mutually perfectly distinguishable states that can be prepared on the system. This assumes
perfect (idealized) preparation and measurement equipment. Alternatively, the dimension can be 
interpreted as the information storage capacity of the system. Very 
abstractly, an encoder is a map $ℰ$ from a set $\cX$ to states of the system and a decoder $\sD$ is a map from states of the system to $\cX$. A pair $\{ ℰ, \sD \}$ of encoder and decoder is perfect if $x = \sD[ℰ[x]]$ for all $x ∈ \cX$. The dimension of a system is then given by the maximum cardinality of a set $\cX$ such that there exist perfect encoders and decoders between $\cX$ and states of the system.

The simplest such system, a ☼{qubit} or two-level system, has dimension two and may
have different physical realizations. For example, the information could be
encoded as the spin degree of freedom of an electron or the polarization of a
photon.

\begin{quote}
  \vspace{-0.5cm}
  \begin{postulate}[State Space]
    \label{po:system}
    Isolated quantum ☼[systems]{system} are modeled as Hilbert spaces with 
    the dimension of the system.
    The system's ☼[state]{quantum state} is represented by a positive semi-definite operator with
    unit trace acting on this space.
  \end{postulate}
\end{quote}

The state of an isolated quantum system $A$ is thus 
fully characterized by all linear functionals on the 
state space, i.e.\ the functionals $\tr{ρ¬A L¬A}$ where $L¬A ∈ \olin{ℋ¬A}$ and $ρ¬A$ is the state 
of the system.

We denote isolated physical systems with capital letters, i.e.\ $A$, $B$, $C$,
and their associated Hilbert spaces with $ℋ¬A, ℋ¬B, ℋ¬C$. The states themselves
are denoted using lowercase greek letters, i.e.\ $ρ¬A$, $σ¬B$, $τ¬C$,
where the subscripts indicate which system is meant when necessary.
The dimension of a system $A$ is denoted by $d¬A := \dim{ℋ¬A}$.

We use $\onorm{ℋ¬A}$ to denote the set of quantum states on $A$, i.e.\ the
set $\onorm{ℋ¬A} := \{ ρ ∈ \opos{ℋ¬A} : \tr{ρ} =  1\}$.
Hence, a state of a system $A$ is represented as an element $ρ¬A ∈ \onorm{ℋ¬A}$.
In addition, we sometimes employ a larger set of states on $ℋ¬A$, the sub-normalized states
$\osub{ℋ¬A} := \{ ρ ∈ \opos{ℋ¬A} : 0 < \tr{ρ} ≤ 1 \}$. These states by themselves do not have a physical interpretation and usually only appear in technical statements
However, they can be seen as normalized quantum states on a Hilbert space $ℋ¬{A'}$ that are projected onto a subspace $ℋ¬A$ of $ℋ¬{A'}$.

We call a positive semi-definite operator $ρ¬A ∈ \opos{ℋ¬A}$ ☼[pure]{state!pure} if it has rank one, i.e.\ if $\rank{ρ¬A} = 1$. Pure operators can be represented as kets $\ket{ρ}¬A ∈ ℋ¬A$, where $\ket{ρ}[A]$ is determined by the relation $ρ¬A = \proj{ρ}[A]$ up to a phase factor. An operator that is not pure is called ☼[mixed]{state!mixed}. We often employ the completely mixed state on a system $A$, which is given by $π¬A := ⅈ¬A / d¬A$.

In most introductions to quantum mechanics\,|\,supposedly for historical reasons\,|\,the theory is first formulated in terms of pure quantum states and then later extended to arbitrary mixed states. From an information theoretic perspective, however, mixed states are more fundamental as they are generalizations of classical ☼[random variables]{random variable}.

\subsubsection{Classical Registers}

Discrete ☼[probability distributions]{probability} can be conveniently represented as states 
of a quantum system. We call these systems ☼[registers]{register} and typically denote them by the letters $X$, $Y$ or $Z$. To a register $X$ and its respective Hilbert space $ℋ¬X$, we associate an orthonormal basis $\{ \ket{x}[X] \}$ where $x ∈ \cX$ and $\cX$ is a set with cardinality $\abs{\cX} = d¬X$. 

A ☼[probability distribution]{probability} on $\cX$ is a map $P¬X: \cX → [0, 1]$ such that $∑_x P¬X(x) = 1$. It can be represented as a quantum state $ρ¬X$ on $X$, i.e.\
\begin{align} 
  ρ¬X = ∑_{x} P¬X(x)\, \proj{x}[X] ¶\,.
\end{align}  
It follows from the properties of $P¬X$ that this state is positive semi-definite and has 
unit trace.
In classical probability theory, a register corresponds to a discrete ☼{random variable} and the 
state of a register to the probability distribution over the random variable.
In this sense, quantum mechanics can be seen as a generalization of classical probability theory.
Moreover, registers will become important to describe the outcomes of measurements, 
as we will see in the following.

\subsection{Separated Systems}

Consider two separate quantum systems, $A$ and $B$, modeled by a Hilbert space $ℋ¬A$ and $ℋ¬B$, 
respectively. We can prepare $d¬A$ different perfectly distinguishable states on system $A$ and 
$d¬B$ different perfectly distinguishable states on system $B$. On the joint system $AB$, we can 
thus prepare $d¬A d¬B$ different perfectly distinguishable states.
Hence, according to Postulate~\ref{po:system}, the dimension of the ☼[joint system]{system!joint} 
is $d¬{AB} = d¬A d¬B$ and it can be modeled using the Hilbert space 
$ℋ¬{AB} \equiv ℋ¬A ⨂ ℋ¬B$ of dimension $d¬A d¬B$. ☼[Joint states]{state!joint} of the two systems are then described by normalized positive semi-definite operators 
$ρ¬{AB} ∈ \onorm{ℋ¬{AB}}$.

Given a state $ρ¬{AB}$ on the joint system, what are the states of the individual systems $A$ and 
$B$? As we have seen before, the state of the system $A$ is fully characterized by the linear 
functionals acting on it, i.e.\ the functionals $\tr{ρ¬{AB} L¬A} = \tr{ρ¬{AB} (L¬A ⨂ ⅈ¬B)} = \tr
{ \tr[B]{\rho¬{AB}} L¬A }$. Here, we introduced the notation $\trace¬B \equiv \trace¬{ℋ¬B}$ for
the partial trace over the subsystem $B$. Since the state on $A$ is thus fully characterized by linear functionals on the operator $\tr[B]{ρ¬{AB}}$, we define the marginal state or ☼{marginal} of $ρ¬{AB}$ on $A$ as $ρ¬A := \tr[B]{ρ¬{AB}}$. Similarly, we define the marginal of $ρ¬{AB}$ on 
$B$ as $ρ¬B := \tr[A]{ρ¬{AB}}$. In the following, the marginals are introduced 
implicitly with the joint state.\footnote{Whenever we introduce a state (e.g.\ $ρ¬{ABC}$) on a joint system, we also 
implicitly introduce all its marginals (e.g.\ $ρ¬{AB} = \tr[C]{ρ¬{ABC}}$ and $ρ¬C = \tr[AB]{ρ¬{ABC}}$).}

\subsubsection{Entanglement}

We call a state $ρ¬{AB}$ of a joint quantum system $AB$ ☼{separable} if it can be 
written in the form
\begin{align}
  ρ¬{AB} = ∑_k σ¬A^k ⨂ τ¬B^k\,, \tn*{where} σ¬A^k ∈ \opos{ℋ¬A} \tn{and} 
  τ¬B^k ∈ \opos{ℋ¬B} ¶[sep-state]\,.
\end{align}
Otherwise, it is ☼{entangled}. The occurrence of entangled states is one of the
most intriguing properties of the formalism of quantum mechanics. 

The prime example of an entangled state is the ☼[maximally
entangled]{entangled} state 
\begin{align}
  \ket{γ}[AA'] = \frac{1}{\sqrt{d¬A}} ∑_i \ket{e_i}[A] ⨂ \ket{e_i}[A'] ¶\,,
\end{align} 
where $\{ \ket{e_i} \}$ is an orthonormal basis of $ℋ¬A \iso ℋ¬{A'}$. This state
cannot be written in the form~§[sep-state] as the following argument, due to Peres and 
Horodecki~\cite{peres96,horodecki96}, shows.
Consider the super-operator $\sT : ρ¬{A'} ↦ ρ¬{A'}ᵀ$, where the transpose is taken with regards to the basis $\{ \ket{e_i} \}$ of $A'$. This super-operator is called the ☼{partial 
transpose} (on $A'$) and it is positive but not completely positive. Applied to separable 
states~§[sep-state], it always results in positive semi-definite states 
$∑_k σ¬A^k ⨂ \big(τ¬B^k\big)ᵀ$.
Applied to $γ¬{AA'}$, however, we get
\begin{align}
  \sT[γ¬{AA'}] = \frac{1}{d¬A} ∑_{i,j} \ket{e_i}\!\bra{e_j} ⨂ \sT \big[\ket{e_i}\!\bra{e_j} \big]
    = \frac{1}{d¬A} ∑_{i,j} \ket{e_i}\!\bra{e_j} ⨂ \ket{e_j}\!\bra{e_i} ¶\,.
\end{align}
This operator is not positive semi-definite. For example, we have
\begin{align}
  \braket{φ|\sT[γ¬{AA'}]|φ} = -\frac{2}{d¬A}, \tn*{where} 
    \ket{φ} = \ket{e_1} ⨂ \ket{e_2} - \ket{e_2} ⨂ \ket{e_1} ¶\,.
\end{align}
Generally, we have seen that a bipartite state is separable only if it remains positive semi-definite under the partial transpose. 
The converse is not true in general.

\subsubsection{Classical-Quantum Systems}

Joint systems where one (or more) subsystems are classical registers are of particular importance.
Consider, for example, the case where a classical register $X$ is described jointly with
a quantum system $A$. The possible ☼[joint states]{state!joint} $ρ¬{XA}$ can be written as\footnote{This describes the full set of states that have a classical marginal on $X$.}
\begin{align}
  ρ¬{XA} = ∑_x P¬X(x)\, \proj{x}[X] ⨂ τ¬A^x\,, \tn*{where} τ¬A^x ∈ \onorm{ℋ¬B} ¶[cq-state]
\end{align}
and $P¬X$ is a probability distribution on $\cX$. These states are called ☼[classical-quantum]{state!CQ} (CQ) states. They are of the form~§[sep-state] and, thus, separable.

A special case occurs when the two classical registers, $X$ and $Y$, are considered jointly. 
In this case, the 
states $ρ¬{XY}$ are of the form
\begin{align}
  ρ¬{XY} &= ∑_{x,y} P¬{XY}(x,y)\, \proj{x}[X] ⨂ \proj{y}[Y] ¶\\
    &= ∑_x P¬X(x)\, \proj{x}[X] ⨂ ∑_y P¬Y^x(y)\, \proj{y}[Y] ¶,
\end{align}
where $P¬{XY}$ is a probability distribution on $\cX \times \cY$ and, for each $x$, $P¬Y^x(y) = P¬{XY}(x,y)/P¬X(x)$ is
the conditional probability distribution on $\cY$ given a fixed $x ∈ \cX$. Such states allow the description of
arbitrarily correlated classical random variables.

\subsubsection{Purifications and Extensions}
☼*{state!pure}

For any state $ρ¬A$ of a system $A$, we can find a ☼{purification} on an auxiliary
system $A'$ with $ℋ¬{A} \iso ℋ¬{A'}$. A purification is a pure state $\ket{ρ} ∈ ℋ¬{AA'}$ 
of the joint system $AA'$ with the property that $ρ¬A = \tr[A']{ρ¬{AA'}}$. More specifically, if $ρ¬A = ∑_i λ_i\, \proj{e_i}[A]$ is the eigenvalue decomposition of $ρ¬A$, then a possible purification is given by $\ket{ρ}[AA'] = ∑_i √{λ_i}\, \ket{e_i}[A] ⨂ \ket{e_i}[A']$.
Purifications of $ρ¬A$ are separable if and only if $ρ¬A$ is pure.
More generally, we call
a (not necessarily pure) state $ρ¬{AA'}$ that satisfies $\tr[A']{ρ¬{AA'}} = ρ¬A$ an ☼{extension} of the
state $ρ¬A$. 

Purifications of CQ states of the form~§[cq-state] can be constructed as
\begin{align}
  \ket{ρ}[XX'AA'] = ∑_x \sqrt{P¬X(x)}\, \ket{x}[X] ⨂ \ket{x}[X'] ⨂ \ket{τ^x}[AA'] ¶\,,
\end{align}
where $\ket{τ^x}$ is a purification of $τ^x$ on $AA'$. We call the two registers
$X$ and $X'$ ☼{coherent classical}. In particular, the marginal states $ρ¬{XAB}$ and
$ρ¬{X'AB}$ are classical on $X$ and $X'$, respectively, and they are equal with regards to the isomorphism $\ket{x}[X] ↦ \ket{x}[X']$.

\subsection{Evolutions and Measurements}

\subsubsection{Evolution}

The ☼{evolution} of a separate quantum system is most generally described by a ☼{quantum channel}. A quantum channel is a linear map (i.e.\ a super-operator) from quantum states on a system $A$ to quantum states 
on a system $B$. Since such super-operators map quantum states onto quantum states, 
they must necessarily be positive and trace-preserving. Moreover, since quantum channels describe the evolution of quantum systems that may be 
part of a larger joint system, they are required to map positive semi-definite states
of any joint system to positive semi-definite states. This implies that they
are completely positive and, thus, ☼[TP-CPMs]{TP-CPM}.

\begin{postulate}[Evolution]
  The ☼{evolution} of quantum systems is described by 
  trace-preserving completely positive maps.
\end{postulate}

An important example of such a map is a time evolution. Here, system $A$ is any system at time $t_0$ and system $B$ the same system at a later time $t_1 > t_0$.
In the traditional treatment of quantum mechanics, the time evolution of a system is 
described by a ☼{unitary} evolution that is induced by the ☼{Hamiltonian} of the system. A unitary evolution is a special case of a TP-CPM and describes the evolution of a closed system, i.e.\ a system that does not interact with any other system.

\subsubsection{Measurement}

A ☼{measurement} of a quantum systems can be described in the above framework of general evolutions.

\begin{postulate}[Measurement]
  A quantum ☼{measurement} is a trace-pre\-ser\-ving completely positive map from
  a system to a classical register containing the measurement outcome and a 
  system that contains the state of the system after measurement.
\end{postulate}

Let $A$ be a quantum system, $X$ the classical ☼{register} 
containing the measurement result and $A'$
the system modeling the quantum system after measurement.
The corresponding measurement map, $ℳ ∈ \tpcpm{ℋ¬A, ℋ¬{XA'}}$, with $ℋ¬{A'} \iso ℋ¬A$, has a Kraus decomposition
\begin{align}
  ℳ : ρ ↦ ∑_k E_k ρ\, E_k†, \tn*{where} E_k ∈ \olin{ℋ¬A, ℋ¬{XA'}} 
  \tn{and} ∑_k E_k† E_k = ⅈ¬A ¶.
\end{align}
Since the resulting state is required to be classical on the register $X$, we further know that
\begin{align}
  ℳ[ρ¬A] = ∑_x P¬X(x)\, \proj{x} ⨂ τ¬{A'}^x 
    \tn*{where} τ¬{A'}^x ∈ \onorm{ℋ¬{A'}} ¶[m/tpcpm].
\end{align}
Here, $P¬X(x)$ is the ☼{probability} that the outcome ``$x$'' is measured and $τ¬{A'}^x$ is the
state of the system conditioned on the event that the outcome ``$x$'' has been measured. 
Due to~§[m/tpcpm], the Kraus operators $E_k$ necessarily have the form $E_k = \ket{x_k} ⨂ F_k$,
where $F_k ∈ \olin{ℋ¬A, ℋ¬{A'}} \iso \olin{ℋ¬A}$ and, thus,
\begin{align}
  τ¬{A'}^x = \frac{1}{P(x)} ∑_{k: x_k = x} F_k ρ¬A F_k† \tn*{and}
    P¬X(x) = \tr<B>{∑_{k: x_k = x} F_k† F_k\, ρ} ¶.
\end{align}
The second equality follows from the fact that $τ¬{A'}^x$ is normalized.

Hence, a measurement is fully specified by the operators $\{ F_k \}$ and the partitioning 
$\{ k : x_k = x \}$. Note that this viewpoint is consistent; specifically, $P¬X$ is a probability
distribution since 
\begin{align}
  0 ≤ ∑_{k: x_k = x} F_k† F_k ≤ ⅈ¬{A} \tn*{and} 
    ∑_x ∑_{k: x_k = x} F_k† F_k = ∑_k F_k† F_k = ⅈ¬A ¶.
\end{align}

Often we are not interested in the specific state after measurement but only the probability 
distribution the measurement induces on the register $X$. In this case, the
measurement is fully characterized by the operators $M_x = ∑_{k: x_k = x} F_k F_k†$, which define the probability $P¬X(x) = \tr{M_x ρ}$. The set $\{ M_x \}$ is called a ☼[positive operator-valued 
measure]{POVM} (POVM).
\begin{definition}[POVM]
  \label{df:POVM}
  A positive operator-valued measure on a quantum system $A$ is a set 
  $\{ M_x \}$ with $M_x ∈ \opos{ℋ¬A}$ and $∑_x M_x = ⅈ¬A$. The corresponding
  measurement TP-CPM $ℳ : ℋ¬A → ℋ¬X$ is
  given by $ρ ↦ ∑_x \tr<b>[A']{M_x\, ρ}\,\proj{x}$
  and the operators $M_x$ are called ☼[POVM elements]{POVM!element}.
\end{definition}
☼*{POVM!element}
We often use the following ☼{Stinespring dilation} of POVM measurements. The isometry
$U ∈ \ouni{ℋ¬A, ℋ¬{A'XX'}}$ maps the state to the classical register $X$ containing the measurement outcome, a ☼{coherent classical} copy of it, $X'$, 
and a possible post-measurement state on $A'$. The post-measurement states that result 
from a POVM are not unique; however, conventionally, one takes $x_k = k$ where $k$ is chosen 
from the same set as $x$ and $F_k = √{M_x}$. Thus,
\begin{align}
  ℳ[ρ] = \tr[X'A']{U ρ\, U†} \tn*{where} U = ∑_x \ket{x}[X] ⨂ \ket{x}[X'] ⨂ √{M_x} ¶\,.
\end{align}

A special case of a POVM occurs when the ☼[POVM elements]{POVM!element} are projectors, i.e.\ $M_x = M_x^2$. We call such a measurement a ☼[projective measurement]{measurement!projective}.
If these projections are of rank $1$, e.g.~$M_x = \proj{x}[X]$, the corresponding measurement TP-CPM takes on the simple
form $ℳ: ρ¬A ↦ ∑_x \braket{x|ρ¬A|x}\, \proj{x}[X]$.

Comparing this to the traditional treatment of quantum measurements, we note that it is not necessary to introduce a separate formalism for measurements and that we can treat measurements as a special case of an evolution. 
This is possible since
we always consider the outcome of a measurement as a ☼{random variable} that is correlated with the state of the system after measurement and do not condition the resulting quantum state on a particular measurement outcome.

\section{Mathematical Toolkit}
\label{se:pre/math}

This section covers the most important mathematical tools used throughout this thesis.
The (smooth) min- and max-entropy introduced in Chapter~\ref{ch:entropies} can be formulated as semi-definite programs. Operator monotone functions are used to explore properties of another class of entropies, generalizations of Rényi entropies, in Chapter~\ref{ch:aep}.

\subsection{Semi-Definite Programs}
\label{sc:premath/sdp}
☼*{semi-definite program|see{SDP}}

This overview is based on the Watrous Lecture Notes~\cite{watrous-ln08} and all proofs can be found there. A ☼[semi-definite program]{SDP} (SDP) is a triple $\{A, B, Ψ \}$, where $A ∈ \oherm{ℋ}$, $B ∈ \oherm{ℋ'}$ and $Ψ ∈ \olin{\olin{ℋ}, \olin{ℋ'}}$ is a super-operator from $ℋ$ to $ℋ'$ that preserves Hermiticity. The following two optimization problems are associated with the semi-definite program:
\begin{align}
  \begin{array}{rlcrl}
    \multicolumn{2}{c}{\underline{\tn{primal problem}}} & \qquad & 
      \multicolumn{2}{c}{\underline{\tn{dual problem}}} \vspace{0.2cm} \\
    \texttt{minimize}: & \ip{A, X} & & \texttt{maximize}: & \ip{B, Y} \\
    \texttt{subject to}: & Ψ[X] ≥ B & & \texttt{subject to}: & Ψ†[Y] ≤ A \\
    & X ∈ \opos{ℋ} &&& Y ∈ \opos{ℋ'}
  \end{array} ¶
\end{align}
We call an operator $X ∈ \opos{ℋ}$ primal feasible if it satisfies $Ψ[X] ≥ B$. Similarly, we say that $Y ∈ \opos{ℋ'}$ is dual feasible if $Ψ†[Y] ≤ A$. Moreover, we denote the optimal solution of the primal problem with $α$ and the optimal solution of the dual problem with $β$. Formally,
\begin{align}
  α &:= \inf \big\{ \ip{A, X} : X ∈ \opos{ℋ}, Ψ[X] ≥ B \big\} ¶\\
  β &:= \sup \big\{ \ip{B, Y} : Y ∈ \opos{ℋ'}, Ψ†[Y] ≤ A \big\} ¶[sdp/alphabeta].  
\end{align}

The following two theorems provide a relation between the primal and dual problems of an SDP.
\begin{theorem}[Weak Duality]
  \label{th:sdp/weak}
  Let $\{A, B, Ψ\}$ be a SDP and $α$, $β$ defined as in~§[sdp/alphabeta]. Then, $α ≥ β$.
\end{theorem}
This implies that every dual feasible operator $Y$ provides a lower bound of 
$\ip{B, Y}$ on $α$ and every primal feasible operator $X$ provides an upper bound of 
$\ip{A, X}$ on $β$.

\begin{theorem}[Strong Duality] 
  \label{th:sdp/strong}
  Let $\{A, B, Ψ\}$ be a SDP and $α$, $β$ 
  defined as in~§[sdp/alphabeta]. Then the following holds:
  \begin{itemize}
    ☛ If $α$ is finite and there exists an operator $Y > 0$ such that $Ψ†[Y] < A$, then $α = β$ 
      and there exists a primal feasible $X$ such that $\ip{A, X} = α$.
    ☛ If $β$ is finite and there exists an operator $X > 0$ such that $Ψ[X] > B$, then $α = β$
      and there exists a dual feasible $Y$ such that $\ip{B, Y} = β$.
  \end{itemize}
\end{theorem}

Optimization problems that can be formulated as semi-definite programs can be efficiently solved numerically.\footnote{For example, using the SeDuMi solver~\cite{SEDUMI} and YALMIP front-end~\cite{YALMIP}.}

\subsection{Operator Monotone Functions}
\label{sc:opmono}

Here, we discuss some useful properties of operator ☼{monotone}, ☼{concave} and ☼{convex} functions. This section is largely based on Chapter V of Bhatia~\cite{bhatia97} and we simply repeat the results here in the form we need in later chapters.

Operator monotone functions preserve the partial order on operators induced by `$≥$' and are, thus, necessarily monotone.

\begin{definition}
  Let $Ω \subseteq ℝ$. A function $f: Ω → ℝ$ is operator monotone on $Ω$ 
  if $A ≥ B$ implies $f(A) ≥ f(B)$ for any Hermitian operators $A, B$ with eigenvalues in $Ω$.
\end{definition}

Similarly, operator concave and convex functions generalize the concept of concavity and convexity to operators.

\begin{definition}
  Let $Ω$ be an interval on $ℝ$. 
  A function $f: Ω → ℝ$ is operator concave on $Ω$ if 
  $f \big( µ A + (1-µ) B) \big) ≥ µ f(A) + (1-µ) f(B)$ for all 
  Hermitian operators $A, B$ with eigenvalues in $Ω$ and all $µ ∈ [0, 1]$.
  The function $f$ is operator convex on $Ω$ if $-f$ is operator concave on $Ω$.
\end{definition}

Prominent examples of such functions include (cf.~Chapter V in~\cite{bhatia97})
\begin{itemize}
  ☛ The logarithm function, which is operator monotone on $ℝ^+$.
  ☛ The function $h : t ↦ - t \log t$ with its extension to $h(0) = \lim_{t → 0} h(t) = 0$, 
  which is operator concave on $ℝ_0^+$.
  ☛ The family of functions $g_α : t ↦ t^α$. These functions are operator concave and operator
  monotone on $ℝ_0^+$ for $α ∈ (0,1]$ and operator convex on $ℝ_0^+$ for $α ∈ [1, 2]$.
  (Note that these functions are convex but not operator convex if $α > 2$.)
\end{itemize}

We start with a straightforward application of Jensen's inequality:
\begin{lemma}
  \label{lm:jensen}
  Let $Ω$ be an interval on $ℝ$ and let $f: Ω → ℝ$ be concave on $Ω$. Then, for any $\ket{φ} 
  ∈ ℋ$ with $\norm{φ} = 1$ and $A ∈ \oherm{ℋ}$ with eigenvalues in $Ω$, we have 
    $\braket{φ | f(A) | φ} ≤ f ( \braket{φ|A|φ} )$.
\end{lemma}

\begin{proof}
  Using the eigenvalue decomposition $A = ∑_i λ_i \proj{e_i}$, we get
  \begin{align}
    \braket{φ | f(A) | φ} 
      = ∑_i f(λ_i)\, \abs{ \braket{φ | e_i} }^2 
        ≤ f \Big( ∑_i λ_i\, \abs{ \braket{φ | e_i} }^2 \Big)  
        = f \big( \braket{φ | A | φ} \big) ¶,
  \end{align}
  where we used that $∑_i \abs{\braket{φ | e_i}}^2 = 1$.
\end{proof}

A comprehensive generalization of Jensen's inequality to operator convex functions is given in~\cite{hansen03}. We state a specialized version here for completeness.
\begin{lemma}[Operator Jensen's Inequality]
  \label{lm:opjensen}
  Let $Ω$ be an interval on $ℝ$ and let $f: Ω → ℝ$ be continuous and operator concave on $Ω$. Then, 
  for any isometry $U: ℋ → ℋ'$ and $A ∈ \oherm{ℋ}$ with eigenvalues in $Ω$, 
  we have
  \begin{align}     
    U f(A) U† ≤ f (U A U†) ¶\,.
  \end{align}
\end{lemma}

%............................

\chapter{The Purified Distance}
\label{ch:pd}

%............................

% short notation for fidelity (expressed as a 1-norm)
\newcommand{\fid}[2]{\norm<?>{√{#1}√{#2}}[1]}
\newcommand{\fidn}[2]{\norm<n>{√{#1}√{#2}}[1]}
\newcommand{\fidb}[2]{\norm<b>{√{#1}√{#2}}[1]}
% complementary projector
\newcommand{\Pip}{Π^{\textnormal{\tiny $\perp$}}}

This chapter is based on~\cite{tomamichel09}, where the purified distance 
was first proposed as a metric on the space of sub-normalized quantum states. 
The usefulness of the purified distance will become apparent when it is applied 
to define the smooth min- and max-entropies in Chapter~\ref{ch:entropies},
providing them with natural properties such as invariance under local isometries
and various data processing inequalities.

\section{Introduction and Related Work}

Smooth ☼[entropies]{entropy}, evaluated for a quantum state $ρ$, are
defined indirectly via an optimization (either a maximization or a minimization)
of an underlying unsmoothed entropy over a set of states that are $ε$-close to
$ρ$, where $ε$ is a small ☼{smoothing parameter}. The resulting quantities are
called ☼{$ε$-smooth} entropies. (See, for example, Chapter~\ref{ch:entropies}, 
where the $ε$-smooth min- and max-entropies are defined in this way.)

Consequently, various definitions of such sets of close
states\,|\,subse\-quently called ☼[$ε$-balls]{$ε$-ball}\,|\,have appeared in the
literature.
However, to the best of our knowledge, none of the existing definitions
simultaneously exhibit the following two properties that are of particular
importance 
in the context of smooth entropies in the quantum regime.

\begin{itemize}
☛ Firstly, the definition of the $ε$-smooth entropies should 
be independent of the Hilbert spaces used to
represent the quantum state $ρ$. 

In particular, embedding $ρ$ into a larger Hilbert space prior to smoothing should leave the $ε$-smooth entropies unchanged. Note that, in general, embedding $ρ$ into a larger Hilbert space offers more flexibility for smoothing as more dimensions orthogonal to the support of $ρ$ become available for the optimization. Indeed, for some $ε$-balls that contain only normalized states, smoothing outside the support becomes advantageous and the smooth entropy thus depends on the Hilbert space representation of $ρ$.

We can avoid this problem by including sub-normalized quantum states in the $ε$-balls. 

☛ Secondly, it will be important that we can define a
ball of pure states that contains purifications of all the states in
a given $ε$-ball. 
This will allow us to establish duality relations
between smooth entropies and is achieved by using a
fidelity-based metric to determine $ε$-closeness.
\end{itemize}

\subsection{Main Contributions}

The following sections introduce a new metric on sub-normalized quantum states, the
☼{purified distance}. We call two quantum states
$ρ$ and $τ$ $ε$-close, denoted $ρ \ecl τ$, if and only if the purified distance between
them is at most $ε$. The purified distance has various interesting properties,
among them are the following.
\begin{result}[Purified Distance]
  The ☼{purified distance} is a metric on sub-normalized states and has the
  following properties:
  \begin{compact}
    ☛ If at least one of the states is normalized, it can be expressed in terms
      of the fidelity as $P(ρ,τ) = √{1 - F(ρ,τ)^2}$.
    ☛ It is an upper bound to the ☼{trace distance}.
    ☛ For any trace non-increasing CPM $ℰ$ and any states $ρ$ and $σ$, we have 
      $ρ \ecl τ ⇒ ℰ[ρ] \ecl ℰ[τ]$.      
    ☛ If $ρ¬{AB}$ is a state and $σ¬A \ecl ρ¬A$ is close to its marginal,
      then there always exists an extension $σ¬{AB}$ of $σ¬A$ with $σ¬{AB} \ecl
      ρ¬{AB}$.
  \end{compact}
\end{result}

\subsection{Outline}

Section~\ref{sc:pd/math} introduces two metrics on the set of sub-normalized states, the
generalized trace distance and the purified distance.
In Section~\ref{sc:pd/prop}, we discuss various properties of the purified distance, including
its relation to the generalized trace distance, its monotonicity under trace non-increasing
CPMs and an adaption of Uhlmann's theorem to the purified distance.
Section~\ref{sc:pd/not} then explains some notational conventions, made possible by the use of the purified distance as a metric, that should help the reader
through the remaining chapters.

\section{Two Metrics for Quantum States}
\label{sc:pd/math}

For the remainder of this chapter, let $ℋ$ be an arbitrary
finite-dimensional Hilbert space.\footnote{Note also that most results of this
chapter have recently been generalized to the framework of general von Neumann
algebras~\cite{furrer11}.}
The two most common measures of distance between normalized quantum states
are the ☼{trace distance} and the ☼{fidelity}. 

\subsection{Generalized Trace Distance}

We start by 
introducing a straight-forward generalization of the trace distance to sub-normalized
quantum states.

\begin{definition}
  \label{df:gtd}
  For $ρ, τ ∈ \osub{ℋ}$, we define the ☼[generalized
  trace distance]{trace distance} between $ρ$ and $τ$ as
  \begin{align}
    D(ρ, τ) := \max \big\{ \trace \{ ρ - τ \}_+ , 
      \trace \{ τ - ρ \}_+ \big\} ¶\,.
  \end{align}
\end{definition}

The generalized trace distance can be expressed alternatively in terms of the
Schatten $1$-norm as
\begin{align} 
  D(ρ, τ) = \frac{1}{2} \norm{\rho - \tau}[1] + 
    \frac{1}{2} \abs<b>{ \trace\rho - \trace\tau } ¶\,
\end{align}
and it is easy to verify that it is a ☼{metric} on $\olin{ℋ}$. 
In the case where both $\rho$ and $\tau$ are normalized states, we recover the
usual definition of the trace distance, $D(ρ, τ) := \frac{1}{2} \norm{ρ -
τ}[1]$. Furthermore, the trace distance has a physical interpretation as the
☼{distinguishing advantage} between the two states. In other words, the
probability $p_{\textrm{dist}}(ρ, τ)$ of correctly distinguishing between two
equiprobable states $ρ$ and $τ$ by any measurement is upper bounded
by~\cite{nielsen00}
\begin{align}
  p_\textrm{dist}(ρ, τ) ≤ \frac{1}{2} \big( 1 + D(ρ, τ) 
    \big) ¶[dist-adv].
\end{align}

\subsection{Generalized Fidelity}
☼*{generalized fidelity|see{fidelity}}

On the other hand, various metrics are derived from the ☼{fidelity}, which is 
given as $F(ρ, τ) = \fidn{ρ}{τ}$ for two normalized states $ρ$ and $τ$. We will
not use the letter $F$ to denote the fidelity hereafter, instead reserving it
for the generalized fidelity defined below. (We will also see that 
the two quantities agree if at least one state is normalized.)
The fidelity has many interesting properties, some of which we will list here
for further reference. The most important properties of the fidelity are summarized in
Table~\ref{tb:fid}.

\begin{table}[!h]
  \rule{\textwidth}{1pt}
  \begin{minipage}{0.98\textwidth}
    \vspace{5pt}
    The following holds for arbitrary positive operators $ρ, τ, σ ∈ \opos{ℋ}$.
    \begin{enum}
☛[i.] The fidelity is ☼{symmetric} in its arguments, $\fidn{ρ}{τ} =
		\fidn{τ}{ρ}$.
☛[ii.] The fidelity is monotonically increasing under the application of TP-CPMs
    (cf.~e.g.~\cite{nielsen00}, Theorem 9.6). This
    implies that, for any TP-CPM $ℰ$, we have
    \begin{align}
      √{\trace\rho}√{\trace\tau} ≥ \fid{ℰ[ρ]}{ℰ[τ]} ≥ \fid{ρ}{τ} ¶[fid/mon].
    \end{align}
    To get the first inequality, we used that $\trace$ is a TP-CPM.
☛[iii.] ☼{Uhlmann's theorem}~\cite{uhlmann85} states that, for any purification $φ$ of
    $ρ$,
    \begin{align}
      \fidn{ρ}{τ} = \max_{\ket{θ}} \abs{\braket{φ|θ}} ¶[fid/uhl],
    \end{align}
    where the maximum is taken over all ☼[purifications]{purification} $θ$ of $τ$. 
☛[iv.] For any ☼{projector} $Π ∈ \opos{ℋ}$, we have
		\begin{align}
			\fidb{Π ρ Π}{τ} = \fidb{Π ρ Π}{Π τ Π} = \fidb{ρ}{Π τ Π} ¶[fid/proj].
		\end{align}
☛[v.] For any $σ ≥ ρ$, we have $\fidn{σ}{τ} ≥ \fidn{ρ}{τ}$.
☛[vi.] For states $ρ = ρ_1 ⨁ ρ_2$ and $τ = τ_1 ⨁ τ_2$, where $ρ_1, τ_1 ∈ \opos{ℋ_1}$ and
    $ρ_2, τ_2 ∈ \opos{ℋ_2}$, we have
    \begin{align}
      \fidb{ρ}{τ} = \fidb{ρ_1}{τ_1} + \fidb{ρ_2}{τ_2} ¶.
    \end{align}
  \end{enum}
  \end{minipage} 
  \vspace{5pt} \\
  \rule{\textwidth}{1pt}
  \vspace{-20pt}
  \caption[Properties of the Fidelity.]{\emph{Properties of the Fidelity.}}
  \label{tb:fid}
\end{table}

Here, we propose a generalization of the fidelity to sub-normalized
states. The generalization is motivated by the observation that sub-nor\-ma\-li\-zed
states can be thought of as normalized states on a larger space which
are projected onto a subspace. Hence, we define the generalized fidelity as the
supremum of the fidelity between such normalized states.
\begin{definition}
  \label{df:gfid}
  For $ρ, τ ∈ \osub{ℋ}$, we define the ☼[generalized
  fidelity]{fidelity!generalized} between $ρ$ and $τ$ as
  \begin{align}
    F(ρ, τ) := \, \sup_{ℋ'} \sup_{\rhob,\taub ∈ \onorm{ℋ'}} 
      \fid{\rhob}{\taub} ¶[gfid]\,,
  \end{align}
  where the supremum is taken over all embeddings $V$ of $ℋ$ into $ℋ'$
  and all states $\rhob, \taub ∈ \onorm{ℋ'}$, such that $ρ$ and $τ$ are
  images of $\rhob$ and $\taub$ under $V†$. (Namely, the states
  satisfy $V† \bar{\rho}\, V = ρ$ and $V† \bar{\tau}\, V = τ$.)
\end{definition}

This expression reduces to the fidelity when at least one state is normalized.
To see this, consider the following alternative expression for the generalized
fidelity.
\begin{lemma}
  \label{lm:alt-gfid} 
  Let $ρ, τ ∈ \osub{ℋ}$. Then,
  \begin{align}
    F(ρ, τ) = F(\rhoh,\tauh) = \fidb{ρ}{τ} + 
      √{(1 - \trace ρ)(1 - \trace τ)} ¶[alt-gfid]\,,
  \end{align}
  where $\rhoh := ρ ⨁ (1 \!-\! \trace ρ)$ and $\tauh := τ ⨁ (1 \!-\!
\trace τ)$.
\end{lemma}
\begin{proof}
  Let $U: ℋ → ℋ'$, $\rhob$ and $\taub$ be an arbitrary candidate for the
  supremum in~§[gfid]. Moreover, let $ℰ$ be the pinching $ρ ↦ Π ρ Π + \Pip ρ
  \Pip$, where $\Pi := UU†$ projects onto the image of $U$ and $\Pip := ⅈ
  - Π$ is its orthogonal complement on $ℋ'$.
  Then, due to the monotonicity property~§[fid/mon], we find
  \begin{align}
    \fid{\rhob}{\taub} &≤ \fid{ℰ[\rhob]}{ℰ[\taub]} ¶\\
    &= \fid{Π \rhob Π}{Π \taub Π} + \fid{\Pip \rhob
			\Pip}{\Pip \taub \Pip} ¶\\
    &≤ \fid{\rho}{\tau} + \sqrt{(1-\trace\rho)(1-\trace\tau)} ¶\,.
  \end{align}
  In particular, the r.h.s.~is an upper bound on $F(ρ, τ)$.
  Finally, it is easy to verify that this upper bound is achieved with 
  the choice $ℋ' \iso ℋ ⨁ ℂ$ as well as $\rhoh$ and $\tauh$.
\end{proof}

\subsection{Purified Distance}

Next, we define a distance measure based on the fidelity, analogous to the one
proposed in~\cite{nielsen04, rastegin06}\footnote{The quantity
  $C(ρ, τ) = \sqrt{1 - F^2(ρ, τ)}$ is introduced
  in~\cite{nielsen04}, where the authors also show that it is a metric on
  $\onorm{ℋ}$. In~\cite{rastegin06}, the same quantity is called sine
  distance and some of its properties are explored.}.
\begin{definition}[Purified Distance]
  \label{df:pd}
  For $ρ, τ ∈ \osub{ℋ}$, we define the
  ☼{purified distance} between $ρ$ and $τ$ as
  \begin{align}
    P(ρ, τ) := √{1 - F(ρ, τ)^2} ¶.
  \end{align}
\end{definition}

Other distance measures based on the fidelity have been investigated in the literature. In particular, the ☼{Bures metric}~\cite{bures69}, 
$B(ρ, τ)^2 := 2 (1 - F(ρ, τ))$, and the ☼{angular distance}, $A(ρ, τ) := \arccos F(ρ, τ)$~\cite{nielsen00}. We prefer the purified distance because it constitutes an upper bound on the trace distance as we will see below and thus inherits its operational interpretation as an upper bound on the ☼{distinguishing advantage}.

The name ``Purified Distance'' is motivated by the fact that, for normalized states $ρ,
τ \in \onorm{ℋ}$, we can write $P(ρ, τ)$ as the minimum
trace distance between purifications $\ket{φ}$ of $ρ$ and
$\ket{θ}$ of $τ$.  More precisely, using Uhlmann's
theorem~§[fid/uhl], we have
\begin{align*}
  P(ρ, τ) &= \sqrt{1 - {F(ρ, τ)}^2}
  = √{1 - \max_{φ, θ} \abs{\braket{φ|θ}}^2} \\
  &= \min_{φ, θ} √{1 - \abs{\braket{φ|θ}}^2}
  = \min_{φ, θ} D(φ, θ) \, .
\end{align*}

The purified distance is a ☼{metric} on the set of sub-normalized states 
according to Definition~\ref{df:metric}.

\begin{proposition}
  \label{pr:pd-metric}
  The purified distance $P(·,·)$ is a metric on $\osub{ℋ}$.
\end{proposition}

\begin{proof}
  Let $ρ, τ, σ ∈ \osub{ℋ}$.
  The condition $P(ρ, τ) = 0 \iff ρ = τ$ can be verified
  by inspection, and symmetry $P(ρ, τ) = P(τ, ρ)$ follows
  from the symmetry of the fidelity.

  It remains to show the triangle inequality, $P(ρ, τ) ≤
  P(ρ, σ) + P(σ, τ)$. Using Lemma~\ref{lm:alt-gfid}, the
  generalized fidelities between $ρ$, $τ$ and $σ$ can be
  expressed as fidelities between the corresponding extensions
  $\rhoh$, $\tauh$ and $\sigmah$. We employ the triangle inequality of the 
  ☼{angular distance}, which can be expressed in terms of the purified 
  distance as $A(ρ, τ) = \arccos F(ρ, τ) = \arcsin P(ρ, τ)$.\footnote{A proof that $A$ 
  is a metric for normalized states is given in~\cite{nielsen00}.}
  This leads to
  \begin{align}
    P(ρ, τ) &= \sin A(\rhoh, \tauh) ¶\\
      &≤ \sin \big( A(\rhoh, \sigmah) + A(\sigmah, \tauh) \big) ¶\\
      &= \sin A(\rhoh, \sigmah) \cos A(\sigmah, \tauh) + 
        \sin A(\sigmah, \tauh) \cos A(\rhoh, \sigmah) ¶[pd/sin-add]\\
      &= P(ρ, σ) F(σ, τ) + P(σ, τ) F(ρ, σ) ¶[pd/tight-triangle] \\
      &≤ P(ρ, σ) + P(σ, τ) ¶\,,
  \end{align}
  where we employed the trigonometric addition formula to get~§[pd/sin-add].
\end{proof}

Note that the purified distance is not an intrinsic metric, i.e.\ given two states $ρ$, $τ$ with $P(ρ, τ) ≤ ε$ it is in general not possible to find intermediate states $σ^λ$ with $P(ρ, σ^λ) = λε$ and $P(σ^λ, τ) = (1-λ) ε$. In this sense, the above triangle inequality is not tight. It is thus sometimes
useful to employ Eq.~§[pd/tight-triangle] instead. For example, given three states $ρ, τ, σ ∈ \osub{ℋ}$ and $0 ≤ ε, \bar{ε} ≤ 1$, we find that 
$P(ρ, σ) ≤ ε$ and $P(σ, τ) ≤ \bar{ε}$ implies
\begin{align}
  P(ρ, τ) ≤ ε √{1 - \bar{ε}^2} + \bar{ε} √{1 - ε^2}  
      ¶[pd/triangle-eps]
\end{align}
if $\arcsin ε + \arcsin \bar{ε} ≤ \frac{π}{2}$. 
This bound is plotted in Figure~\ref{fg:tight}.

\begin{figure}
  \centering \includegraphics[scale=1.0]{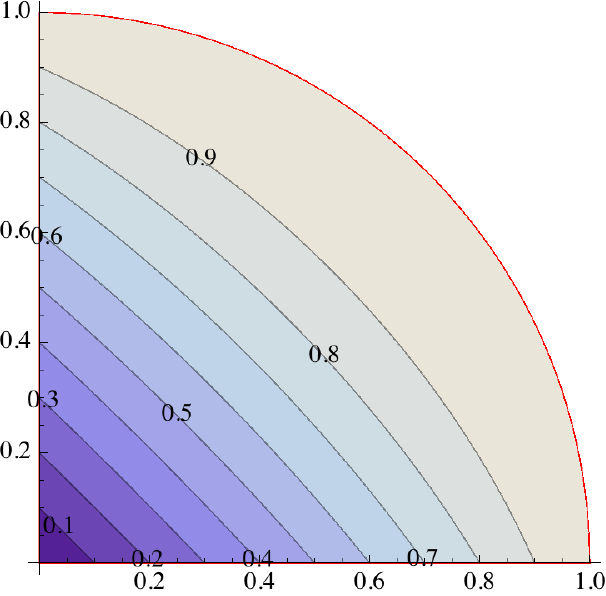}
  \caption[Improved Triangle Inequality of the Purified Distance.]{\emph{Improved Triangle Inequality.} Contour plot of the expression in~Eq.~§[pd/triangle-eps] 
  with axes $ε$ and $\bar{ε}$. The boundary contains all pairs $\{ ε, \bar{ε} \}$ that
  lead to the trivial bound $P(ρ, τ) ≤ 1$.
  This shows that Eq.~§[pd/tight-triangle] gives significantly tighter results than the triangle 
  inequality for the purified distance if $ε$ and $\bar{ε}$ get large. For example, if
  $P(ρ, σ) = P(σ, τ) = 0.5$, we find $P(ρ, τ) ≤ 0.87$ instead of the trivial $P(ρ, τ) ≤ 1$.}
  \label{fg:tight}
\end{figure}

\section{Properties of the Purified Distance}
\label{sc:pd/prop}

The purified distance has simple upper and lower bounds in terms of the generalized trace distance.
\begin{proposition}
  \label{pr:pd-gtd-bounds}
  Let $ρ, τ ∈ \osub{ℋ}$. Then 
  \begin{align}
    D(ρ, τ) ≤ P(ρ, τ) ≤ √{2 D(ρ, τ) - D(ρ, τ)^2} ≤ √{2 D(ρ, τ)} ¶\,.
  \end{align}
\end{proposition}
\begin{proof}
  We express the quantities using the normalized extensions $\rhoh$
  and $\tauh$ of Lemma~\ref{lm:alt-gfid} to get
  \begin{align}
    P(ρ, τ) &= √{1 - F(\rhoh, \tauh)^2} ≥ D(\rhoh, \tauh) = D(ρ, τ) \quad \tn{and} ¶\\
    P(ρ, τ)^2 &= 1 - F(\rhoh, \tauh)^2 ≤ 1 - \big(1 - D(\rhoh, \tauh) \big)^2 ¶\\
      &= 2D(ρ, τ) - D(ρ, τ)^2 ≤ 2 D(ρ, τ) ¶,
  \end{align}
  where we employed the Fuchs -- van de Graaf inequalities 
  $1 - F(\rhoh, \tauh) ≤ D(\rhoh, \tauh) ≤ \sqrt{1 - F(\rhoh,
    \tauh)^2}$~(cf.~\cite{fuchs96} and~\cite{nielsen00}, Section 9.2.3).
\end{proof}

One very useful property of the purified distance is that it does not
increase under simultaneous application of a quantum operation on both
states.  We consider the class of ☼{trace non-increasing} CPMs, which
includes ☼[TP-CPMs]{TP-CPM}, ☼[projections]{projector} and 
☼[partial isometries]{isometry!partial}.
\begin{theorem}[Monotonicity of Purified Distance]
  \label{th:pd-mono}
  Let $ρ, τ ∈ \osub{ℋ}$ and $ℰ: ℋ → ℋ'$ a trace non-increasing CPM. Then,
  \begin{align}
    F\big(ℰ[ρ], ℰ[τ]\big) ≥ F(ρ, τ) \tn*{and} 
    P\big(ℰ[ρ], ℰ[τ]\big) ≤ P(ρ, τ) ¶.
  \end{align}
\end{theorem}
\begin{proof}
  Remember that a trace non-increasing CPM $ℰ : ℋ → ℋ'$ can be decomposed into an isometry   
  $U: ℋ → ℋ' ⨂ ℋ''$ followed by a projection $Π ∈ \opos{ℋ' ⨂ ℋ''}$ and a 
  partial trace over $ℋ''$ (cf.~Lemma~\ref{lm:stinespring}).  
  Isometries and the partial trace are TP-CPMs and, hence, it suffices to show that $F(ρ,
  τ) ≤ F\big(ℰ[ρ], ℰ[τ]\big)$ for TP-CPMs and projections.

  First, let $ℰ$ be a TP-CPM. Using~§[alt-gfid] and the
  monotonicity under TP-CPMs of the fidelity~§[fid/mon], we see that
  \begin{align}
    F(ρ, τ) &= \fid{ρ}{τ} + √{(1-\trace ρ)(1-\trace τ)} ¶\\
    &≤ \fid{ℰ[ρ]}{ℰ[τ]} + √{(1-\trace ρ)(1-\trace τ)} ¶\\
    &= F\big(ℰ[ρ], ℰ[τ]\big) ¶.
  \end{align}

  Next, consider a projection $Π ∈ \opos{ℋ}$ and the CPM
  $ℰ: ρ ↦ Π ρ Π$. Following the definition of the generalized fidelity in~§[gfid], 
  we write $F(ρ, τ) = \sup\, \fid{\rhob}{\taub}$, where the supremum is taken over
  all normalized extensions $\{ ℋ', \rhob, \taub \}$ of $\{ \h, \rho, \tau
  \}$. Since all normalized extensions of $\{ \h, \rho, \tau \}$ are obviously also
  normalized extensions of $\big\{\!\supp{Π}, Π ρ Π, Π τ Π \big\}$,
  we find $F\big(Π ρ Π, Π τ Π\big) ≥ F(ρ, τ)$.
    
  Finally, the second statement trivially follows from the first one 
  by definition of the purified distance.
\end{proof}

The main advantage of the purified distance (and other metrics based on the fidelity) 
over the trace distance is that we can always find extensions and purifications without
increasing the distance. This is captured in the following two results. ☼*{Uhlmann's theorem}
\begin{theorem}[Uhlmann's Theorem for Purified Distance]
  \label{th:pd-uhl}
  Let $ρ, τ ∈ \osub{ℋ}$, $ℋ' \iso ℋ$ and $\ket{φ} ∈ ℋ ⨂ ℋ'$ be a purification 
  of $ρ$. Then, there exists a purification $\ket{θ} ∈ ℋ ⨂ ℋ'$ of $τ$ that satisifies 
  $P(ρ, τ) = P(φ, θ)$.
\end{theorem}
\begin{proof}
  We use Uhlmann's theorem for the fidelity~§[fid/uhl] to choose $\ket{θ} ∈ ℋ ⨂ ℋ'$ such
  that $\fid{ρ}{τ} = \abs{\braket{φ|θ}}$ holds. Then, due to~§[alt-gfid], we have 
  $F(ρ, τ) = F(φ, θ)$ as well as $P(ρ, τ) = P(φ, θ)$.
\end{proof}
\begin{corollary}
  \label{co:pd-ext}
  Let $ρ, τ ∈ \osub{ℋ}$ and $\rhob ∈ \osub{ℋ ⨂ ℋ'}$ be an extension of $\rho$. Then, 
  there exists an extension $\taub ∈ \osub{ℋ ⨂ ℋ'}$ of $τ$ with $P(ρ, τ) = P(\rhob,
  \taub)$.
\end{corollary}
\begin{proof}
  Let $ℋ'' \iso ℋ ⨂ ℋ'$ be an auxiliary Hilbert space and
  $φ ∈ ℋ ⨂ ℋ' ⨂ ℋ''$ be a purification of
  $\rhob$. We introduce a purification $θ ∈ ℋ ⨂ ℋ' ⨂ ℋ''$
  of $τ$ with $P(φ, θ) = P(ρ, τ)$ using Uhlmann's theorem for the purified distance above 
  and define $\taub = \tr[ℋ'']{θ}$. However, due to monotonicity
  (cf.\ Theorem~\ref{th:pd-mono}), we have $P(φ, θ) \geq P(\rhob, \taub) \geq P(ρ, τ)$, which implies
  that all three quantities must be equal.
\end{proof}

The next lemma offers an explicit construction that satisfies Corollary~\ref{co:pd-ext}.
This was shown in~\cite{dupuis10} and we provide the proof here for completeness.

\begin{lemma}
  \label{lm:pd-ext-constr}
  Let $ρ, τ ∈ \osub{ℋ}$ and let $\rhob ∈ \osub{ℋ ⨂ ℋ'}$ be an extension of $ρ$. Then, 
  there exists an operator $X ∈ \olin{ℋ}$ such that $\taub = X \rhob X† ∈ 
  \osub{ℋ ⨂ ℋ'}$ is an extension of $τ$ with $P(ρ, τ) = P(\rhob, \taub)$.
\end{lemma}

\begin{proof}
  We show the theorem for the case when $\rhob$ is pure. The general statement then follows 
  from the same arguments outlined in the proof of Corollary~\ref{co:pd-ext}.
  Let $U \abs{√{ρ}√{τ}}$ be the ☼[polar decomposition]{decomposition!polar} of $√{ρ}√{τ}$. 
  Then, we define $X := τ^{½} U ρ^{-½}$ using the ☼[generalized inverse]{inverse!generalized}. 
  Clearly. $\taub = X \rhob X†$ is an extension of $τ$ as $\tr[ℋ']{\taub} = 
  X \rho X† = τ$. Furthermore, we find
  \begin{align}
    \fidb{\rhob}{\taub} = \abs<b>{\braket{\rhob|\taub}} = \abs<b>{\braket{\rhob|X|\rhob}} 
    = \abs<b>{\tr{X ρ}} = \abs<b>{\tr{U √{ρ}√{τ}}} = \fidb{ρ}{τ} ¶.
  \end{align}
  The equality of the purified distance then follows by~Lemma~\ref{lm:alt-gfid} and 
  the definition of the purified distance.
\end{proof}

The following lemma (see also~\cite{berta10,tomamichel11}) gives a bound on the distance between 
a state and the ☼[projection]{projector} of that state onto a subspace.
\begin{lemma}
  \label{lm:pd-proj}
  Let $ρ ∈ \osub{ℋ}$ and let $Π ∈ \opos{ℋ}$ be a projector, then
  \begin{align}
  P \big( ρ, Π ρ Π \big) ≤ √{2\, \tr{\Pip ρ} - \tr{\Pip ρ}^2} ¶\,,
  \end{align}
  where $\Pip := ⅈ - Π$ is the complement of $Π$ on $ℋ$.
\end{lemma}

\begin{proof}
  The generalized fidelity between the two states can be bounded using
  $\tr{Π ρ} ≤ \tr{ρ}$ and $\fid{ρ}{Π ρ Π} = \tr{Π ρ}$, which follows from~§[fid/proj]. 
  We have
  \begin{align}
    F(ρ, Π ρ Π) &≥ \tr{Π ρ} + 1 - \tr{ρ} = 1 - \tr{\Pip ρ} ¶\,.
  \end{align}
  The desired bound on the purified distance follows from its definition.
\end{proof}

%..........................

\section{Notational Conventions}
\label{sc:pd/not}
☼*{notation}

We will often use the notation $ρ \ecl τ$ (in words, ``$ρ$ is ☼{$ε$-close} to $τ$'') to say that $P(ρ, τ) ≤ ε$, where $ε \ll 1$ is a small parameter and $ρ, τ ∈ \osub{ℋ}$. 

The following intuitive notational conventions will be used widely in the following chapters and shortens many proofs.
Let $ℋ¬{AB}$ be a bipartite Hilbert space. If $ρ¬{AB} ∈ \osub{ℋ¬{AB}}$ and $\rhot¬{A} ∈ \osub
{ℋ¬A}$ are defined, we implicitly define $\rhot¬{AB}$ as the extension of $\rhot¬A$ that has 
minimum purified distance to $ρ¬{AB}$ according to Corollary~\ref{co:pd-ext} and 
Lemma~\ref{lm:pd-ext-constr}. This implies that $P(\rhot¬{AB}, ρ¬{AB}) = P(\rhot¬A, ρ¬B)$.

Generally, states with the same Greek letter will be treated in this 
way, where the distance of the modified letter states (e.g.\ $\rhot, \sigmah, \taub$) is always 
measured with regards to the unmodified letter state (e.g.\ $ρ, σ, τ$). Consistent with that, if 
an extension is introduced that is not of minimum distance, we will always use another greek letter to denote it. 
Due to the arguments above, it is usually sufficient to write $P(\rhot, ρ)$ without mentioning the subspaces on which the states are compared.

These minimum distance extensions often inherit properties of the original state. For the example discussed above, if $ρ$ is classical on $B$ then $\rhot$ is also classical on $B$. This follows directly from the monotonicity (cf.\ Theorem~\ref{th:pd-mono}) of the purified distance under trace non-increasing maps, since $P(ℳ[\rhot], ρ) = P(ℳ[\rhot], ℳ[ρ]) ≤ P(\rhot, ρ)$, where $ℳ$ measures in the classical basis of $B$ and thus leaves $ρ$ invariant. Using the same argument, one can show that the minimum distance extension $\rhot¬{AB}$ lives in the subspace 
$ℋ¬A ⨂ \supp{ρ¬B}$ of $ℋ¬{AB}$.

Furthermore, if $ℰ : ℋ¬A → ℋ¬B$ is a trace non-increasing map and $τ¬B = ℰ[ρ¬A]$, then, the definition of a state $\rhot \ecl ρ$ also implicitly defines a state $\taut = ℰ[\rhot¬A] \ecl τ$.
This will often be used when $ℰ$ is an isometry. In this case, we often give the states the same letter, for example if $U = ∑_x \ket{x}[X] ⨂ \ket{x}[X']\bra{x}[A]$ is the isometry that purifies a projective measurement of the system $A$ in the basis $\{ \ket{x} \}$, we use $ρ¬{XX'} = U ρ¬A U†$ to denote the mapped state.

%............................

\chapter{Min- and Max-Entropies}
\label{ch:entropies}

This chapter formally introduces the min- and max-entropies for quantum states and
discusses some of their properties. We provide a plethora of different expressions
for the min- and max-entropy and introduce the interpretation of the min-entropy as a guessing
probability. Moreover, we explore the classical limits of the min- and max-entropy and 
investigate their continuity and the concavity of the max-entropy.

\section{Introduction and Related Work}

We have seen in the introduction that generalizations of the classical ☼[Rényi 
$α$-entropies]{entropy!Rényi}~\cite{renyi61}
can be used to characterize different information theoretic tasks in the ☼{one-shot}
setting. For a discrete probability distribution $P¬X$ over a set $\cX$, the Rényi $α$-entropies are defined as
\begin{align}
  \hh{α}{X}[\tinyP] := \frac{1}{1 - α} \log ∑_x P¬X(x)^α, \tn*{where} α ∈ (0, 1) \cup (1, ∞) ¶. 
\end{align}
These entropies have a trivial generalization to the quantum setting, 
which, for a state
$ρ ∈ \onorm{ℋ¬{A}}$, is given as
\begin{align}
  \hh{α}{A}[ρ] := \frac{1}{1 - α} \log \tr<b>{ρ¬A^α} ¶[renyi-quant].
\end{align}
The range of allowed $α$ can be extended to include $0$ and $∞$ by taking the respective limits
of~§[renyi-quant]. This leads to the expressions
\begin{align}
  \hh{∞}{A}[ρ] := - \log \norm{A}[∞] \tn*{and} \hh{0}{A}[ρ] := \log \rank{ρ¬A} ¶.
\end{align}
Furthermore, taking the limit to $α = 1$ from both sides 
recovers the ☼[von Neumann entropy]{entropy!von Neumann}; hence, we set
$H_1 \equiv H$ and have now defined a spectrum of entropies for $α ∈ [0, ∞]$.
These entropies are monotonically decreasing in the parameter $α$, i.e.
\begin{align}
  α ≥ β \implies \hh{α}{A}[ρ] ≤ \hh{β}{A}[ρ] \quad \tn{for all} ρ ∈ \onorm{ℋ¬A} ¶.
\end{align}

One of the first questions to answer is now whether we need to consider the whole spectrum of Rényi $α$-entropies for our framework for non-asymptotic information theory. This was
answered in the negative by Renner and Wolf~\cite{renner04,rennerwolf05}.
%(see also earlier work by Cachin~\cite{cachin97}). 
They show that if we allow
a small variation of the state of the system\,|\,in the following called ☼{smoothing}\,|\,these 
entropies can be separated into three classes, the elements of each being approximately equal.
The three classes are: the von Neumann entropy, the $α$-entropies with $α < 1$ and the $α$-entropies with $α > 1$.

To follow their argument, we define ☼{$ε$-smooth} Rényi $α$-Entropies,
\begin{align}
  \hh*{α}{A}[ρ] = \begin{cases} 
    \min_{\rhot}\, \hh{α}{A}[\rhot] & \tn{if} α < 1 \\
    \max_{\rhot}\, \hh{α}{A}[\rhot] & \tn{if} α > 1
   \end{cases} ¶, \tn*{where} 0 ≤ ε < 1
\end{align}
and the optimization in each case is over an ☼{$ε$-ball} of close states, $\rhot \ecl ρ$. Note that we smooth in the direction of the von Neumann entropy in both ranges, $α < 1$ and $α > 1$.

These entropies now satisfy the following inequalities. 
(The proof of these statements can be adapted from results in~\cite{renner04,rennerwolf05}.)
\begin{align}
  \hh{½}[2ε+ε']{A}[ρ] - \frac{1}{1\! -\! α} \log \frac{2}{ε^2} &≤
    \hh{α}[ε+ε']{A}[ρ] ≤ \hh{½}[ε']{A}[ρ] + 2 \log \frac{2}{ε^2}
    \ \ (α < 1) ¶,\\
  \hh{∞}[ε+ε']{A}[ρ] + \frac{1}{α\! -\! 1} \log \frac{2}{ε^2} &≥ 
    \hh{α}[ε']{A}[ρ] ≥ \hh{∞}[ε']{A}[ρ] 
    \ \ (α > 1) ¶[alpha-equal].
\end{align}

Note that the deviation terms in $α$ and $ε$ do not depend on properties of the state. Hence, if the entropies are large enough, these terms will be negligible in comparison.

On one hand, the second statement of Eq.~§[alpha-equal] thus shows that smooth Rényi entropies of order $α > 1$ are well approximated by the smooth Rényi entorpy of order $∞$. 
We choose the ☼[min-entropy]{entropy!min-entropy}, $\hmin{A}[ρ] = \hh{∞}{A}[ρ]$, as the representative of this class of Rényi entropies. The choice of the Rényi $∞$-entropy is motivated by its operational interpretation as a guessing probability~(see~\cite{koenig08} and 
Section~\ref{se:min/guessing}) as well as the fact that its quantum generalization has a simple
form that can be expressed as a ☼{semi-definite program}.

On the other hand, the first statement of Eq.~§[alpha-equal] 
implies that the smooth Rényi entropies 
of order $α < 1$ are well approximated by the smooth Rényi entropy of order $\frac{1}{2}$. 
We call this entropy the ☼[max-entropy]{entropy!max-entropy}, $\hmax{A}[ρ] = \hh{½}{A}[ρ]$.
The choice of this entropy as a representative of the class may seem 
arbitrary at this point. Indeed, it could be argued that the Rényi entropy of order $0$ is also a natural choice, as it characterizes such tasks as the amount of memory needed to store the output of a source perfectly in the ☼{one-shot} setting.\footnote{In fact, the initial extension to the fully quantum setting was done for the Rényi entropy of order $0$~\cite{renner05}.} However, the choice of $\frac{1}{2}$ is motivated by the duality relation of the conditional min- and max-entropies, which holds for this choice of the max-entropy (see~\cite{koenig08} and Lemma~\ref{lm:min-max/dual}).

Conditional Rényi entropies can be defined in various ways. In analogy with the von Neumann entropy, we might be tempted to define $\hh{α}{A|B} = \hh{α}{AB} - \hh{α}{B}$. However,
to the best of our knowledge, this definition is not very useful to characterize information theoretic tasks
in the ☼{one-shot} setting.
In this chapter, we propose a natural generalization of the min-entropy that is motivated
by its operational interpretation as a guessing probability. 
Our generalization of the max-entropy then 
follows immediately from the duality of the min- and max-entropies.

Quantum generalizations of the min- and max-entropies were first considered by 
Renner and König~\cite{rennerkoenig05,renner05} in order to investigate security in quantum cryptography and related
tasks, e.g.\ information reconciliation. 
They considered a generalization of the Rényi-entropy of order $0$ and two
variations of the conditional min-entropy as well as a different method of smoothing than
the one proposed in this thesis. The smooth entropy framework has been consolidated since 
and we attempt to summarize the most important results in this and the following chapter.

\subsection{Main Results}

The main result of this chapter is a collection of expressions for the quantum conditional min- and max-entropies.

\begin{result}[Expressions for the Min- and Max-Entropy]
  Let $ρ¬{ABC}$ be a pure quantum state, then
  \begin{align}
    \hmax{A|C}[ρ] &= \log \min \big\{ \norm<b>{Z¬C}[∞] \!: Z¬{AC} ∈ \opos{ℋ¬{AC}} 
        \wedge ρ¬{ABC} ≤ Z¬{AC} ⨂ ⅈ¬B \big\} ¶\\
      &= \log d¬A \max_{σ}\,  F^2 \big( ρ¬{AC}, π¬A ⨂ σ¬C \big) ¶\\
      &= \log d¬A \!\max¬{B → B'B''}\!\! \max_{τ} \, F^2 \big( ρ¬{AB'B''} , 
        γ¬{AB'} ⨂ τ¬{B''} \big) ¶\\
      &= \log \min \big\{ \tr{σ} : σ ∈ \opos{ℋ¬B} \wedge ρ¬{AB} ≤ ⅈ¬A ⨂ σ¬B \big\} ¶\\
      &= \min_{σ} \log \norm<b>{σ¬B^{-½} ρ¬{AB}\, σ¬B^{-½}}[∞] ¶\\
      &= \min_{σ}\, \inf \big\{ λ ∈ ℝ : ρ¬{AB} ≤ 2^{λ} ⅈ¬A ⨂ σ¬B \big\} ¶\\
      &= - \hmin{A|B}[ρ] ¶,
  \end{align}
  where $σ, τ$ are quantum states and $B → B'B''$ is an embedding.
\end{result}

In particular, this result can be interpreted as follows. The expression $\max¬{B → B'B''} \max_{τ} \, F \big( ρ¬{AB'B''}, γ¬{AB'} ⨂ τ¬{B''} \big)$ measures the fidelity with a state that corresponds to an omniscient observer $B = B'B''$ of the system $A$. Any such observer necessarily controls a system $B'$ that is fully entangled with the system $A$ and may, in addition, control a system $B''$ that is uncorrelated with $A$. Moreover, the expression $\max_{σ} F \big( ρ¬{AC}, π¬A ⨂ σ¬C \big)$ measures the fidelity with a state that corresponds to an ignorant observer of the system $A$. Such an observer $C$ necessarily holds a state $σ¬C$ that is product with the system $A$.
Since these two quantitates are equal, we find the following: For any pure state $ρ¬{ABC}$, the marginal $ρ¬{AB}$ is close to an omniscient observer $B$ of the system $A$ if and only if $ρ¬{AC}$ is close to an ignorant observer $C$ of system $A$.

\subsection{Outline}

In Section~\ref{se:min/minmax} we formally introduce the conditional min- and max-entropies and show how they can be expressed as ☼[semi-definite programs]{SDP}.
We also formally introduce the von Neumann entropy.
In Section~\ref{se:min/class} we evaluate the conditional min- and max-entropies for classical probability distributions, and in Section~\ref{se:min/guessing}, we explore the interpretation of the min-entropy as a guessing probability.
Section~\ref{se:min/prop} then discusses various properties of the min- and max-entropy. Most importantly, we give first bounds on the min- and max-entropies and show that these entropies are continuous functions of the state.

\section{Min- and Max-Entropies}
\label{se:min/minmax}

\subsection{The Min-Entropy}
☼*{entropy!min-entropy} ☼*{min-entropy|see{entropy}}

Here, we start with specific definitions of the min- and max-entropy and then develop a variety of alternative expressions.

\begin{definition}
  \label{df:min-entropy}
  Let $ρ¬{AB} ∈ \osub{ℋ¬{AB}}$. The min-entropy of $A$ conditioned on $B$ of the 
  state $ρ¬{AB}$ is
  \begin{align}
    \hmin{A|B}[ρ] :=
      \max_{σ}\, \sup \big\{ λ ∈ ℝ : ρ¬{AB} ≤ 2^{-λ} ⅈ¬A ⨂ σ¬B \big\} ¶[min/def],
  \end{align}
  where the maximum is taken over all states $σ ∈ \osub{ℋ¬B}$.
\end{definition}

Note that there exists a feasible $λ$ only if $\supp{σ¬B} \supseteq \supp{ρ¬B}$. However, if this 
condition on the support is satisfied, there exists a feasible $λ_* = -\log \norm<b>{σ¬B^{-½} 
ρ¬{AB} σ¬B^{-½}}[∞]$ which achieves the supremum. The min-entropy can thus alternatively be 
written as
\begin{align}
  \hmin{A|B}[ρ] = \max_{σ} - \log \norm<b>{σ¬B^{-½} ρ¬{AB} σ¬B^{-½}}[∞] ¶[min/inv],
\end{align}
where we use the generalized inverse and the maximum is taken over all $σ¬B ∈ \osub{ℋ¬B}$ 
with $\supp{σ¬B} \supseteq \supp{ρ¬B}$. 

We can also reformulate~§[min/def] as a ☼[semi-definite program]{SDP} (SDP). (Semi-Definite Programs are introduced in Section~\ref{sc:premath/sdp}.)
For this purpose, we include $2^{-λ}$ in $σ¬B$ and allow the new $σ¬B$ to be an arbitrary positive semi-definite operator. The min-entropy is then given by
\begin{align}
  \hmin{A|B}[ρ] = -\log \min \big\{ \tr{σ} : 
    σ ∈ \opos{ℋ¬B} \wedge ρ¬{AB} ≤ ⅈ¬A ⨂ σ¬B \big\} ¶
\end{align}
and the optimization problem thus has an efficient numerical solver.
In particular, we consider the SDP for the expression $2^{-\hmin{A|B}[ρ]}$.
\begin{align}
  \begin{array}{rlcrl}
    \multicolumn{2}{c}{\underline{\tn{primal problem}}} & \quad & 
      \multicolumn{2}{c}{\underline{\tn{dual problem}}} \vspace{0.2cm} \\
    \texttt{minimize}: & \ip{ⅈ¬B, σ¬B} & & \texttt{maximize}: & \ip{ρ¬{AB}, X¬{AB}} \\
    \texttt{subject to}: & ⅈ¬A ⨂ σ¬B ≥ ρ¬{AB} & & \texttt{subject to}: & \tr[A]{X¬{AB}} ≤ ⅈ¬B \\
    & σ¬B ∈ \opos{ℋ¬B} &&& X¬{AB} ∈ \opos{ℋ¬{AB}}
  \end{array} ¶[min/sdp]
\end{align} 

Clearly, the dual problem has a finite solution; in fact, it is easy to verify that 
$\ip{ρ¬{AB}, X¬{AB}} ≤ \tr{X¬{AB}} ≤ d¬B$. Furthermore, there exists a 
$σ¬B > 0$ with $ⅈ¬A ⨂ σ¬B > ρ¬{AB}$. Hence, strong duality (Theorem~\ref{th:sdp/strong}) 
applies and the primal and dual solution are equivalent.

Let us now investigate the dual problem more closely. We can replace the inequality in the condition $X¬{B} ≤ ⅈ¬B$ by an equality since adding a positive part to $X¬{AB}$ only increases $\ip{ρ¬{AB}, X¬{AB}}$. Hence, $X¬{AB}$ can be interpreted as a ☼[Choi-Jamiolkowski]{isomorphism!Choi-Jamiolkowski} state of a completely positive ☼{unital} map~(cf.\ Section~\ref{sc:choi-jami}) from $ℋ¬{B'} \iso ℋ¬A$ to $ℋ¬B$. Let $ℰ†$ be that map, then
\begin{align}
  2^{-\hmin{A|B}[ρ]} = \max_{ℰ†} \ip<b>{ρ¬{AB}, ℰ†[Γ¬{AB'}]} = d¬A \max_{ℰ} 
    \ip<b>{ℰ[ρ¬{AB}], γ¬{AB'}} ¶\,,
\end{align}
where the second maximization is over all ☼[TP-CPMs]{TP-CPM} $ℰ$ from $B$ to $B'$, i.e.\ all super-operators whose adjoint is completely positive and unital from $B'$ to $B$. 
The fully entangled state $γ = Γ/d¬A$ is pure and normalized, hence, we can write~\cite{koenig08}
\begin{align}
  \hmin{A|B}[ρ] = - \log d¬A \max_{ℰ} F^2 \big( ℰ[ρ¬{AB}], γ¬{AB'} \big) ¶[min/koenig],
\end{align}
where the maximum is taken over all TP-CPMs from $B$ to $B'$. (Note that $γ$ is defined as the fully entangled in an arbitrary but fixed basis of $ℋ¬A \iso ℋ¬{B'}$. The expression is invariant under the choice of basis, since the fully entangled states can be converted into each other by a unitary appended to $ℰ$.) We write $F$ for the generalized fidelity, which corresponds to the fidelity in this case as $γ$ is normalized.

Alternatively, we can interpret $X¬{AB}$ as the Choi-Jamiolkowski state of a TP-CPM map from $ℋ¬{A'} \iso ℋ¬{B}$ to $ℋ¬{A}$. This immediately leads to the relation
\begin{align}
  \hmin{A|B}[ρ] &= - \log d¬B \max_{ℰ} \ip<b>{ρ¬{AB}, ℰ[γ¬{A'B}]} ¶\,,
%   &= - \log d¬B \max_{ℰ} F^2 \big( ρ¬{ABC} , ℰ[γ¬{A'B}] ⨂ ⅈ¬C \big) ¶\,.
\end{align}
where the maximization is over all TP-CPMs from $A'$ to $A$.

We may now decompose the TP-CPMs of~§[min/koenig] into their Stinespring dilation: an isometry $U : ℋ¬B → ℋ¬{B'B''}$ followed by a partial trace over $ℋ¬{B'}$. Uhlmann's theorem now implies that there exists an extension 
of $γ¬{AB'}$ to $B''$ such that $F( U ρ¬{AB} U† ,\, γ¬{AB'B''} ) = F( ℰ[ρ¬{AB}] ,\, γ¬{AB'} )$. Since such extensions of a pure state are necessarily of the form $γ¬{AB'B''} = γ¬{AB'} ⨂ τ¬{B''}$, we recover the following expression for the min-entropy
\begin{align}
  \hmin{A|B}[ρ] = - \log\, d¬A\! \max¬{B → B'B''} \max_{τ} \, F^2 \big( ρ¬{AB'B''} ,\, 
    γ¬{AB'} ⨂ τ¬{B''}  \big) ¶[min-omni],
\end{align}
where the maximization is over all isometries from $B$ to $B'B''$ and 
states $τ ∈ \onorm{ℋ¬{B''}}$.

Using the expression in~§[min-omni], the min-entropy can be interpreted as a 
measure of distance to a state describing an observer $B$ that is ☼[omniscient]{observer!omniscient} about $A$.
Such an observer must necessarily hold a state $γ$ that is fully entangled with $A$ and
may, in addition, hold an arbitrary state $τ$ that is uncorrelated with $A$. The
min-entropy now evaluates the distance (in terms of the fidelity) of $ρ$ to the closest such 
state.

Finally, we introduce the quantity $\hmin+{A|B}$, which is a trivial lower bound on
$\hmin{A|B}$ and is sometimes used instead of $\hmin{A|B}$~\cite{renner05}.
\begin{align}
  \hmin+{A|B}[ρ] := - \log \norm<b>{ρ¬B^{-½} ρ¬{AB} ρ¬B^{-½}}[∞] ≤ \hmin{A|B}[ρ] ¶\,.
\end{align}
The inequality follows by the choice $σ¬B = ρ¬B$ in~§[min/inv]. It has been shown that
the smooth versions of $\hmin+{A|B}$ and $\hmin{A|B}$ are equivalent up to terms in the
smoothing paramter~\cite{tomamichel10}.

\subsection{The Max-Entropy}
☼*{entropy!max-entropy} ☼*{max-entropy|see{entropy}}

We use the following definition of the max-entropy.

\begin{definition}
  \label{df:max-entropy}
  Let $ρ¬{AB} ∈ \osub{ℋ¬{AB}}$. The max-entropy of $A$ conditioned on $B$ of the 
  state $ρ¬{AB}$ is
  \begin{align}
    \hmax{A|B}[ρ] := \max_{σ}\, \log\, \fidb{ ρ¬{AB} }{ ⅈ¬A ⨂ σ¬B }^2 ¶[max/def],
  \end{align}
  where the maximum is taken over all states $σ ∈ \osub{ℋ¬B}$.
\end{definition}

Since the maximum is taken for normalized states $σ ∈ \onorm{ℋ¬B}$, 
we may rewrite this as
\begin{align}
  \hmax{A|B}[ρ] = \max_{σ}\, \log\, d¬A F^2 \big( ρ¬{AB}, π¬A ⨂ σ¬B \big) ¶.
\end{align}
Contrasting this to the min-entropy in~§[min-omni], the max-entropy can be seen as a measure
of proximity of $ρ$ to a state describing an observer $B$ that is ignorant about $A$. Such an observer necessarily holds a state that is product with the state on $A$ and the max-entropy
evaluates the fidelity with the closest such state.

Introducing an arbitrary ☼{purification} $ρ¬{ABC}$ of $ρ¬{AB}$ and applying ☼{Uhlmann's theorem}, we rewrite this as the following optimization problem.
\begin{align}
  2^{\hmax{A|B}[ρ]} = d¬A \max_{τ}\ \braket{ ρ¬{ABC} | τ¬{ABC} | ρ¬{ABC} } ¶\,,
\end{align}
where $τ$ has the marginal $τ¬{AB} = π¬A ⨂ σ¬B$ for some $σ ∈ \osub{ℋ¬B}$. This is the dual problem of the following SDP:
\begin{align}
  \begin{array}{rlcrl}
    \multicolumn{2}{c}{\underline{\tn{primal problem}}} && 
      \multicolumn{2}{c}{\underline{\tn{dual problem}}} \vspace{0.2cm} \\
    \texttt{minimize}: & µ & & 
      \texttt{maximize}: & \ip{ρ¬{ABC}, Y¬{ABC}} \\
    \texttt{subject to}: & µ ⅈ¬B ≥ \tr[A]{Z¬{AB}} & & 
      \texttt{subject to}: & \tr[C]{Y¬{ABC}} ≤ ⅈ¬A ⨂ σ¬B \\
    & Z¬{AB} ⨂ ⅈ¬C ≥ ρ¬{ABC} &&& \tr{σ¬B} ≤ 1 \\
    & Z¬{AB} ∈ \opos{ℋ¬{AB}} &&& Y¬{ABC} ∈ \opos{ℋ¬{ABC}} \\
    & µ ≥ 0 &&& σ¬{B} ∈ \opos{ℋ¬{B}}
  \end{array} ¶
\end{align} 

Again, it is easy to verify that the dual problem has a finite solution. To see this, note that
$\tr{Y} ≤ d¬A$ due to the constraints in the dual problem, hence, the maximum cannot exceed $d¬A$ for normalized states. Moreover, we can easily construct a primal feasible solution with $Z¬{AB} ⨂ ⅈ¬C > ρ¬{ABC}$ and $µ ⅈ¬B > Z¬B$. Hence, strong duality (Theorem~\ref{th:sdp/strong}) applies and the primal and dual solution are equivalent.

The primal problem can be rewritten by noting that the optimization over $µ$ corresponds to evaluating the $∞$-norm of $Z¬B$.
\begin{align}
  \hmax{A|B}[ρ] = \log \min \Big\{ \norm<b>{Z¬B}[∞] \!: Z¬{AB} ⨂ ⅈ¬C
    ≥ ρ¬{ABC},\, Z¬{AB} ∈ \opos{ℋ¬{AB}} \Big\} ¶[max/minimize]\,.
\end{align}

This can be used to prove upper bounds on the max-entropy. For example, the quantity $\hmax+{A|B}$\,|\,which is sometimes used instead of the max-entropy~\cite{renner05}\,|\,is an upper bound on $\hmax{A|B}$.
\begin{align}
  \hmax+{A|B}[ρ] := \log \max_{σ} \tr{ Π^{ρ¬{AB}} ⅈ¬A ⨂ σ¬B } ≥ \hmax{A|B}[ρ] ¶\,.
\end{align}
This follows from~§[max/minimize] by the choice $Z¬{AB} = Π^{ρ¬{AB}}$, which is the projector onto the support of $ρ¬{AB}$. Note also that smooth versions of $\hmax+{A|B}$ and $\hmax{A|B}$ are equivalent up to terms in the smoothing parameter~\cite{tomamichel10}.

Furthermore, König et al.~\cite{koenig08} showed that the max-entropy can be expressed as a min-entropy of the purified state.
\begin{lemma}
  \label{lm:min-max/dual}
  Let $ρ ∈ \osub{ℋ¬{ABC}}$ be pure. Then,
  \begin{align}
    \hmax{A|B}[ρ] = - \hmin{A|C}[ρ] ¶.
  \end{align}
\end{lemma}

\begin{proof}
  We show that $2^{\hmax{A|B}}/d¬A = 2^{-\hmin{A|C}} / d¬A$ using the 
  expression~§[min-omni] for the min-entropy and Def.~\ref{df:max-entropy} for the 
  max-entropy, i.e.\ we will show that
  \begin{align}
    \max¬{σ} \fidb{ρ¬{AB}}{π¬A ⨂ σ¬B} = \max¬{C → C'C''} \max_{τ}\ 
      \fidb{ρ¬{AC'C''}}{γ¬{AC'} ⨂ τ¬{C''}} ¶,
  \end{align}
  where the maximization is over quantum states $σ ∈ \osub{ℋ¬B}$ and 
  $τ¬{C''} ∈ \osub{ℋ¬{C''}}$ as well as embeddings $C → C'C''$.

  Due to Uhlmann's theorem~§[fid/uhl] and the fact that the fidelity cannot decrease 
  under partial trace, we have
  \begin{align}
    \fidb{ρ¬{AB}}{π¬A ⨂ σ¬B} 
      &= \max¬{C → C'C''}\ \fidb{ρ¬{ABC'C''}}{γ¬{AC'} ⨂ φ¬{BC''}} ¶\\
      &≤ \max¬{C → C'C''} \max_{τ}\ \fidb{ρ¬{AC'C''}}{γ¬{AC'} ⨂ τ¬{C''}} ¶,
  \end{align}
  where $\ket{φ}$ is any purification of $σ$ and the fully entangled state $\ket{γ}$ 
  purifies the fully mixed state $π$
  by definition. Since this holds for every $σ$, it particularly holds when $σ$
  maximizes the fidelity in the definition of the max-entropy.
  On the other hand,
  \begin{align}
    \fidb{ρ¬{AC'C''}}{γ¬{AC'} ⨂ τ¬{C''}}
      &= \max_{\ket{φ}}\ \fidb{ρ¬{ABC'C''}}{γ¬{AC'} ⨂ φ¬{BC''}} ¶\\
      &≤ \max_{σ}\ \fidb{ρ¬{AB}}{π¬{A} ⨂ σ¬{B}} ¶.
  \end{align}
  Since this holds for all embeddings $C → C'C''$ and all states $τ$, it
  particularly holds for the tuple that achieves the maximum fidelity in~§[min-omni]. 
  Thus, we established the equality of the two expressions.
\end{proof}

\subsection{The von Neumann Entropy}
☼*{von Neumann entropy|see{entropy}}

For completeness, we also define the ☼[von Neumann entropy]{entropy!von Neumann} for sub-normalized states.

\begin{definition}
  \label{df:vn-entropy}
  Let $ρ¬{AB} ∈ \osub{ℋ¬{AB}}$. Then, the von Neumann entropy is
  \begin{align}
    \hvn{A|B}[ρ] := H(ρ¬{AB}) - H(ρ¬B)\,, \tn*{where} H(ρ) := -\tr{ρ \log ρ} ¶\,.
  \end{align}
\end{definition}

We may rewrite this as an optimization problem. For any $σ ∈ \osub{ℋ¬B}$ 
with $\tr{σ} ≤ \tr{ρ}$, we find
\begin{align}
  \hvn{A|B}[ρ] &= \tr<b>{ρ¬{AB} (ⅈ¬A ⨂ \log σ¬B - \log ρ¬{AB})} - 
    \tr<b>{ρ¬B (\log σ¬B - \log ρ¬B )} ¶\\
  &≥ \tr<b>{ρ¬{AB} (ⅈ¬A ⨂ \log σ¬B - \log ρ¬{AB})} ¶.
\end{align}
We used that the second term, the relative entropy 
$\tr{ρ¬B (\log σ¬B - \log ρ¬B)}$, is non-negative thanks to Klein's 
inequality~\cite{klein31}. Furthermore, the term vanishes for the choice $σ¬B = ρ¬B$. 
If we optimize the expression over all $σ ∈ \osub{ℋ¬B}$ with $\tr{σ} ≤ \tr{ρ}$, we therefore get
\begin{align}
    \hvn{A|B}[ρ] = \max_{σ} \tr<b>{ρ¬{AB} (ⅈ¬A ⨂ \log σ¬B - \log ρ¬{AB})} ¶[vn/klein].
\end{align}

We have seen that the min- and max-entropy satisfy a ☼{duality} relation in 
Lemma~\ref{lm:min-max/dual}.
Here, we show that the von Neumann entropy also satisfies a duality relation. Note that 
$H(ρ) = ∑_i λ_i \log \frac{1}{λ_i}$ is a function of the ☼[eigenvalues]{eigenvalue} $\{ λ_i \}$
of $ρ$ only. Due to the ☼[Schmidt decomposition]{decomposition!Schmidt} of any pure state
$ρ¬{ABC}$, it holds that
\begin{align}
  \hvn{A|B}[ρ] = H(ρ¬{AB}) - H(ρ¬B) = H(ρ¬C) - H(ρ¬{AC}) = -\hvn{A|C} ¶[vn/dual]\,.
\end{align}

\section{Relation to Classical Entropies}
\label{se:min/class}

Here, we evaluate the min- and max-entropy for states on two classical ☼[registers]{register}, $X$ and $Y$. Their content is described by a joint probability distribution $P¬{XY}$ encoded in the state
\begin{align}
  ρ¬{XY} = ∑_{x,y} P¬{XY}(x,y)\, \proj{x}[X] ⨂ \proj{y}[Y] ¶\,.
\end{align}

\subsection{Classical Min-Entropy}

In order to evaluate the min-entropy $\hmin{X|Y}[\tinyP] = \hmin{X|Y}[ρ]$ 
for this state, we consider the SDP for the min-entropy.
We use the operators
\begin{align}
  σ¬Y = ∑_y P¬{XY}(x_{*}^y, y)\, \proj{y} \tn*{and} X¬{AB} = ∑_y \proj{x_{*}^y} ⨂ \proj{y} ¶\,,
\end{align} 
where $x_{*}^y$ is the $x$ that maximizes $P¬{XY}(x, y)$ for a given $y$. It is easy to verify that $σ¬Y$ is primal feasible and $X¬{AB}$ is dual feasible. Since they both give the same (upper and lower, respectively) bound on the min-entropy, we get the equality
\begin{align}
  \hmin{X|Y}[\tinyP] = - \log ∑_y P¬{XY}(x_{*}^y, y) 
    = - \log ∑_y P¬{Y}(y)\, 2^{-\hmin{X}[\tinyP^y]} ¶[min/class],
\end{align}
where $\hmin{X}[\tinyP^y] = \log \max_x P¬X^y(x)$ is the min-entropy of $X$ evaluated for the conditional probability distribution $P¬X^y(x) = P¬{XY}(x, y)/P¬{Y}(y)$.

\subsection{Classical Max-Entropy}

For the max-entropy $\hmax{X|Y}[\tinyP] = \hmax{X|Y}[ρ]$ the calculation is a bit more involved. First, note that we can assume that $σ¬Y$ in the optimization of Def.~\ref{df:max-entropy} is classical in the same basis as $ρ¬Y$. If the fidelity $F(ρ¬{XY}, ⅈ¬X ⨂ σ¬Y)$ is maximal for a given $σ¬Y$, we can always measure both $ρ¬Y$ and $σ¬Y$ in the classical basis of $ρ¬Y$. This operation cannot decrease the fidelity due to~§[fid/mon] and the measured state $σ¬Y^*$ thus achieves the optimum and is of the form
\begin{align}
  σ¬Y^* = ∑_y Q¬Y(y)\, \proj{y} \,, \tn*{where $Q¬Y$ is a probability distribution.} ¶
\end{align}
Hence, we maximize over all probability distributions $Q¬Y$ so that
\begin{align}
  \hmax{X|Y}[\tinyP] = \max_{Q¬Y}\, \log \Big( ∑_{x,y} √{P¬{XY}(x, y)} √{Q¬Y(y)} \Big)^2 ¶.
\end{align}
It is straightforward to verify\,|\,for example using the method of Lagrange multipliers\,|\,that the optimal $Q¬Y$ is proportional to the map 
$y ↦ \big( ∑_x √{P¬{XY}(x, y)} \big)^2$. Hence,
if we insert the (normalized) optimal $Q¬Y$ into the above equation, we get
\begin{align}
  \hmax{X|Y}[\tinyP] = \log \Big( ∑_y \Big( ∑_x √{P¬{XY}(x,y)} \Big)^2 \Big) 
    = \log ∑_y P¬Y(y)\, 2^{\hmax{X}[P^y]} ¶,
\end{align}
where $\hmax{X}[\tinyP^y] = 2 \log ∑_x √{P¬X^y(x)} = \hh{½}{X}[\tinyP^y]$ is the Rényi entropy of order $½$.

\section{Guessing Probability}
\label{se:min/guessing}

The classical min-entropy $\hmin{X|Y}[\tinyP]$ in~§[min/class] can be interpreted as a ☼{guessing probability}. Consider an observer with access to $Y$. What is the probability that this observer guesses $X$ correctly, using his optimal strategy? The optimal strategy of the observer is clearly to guess the $x$ with the highest probability conditioned on his observation $y$. As before, we denote the probability distribution of $x$ conditioned on a fixed $y$ by $P¬{X}^y$.
Then, the (average) guessing probability is given by 
\begin{align}
  ∑_y P¬{Y}(y)\, \max_x P¬{X}^y(x) = 2^{-\hmin{X|Y}[\tinyP]}
\end{align}

It was noted by König \emph{et.\ al.}~\cite{koenig08} that this interpretation of the min-entropy extends to the case where $Y$ is replaced by a quantum system $B$ and the allowed strategies include any measurement of $B$.

Consider a ☼[CQ-state]{state!CQ} $ρ¬{XB} = ∑_x P¬X(x) \proj{x} ⨂ ω¬B^x$ with 
$ω^x ∈ \onorm{ℋ¬B}$. For states of this form, the expression in~§[min/koenig] 
simplifies to
\begin{align}
  2^{-\hmin{X|B}[ρ]} &= \max_{ℰ} \braket<b>{Γ¬{XB'} | 
      \sum_x P¬X(x)\, \proj{x} ⨂ ℰ[ω¬B^x] | Γ¬{XB'}} ¶\\
    &= \max_{ℰ} ∑_x P¬X(x)\, \braket<b>{x | ℰ[ω¬B^x] | x} ¶\,.
\end{align}
This expression reaches its minimum when $ℰ$ is a measurement, i.e.\ a map from $B$ to a register $ℋ¬{B'} \iso ℋ¬X$. Moreover, the expression optimizes the probability that $ℰ[ω¬B^x]$ is mapped to $\ket{x}$ over all such measurement. We can interpret $\ket{x}¬{B'}$ as the observer's guess of the value $x$ and thus $2^{-\hmin{X|B}[ρ]}$ is the guessing probability in the sense described above.\footnote{Note that the special case when $B$ is classical is recovered by setting $ℰ : \ket{y} ↦ \ket{x_*^y}$.}

\section{Properties of the Min- and Max-Entropy}
\label{se:min/prop}
☼*{entropy!min-entropy} ☼*{entropy!max-entropy}

Many properties of the min- and max-entropies (and their smooth variants, introduced in 
Section~\ref{sc:smooth}) can be derived from properties of the SDP for the min-entropy, 
Eq.~§[min/sdp].
Let us consider the functional $Φ¬{A|B} : ρ¬{AB} ↦ 2^{-\hmin{A|B}[ρ]}$, given by the SDP, which we extend to arbitrary Hermitian arguments $K ∈ \oherm{ℋ¬{AB}}$. We use the shorthand notation $Φ \equiv Φ¬{A|B}$ in the following. We restate the SDP here for the 
convenience of the reader.
\begin{align}
  \begin{array}{rlcrl}
    \multicolumn{2}{c}{\underline{\tn{primal problem}}} & \quad & 
      \multicolumn{2}{c}{\underline{\tn{dual problem}}} \vspace{0.2cm} \\
    \texttt{minimize}: & \ip{ⅈ¬B, σ¬B} & & \texttt{maximize}: & \ip{K¬{AB}, X¬{AB}} \\
    \texttt{subject to}: & ⅈ¬A ⨂ σ¬B ≥ K¬{AB} & & \texttt{subject to}: & \tr[A]{X¬{AB}} ≤ ⅈ¬B \\
    & σ¬B ∈ \opos{ℋ¬B} &&& X¬{AB} ∈ \opos{ℋ¬{AB}}
  \end{array} ¶
\end{align} 

We may express the min- and max-entropies of $A$ conditioned on $B$ of a pure state 
$ρ ∈ \osub{ℋ¬{ABC}}$ as follows. 
The min-entropy is given as $\hmin{A|B}[ρ] = - \log Φ¬{A|B}(ρ¬{AB})$ and the max-entropy as $\hmax{A|B}[ρ] = \log Φ¬{A|C}(ρ¬{AC})$, by duality. 
The most important properties of the functional $Φ$ are listed in Table~\ref{tb:Phi}. 

\begin{table}[!ht]
  \rule{\textwidth}{1pt}
  \begin{minipage}{0.98\textwidth}
    \vspace{5pt}
    Let $K, L ∈ \oherm{ℋ¬{AB}}$ and $M ∈ \opos{ℋ¬{AB}}$.
    \begin{enum}

☛[i.] Multiplication with scalar: Let $λ ≥ 0$. Then, $Φ(λ K) = λ\,Φ(K)$.

☛[ii.] Monotonicity: $K ≥ L ⇒ Φ(K) ≥ Φ(L)$.

☛[iii.] Sub-Additivity: $Φ(K + L) ≤ Φ(K) + Φ(L)$. Furthermore, equality 
  holds if $K¬B$ and $L¬B$ have orthogonal support.
  
☛[iv.] Convexity: 
  $Φ(λ K + (1 - λ) L) ≤ λ Φ(K) + (1-λ) Φ(L)$ for $λ ∈ [0, 1]$.

☛[v.] Invariance under Isometries: 
  $Φ$ is invariant under local isometries on the $A$ or $B$ system. Namely,
  $Φ¬{A'|B'} \big( (U ⨂ V) K (U† ⨂ V†) \big) = Φ(K)$, where $U : ℋ¬A → ℋ¬{A'}$ 
  and $V: ℋ¬B → ℋ¬{B'}$ are isometries.
  
☛[vi.] Data Processing: 
  Let $ℰ$ be a trace non-increasing CPM from 
  $B$ to $B'$ and let $\sF$ be a sub-unital 
  CPM from $A$ to $A'$ ($\sF[ⅈ¬A] ≤ ⅈ¬{A'}$).
  Then, we have
  $Φ¬{A|B'} \big( ℰ[M] \big) ≤ Φ(M)$ and $Φ¬{A'|B} \big ( \sF[M] \big) ≤ Φ(M)$.
  
☛[vii.] Data Processing on Extension: Let $ℰ$ be a trace 
  non-increasing CPM from $C$ to $C'$. Then, $Φ(ℰ[M]) ≤ Φ(M)$.
  
☛[viii.] Tensor Product: 
  $Φ¬{A_1A_2|B_1B_2}(L ⨂ K) = Φ¬{A_1|B_1}(L)\, Φ¬{A_2|B_2}(K)$.

    \end{enum}
  \end{minipage}
  \vspace{5pt} \\
  \rule{\textwidth}{1pt}
  \vspace{-20pt}
  \caption[Properties of $Φ$.]{\emph{Properties of $Φ$.}}
  \label{tb:Phi}
\end{table}

Some properties can be readily verified 
by close inspection of the SDP and we will only provide their proof when necessary.

\begin{proof}[Proof of Property iii.]
  The inequality follows from the fact that dual feasibility 
  of $X¬{AB}$ is independent of the argument of
  $Φ$. Hence, if $X¬{AB}$ is optimal for $K + L$, we have 
  $Φ(K + L) = \ip{K, X} + \ip{L, X} ≤ Φ(K) + Φ(L)$.

  For equality, consider the optimal primal for $K + L$, which we denote $σ¬B$. Then, 
  $Π^{L¬B} σ¬B Π^{L¬B}$ is primal feasible for $L$ and $Π^{K¬B} σ¬B Π^{K¬B}$
  is primal feasible for $K$. Hence, if these projectors are orthogonal, we have 
  \begin{align}
    Φ(K) + Φ(L) &= \tr{(Π^{K¬B} + Π^{L¬B}) σ¬B} ≤ \tr{σ} = Φ(K + L) ¶. \qedhere
  \end{align}
\end{proof}

\begin{proof}[Proof of Property vi.]
  Let $σ¬B$ be the optimal primal for $M$. Then, $ℰ[σ]$ is primal feasible for $ℰ[M]$ since
  \begin{align}
    ⅈ¬A ⨂ σ¬B ≥ M¬{AB} \implies ⅈ¬A ⨂ ℰ[σ¬B] ≥ ℰ[M¬{AB}] ¶\,.
  \end{align}
  Moreover, $\tr{ℰ[σ]} ≤ \tr{σ}$, concluding the proof of the first statement.   
  To prove the second statement, note that $σ$ is primal feasible for $\sF[M]$ since
  \begin{align}
     ⅈ¬A ⨂ σ¬B ≥ M¬{AB} \implies ⅈ¬{A'} ⨂ σ¬B ≥ \sF[ⅈ¬A] ⨂ σ¬B ≥ \sF[M¬{AB}] ¶\,,
  \end{align}
  where we applied the defining property of sub-unital maps.
\end{proof}

\begin{proof}[Proof of Property vii.]
  Lemma~\ref{lm:pt-norm-bound} establishes that $\tr[C]{ℰ[M]} ≤ \tr[C]{M}$. The property then
  follows from monotonicity.
\end{proof}

\begin{proof}[Proof of Property viii.]
  If $σ¬{B_1}$ is primal optimal for $L$ and $σ¬{B_2}$ is primal optimal for $K$, 
  then $σ¬{B_1} ⨂ σ¬{B_2}$ is primal feasible for $L ⨂ K$. Moreover, if $X¬{A_1B_1}$ is dual
  optimal for $L$ and $X¬{A_2B_2}$ is dual optimal for $K$, then $X¬{A_1B_1} ⨂ X_{A_2B_2}$ is
  dual feasible for $L ⨂ K$. Hence, the equality follows from properties of the inner product.
\end{proof}

\subsection{First Bounds and Order}

We first show the following technical lemma.

\begin{lemma}
  \label{lm:min/func-bounds}
  Let $K ∈ \oherm{ℋ¬{AB}}$, $d¬A^* = \rank{ \trace¬B \{ K \}_+ } ≤ d¬A $ and $d¬B^* = 
  \rank{ \trace¬A \{ K \}_+} ≤ d¬B$. Then, 
  \begin{align}
    \frac{1}{d¬A^*} \trace\{K\}_+ ≤ Φ(K) ≤ \min \{d¬A^*, d¬B^*\} \trace\{K\}_+ ¶.
  \end{align}
\end{lemma}

\begin{proof}
  Using the eigenvalue decomposition $K = ∑_i λ_i\, \proj{φ^i}$ and Properties i.\ 
  to iii., we find
  \begin{align}     
    Φ(K) &≤ Φ \bigg(\, ∑_{i: λ_i > 0} λ_i\, φ^i \bigg)
        ≤ ∑_{i: λ_i > 0} λ_i\, Φ \big( φ^i \big) ¶\,.
  \end{align}
  Since the $φ^i$ are normalized and pure, we have, $φ^i ≤ ⅈ¬A ⨂ Π^{φ¬B^i}$ and, 
  thus, $σ¬B = Π^{φ¬B^i}$ is primal feasible. This means that
  $Φ(φ^i) ≤ \tr{Π^{φ¬B^i}} = \tr{Π^{φ¬A^i}} ≤ \min\{d¬A^*, d¬B^*\}$, where the equality 
  follows from the Schmidt decomposition and the latter inequality from 
  \begin{align}
    \supp{φ¬B^i} \subseteq \supp{ \trace¬A\{K\}_+ } \!\tn*{and}\!
      \supp{φ¬A^i} \subseteq \supp{ \trace¬B\{K\}_+ } ¶.
  \end{align}
  This concludes the proof of the upper bound. 
    
  On the other hand, we use that 
  \begin{align}
    X¬{AB} = \frac{1}{d¬A^*}\, Π^{\trace¬B\{K\}_+} 
      ⨂ Π^{\trace¬A\{K\}_+} ≥ \frac{1}{d¬A^*}\, Π^{\{K\}_+} ¶.
  \end{align}
  is dual feasible. Hence, $Φ(K) ≥ \tr{K¬{AB} X¬{AB}} ≥ \frac{1}{d¬A^*} \trace\{K\}_+$.
\end{proof}

Using these properties, we now establish bounds on the conditional min- and max-entropies in terms of the support of the marginals of the state. This also establishes that\,|\,for 
normalized states\,|\,the min-entropy is always smaller than the von Neumann entropy, 
which is in turn smaller than the max-entropy.
\begin{proposition}
  \label{pr:min-max-bounds}
  Let $ρ¬{AB} ∈ \osub{ℋ¬{AB}}$, $d¬A^* = \rank{ρ¬A}$, $d¬B^* = \rank{ρ¬B}$ and 
  $d¬{AB}^* = \rank{ρ¬{AB}}$. Then, using $t = 1/\tr{ρ}$, 
  the following bounds hold:
  \begin{align}
    -\log \min\{d¬A^*, d¬B^*\} &≤ \hmin{A|B}[ρ] - \log t ¶\\
      &≤ t\, \hvn{A|B}[ρ] ¶\\
      &≤ \hmax{A|B}[ρ] + \log t ≤ \log \min\{d¬A^*, d¬{AB}^*\} ¶.
  \end{align}
\end{proposition}

\begin{proof}
  The lower bound on the min-entropy follows directly from Lem\-ma~\ref{lm:min/func-bounds}.
  To get the second inequality, we consider the tuple $\{ σ¬B, λ \}$ that optimizes the
  min-entropy in Def.~\ref{df:min-entropy}. Hence, $ⅈ¬A ⨂ \tr{ρ} σ¬B ≥ \tr{ρ} 2^λ ρ¬{AB}$.
  We now plug $\tr{ρ} σ¬B$ into~§[vn/klein], providing us with a lower bound on the von
  Neumann entropy.
  \begin{align}
    \hvn{A|B}[ρ] &≥ \tr<b>{ρ¬{AB} \big( ⅈ¬A ⨂ \log \big( \tr{ρ} σ¬B \big) - 
        \log ρ¬{AB} \big)} ¶\\
%      &= \tr<b>{ρ¬{AB} \big( \log \big( ⅈ¬A ⨂ \tr{ρ} σ¬B \big) - 
%        \log ρ¬{AB} \big)} ¶\\
      &≥ \tr<b>{ρ¬{AB} \big(\log \big(\tr{ρ} 2^λ ρ¬{AB}\big) - \log ρ¬{AB}\big)} ¶\\
      &= \tr{ρ} \big(λ + \log \tr{ρ} \big)¶,
  \end{align}
  where we used the operator monotonicity of the $\log$ function (cf.~Section~\ref{sc:opmono}) 
  to get the second inequality. This establishes the second inequality of the statement of the 
  lemma. 
  
  The third and fourth inequality follow symmetrically from the first two inequalities 
  applied to the marginal state
  $ρ¬{AC}$ of any purification $ρ ∈ \osub{ℋ¬{ABC}}$ of $ρ$ together with the
  duality relations $\hmax{A|B} = -\hmin{A|C}$ and $\hvn{A|B} = -\hvn{A|C}$. Note also that 
  $d¬C^* = d¬{AB}^*$ due to the Schmidt decomposition.
\end{proof}

For example, the above lemma implies that the min- and max-entropy of 
normalized pure states is at most zero and that the min-entropy cannot
exceed the max-entropy for normalized states.

\subsection{Continuity}

For classical-quantum states, the operational interpretation of the conditional min-entropy as a guessing probability (cf.\ \cite{koenig08}) already implies its
continuity in the state. To see this, note that a discontinuity in
the guessing probability could be detected experimentally using a
fixed number of trials (the number depending only on the required
precision), hence giving us the means to distinguish between
arbitrarily close states for a cost (in terms of the number of
trials) independent of their distance. For sufficiently close
states, this would contradict the upper bound on the distinguishing
advantage, Eq.~§[dist-adv]. Here, we make this statement more
precise.

\begin{proposition}
  \label{pr:cont} 
  Let $ρ, τ ∈ \osub{ℋ¬{AB}}$ and $δ := D(ρ, τ)$, then
  \begin{align}
    \abs<b>{\hmin{A|B}[ρ] - \hmin{A|B}[τ]} ≤ δ\, \frac{ d¬A \min\{d¬A, d¬B\}}
      {\ln 2 \cdot \min\{ \tr{ρ}, \tr{σ} \}} ¶\,.
  \end{align}
\end{proposition}

\begin{proof}
  We use continuity of the functional $Φ$ to obtain
  \begin{align}
    Φ(τ) &= Φ(ρ + (τ - ρ)) ≤ Φ(ρ) + Φ(τ - ρ) ¶\\
      &≤ Φ(ρ) + \min\{d¬A, d¬B\} \trace \{τ - ρ\}_+ ≤ Φ(ρ) + \min\{d¬A, d¬B\}\, δ\, .
  \end{align}
  Note that $Φ > 0$ for all states in $\osub{ℋ¬{AB}}$.
  Taking the logarithm and using the bound $\ln (a + x) ≤ \ln a + \frac{x}{a}$, we find
  \begin{align}
    \log Φ(τ) - \log Φ(ρ) ≤ δ\, \frac{\min\{d¬A, d¬B\}}{\ln 2 · Φ(ρ)}
      ≤ δ\, \frac{d¬A \min\{d¬A, d¬B\}}{\ln 2 · \tr{ρ}} ¶\,.
  \end{align}
  The same argument also applies on exchange of $ρ$ and $τ$
  and we obtain the statement of the lemma by substituting $\hmin{A|B}[ρ] = -\log Φ(ρ)$.
\end{proof}

\begin{remark}
  The above result is tight in the following sense: Consider a system
  with Hilbert spaces $ℋ¬A$ and $ℋ¬B = ℋ¬{A'} ⨁ ℋ¬{B'}$, where 
  $ℋ¬{A'} \iso ℋ¬A$. Let $γ¬{AB}$ be the normalized fully entangled state on
  $ℋ¬A ⨂ ℋ¬{A'}$ and $σ¬B ∈ \osub{ℋ¬{B'}}$ be orthogonal
  to $γ¬B$. The choice $ρ = π¬A ⨂ σ¬B$ and $τ = ρ + δ γ$ for some small
  $δ > 0$ leads to $D(ρ, τ) = δ$,
  \begin{align}
    Φ(ρ) = \frac{\tr{ρ}}{d¬A} \tn*{and} \quad Φ(τ) = Φ(ρ) + \delta\, \min\{d¬A, d¬B\} ¶.
  \end{align}
  Taking the logarithm (for small $δ$) leads to
  \begin{align}  
    \log Φ(τ) - \log Φ(ρ) ≈ \delta\, \frac{\min\{d¬A, d¬B\}}{\ln 2 · Φ(ρ)} = 
      \delta\, \frac{d¬A \min\{d¬A, d¬B\}}{ \ln 2 · \tr{ρ}} ¶.
  \end{align}
\end{remark}

Proposition~\ref{pr:cont} implies that the conditional min-entropy
is uniformly (Lipschitz) continuous on the set of normalized states
and in any $ε$-ball. Since $D(ρ, τ) ≤ P(ρ, τ)$,
the bound also holds for $δ = P(ρ, τ)$. Due to the duality between min-
and max-entropy, Lipschitz continuity in $P(ρ, τ)$ also follows for 
the max-entropy.

\subsection{Conditioning on Classical Information}

Let us, more generally, consider a general state between two quantum systems, $A$ and $B$, and a classical register, $K$. Such a state has the form 
\begin{align}
  ρ¬{ABK} = ∑_k p_k\, \proj{k} ⨂ τ¬{AB}^k , \tn*{where} τ^k ∈ \osub{ℋ¬{AB}} ¶[qqc-state] 
\end{align}
and $\{ p_k \}$ is a probability distribution.

\begin{proposition}
  \label{pr:classical-side-info}
  Let $ρ ∈ \osub{ℋ¬{ABK}}$ be of the form~§[qqc-state]. Then,
  \begin{align}
    \hmin{A|BK}[ρ] &= -\log \big(  ∑_k p_k\, 2^{-\hmin{A|B}[τ^k]} \big) \tn*{and} ¶\\
    \hmax{A|BK}[ρ] &= \log \big(  ∑_k p_k\, 2^{\hmax{A|B}[τ^k]} \big)  ¶.
  \end{align}
\end{proposition}

\begin{proof}
  The first statement follows directly from the sub-additivity of $Φ$, Property ii., and
  the fact that the marginal states on the conditioning system, 
  $τ¬B^k ⨂ \proj{k}$, are orthogonal. Thus, 
  \begin{align}
    Φ¬{A|BK}(ρ) = ∑_k p_k\, Φ¬{A|BK}(τ^k ⨂ \proj{k}) = ∑_k p_k\, Φ¬{A|B}(τ^k) ¶.
  \end{align} 
  Note that we can remove the trivial register $K$ in the last expression since $Φ$ is invariant
  under the isometry $τ ↦ τ ⨂ \proj{k}$.

  The second statement follows by the duality relation of the min- and max-entropy 
  (Lemma~\ref{lm:min-max/dual}). We consider purifications of $τ^k$ on $C$ and introduce
  a purification of $ρ$ as
  \begin{align}
    \ket{ρ}[ABCKK'] = ∑_k √{p_k}\, \ket{τ^k}[ABC] ⨂ \ket{k}[K] ⨂ \ket{k}[K'] ¶[conc/rho]\,,
  \end{align}
  where $ℋ¬{K'} \iso ℋ¬K$. Using this state and the fact that its marginal $ρ¬{ACK'}$ is of the 
  form~§[qqc-state], we find
  \begin{align}
    \hmax{A|BK}[ρ] &= - \hmin{A|CK'}[ρ] = \log \Big( ∑_k p_k\, 2^{-\hmin{A|C}[τ^k]} \Big) ¶\\
       &= \log \Big( ∑_k p_k\, 2^{\hmax{A|B}[τ^k]} \Big) ¶. \qedhere
  \end{align} 
\end{proof}

\subsection{Concavity of the Max-Entropy}

The following bounds follow directly from the above lemma and the ☼[concavity]{concave} of 
the $\log$ function. For a state of the form~§[qqc-state], we have
\begin{align}
   \hmin{A|BK}[ρ] &≤ ∑_k p_k\, \hmin{A|B}[τ^k] \tn*{and} ¶\\
   \hmax{A|BK}[ρ] &≥ ∑_k p_k\, \hmax{A|B}[τ^k]  ¶.
\end{align}
Moreover, $\hmax{A|B}[ρ] = \log Φ¬{A|CKK'}(ρ) ≥ \log Φ¬{A|CK'}(ρ) = \hmax{A|BK}[ρ]$ due to 
☼{data processing} of $Φ$ applied to the purification of $ρ$ given in~§[conc/rho].
Together with the fact that $ρ¬{AB} = \sum_k p_k\, τ¬{AB}^k$, this implies that the 
max-entropy is a concave function of the state.

Furthermore, we have seen that $Φ¬{A|B}$ (which corresponds to the guessing probability if
$A$ is classical) is convex in the state. However, numerical evidence suggests
that the min-entropy itself is neither ☼{concave} nor ☼{convex} in the state.

%----------------------

%............................ 

\chapter{Smooth Entropies}
\label{ch:smooth}

Smooth entropies are defined as optimizations of the min- and max-entropy over a set of close states. In this chapter, based on~\cite{tomamichel09}, we propose that this closeness should be measured in terms of the purified distance. This endows the smooth entropies with many
useful properties, for example invariance under isometries and a duality relation. Among many other properties, we also show that the smooth entropies satisfy a ☼{data processing} inequality 
and chain rules. 
☼*{entropy!smooth}

\section{Introduction and Related Work}

Generally, smooth entropies are defined as optimizations of the underlying 
entropy over a ball of states close to the state under
consideration. In particular, the ☼[smooth min- and max-entropies]{entropy!smooth} 
are defined as
\begin{align}
  \hmin*{A|B}[ρ] = \max_{\rhot} \, \hmin{A|B}[\rhot] \!\tn*{and}\!
  \hmax*{A|C}[ρ] = \min_{\rhot} \, \hmax{A|C}[\rhot] ¶,
\end{align}
where the optimization is over an ☼{$ε$-ball} of states $\rhot$ close to $ρ$.
The conditional smooth entropies for quantum states are first introduced by Renner~\cite{renner05}. In this work, various properties of the smooth entropies are shown, among them chain rules and various data processing inequalities.

It is important to note that the metric used in~\cite{renner05} 
to define the $ε$-ball was a generalization of the trace distance to sub-normalized states.
Subsequently, we found that a metric based on the fidelity
would be more suitable~\cite{wullschleger08pc}. Using such a metric, we can always find 
extensions, due to Uhlmann's theorem, that are as close as 
their marginals. This property does not hold for the trace distance metric 
used in~\cite{renner05}. Note also that Renner used a max-entropy based on the Rényi entropy of order $0$ as a basis for the smooth max-entropy and was thus unable to harness the duality relation, which simplifies many proofs.

\subsection{Main Contributions}

Using the ☼{purified distance} as a metric, we are able to show many properties of the smooth
entropies in a much more direct way as compared to~\cite{renner05}. 
In particular, this smoothing allows us extend the ☼{duality} relation of the 
min- and max-entropy 
(cf.\ Lemma~\ref{lm:min-max/dual}) to $ε$-smooth entropies.

\begin{result}[Duality Relation]
  \label{res:duality}
  For any $ε ≥ 0$ and $ρ¬{ABC}$ pure, we have
  \begin{align}
    \hmin*{A|B}[ρ] = -\hmax*{A|C}[ρ] ¶.
  \end{align}
\end{result}

Another result we want to highlight is a ☼{data processing} inequality. 
It states
that the uncertainty about the $A$ system gets at most larger when we 
condition on less side information. 

\begin{result}[Data Processing Inequality]
  \label{res:dataproc}
  For any $ε ≥ 0$, we have
  \begin{align}
    \hmin*{A|BC}[ρ] ≤ \hmin*{A|B}[ρ] \tn*{and} 
      \hmax*{A|BC}[ρ] ≤ \hmax*{A|B}[ρ]¶.
  \end{align}
\end{result}

This is a special case of a more general 
data processing inequality that holds for any
☼{TP-CPM} applied to the system we condition on.

This section also includes a variety of other results. For example, various properties of
smooth entropies of classical ☼[registers]{register} are discussed and should be very 
helpful in many applications. For example, we show that the smoothing on registers can be restricted to valid classical states without loss of generality. This implies that
the smooth entropies are also well-defined in their classical limit.

\subsection{Outline}

In Section~\ref{sc:smooth}, we introduce the smooth min- and max-entropies and show several properties of the smoothing. 
Section~\ref{sc:smooth/prop} discusses some properties of the smooth entropies, including their invariance under isometries and a duality relation.
One of the most striking properties, a generalized data processing inequality, is shown in Section~\ref{sc:smooth/data-proc}.
Then, Section~\ref{sc:smooth/class} explores the properties of the smooth entropies of classical registers.
Finally, Section~\ref{sc:chain} summarizes some chain rules that were recently proven for the
smooth entropies.

\section{Smooth Min- and Max-Entropies}
\label{sc:smooth}

The smooth entropies of a state $ρ$ are defined as optimizations over the min- and max-entropies
of states $\rhot \ecl ρ$, i.e.\ over states that are close to $ρ$.
Here, we define the smooth min- and max-entropies and explore some properties 
of the smoothing. In particular, we show that our definition of smoothing 
allows us to extend the duality relation of the
min- and max-entropy to the smooth entropies.

\subsection{The Smooth Entropies}
☼*{entropy!smooth}
☼*{smoothing}

We introduce sets of $ε$-close states that will then later be used to define the smooth entropies.
\begin{definition}[$ε$-Ball]
  \label{df:ball}
  Let $ρ ∈ \osub{ℋ}$ and $0 ≤ ε < √{\tr{ρ}}$. We define the 
  $ε$-ball of operators on $ℋ$ around $ρ$ as
  \begin{align}
    \ball{ℋ; ρ} := \{ τ ∈ \osub{ℋ} : P(τ, ρ) ≤ ε \} ¶.
  \end{align}
  Furthermore, we define the $ε$-ball of pure states around $ρ$ as $\ball*{ℋ; ρ} := \{ τ ∈ \ball{ℋ ; ρ} : \rank{τ} = 1 \}$.
\end{definition}

For the remainder of this chapter, we will assume that $ε$ is sufficiently small so that 
$ε < √{\tr{ρ}}$ is always satisfied. Furthermore, if the Hilbert space used is obvious from 
context, we will omit it and simply use the notation $\ball{ρ}$. We now list some properties of 
this $ε$-ball that come in addition to the properties of the underlying purified distance metric.

\begin{enum}

☛[i.] The set $\ball{ℋ; ρ}$ is compact and convex.

\begin{proof}
  The set is closed and bounded, hence compact.  For convexity, we
  require that, for any $λ ∈ [0, 1]$ and $σ, τ ∈ \ball{ρ}$, 
  the state $ω := λ σ + (1 - λ) τ$ is also in $\ball{ρ}$. We define
  $\omegah = ω ⨁ (1\!-\!\trace ω)$ and analogously
  $\rhoh$, $\sigmah$ and $\tauh$. Thus, $\omegah = λ \sigmah +
  (1 - λ) \tauh$ by linearity.  By assumption, we have
  $F(\sigmah, \rhoh) ≥ √{1 - ε^2}$ and $F(\tauh, \rhoh) ≥
  \sqrt{1 - ε^2}$.  We use the concavity of the 
  fidelity~(cf.~\cite{nielsen00}, Section 9.2.2) to find
  \begin{align}
    P(ω, ρ) &= √{1 - F(\omegah, \rhoh)^2} ¶\\
    &= √{1 - F(λ \sigmah + (1 \!-\! λ) \tauh, \rhoh)^2} ¶\\
    &≤ √{1 - \big( λ F(\sigmah, \rhoh) + (1 \!-\! λ)
      F(\tauh, \rhoh) \big)^2} ≤ ε ¶\,.
  \end{align}
  Therefore, $ω ∈ \ball{ρ}$, as required.
\end{proof}

%%%%%%%%%%%%%%%

☛[ii.] Normalized states in $\ball{ℋ; ρ}$ are not
  distinguishable from a normalized state $ρ$ with probability more
  than $\frac{1}{2}(1 + \eps)$.
  
\begin{proof}
  By Proposition~\ref{pr:pd-gtd-bounds}, $τ ∈ \ball{ρ}$ implies $D(τ, ρ) ≤ P(τ, ρ) ≤ ε$. The
  statement then follows from~§[dist-adv].
\end{proof}

%%%%%%%%%%%%%%%

☛[iii.] The ball grows monotonically in the smoothing parameter
  $ε$, namely $ε < ε' ⇒ \ball{ℋ; ρ} \subset \ball[ε']{ℋ; ρ}$. 
  Furthermore, $\ball[0]{ℋ; ρ} = \{ ρ \}$.

\end{enum}

The ☼[smooth entropies]{entropy!smooth} are now defined as follows. 
\begin{definition}[Smooth Entropies]
  \label{df:smooth}
  Let $ρ ∈ \osub{ℋ¬{AB}}$ and $ε ≥ 0$. Then, we define the $ε$-smooth min- and 
  max-entropies of $A$ conditioned on $B$ of the state $ρ$ as
  \begin{align}
    \hmin*{A|B}[ρ] &:= \max_{\rhot} \hmin*{A|B}[\rhot] \tn*{and} ¶\\
      \hmax*{A|B}[ρ] &:= \min_{\rhot} \hmax*{A|B}[\rhot] ¶,
  \end{align}
  where the optimization is over all states $\rhot ∈ \ball{ρ¬{AB}}$ in both cases.
\end{definition}

Note that the extrema can be achieved due to compactness of the $ε$-ball (cf. Property~i.).
We usually use $\rhot¬{AB}$ to denote the state that achieves the extremum, e.g.,
for the min-entropy, there exists a state $\rhot¬{AB} ∈ \ball{ρ¬{AB}}$ such that 
$\hmin*{A|B}[ρ] = \hmin{A|B}[\rhot]$.
The state $\rhot$ is $ε$-indistinguishable from $ρ$ in the sense described
in Property~ii.
Moreover, the smooth min-entropy is monotonically increasing in $ε$ due to
Property~iii.\ of the $ε$-ball and we recover the
non-smooth min-entropy in the limit $ε = 0$, i.e.\ 
$\hmin[0]{A|B}[ρ] = \hmin{A|B}[ρ]$. Similarly, the 
smooth max-entropy is monotonically decreasing in $ε$ 
and $\hmax[0]{A|B}[ρ] = \hmax{A|B}[ρ]$.

If $ρ¬{ABC}$ is normalized, the optimization problems defining the smooth min- and max-entropies can be formulated as SDPs. To see this, note that the restrictions
on the smoothed state $\rhot$ are linear in the purification $\rho¬{ABC}$ of $ρ$. In 
particular, consider the condition $P(ρ, \rhot) ≤ ε$ on $\rhot$, or, equivalently, $F^2(ρ, \rhot) ≥ 1 - ε^2$. If $ρ¬{ABC}$ is normalized, then the squared fidelity can be expressed 
as $F^2(ρ, \rhot) = \tr{ρ¬{ABC}\,\rhot¬{ABC}}$. 

We give the primal of the SDP for $2^{-\hmin*{A|B}[ρ]}$ as an example. This SDP is parametrized by 
an (arbitrary) purification $ρ ∈ \onorm{ℋ¬{ABC}}$.
\begin{align}
  \begin{array}{rlcrl}
    \multicolumn{2}{c}{\underline{\tn{primal problem}}} \vspace{0.2cm} \\
    \texttt{minimize}: & \ip{ⅈ¬B, σ¬B} \\
    \texttt{subject to}: & ⅈ¬A ⨂ σ¬B ≥ \tr[C]{\rhot¬{ABC}} \\
    & \tr{\rhot¬{ABC}} ≤ 1 \\
    & \tr{\rhot¬{ABC} \rho¬{ABC}} ≥ 1 - ε^2 \\
    & \rhot¬{ABC} ∈ \opos{ℋ¬{ABC}} \\
    &  σ¬B ∈ \opos{ℋ¬B}
  \end{array} ¶
\end{align}
This program allows us to efficiently compute the smooth min-entropy.

\subsection{Remarks on Smoothing}
☼*{smoothing}

For both the smooth min- and max-entropy, 
we can restrict the optimization in Def.~\ref{df:smooth} 
to states in the support of $ρ¬A ⨂ ρ¬B$.
\begin{proposition}
  \label{pr:smoothing/support}
  For any state $\rho ∈ \osub{ℋ¬{AB}}$, there exist optimal states 
  $\rhot, \rhob ∈ \ball{ρ¬{AB}}$ in $\supp{ρ¬A} ⨂ \supp{ρ¬B}$ such that
  $\hmin*{A|B}[ρ] = \hmin{A|B}[\rhot]$ and $\hmax*{A|B}[ρ] = \hmax{A|B}[\rhob]$.
\end{proposition}
\begin{proof}
  Let $ρ¬{ABC}$ be any purification of $ρ$ on $C$, where $ℋ¬C \iso ℋ¬{AB}$. Moreover, let 
  $Π = Π^{ρ¬A} ⨂ Π^{ρ¬B}$ be the projector 
  onto the support of $ρ¬A ⨂ ρ¬B$. Recall that we 
  can express the min- and max-entropies in terms of the functional $Φ$, i.e.\
  \begin{align}
    \hmin{A|B}[ρ] = -\log Φ(ρ) \tn*{and}
      \hmax{A|B}[ρ] = \log Φ¬{A|C}(ρ) ¶[smooth/phi-def].
  \end{align}
  
  For the min-entropy, first consider any state $\rhot' ∈ \ball{ℋ¬{AB}, ρ¬{AB}}$ 
  that achieves the maximum in Def.~\ref{df:smooth}, i.e.\ $\hmin*{A|B}[ρ] = - \log Φ(\rhot')$. 
  We have $Φ(\rhot') ≥ Φ(Π \rhot' Π)$ due to the data processing property of $Φ$
  (Property vi.). This implies
  that projecting onto $Π$ will not decrease the min-entropy.
  Moreover, since $ρ =  Π ρ Π$, we find
  $P(Π \rhot' Π, ρ) ≤ P(\rhot', ρ) ≤ ε$ due to the monotonicity 
  of the purified distance under trace non-increasing maps. Hence, 
  $\rhot = Π \rhot' Π ∈ \ball{ρ¬{AB}}$ and, thus, necessarily
   $\hmin{A|B}[\rhot] = \hmin*{A|B}[ρ]$.

  For the max-entropy, consider a state $\rhob' ∈ \ball{ρ¬{ABC}}$
  whose marginal achieves the minimum in Def.~\ref{df:smooth}. (Note that every
  $\rhob¬{AB} ∈ \osub{ℋ¬{AB}}$ has a purification in this $ε$-ball around $ρ¬{ABC}$ due
  to Theorem~\ref{th:pd-uhl}.) Thus, 
  \begin{align}
    \hmax*{A|B}[ρ] = \log Φ¬{A|C}(\rhob') ≥ \log Φ¬{A|C}(Π^{ρ¬A} \rhob' Π^{ρ¬A}) ≥ 
      \log Φ¬{A|C}(Π \rhob' Π) ¶.
   \end{align}
   Here, the first inequality follows from the data processing property of $Φ$ (property v.) 
   and the second one from data processing on the purifying system (Property vii.).
   Since $\rhob = Π \rhob' Π ∈ \ball{ρ¬{ABC}}$, we get
   $\hmax{A|B}[\rhob] = \hmax*{A|B}[ρ]$.
\end{proof}
Note that these optimal states are not necessarily normalized. In fact, it is in general not possible to find a normalized state in the support of $ρ¬A ⨂ ρ¬B$ that achieves the optimum. However, if $ρ$ is normalized, we can always find normalized optimal states 
if we embed the systems $A$ and $B$
into large enough Hilbert spaces that allow 
smoothing outside the support of $ρ¬A ⨂ ρ¬B$.
For the min-entropy, this is intuitively true since adding weight in a space orthogonal to $A$, 
if sufficiently diluted, will neither affect the min-entropy nor the purified distance. For the max-entropy, the result follows from the duality of the entropies.
\begin{lemma}
  \label{lm:smoothing-normalized}
  Let $ρ ∈ \onorm{ℋ¬{AB}}$. Then, there exists an embedding from $ℋ¬A$ to $ℋ¬{A'}$ 
  and a normalized state $\rhoh¬{A'B} ∈ \ball{ρ¬{A'B}}$ such that 
  $\hmin{A'|B}[\rhoh] = \hmin*{A|B}[ρ]$.
  Moreover, there exist embeddings $ℋ¬A$ to $ℋ¬{A'}$ and $ℋ¬B$ to
  $ℋ¬{B'}$ and a state $\rhoh¬{A'B'} ∈ \ball{ρ¬{A'B'}}$ such that 
  $\hmax{A'|B'}[\rhoh] = \hmax*{A|B}[ρ]$.
\end{lemma}

\begin{proof}
  To prove the property for the min-entropy, let $\{ \rhot¬{AB}, σ¬B \}$ be such that 
  they maximize the smooth min-entropy $λ = \hmin*{A|B}[ρ]$, i.e.\ we have $\rhot¬{AB} ≤ 
  2^{-λ} ⅈ¬A ⨂ σ¬B$. Then we introduce an auxiliary Hilbert space $ℋ¬{\bar{A}}$ with
  dimension $d¬{\bar{A}}$ to be defined and embed $ℋ¬A$ into $ℋ¬{A'} \iso ℋ¬A ⨁ ℋ¬{\bar{A}}$.
  The state $\rhoh¬{A'B} = \rhot¬{AB} ⨁ (1-\trace{\rhot})\, π¬{\bar{A}} ⨂ σ¬B$, where 
  $π¬{\bar{A}} = ⅈ¬{\bar{A}}/d¬{\bar{A}}$,
  satisfies the required conditions. We have
  \begin{align}
    \rhoh¬{A'B} = \rhot¬{AB} ⨁ (1-\trace{\rhot})\, π¬{\bar{A}} 
      ⨂ σ¬B ≤ 2^{-λ} (ⅈ¬{A} ⨁ ⅈ¬{\bar{A}}) ⨂ σ¬B ¶ 
  \end{align}
  if $2^{λ} (1 - \trace{\rhot}) ≤ 2^{λ} ε ≤ d¬{\bar{A}}$, i.e.\ if $d¬{\bar{A}}$ is chosen 
  large enough. This implies that the $\hmin{A'|B}[\rhoh]$ 
  of $\rhoh$ is at least $λ$. Moreover, $F(\rhoh, ρ) = F(\rhot, ρ)$ is not affected by adding
  weight into an orthogonal subspace.
  
  The equivalent statement for the max-entropy follows by duality. 
  For any purification $ρ¬{ABC}$, we have $\hmax*{A|B}[ρ] =
  -\hmin*{A'|C}[\rhoh]$. The state $\rhoh$ has a purification on $B'$ (in fact, this defines
  the embedding) that is $ε$-close $ρ¬{A'B'C}$. Hence, the result follows.
\end{proof}

\section{Properties of the Smooth Entropies}
\label{sc:smooth/prop}

\subsection{Invariance under Isometries}
☼*{unitarily invariant}

The $ε$-smooth min- and max-entropies are independent of the Hilbert spaces used to
represent the density operator locally, as the following lemma
shows.

\begin{proposition}
  \label{pr:smooth-iso}
  Let $ε ≥ 0$ and $ρ¬{AB} ∈ \osub{ℋ¬{AB}}$. Then, for all embeddings $ℋ¬A → ℋ¬{A'}$ and 
  $ℋ¬B → ℋ¬{B'}$, the embedded state $ρ¬{A'B'}$ satisfies
  \begin{align}
    \hmin*{A|B}[ρ] = \hmin*{A'|B'}[ρ] \tn*{and} \hmax*{A|B}[ρ] = \hmax*{A'|B'}[ρ] ¶\,.
  \end{align}
\end{proposition}

\begin{proof}
  Let $U: ℋ¬A → ℋ¬{A'}$ and $V: ℋ¬B → ℋ¬{B'}$ be above-mentioned embeddings and let $ρ¬{ABC}$ 
  be a purification of $ρ$. We also introduce the state $ρ¬{A'B'C} = (U ⨂ V) ρ¬{ABC} (U† ⨂ V†)$, 
  which purifies $ρ¬{A'B'}$.
  
  We first consider the special case $ε = 0$. 
  The property for the min-entropy then follows 
  directly from the invariance of $Φ$, i.e.\ $Φ¬{A|B}(ρ¬{AB}) = Φ¬{A'|B'}(ρ¬{A'B'})$.
  For the max-entropy, we use $Φ¬{A|C}(ρ¬{AC}) = Φ¬{A'|C}(ρ¬{A'C})$, which establishes the 
  equivalence of the max-entropies expressed in terms of the purifications $ρ¬{ABC}$ 
  and $ρ¬{A'B'C}$, respectively.
  
  To extend this to $ε > 0$, we apply the following argument separately to the 
  statement for the smooth
  min-entropy and the statement for the smooth max-entropy.
  The argument for the min-entropy goes as follows. We first introduce states
  $\rhoh¬{AB} ∈ \ball{\supp{ρ¬A ⨂ ρ¬B}; ρ¬{AB}}$ and $\rhob¬{A'B'} ∈ 
  \ball{\supp{ρ¬{A'} ⨂ ρ¬{B'}}; ρ¬{A'B'}}$
  that maximize $\hmin*{A|B}$ and $\hmin*{A'|B'}$, respectively (cf.\ 
  Proposition~\ref{pr:smoothing/support}).
  Since the purified distance is non-increasing under trace 
  non-increasing CPMs (cf.\ Theorem~\ref{th:pd-mono}), we immediately find that 
  \begin{align}
    \rhoh¬{A'B'} = (U ⨂ V) \rhoh¬{AB} (U† ⨂ V†) ∈ \ball{ℋ¬{A'B'}; ρ¬{A'B'}} ¶
  \end{align}  
  is a candidate for the optimization of $\hmin*{A'|B'}$. 
  Hence, 
  \begin{align}
    \hmin*{A'|B'}[ρ] ≥ \hmin{A'|B'}[\rhoh] = \hmin{A|B}[\rhoh] = \hmin*{A|B}[ρ] ¶\,.
  \end{align}
  The same argument applies in the converse direction. There, we define $\rhob¬{AB}$ 
  as the pre-image of $\rhob¬{A'B'}$ under the isometry $U ⨂ V$. (This is possible
  since $\rhob¬{A'B'}$ lies in the support of $ρ¬{A'} ⨂ ρ¬{B'}$ and, thus, in the
  image of $U ⨂ V$.)
  This establishes
  \begin{align}
    \hmin*{A|B}[ρ] ≥ \hmin{A|B}[\rhob] = \hmin{A'|B'}[\rhob] = \hmin*{A|B}[ρ] ¶\,.
  \end{align}
  Therefore, equality holds. 
\end{proof}

Note that the above extension of a non-smooth argument to a smooth argument 
is quite generic and only relies on the monotonicity of the purified distance under trace 
non-increasing maps. We will often use variations of the above technique to lift proofs 
for $ε = 0$ to $ε > 0$.

\subsection{Duality of Smooth Entropies}

The ☼{duality} relation in Lemma~\ref{lm:min-max/dual} extends to smooth entropies.

\begin{theorem}
  \label{th:smooth-dual}
  Let $ε ≥ 0$ and let $ρ ∈ \osub{ℋ¬{ABC}}$ be pure. Then,
  \begin{align}
    \hmax*{A|B}[ρ] = -\hmin*{A|C}[ρ] ¶\,.
  \end{align}
\end{theorem}

\begin{proof}
  According to Proposition~\ref{pr:smooth-iso}, the smooth entropies are invariant under embeddings.
  Let $ρ¬{ABC'}$ be an embedding of $ρ¬{ABC}$ into a space $ℋ¬{AB} ⨂ ℋ¬{C'}$ with $\dim{ℋ¬{C'}} = 
  \max \{ d¬{C} , d¬{AB} \}$. Then,
  \begin{align}
    \hmax*{A|B}[ρ] &= \min_{\rhot\,∈\,\ball{ρ¬{AB}}} \hmax{A|B}[\rhot] 
        \,= \min_{\rhot\,∈\,\ball*{ρ¬{ABC'}}} \hmax{A|B}[\rhot] ¶\\
      &= -\!\! \max_{\rhot\,∈\,\ball*{ρ¬{ABC'}}} \hmin{A|C'}[\rhot] 
        \,≥ -\!\! \max_{\rhot\,∈\,\ball{ρ¬{AC'}}} \hmin{A|C'}[\rhot] ¶\\
      &= - \hmin*{A|C'}[ρ] = - \hmin*{A|C}[ρ] ¶\,.
  \end{align}
  We used that $ℋ¬{C'}$ is chosen large enough to accommodate all purifications
  of states $\rhot ∈ \ball{ρ¬{AB}}$ (cf.\ Theorem~\ref{th:pd-uhl}). However, this also
  means that not all purifications of states $\rhot¬{AC'} ∈ \ball{ρ¬{AC'}}$ can be found
  in $\ball*{ρ¬{ABC'}}$, which leads to the inequality.
  
  On the other hand, we may consider the embedding $ρ¬{AB'C}$ into a space $ℋ¬{AC} ⨂ ℋ¬{B'}$ 
  with $\dim{ℋ¬{B'}} = \max\{ d¬B, d¬{AC} \}$. Then, using the same arguments
  as above,
  \begin{align}
    \hmax*{A|B}[ρ] &= \hmax*{A|B'}[ρ] 
        \,= \!\!\min_{\rhot\,∈\,\ball{ρ¬{AB'}}} \hmax{A|B'}[\rhot] ¶\\
      &≤ \!\min_{\rhot\,∈\,\ball*{ρ¬{AB'C}}} \hmax{A|B'}[\rhot]
        \,= -\!\!\max_{\rhot\,∈\,\ball*{ρ¬{AB'C}}} \hmin{A|C}[\rhot] ¶\\
      &= -\!\!\max_{\rhot\,∈\,\ball{ρ¬{AC}}} \hmin{A|C}[\rhot]
        \,= - \hmin*{A|C}[ρ] ¶\,.\qedhere
  \end{align}
\end{proof}

\subsection{Relation between Smooth Entropies}

We have established in Proposition~\ref{pr:min-max-bounds} that the min-entropy cannot exceed
the max-entropy for normalized states. Here, this result is extended to smooth entropies.
(See also~\cite{vitanov12} for an alternative proof, which inspired the one provided
here.)
\begin{proposition}
  \label{pr:min-max-smooth}
  Let $ρ ∈ \onorm{ℋ¬{AB}}$ and $ε, ε' ≥ 0$ s.t.~$ε + ε' < 1$. 
  Then,
  \begin{align}
    \hmin*{A|B}[ρ] ≤ \hmax[ε']{A|B}[ρ] + \log \frac{1}{1 - (ε + ε')^2} ¶\,.
  \end{align}
\end{proposition}
\begin{proof}
  First, note that we can always embed $ℋ¬A$ into a larger space $ℋ¬{A'}$ such that 
  there exists a normalized state $\rhot¬{A'B} ∈ \ball{ρ¬{A'B}}$ 
  that maximizes $\hmin*{A'|B} = \hmin*{A|B}$ (cf.\ 
  Lemma~\ref{lm:smoothing-normalized}).
  Similarly, by duality, there exists a normalized state $\rhob¬{A'B'} ∈ \ball[ε']{ρ¬{A'B'}}$ that 
  minimizes $\hmax[ε']{A'|B'} = \hmax[ε']{A|B}$. The system $B'$ contains $B$ and the 
  purification of the 
  additional weight introduced to normalize the dual smooth min-entropy.
  
  Hence, by definition of the smooth min-entropy, there exists a state 
  $σ ∈ \onorm{ℋ¬{B'}}$ such that 
  $\rhot¬{A'B'} ≤ 2^{-λ} ⅈ¬{A'} ⨂ σ¬{B'}$ with $λ = \hmin*{A'|B'}[ρ]$. Thus,
  \begin{align}
    \hmax[ε']{A|B}[ρ] &= \hmax[ε']{A'|B'}[ρ] ≥ \log \fidb{\rhob¬{A'B'}}{ⅈ¬{A'} ⨂ σ¬{B'}}^2 ¶\\
      &≥ λ + \log \fidb{\rhob¬{A'B'}}{\rhot¬{A'B'}}^2 
      = λ + \log \big (1 - P^2(\rhob¬{A'B'}, \rhot¬{A'B'}) \big) ¶\\
      &≥ \hmin*{A|B}[ρ] - \log \frac{1}{ 1 - (ε + ε')^2 } ¶\,.
  \end{align}
  The first inequality follows from the definition of the smooth max-entropy, together 
  with the fact that we took a particular $σ$ instead of optimizing over all $σ$. The second
  inequality is a simple application of the operator inequality that defines $\hmin*{A'|B}$.
  Finally, the triangle inequality for the purified distance establishes 
  $P(\rhob, \rhot) ≤ ε + ε'$.
\end{proof}

\begin{remark}
  \label{rm:min-max-smooth}  
  The term $-\log \big( 1 - (ε + ε')^2 \big)$ can be 
  reduced by the use of Eq.~§[pd/tight-triangle] instead of the triangle inequality for the 
  purified distance. Hence,
  \begin{align}
    \hmin*{A|B}[ρ] ≤ \hmax[ε']{A|B}[ρ] + 
      \log \frac{1}{1 - \left( ε √{1-ε'^2} + ε' √{1-ε^2} \right)^2} ¶.
  \end{align}
  The range of allowed pairs $\{ε, ε'\}$ is extended to those satisfying
    $\arcsin(ε) + \arcsin(ε') < \frac{π}{2}$.
  In particular, this means that the term is finite if 
  we choose $ε' = 1 - ε$, for any $0 < ε < 1$. (See also Figure~\ref{fg:tight}.)
\end{remark}

Proposition~\ref{pr:min-max-smooth} implies that smoothing states that have similar min- and max-entropies has almost 
no effect. More precisely, let $ρ ∈ \onorm{ℋ¬{AB}}$ be such that $\hmin{A|B}[ρ] = \hmax{A|B}[ρ] - δ$. Then, 
\begin{align}
  \hmin*{A|B}[ρ] ≤ \hmax{A|B}[ρ] - \log ( 1 - ε^2 ) = \hmin{A|B}[ρ] + δ - \log ( 1 - ε^2 ) ¶.
\end{align}
For $δ = 0$, this inequality is tight. Note that the smoothed state
$\rhot = ρ (1-ε^2) \ecl ρ$ reaches equality in this case.
A corresponding inequality can be derived for the
smooth max-entropy.

\section{Data Processing Inequalities}
\label{sc:smooth/data-proc}
☼*{data processing}

We expect measures of uncertainty about the system $A$ given side
information $B$ to be non-decreasing under local physical operations 
(e.g.\ measurements or unitary evolutions) applied to the $B$ system. Such operations can
be described most generally by ☼[TP-CPMs]{TP-CPM}. Here, we show that the smooth entropies, 
$\hmin*{A|B}$ and $\hmax*{A|B}$, have this property.

Another data processing inequality concerns rank-$1$
☼[projective measurements]{measurement!projective} of the 
system $A$. Such measurements can be
described in terms of an orthonormal basis $\{ \ket{x} \}$ of
$ℋ¬A$ and a measurement TP-CPM $ℳ ∈ \tpcpm{ℋ¬A, ℋ¬X}$ from $ℋ¬A$ to $ℋ¬X
\iso ℋ¬A$, which maps $ρ$ to $∑_x \braket{x|ρ|x}\,\proj{x}$. We expect 
that the uncertainty about the system $A$ as well as the smooth entropies, 
$\hmin*{A|B}$ and $\hmax*{A|B}$, do not decrease when such a measurement 
is executed on the $A$ system.

In fact, we show a more general theorem that encompasses the two example data processing
properties above. (Note, in particular, that $ℳ$ is a ☼{unital} map.)

\begin{theorem}[Generalized Data Processing]
  \label{th:data-proc} 
  Let $ε ≥ 0$ and $ρ ∈ \osub{ℋ¬{ABC}}$. 
  Moreover, let $ℰ$ be a sub-unital and trace non-increasing 
  CPM from $ℋ¬A$ to $ℋ¬A'$, let $\sF$ be a trace non-increasing CPM from
  $ℋ¬B$ to $ℋ¬{B'}$ and let $\sG$ be a trace non-increasing CPM from $ℋ¬C$ to $ℋ¬{C'}$. 
  Then, the state $τ¬{A'B'C'} = (\sG \circ \sF \circ ℰ)[ρ¬{ABC}]$ satisfies
  \begin{align}
    \hmin*{A|B}[ρ] ≤ \hmin*{A'|B'}[τ] ¶\,.
  \end{align}
  Furthermore, if $ℰ$, $\sF$ and $\sG$ are also trace preserving, then
  \begin{align}
    \hmax*{A|B}[ρ] ≤ \hmax*{A'|B'}[τ] ¶\,.
  \end{align}
\end{theorem}

\begin{proof}
  We first prove the result for the smooth min-entropy.
  Let $\rhot ∈ \ball{ρ¬{AB}}$ be the
  state that maximizes the smooth min-entropy. Then, data processing 
  of $Φ$ (Properties vi. and vii.) implies
  \begin{align}
    \hmin*{A|B}[ρ] &= \hmin{A|B}[\rhot] ¶\\
      &= -\log Φ¬{A|B}(\rhot) ¶\\
%      &≤ -\log Φ¬{A|B'}\big(ℰ[\rhot]\big) ¶\\
      &≤ -\log Φ¬{A'|B'}\big((\sF \circ ℰ)[\rhot]\big) ¶\\
      &≤ -\log Φ¬{A'|B'}\big((\sG \circ \sF \circ ℰ)[\rhot]\big) ¶\,.
 \end{align}
  We now introduce the state $\taut = (\sG \circ \sF \circ ℰ)[\rhot]$,
  which is $ε$-close to $\tau$ due to the monotonicity of the purified distance under
  trace non-increasing maps (cf.\ Theorem~\ref{th:pd-mono}). Thus, we conclude
  \begin{align}
    \hmin*{A|B}[ρ] &≤ -\log Φ¬{A'|B'}(\taut) 
      = \hmin{A'|B'}[\taut] ≤ \hmin*{A'|B'}[τ] ¶\,.
  \end{align}
  
  To prove the result for the max-entropy,  we take advantage of the ☼{Stinespring dilation}
  of two TP-CPM maps $ℰ$ and $\sF$. Namely, we introduce the isometries $U: A → A'A''$ and
  $V: B → B'B''$ and the state $τ¬{A'A'B'B''} = (U ⨂ V) ρ¬{AB} (U† ⨂ V†)$ of 
  which $τ¬{A'B'}$ is a marginal. Let $\taut ∈ \ball{τ¬{A'A''B'B''}}$ be the state that minimizes
  the smooth max-entropy $\hmax*{A'|B'}[τ]$. Then,
  \begin{align}
    \hmax*{A'|B'}[τ] &= \max_{σ} \log \fidb{\taut¬{A'B'}}{ⅈ¬{A'} ⨂ σ¬{B'}}^2 ¶\\
      &≥ \max_{σ} \log \fidb{\taut¬{A'B'}}{\tr[A'']{Π¬{A'A''}} ⨂ σ¬{B'}}^2 ¶[data-proc1]\,.
  \end{align}
  We introduced the projector $Π¬{A'A''} = U U†$ onto the image of $U$, which exhibits
  the following property  due to the fact that $ℰ$ is sub-unital:
  \begin{align}
    \tr[A'']{Π} = \tr[A'']{U ⅈ¬A U†} = ℰ[ⅈ¬A] ≤ ⅈ¬{A'} ¶.
  \end{align}
  The inequality in~§[data-proc1] is then a result of the fact that the fidelity
  is non-increasing when an argument $A$ is replaced by a smaller argument $B ≤ A$.
  (Property v.\ of the fidelity). Next, we use the monotonicity of the fidelity under
  partial trace~§[fid/mon] to bound~§[data-proc1] further.
  \begin{align}
    \hmax*{A'|B'}[τ] &≥ \max_{σ} \log \fidb{\taut¬{A'A''B'B''}}{Π¬{A'A''} 
        ⨂ σ¬{B'B''}}^2 ¶\\
      &= \max_{σ} \log \fidb{Π¬{A'A''} \taut¬{A'A''B'B''} Π¬{A'A''}}{ⅈ¬{A'A''} 
        ⨂ σ¬{B'B''}}^2 ¶\\
      &= \hmax{A'A''|B'B''}[\tauh] ¶\,.
  \end{align}
  Finally, we note that  the projector $\tauh = Π \taut Π ∈ \ball{ρ¬{A'A''B'B''}}$ due
  to the monotonicity of the purified distance under trace non-increasing maps. Hence,
  we established $\hmax*{A'|B'}[τ] ≥ \hmax*{A'A''|B'B''}[τ] = \hmax{A|B}[ρ]$, where the
  last equality follows due to the invariance of the max-entropy under embeddings 
  (Proposition~\ref{pr:smooth-iso}).
\end{proof}

Note that a generalization of the data processing results to trace non-increasing maps is not possible for the max-entropy. For example, the max-entropy can decrease when a projection is applied to the $B$ system.

\section{Classical Information}
\label{sc:smooth/class}

Here we discuss smooth entropies of quantum states that encode partially classical information. To maintain full generality in the following arguments, we consider states on a four-partite system composed of $X$, $Y$, $A$ and $B$, where $X$ and $Y$ are classical ☼[registers]{register}. The smooth entropies are evaluated for the classical register $X$ and the quantum register $A$ conditioned on classical side information $Y$ and quantum side information $B$.

In the following, we say that the state $ρ ∈ \osub{ℋ¬{XYAB}}$ is ☼[classical-classical-quantum-quantum (CCQQ)]{state!CQ}☼*{CQ|see{state}} if it can written in the form
\begin{align}
	ρ¬{XYAB} = ∑_{x,y} p_{xy}\, \proj{x}[X] ⨂ \proj{y}[Y] ⨂ ω¬{AB}^{xy}\,, 
	\ \tn{where}\ ω¬{AB}^{xy} ∈ \osub{ℋ¬{AB}},
	¶[CCQQ]
\end{align}
the kets $\{ \ket{x} \}$, $\{ \ket{y} \}$ are orthonormal bases of $ℋ¬X$ and $ℋ¬Y$, respectively, and $\{ p_{xy} \}$ is a probability distribution. This state has the property that it remains invariant under a projective measurement of $X$ in the basis $\{ \ket{x} \}$ and under a projective measurement of $Y$ in the basis $\{ \ket{y} \}$. Formally, we define  measurement ☼[TP-CPMs]{TP-CPM}
\begin{align}
  \sM¬X : ρ ↦ ∑_x \braket{x|ρ|x}[X] \proj{x}[X] 
  \ \tn{and}\
  \sM¬Y : ρ ↦ ∑_y \braket{y|ρ|y}[Y] \proj{y}[Y]  ¶[smooth/class/measure]
\end{align}
and note that $\sM¬X [ ρ ] = \sM¬Y [ ρ ] = ρ$. Furthermore, given an arbitrary state $τ ∈ \osub{ℋ¬{XYAB}}$, it is easy to verify that the measured state $(\sM¬Y \circ \sM_X) [τ]$ is of CCQQ form~§[CCQQ]. 

Partially classical states of the form~§[CCQQ] can be purified in a way that preserves their structure. For this, we introduce purifying Hilbert spaces $ℋ¬{X'} \iso ℋ¬X$, $ℋ¬{Y'} \iso ℋ¬Y$ and $ℋ¬C \iso ℋ¬{AB}$. Then, one may purify the states $ω¬{AB}^{xy}$ individually on $\osub{ℋ¬{ABC}}$, which allows us to write down a purification $\ket{ρ} ∈ ℋ¬{XX'YY'ABC}$ of the form
\begin{align}
  \ket{ρ}[XX'YY'ABC] = ∑_{x,y} √{p_{xy}}\, \ket{x}[X] ⨂ \ket{x}[X'] ⨂ \ket{y}[Y] ⨂ 
  \ket{y}[Y'] ⨂ \ket{ω^{xy}}[ABC] ¶[CCQQ-pure]\, .
\end{align}

It is easy to verify that purifications of this form commute with the following projectors,
\begin{align}
  Π¬{XX'} := ∑_{x} \proj{x}[X]\! ⨂\! \proj{x}[X']\ \tn{and}\ 
  Π¬{YY'} := ∑_{y} \proj{y}[Y]\! ⨂\! \proj{y}[Y'] ¶[smooth/class/proj]\,.
\end{align}
In the converse, any state $τ ∈ \osub{ℋ¬{XX'YY'AB}}$ that commutes with both $Π¬{XX'}$ and $Π¬{YY'}$ has a CCQQ marginal $τ¬{XYAB}$. We say that the two systems $X$ and $X'$ as well as $Y$ and $Y'$ are ☼{coherent classical} pairs. 

Some important properties of min- and max-entropies of coherent classical states are discussed in Appendix~\ref{ap:lemmas}, Section~\ref{se:lemmas/cc} and will be used in the following.

\subsection{Smoothing of Classical States}

First, we show that\,|\,in order to smooth min- and max-entropies\,|\,it is sufficient to consider a ball of close states that are classical on the same subsystems as the original state. This is formalized in the following proposition.
\begin{proposition}[Classical Smoothing]
  \label{pr:class-smooth}
  Let $ρ ∈ \osub{ℋ¬{XYAB}}$ be classical on $X$ and $Y$ and let $\eps ≥ 0$. Then,
  there exist states $\rhob, \rhot ∈ \ball{ρ}$ that are classical on $X$ and $Y$ such that
  \begin{align}
    &\hmin*{XA|YB}[ρ] = \hmin{XA|YB}[\rhob] \tn*{and} ¶\\
    &\hmax*{XA|YB}[ρ] = \hmax{XA|YB}[\rhot] ¶\,.
  \end{align}
\end{proposition}
\begin{proof}
  We first prove the statement for the min-entropy.
  By definition of the smooth min-entropy, there exists a state 
  $\rhoh ∈ \ball{ρ}$ such that $\hmin*{XA|YB}[ρ] = \hmin{XA|YB}[\rhoh]$. 
  We propose the CCQQ state $\rhob = (\sM¬X \circ \sM¬Y) [\rhoh]$ 
  using the measurement in~§[smooth/class/measure] as a candidate. Due to the data processing 
  inequalities in Theorem~\ref{th:data-proc}, we have
  \begin{align}
    \hmin{XA|YB}[\rhoh] ≤ \hmin{XA|YB}[\rhob] ¶\,.
  \end{align} 
  Furthermore, $\rhob \ecl ρ$ 
  due to Theorem~\ref{th:pd-mono} and, hence, $\rhob$ satisfies 
  all conditions of the proposition.

  To prove the statement for the max-entropy, we consider the dual problem 
  for the min-entropy using the purification $\ket{ρ}$ in~§[CCQQ-pure].
  Namely, we need to show that there exists a state $\rhot ∈ \ball{ρ}$ 
  that satisfies
  \begin{align}
    \hmin*{XA|X'Y'C}[ρ] = \hmin{XA|X'Y'C}[\rhot] ¶\,,
  \end{align}  
  commutes with $Π¬{XX'}$, and is classical on $Y'$. 
  Due to the first part of this lemma, there 
  exists a state~$\rhoh ∈ \ball{ρ}$ such that $\hmin*{XA|X'Y'C}[\rho] = 
  \hmin{XA|X'Y'C}[\rhoh]$ and $\rhoh$ is classical on $Y'$. We propose the
  state $\rhot = Π¬{XX'} \rhoh\, Π¬{XX'}$ as a candidate. Then,
  Lemma~\ref{lm:smooth/class/class-entangled-min} and 
  Theorem~\ref{th:pd-mono} establish that $\rhot$ satisfies all 
  conditions, which conclude the proof.
\end{proof}

\subsection{Entropy of Classical Information}

Here, we bound the smooth entropies of a classical system $X$. For full generality, we consider the state $ρ¬{XAB}$, which is simply the state~§[CCQQ] with trivial classical subsystem $Y$.
\begin{proposition}
  \label{pr:class/bounds-1}
  Let $ρ ∈ \osub{ℋ¬{XAB}}$ be classical on $X$ and let $0 ≤ ε < 1$. Then,
  \begin{align}
    \hmin*{A|B}[ρ] &≤ \hmin*{XA|B}[ρ] ≤ \hmin*{A|B}[ρ] + \log d¬{X} \tn*{and} 
      ¶[class/bounds-1/min]\\
    \hmax*{A|B}[ρ] &≤ \hmax*{XA|B}[ρ] ≤ \hmax*{A|B}[ρ] + \log d¬{X}
      ¶[class/bounds-1/max].
  \end{align}
\end{proposition}
\begin{proof}
  To show the first inequality of §[class/bounds-1/min], 
  we consider the dual problem for the max-entropy.
  Namely, we need to show that 
  \begin{align}
    \hmax*{XA|X'C}[ρ] ≤ \hmax*{A|CXX'}[ρ] ¶
  \end{align} 
  for the state $ρ¬{XX'AC}$ that is coherent classical between $X$ and $X'$. 
  To do this, note that there exists a $\rhot ∈ \ball{ρ}$ with 
  $\hmax*{A|CXX'}[ρ] = \hmax{A|CXX'}[\rhot]$ that has support on 
  $\supp{\rhoA} ⨂ \supp{\ii{ρ}{XX'C}}$ (cf.\ Proposition~\ref{pr:smoothing/support}) and, thus,
  commutes with $Π¬{XX'}$. The first inequality now follows after we apply 
  Lemma~\ref{lm:smmoth/class/class-entangled-max} to $\rhot$ and realize that 
  \begin{align}
    \hmax*{XA|X'C}[ρ] ≤ \hmax{XA|X'C}[\rhot] ≤ \hmax{A|XX'C}[\rhot] ¶\,.
  \end{align}

  Since smoothing can be restricted to states that are classical on $X$ 
  (cf.~Proposition~\ref{pr:class-smooth}), there exist states 
  $\rhob ∈ \ball{ρ}$ and $σ ∈ \onorm{ℋ¬B}$ such that
  \begin{align}
    \rhob¬{XAB} = ∑_x \bar{p}_x\, \proj{x}[X] ⨂ \taub¬{AB}^x ≤ 
      2^{-\hmin*{XA|B}[ρ]}\, ⅈ¬{XA} ⨂ σ¬B ¶\,.
  \end{align}
  Since both sides of the inequality are block diagonal in $X$, this implies that  
 \begin{align}
    &\forall_x : \bar{p}_x\, \taub¬{AB}^x ≤ 2^{-\hmin*{XA|B}[ρ]}\, ⅈ¬{A} ⨂ σ¬B ¶\\
    ⇒\ &\rhob¬{AB} = ∑_x \bar{p}_x\, \taub¬{AB}^x ≤ d¬{X}\, 
      2^{-\hmin*{XA|B}[ρ]}\, ⅈ¬{A} ⨂ σ¬B ¶\,.
  \end{align}
  This implies the second inequality of~§[class/bounds-1/min].

  The first inequality of~§[class/bounds-1/max] is equivalent to the statement
  \begin{align*}
    \hmin*{A|XX'C}[ρ] ≥ \hmin*{XA|X'C}[ρ] 
  \end{align*}
  for a state $ρ$ that is coherent classical between $X$ and $X'$. The inequality then 
  follows directly
  from Lemma~\ref{lm:smooth/class/class-entangled-min} applied to the state $\taub ∈ \ball{ρ}$
  that maximizes $\hmin*{XA|X'C}[ρ]$ and Theorem~\ref{th:pd-mono}, which 
  establishes that $\rhob = Π¬{XX'} \taub\, Π¬{XX'}$ is a candidate 
  for the maximization of $\hmin*{A|XX'C}[ρ]$.
  
  The second inequality of~§[class/bounds-1/max] is shown as follows.
  By the definition of the smooth max-entropy, there exists a state 
  $\rhot ∈ \ball{ρ}$ and a state $σ ∈ \onorm{ℋ¬{B}}$ such that
  \begin{align}
    \hmax*{A|B}[ρ] &≥ 2 \log F \big( \rhot¬{AB} , ⅈ¬A ⨂ σ¬B \big) ¶\\
    &≥ 2 \log F \big( \rhot¬{XAB} , \frac{1}{d¬{X}} ⅈ¬{XA} ⨂ σ¬B \big) 
    ≥ \hmax*{XA|B}[ρ] - \log d¬{X} ¶\,,
  \end{align}
  where $σ$ maximizes $\hmax{XA|B}[\rhot]$ and $\rhot$ minimizes $\hmax*{A|B}$.
\end{proof}

\subsection{Conditioning on Classical Information}

We investigate the maximum amount of information a classical register $X$ can contain about
a quantum state $A$. (Various relations of this type have appeared in the literature, see, 
e.g.~\cite{renesrenner10, winkler11}. The following lemma generalizes these results.)

\begin{proposition}
  \label{pr:class/bounds-2}
  Let $ρ ∈ \osub{ℋ¬{XAB}}$ be classical on $Y$ and let $0 ≤ ε < 1$. Then,
  \begin{align}
    \hmin*{A|B}[ρ] &≥ \hmin*{A|BY}[ρ] ≥ \hmin*{A|B}[ρ] - \log d¬{Y} \tn*{and} 
      ¶[class/bounds-2/min]\\
    \hmax*{A|B}[ρ] &≥ \hmax*{A|BY}[ρ] ≥ \hmax*{A|B}[ρ] - \log d¬{Y} 
      ¶[class/bounds-2/max].
  \end{align}
\end{proposition}

\begin{proof}
  The first inequalities of both~§[class/bounds-2/min] and~§[class/bounds-2/max] directly follow
  from data processing (cf.\ Theorem~\ref{th:data-proc}) applied to the TP-CPM $\trace¬{Y}$.
  
  Let $\rhot ∈ \ball{ρ¬{ABY}}$ be a state that achieves the maximal min-entropy for
  $\hmin*{AY|B}[ρ] = \hmin{AY|B}[\rhot] = λ$. 
  %This state can be chosen classical on $Y$ 
  %due to Theorem~\ref{th:class-smooth}. 
  In particular, there exists a state $σ ∈ \osub{ℋ¬B}$ s.t.\
  \begin{align}
    %\rhot¬{ABY} = \sum_y \proj{y} ⨂ \taut¬{AB}^y ≤ 2^{-λ} ⅈ¬{AY} ⨂ σ¬B ¶.
    \rhot¬{ABY} ≤ 2^{-λ} ⅈ¬{AY} ⨂ σ¬B = d¬Y\, 2^{-λ} ⅈ¬A ⨂ π¬Y ⨂ σ¬B¶.
  \end{align}
  This implies that $\hmin*{A|BY}[ρ] ≥ \hmin*{AY|B}[ρ] - \log d¬Y$. The second
  inequality in~§[class/bounds-2/min] then directly follows from 
  Proposition~\ref{pr:class/bounds-1}.
  
  To establish the second inequality in~§[class/bounds-2/max], we prove the dual property, i.e.\
  for the purification $ρ¬{ABCYY'}$ that is coherent classical between $Y$ and $Y'$, we show
  that $\hmin*{A|CY'}[ρ] ≤ \hmin*{A|CYY'}[ρ] + \log d¬Y$. Let $\rhot ∈ \ball{ρ¬{ACY'}}$ 
  and $σ ∈ \osub{ℋ¬{CY'}}$ be such that
  \begin{align}
    \rhot¬{ACY'} ≤ 2^{-λ} ⅈ¬A ⨂ σ¬{CY'}, \tn*{where} λ = \hmin*{A|CY'}[ρ] ¶.
  \end{align}
  Lemma~\ref{lm:op-bound} implies that any extension $\rhot¬{ACYY'} ∈ \ball{ρ}$ satisfies 
  \begin{align}
    \rhot¬{ACYY'} ≤ d¬Y\, \rhot¬{ACY'} ⨂ ⅈ¬{Y} ≤ 2^{-λ} d¬Y ⅈ¬{YA} ⨂ σ¬{CY'} ¶.
  \end{align}
  Hence, we find the following inequality
  \begin{align}
    \hmin{AY|CY'}[\rhot] ≥ \hmin*{A|CY'}[ρ] - \log d¬Y ¶.
  \end{align}
  An application of Lemma~\ref{lm:smooth/class/class-entangled-min} to the min-entropy 
  on the lhs.\ establishes
  \begin{align}
    \hmin{A|CYY'}[\taut] ≥ \hmin*{A|CY'}[ρ] - \log d¬Y, \tn*{where} 
      \taut = Π¬{YY'} ρ Π¬{YY'} ¶
  \end{align}
  is coherent classical between $Y$ and $Y'$. Since $ρ$ has the same property, 
  Theorem~\ref{th:pd-mono} implies that $\taut ∈ \ball{ρ}$, which concludes the proof.
\end{proof}

\subsection{Functions on Classical Registers}

Let us again consider the state $ρ¬{AXBY}$ of~§[CCQQ]. On the one hand, 
applying a function $f: Y → Z$ on the register $Y$ is a special case of data processing 
(cf.\ Theorem~\ref{th:data-proc}) and we immediately find
\begin{align}
  \hmin*{A|BY}[ρ] ≤ \hmin*{A|BZ}[ρ] \tn*{and} \hmax*{A|BY}[ρ] ≤ \hmax*{A|BZ}[ρ] ¶.
\end{align}

On the other hand, Proposition~\ref{pr:class/bounds-2} can be used to show that applying
a function on the register $X$ cannot increase the smooth entropies. (Note that for
the min-entropy, using its interpretation as a ☼{guessing probability}, 
this corresponds to the intuitive statement that it is
always at least as hard to guess the input of a function than to guess its output.)

\begin{proposition}
  \label{pr:func}
  Let $ρ¬{XAB} = ∑ p_x\, \proj{x} ⨂ ω¬{AB}^x$ with $ω¬{AB}^x ∈ \osub{ℋ¬{AB}}$ 
  be classical on $X$. Furthermore, let $0 ≤ ε < 1$ and let $f: X → Z$ be
  a function. Then, the state $τ¬{ZAB} = ∑ p_x\, \proj{f(x)}[Z] ⨂ ω¬{AB}^x$ satisfies
  \begin{align}
    \hmin*{ZA|B}[τ] ≤ \hmin*{XA|B}[ρ] \tn*{and} \hmax*{ZA|B}[τ] ≤ \hmax*{XA|B}[ρ] ¶.
  \end{align}
\end{proposition}

\begin{proof}
  A possible ☼{Stinespring dilation} of $f$ is given by the isometry 
  $U: \ket{x}[X] ↦ \ket{x}[X'] ⨂ \ket{f(x)}[Z]$ followed by a partial trace over $X'$.
  Applying $U$ on $ρ¬{XAB}$, we get
  \begin{align}
    τ¬{XZAB} := U ρ¬{XAB} U† = ∑ p_x\, \proj{x}[X'] ⨂ \proj{f(x)}[Y] ⨂ ω¬{AB}^x ¶
  \end{align}
  which is classical on $X'$ and $Z$ and an extension of $τ¬{ZAB}$. Hence, the invariance under
  isometries of the smooth entropies (cf.\ Proposition~\ref{pr:smooth-iso}) in conjunction with 
  Proposition~\ref{pr:class/bounds-2} implies
  \begin{align}
    \hmin*{XA|B}[ρ] = \hmin*{XZA|B}[τ] ≥ \hmin*{ZA|B}[τ] ¶.
  \end{align}
  An equal argument for the smooth max-entropy concludes the proof.
\end{proof}

\section{Chain Rules}
\label{sc:chain}

The ☼{chain rule} for the von Neumann (or Shannon) entropy is the relation $\hvn{AB|C} = \hvn{A|BC} + \hvn{B|C}$. To see that the von Neumann entropy satisfies this relation, simply substitute the definition of the conditional entropy, i.e.\ $\hvn{A|B} = \hvn{AB} - \hvn{B}$.
For the smooth min- and max-entropies, this rule does not hold with equality. 

Instead, we give a collection of inequalities that replace the chain rule for the smooth min- and max-entropies. These were recently derived in~\cite{vitanov11,vitanov12}. (See 
also~\cite{tomamichel10,berta10} for preliminary results.)
These chain rules generalize some of the above results for classical registers to general quantum systems; however, they introduce an additional, necessary smoothing parameter, $ε$, and correction terms in $\log (2/ε^2)$ that do not appear in the results of the previous section. This is due to the fact that, previously, we gave bounds in terms of the dimension of the classical register instead of its max-entropy. This simplifies the analysis and is sufficient in many applications involving classical information.

The chain rules provided here should thus be considered complementary 
the other results of this chapter.
\begin{theorem}
  \label{th:chain-rules}
  Let $ε > 0$, $ε' ≥ 0$, $ε'' ≥ 0$ and $ρ ∈ \osub{ℋ¬{ABC}}$. Then,
  the following chain rules hold
  \begin{align}
    \hmin[ε+2ε'+ε'']{AB|C}[ρ] &≥ \hmin[ε']{A|BC}[ρ] + \hmin[ε'']{B|C}[ρ] - \log \frac{2}{ε^2} ¶[cr/min-ab],\\
    \hmax[ε+ε'+2ε'']{AB|C}[ρ] &≤ \hmax[ε']{A|BC}[ρ] + \hmax[ε'']{B|C}[ρ] + \log \frac{2}{ε^2} ¶[cr/max-ab],\\
    \hmin[ε+3ε'+2ε'']{A|BC}[ρ] &≥ \hmin[ε']{AB|C}[ρ] - \hmax[ε'']{B|C}[ρ] - 2 \log \frac{2}{ε^2} ¶[cr/min-a],\\
    \hmax[2ε+ε'+2ε'']{A|BC}[ρ] &≤ \hmax[ε']{AB|C}[ρ] - \hmin[ε'']{B|C}[ρ] + 3 \log \frac{2}{ε^2} ¶[cr/max-a],\\
    \hmin[2ε+ε'+2ε'']{B|C}[ρ] &≥ \hmin[ε']{AB|C}[ρ] - \hmax[ε'']{A|BC}[ρ] - 3 \log \frac{2}{ε^2} ¶[cr/min-b],\\
    \hmax[ε+3ε'+2ε'']{B|C}[ρ] &≤ \hmax[ε']{AB|C}[ρ] - \hmin[ε'']{A|BC}[ρ] + 2 \log \frac{2}{ε^2} ¶[cr/max-b].
  \end{align}
\end{theorem}

\noindent For a proof, we refer to~\cite{vitanov11}. Note that Eqs.~§[cr/max-ab], §[cr/max-a] and~§[cr/max-b] follow by duality of the smooth entropies from Eqs.~§[cr/min-ab], §[cr/min-b] and~§[cr/min-a], respectively.

%............................

\chapter{The Quantum Asymptotic Equipartition Property}
\label{ch:aep}

The classical asymptotic equipartition property (AEP) is the statement that, in the limit of a large number of identical repetitions of a random experiment, the output sequence is virtually certain to come from the typical set, each member of which is almost equally likely. In this chapter, expanding on previous results in~\cite{tomamichel08}, a fully quantum generalization of this property is shown, where both the output of the experiment and side information are quantum. We give an explicit bound on the convergence, which is independent of the dimensionality of the side information. 
☼*{AEP}
☼*{asymptotic equipartition|see{AEP}}

%...........................

\section{Introduction and Related Work}
\label{se:aep/res}

One of the pivotal results of classical information theory is the ☼[asymptotic
equipartition property]{AEP} (AEP). This result justifies
the use of the Shannon entropy to characterize many information theoretic tasks.
It concerns a sequence of ☼[independent and identically distributed]{i.i.d.}
(i.i.d.) random variables $X^n = (X_1, X_2, \dots, X_n)$ and states
that~\cite{cover91}
\begin{align}
   - \frac{1}{n} \log P¬{X^n}(x^n) → \hvn{X}[\tinyP] \quad \tn{in probability} ¶.
\end{align}
Here, $P¬{X^n}(x^n) = P¬{X^n}(x_1, x_2, \dots, x_n) = P¬X(x_1) P¬X(x_2) \dots P¬X(x_n)$ denotes
the probability of the independent events $x_1, x_2, \dots, x_n$.
More precisely, the AEP states that, for any $0 < ε < 1$, any $μ > 0$ and for large enough $n$, a randomly
chosen sequence $(x_1, x_2, \dots, x_n)$ is with probability more than $1 - ε$
in a ☼{typical set} of sequences that satisfy
\begin{align}
  H(X) - μ < - \frac{1}{n} \log P¬{X^n}(x^n) < H(X) + μ ¶[AEP/it]\,.
\end{align}
This typical set of sequences of length $n$ is denoted $A^{n}_{μ}$.

The AEP can be stated equivalently in terms of entropies. Recall the
min- and max-entropies of a classical random variable $X^n$, which
satisfy
\begin{align}
  \hmin{X^n}[\tinyP] &= \hh{∞}{X^n}[\tinyP] = \min_{x^n} - \log P¬{X^n}(x^n) \tn*{and} ¶\\
  \hmax{X^n}[\tinyP] &= \hh{½}{X^n}[\tinyP] ≤ \max_{x^n\!,\, P(x^n) ≠ 0} 
     \!- \log P¬{X^n}(x^n) ¶[AEP/min-max-bounds]\,.
\end{align}

The bounds in~§[AEP/it] can now be expressed in terms of the smooth
min- and max-entropy. For this purpose, we consider the probability
distribution $Q$ which restricts $P$ to the typical set. It is
formally defined as
\begin{align}
  Q¬{X^n}(x^n) = \begin{cases}
             P¬{X^n}(x^n) / c & \tn{if} x^n ∈ A^{n}_{μ} \\ 
             0 & \tn{else}
           \end{cases}, \tn*{where}\! c = ∑_{x^n ∈ A^{n}_{μ}} P¬{X^n}(x^n) ¶\,.
\end{align}

The AEP now states that for large enough $n$, the fidelity between
the distributions $P$ and $Q$ satisfies
\begin{align}
  F(P, Q)^2 = \Big( ∑_{x^n} √{P¬{X^n}(x^n) Q¬{X^n}(x^n)} \Big)^2 = ∑_{x^n ∈ A^{n}_{μ}}
P¬{X^n}(x^n) > 1 - ε ¶\,, 
\end{align}
and, thus, the purified distance is upper bounded by $√{ε}$. Finally, the AEP in~§[AEP/it], in conjunction with~§[AEP/min-max-bounds],
implies that, for any $0 < ε < 1$, $μ > 0$ and large enough $n$,
\begin{align}
  \frac{1}{n} \hmin[√{ε}]{X^n}[\tinyP] &> \frac{1}{n} \hmin{X^n}[\tinyQ] > \hvn{X}[\tinyP] - μ
    \tn*{and} ¶\\
  \frac{1}{n} \hmax[√{ε}]{X^n}[\tinyP] &< \frac{1}{n} \hmax{X^n}[\tinyQ] < \hvn{X}[\tinyP] + μ
    ¶[AEP/ent]\,.
\end{align}

Conversely, its entropic form~§[AEP/ent] implies an AEP of the
form~§[AEP/it]. We only roughly sketch this argument here. First note that the
fact that the smooth min-entropy converges to the von Neumann entropy implies
the existence of a set $\tilde{A}^{n}$ with the property $\max \{P(x^n) : x^n ∈
\tilde{A}^{n} \} ≈ 2^{-n \hvn{X}}$ and $x^n ∈ \tilde{A}^{n}$ with probability almost
$1$. Similarly, the convergence of the max-entropy (together with the fact that
the smooth max-entropy is related to smooth Rényi entropy of order $0$) implies
the existence of another set $\bar{A}^{n}$ with $|\bar{A}^{n}| ≈ 2^{n \hvn{X}}$
and $x^n ∈ \bar{A}^{n}$ with probability almost $1$.
Since both these sets occur with probability almost $1$, it is easy to see
from the union bound that their intersection, $\tilde{A}^{n} \cap \bar{A}^{n}$, constitutes a typical
set.

\begin{figure}
  \centering
  \includegraphics{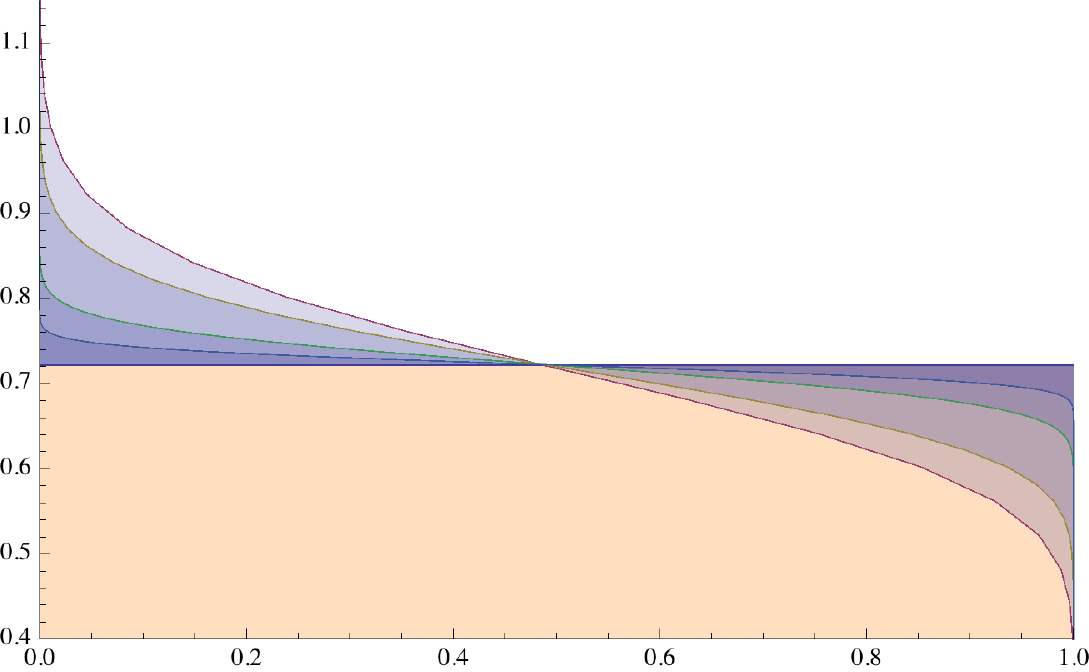}
  \caption[Emergence of the Typical Set.]{\emph{Emergence of Typical Set.}
  We consider $n$ independent Bernoulli trials with $p = 0.2$ and denote the probability that an 
  event $x^n$ (a bit string of length $n$) occurs by $P¬{X^n}(x^n)$.
  The plot shows the surprisal per round, $- \frac{1}{n} \log P¬{X^n}(x^n)$, over the cumulated 
  probability of the events. 
  The curves for $n = \{ 50, 100, 500, 2500 \}$ converge to the von Neumann entropy, 
  $\hvn{X} \approx 0.72$. This indicates that, for large $n$, most (in probability) 
  events are close to typical and have surprisal $\hvn{X}$ per round.
  
  The min-entropy, $\hmin{X} \approx 0.32$, constitutes the minimum of the curves while the max-entropy, $\hmax{X} \approx 0.85$, is upper bounded by their maximum. Moreover, the respective $ε$-smooth entropies, $\frac{1}{n} \hmin*{X^n}$ and $\frac{1}{n} \hmax*{X^n}$, can be approximately obtained by cutting off a probability $ε$ from each side of the $x$-axis and taking the minima or maxima of the remaining curve. Clearly, the $ε$-smooth entropies converge to the von Neumann entropy. ☼*{typical set}} 
  \label{fg:aep}
\end{figure}

The entropic form of the AEP explains the crucial role of the ☼[von Neumann
entropy]{entropy!von Neumann} to describe information theoretic tasks. 
While operational quantities
in information theory (such as the amount of extractable randomness, the minimal
length of compressed data and channel capacities) can naturally be expressed in
terms of smooth entropies in the ☼{one-shot} setting, the von Neumann entropy is 
recovered if we consider a large number of independent repetitions of the
task.

Moreover, the entropic approach to asymptotic equipartition lends itself to a
generalization to the quantum setting. Note that the traditional approach, 
which considers the AEP 
as a statement about (conditional) probabilities, does not have a natural quantum 
generalization due to the fact that
we do not know a suitable generalization of conditional probabilities to quantum
side information.

Here, we want to show that the
smooth conditional min- and max-entropies converge to the von Neumann entropy in 
the ☼{i.i.d.\ limit}.
Recall that the conditional von Neumann entropy of a state $ρ$ is defined as 
\begin{align}
  \hvn{A|B}[ρ] = H(ρ¬{AB}) - H(ρ¬{B}), \tn*{where} H(ρ) = - \tr{ρ \log ρ} ¶[vn].
\end{align}
This convergence can be shown in several ways, for example through the use of
chain rules for smooth entropies~\cite{beaudry11}, which reduce the problem
to the classical AEP, or through the use of typical subspaces
(cf.\ e.g.~\cite{nielsen00}). While these techniques achieve the desired
asymptotic limit, they fail to give good bounds on the convergence for finite
$n$. More precisely, we are interested in the difference between the smooth
entropies and the von Neumann entropy for finite $n$. This distance is in
general a function of $n$, $ε$ and some properties of the quantum state under
consideration.

We call such a relation a fully quantum AEP because both the $A$ and $B$ systems
are general quantum systems.

For the smooth conditional min- and max-entropy, the convergence for finite $n$ was first
analyzed by Holenstein and Renner~\cite{holensteinrenner06} for classical probability
distributions. Renner also generalized these arguments to the
quantum setting~\cite{renner05}. He shows that\footnote{Note that the
smoothing of the min-entropy was defined differently in~\cite{renner05}.},
for any $ε > 0$ and for tensor product states $ρ¬{AB}^{⨂n} = ρ¬{AB} ⨂ ρ¬{AB} ⨂
\dots ⨂ ρ¬{AB}$, it holds that
\begin{align}
  \frac{1}{n} \hmin*{A^n|B^n}[ρ] ≥ \hvn{A|B}[ρ] - \frac{δ}{√{n}} ¶[qaep/renato],
\end{align}
where $δ$ scales with the product of $\hh{0}{A}[ρ]$ and $√{\log(1/ε)}$.\footnote{Note that
$\hh{0}{A}[ρ] = d¬A$ if the state $ρ$ has full support on $ℋ¬A$; thus, this bound
depends on the dimension of the $A$ system.}

On a technical level, our results are related to recent findings in quantum hypothesis testing~\cite{audenaert07,audenaert07-3,ogawa00}.

\subsection{Main Contributions}

This chapter is based on~\cite{tomamichel08}; however, the results presented here are strictly more general as we extended certain proofs to relative entropies and a more general class of
operators. This allows us to plug the inequalities into other arguments, where (e.g.~smoothed) states are not necessarily normalized. 

The main result of this chapter
establishes that $δ$ in~§[qaep/renato] 
scales with $\hmax{A|B}$
and $\hmax{A|R}$ instead of $\hh{0}{A}$, where $R$ is a reference system
purifying $ρ¬{AB}$. These conditional entropies measure the correlations
between the subsystems $A$, $B$ and $R$ and are often much smaller than
$\hh{0}{A}$. In particular, conditional entropies do not depend on the Hilbert
space dimension of any subsystem, which is useful in the context of quantum
cryptography, where these dimensions are generally unknown.\footnote{Note also that
the appearance of the reference system in the convergence rate is not entirely
unexpected. This can be seen from the fact that~§[qaep/renato] directly implies
a bound on the max-entropy, i.e.\
\begin{align}
  \frac{1}{n} \hmax*{A^n|R^n} ≤ \hvn{A|R} + \frac{δ}{√{n}} ¶,
\end{align}
by duality of the entropies. Hence, it is no surprise that this duality also appears in
the convergence rate.}

\begin{result}
  \label{re:aep-bound}
  For any $ε > 0$ and $n$ large enough, we have
  \begin{align}
    \frac{1}{n} \hmin*{A^n|B^n} &≥ \hvn{A|B} - \frac{δ(ε,υ)}{√{n}} \tn*{and} ¶\\
    \frac{1}{n} \hmax*{A^n|B^n} &≤ \hvn{A|B} + \frac{δ(ε,υ)}{√{n}} ¶,
  \end{align}
  where $δ(ε,υ) = 4 \log υ √{\log (2/ε^2)}$ and 
  $υ = √{2^{\hmax{A|B}}} + √{2^{\hmax{A|R}}} + 1$.
\end{result}

Together with converse bounds, this implies the fully quantum asymptotic
equipartition property.

\begin{result}
  \label{re:aep}
  The smooth entropies converge to the von Neumann entropy in the asymptotic limit of many copies. 
  For any $0 < ε < 1$, we have
  \begin{align}
    \lim_{n → ∞} \frac{1}{n} \hmin*{A^n|B^n} 
      = \lim_{n → ∞} \frac{1}{n} \hmax*{A^n|B^n} = \hvn{A|B} ¶\,.
  \end{align}
\end{result}

Note, in particular, that our statement holds for any $0 < ε < 1$. This improves on
 the previous results in~\cite{tomamichel08}, where the converse bound was only
shown in the limit $ε → 0$. Additionally, we improve many bounds for the case where the smoothing parameter is close to $1$. 
This extension is important to show strong converse statements, as we will
see in Section~\ref{se:strong}.

\subsection{Outline}

The remainder of the chapter is organized as follows.
In Section~\ref{se:aep/class}, we explore the special case of classical registers without side
information and introduce the proof techniques, based on relations between the smooth min-entropy 
and ☼[Rényi entropies]{entropy!Rényi}, that will also lead to the fully 
quantum AEP. In Section~\ref{se:aep/rel}, we introduce quantum generalizations of the 
classical Rényi $α$-entropies. (See also Appendix~\ref{ap:renyi}, where some
of their properties are discussed.)
Then, in Section~\ref{se:aep/min-renyi}, we show, analogous to the classical proof, that the
smooth min-entropy can be bounded in terms of these entropies. Section~\ref{se:aep/renyi-vn}
then bounds the difference between Rényi $α$-entropies and the von Neumann entropies for
$α$ close to $1$. Finally, these results are combined in Section~\ref{se:aep/aep} to
prove the fully quantum AEP.

\section{Sketch of a Classical Proof}
\label{se:aep/class}

The relation between the traditional formulation of the AEP in terms of probabilities~\cite{cover91} and its entropic formulation was already explored in the preceding section.
To give an intuitive idea of the techniques used in the following, we first prove
a special case of the AEP for classical registers and take an additional limit 
$ε → 0$. More precisely, we show that
\begin{align}
  &\lim_{ε → 0} \lim_{n → ∞} \frac{1}{n} \hmin*{X^n} = \hvn{X} ¶\,.
\end{align}

We employ the Rényi $α$-entropies~\cite{renyi61}
\begin{align}
  \hh{α}{X} := \frac{1}{1 - α} \log ∑_{x ∈ \cX} P(x)^α, \quad α ∈ (0, 1) \cup (1, ∞) 
    ¶[aep/renyi] \,,
\end{align}
for which $H_∞$ ($α → ∞$), $H_0$ ($α → 0$) and the Shannon entropy ($α → 1$) are defined as limits. Furthermore, the entropies $H_α$ are monotonically decreasing in $α$ and,
as shown in~\cite{renner04}, the Rényi entropies with $α > 1$ are close to the smooth
min-entropy in the sense that
\begin{align}
  \hmin*{X} ≥ \hh{α}{X} - \frac{1}{α - 1} \log \frac{1}{ε}, \qquad α > 1 ¶
\end{align}
while those with $α < 1$ are close to the smooth max-entropy.
Note that the error term $\frac{1}{α - 1} \log 1/ε$ above diverges 
when we try to recover the Shannon entropy. However, in the case of an i.i.d.\ sequence we find
\begin{align}
  \frac{1}{n} \hmin*{X^n} ≥ \hh{α}{X} - \frac{1}{n(α - 1)} \log \frac{1}{ε} 
    ¶[aep/class/min-renyi],
\end{align}
where we have used $\hh{α}{X^n} = n \hh{α}{X}$. We proceed by bounding the entropy $\frac{1}{n} \hmin*{X^n}$ in the limit $n → ∞, ε → 0$ from above and below. To get the lower bound, we choose $α = 1 + 1/√{n}$ and take the limit $n → ∞$ in~§[aep/class/min-renyi].
An upper bound in this limit follows directly from the fact that $\hmin{X} ≤ \hvn{X}$ (cf.\ Proposition~\ref{pr:min-max-bounds} or~\cite{renyi61})
and the continuity of the Shannon entropy.\footnote{See, e.g.\ Fannes~\cite{fannes73,alicki03}, where the relevant continuity property is shown for (conditional) von Neumann entropies.}

A similar argument shows that the smooth max-entropy converges to the 
von Neumann entropy in the asymptotic limit.

%...........................

\section{Quantum Relative Entropies}
\label{se:aep/rel}

We prove the asymptotic equipartition property for relative entropies, which
are introduced here. Conditional entropies can be seen as 
special cases of relative entropies.

\subsection{Quasi-Entropies}

A very general class of classical relative entropies are the $f$-divergences,
originally introduced by Csiszár~\cite{csiszar72}. They have been generalized to
the quantum setting by Petz~\cite{petz86}, who calls them quasi-entropies. (See
also~\cite{ohya93,petz09,hiai10} for an overview of recent results on
these entropies.) The crucial observation is that some of the most interesting
mathematical properties of the ☼[von Neumann]{entropy!von Neumann} and ☼[Shannon]{entropy!Shannon} entropies are a direct
consequence of the (operator) concavity of the function $h: t ↦ - t \log t$ that
defines the functional $H(ρ) = \tr{h(ρ)}$. Hence, the following
generalization of this functional was investigated.

\begin{definition}[Quasi-Entropy]
  \label{df:quasi}
  Let $A, B ∈ \opos{ℋ}$ and let $f: ℝ_0^+ → ℝ$ be continuous. 
  Then, the $f$-quasi-entropy of $A$ relative to $B$ is
  \begin{align}
    S_f(A \,\|\, B) := \lim_{ξ → 0}\, \braket<B>{Γ|√{B + ξ ⅈ} ⨂ ⅈ\, f \Big( (B + ξ ⅈ)\inv 
      ⨂ Aᵀ \Big) √{B + ξ ⅈ} ⨂ ⅈ\,|Γ} ¶,
  \end{align}
  where $\ket{Γ} = \sum_i \ket{i} ⨂ \ket{i}$ and $\{ \ket{i} \}$ is an
orthonormal basis of $ℋ$, 
  with respect to which the transpose is defined.
\end{definition}
Unconditional quasi-entropies are recovered when we substitute $B = ⅈ$.
Let us consider a few examples of such functions. For this purpose, let $A, B ∈ \opos{ℋ}$ 
be two positive operators whose supports satisfy $\supp{A} \subseteq
\supp{B}$, such that the limit in Def.~\ref{df:quasi} is always finite and $B\inv$ can be interpreted as a ☼{generalized inverse}. First, using $h: t ↦ - t
\log t$ extended to $ℝ_0^+$ using $h(0) = \lim_{t → 0} h(t) = 0$, the
von Neumann ☼{relative entropy} is given as
\begin{align}
  S(A \| B) := S_h(A \| B) = \tr<b>{A (\log B - \log A)} ¶[vn/rel]\,.
\end{align}
Note that we omit a minus sign here that is present in the conventional definition of the relative entropy.
To derive this well-known expression from Def.~\ref{df:quasi}, we used that
$\braket{Γ| X ⨂ ⅈ |Γ} = \tr{X}$ and $(X ⨂ ⅈ) \ket{Γ}
= (ⅈ ⨂ Xᵀ) \ket{Γ}$ for any operator $X ∈ \olin{ℋ}$.

\subsection{Relative Rényi entropies}

As another example, we recover generalizations of the Rényi entropies of
order $α ∈ (0, 1) \cup (1, ∞)$. This is achieved using $g_α : t ↦ t^α$ and
\begin{align}
  S_α(A \| B) := \frac{1}{1-α} \log S_{g_α}(A \| B) 
    = \frac{1}{1-α} \log \tr{A^α B^{1-α}} ¶[alpha/rel]\,,
\end{align}
where, for $α > 1$, we use the generalized inverse of $B$. We may continuously extend the range of $α$ to the limits $α → 0$ and $α → ∞$. Moreover, the von Neumann relative entropy emerges in the limit $α → 1^{±}$ if $\tr{A} = 1$. Hence, we can continuously extend the range of valid parameters to $α ∈ ℝ_0^+$ by setting $S_1 \equiv S$ in this case.
Some properties of these entropies are discussed in Appendix~\ref{ap:renyi}.

\subsection{Relative Min- and Max-Entropy}

In addition to this, we will also need a relative entropy version of the min- and max-entropies. The relative min-entropy was introduced by Datta~\cite{datta08}\footnote{There, it appears under the name $D_{\max} \equiv -S_{\min}$.} and is directly related to the conditional min-entropy defined in~\cite{renner05}.
The relative min- and max-entropies of $A ∈ \opos{ℋ}$ relative to $B ∈ \opos{ℋ}$ are given by
\begin{align}
  S_{\min}(A \,\| B) &:= \sup \big\{ λ : A ≤ 2^{-λ} B \big\} \tn*{and} ¶[min/rel]\\
  S_{\max}(A \,\| B) &:= \log \fidb{A}{B}^2 ¶.
\end{align}

We also need a smoothed version of the relative min-entropy, which we define for any $ρ ∈ \osub{ℋ}$ and $σ ∈ \opos{ℋ}$.
\begin{align}
  S_{\min}^{ε}(ρ \,\| σ) := \max_{\tilde{ρ} ∈ \ball{ρ}} S_{\min}(\tilde{ρ} \,\| σ) 
    ¶[smooth/rel]\,.
\end{align}

Clearly, the (smooth) min- and the max-entropy (as defined in Chapter~\ref{ch:entropies}) of $A$ conditioned on $B$, $\hmin*{A|B}[ρ]$ and $\hmax{A|B}[ρ]$, can be recovered by the substitution $A = ρ¬{AB}$ and $B = ⅈ¬A ⨂ σ¬{B}$, where $σ¬{B}$ is maximized over $\onorm{ℋ¬B}$.

%...........................

\section{Lower Bounds on the Smooth Min-Entropy}
\label{se:aep/min-renyi}
☼*{entropy!smooth}

The following lemma gives a first lower bound on the smooth relative entropy~§[smooth/rel]. (A similar, less general lemma can be found in~\cite{dattarenner08}.)

\begin{lemma}
  \label{lm:aep/smooth-bound}
  Let $ρ ∈ \osub{ℋ}, σ ∈ \opos{ℋ}$ and $λ ≥ S_{\min}(ρ \,\| σ)$. Then,
  \begin{align}
    S_{\min}^{ε}( ρ \,\| σ) ≥ λ \,, \!\!\tn*{where} \!\! ε = √{2 \tr{Δ} -\tr{Δ}^2} \
      \tn{and}\ Δ = \{ ρ - 2^{-λ} σ \}_+ ¶.
  \end{align}
\end{lemma}

\begin{proof}
  We first choose $\rhot$, bound $S_{\min}^{ε}(\rhot \,\| σ)$, and then show that 
  $\rhot ∈ \ball{ρ}$. We use the abbreviated notation $Λ := 2^{-λ} σ$ and set
  \begin{align}
    \rhot := G ρ G†, \tn*{where} G := Λ^{\frac{1}{2}} (Λ + Δ)^{-\frac{1}{2}} ¶\,,
  \end{align}
  where we use the ☼{generalized inverse}. From the definition of $Δ$, we have $ρ ≤ Λ + Δ$;
  hence, $\rhot ≤ Λ$ and $S_{\min}(\rhot \,\| σ) ≥ λ$.

  Let $\ket{ψ}$ be a purification of $ρ$, then $(G ⨂ ⅈ) \ket{ψ}$ is a purification of $\rhot$ 
  and, using Uhlmann's theorem~§[fid/uhl], we find a bound on the (generalized) fidelity:
  \begin{align}
    F(\rhot, ρ) &≥ \abs{\braket{ψ|G|ψ}} + √{(1 - \tr{ρ})(1 - \tr{\rhot})} ¶\\
      &≥ \Re \big\{ \tr{G ρ} \big\} + 1 - \tr{ρ} = 1 - \tr<b>{(ⅈ - \bar{G}) ρ} ¶\,,
  \end{align}
  where we introduced $\bar{G} := \frac{1}{2}(G + G†)$. This can be simplified further 
  after we note that $G$ is a contraction.\footnote{A ☼{contraction} $G$ is an operator with 
  operator norm $\norm{G}{} \leq 1$.} To see this, we multiply $Λ ≤ Λ + Δ$ with 
  $(Λ + Δ)^{-\frac{1}{2}}$ from left and right to get
  \begin{align}
    G† G = (Λ + Δ)^{-\frac{1}{2}} Λ (Λ + Δ)^{-\frac{1}{2}} ≤ ⅈ.
  \end{align}
  Furthermore, $\bar{G} ≤ ⅈ$, since $\norm{\bar{G}} ≤ 1$ by the triangle inequality and 
  $\norm{G} = \norm{G†} ≤ 1$. 
  %Thus, $\tr{\bar{G} ρ} ≤ 1$.  % it apperas we don't need this anymore
  Moreover,
  \begin{align}
    \tr<b>{(ⅈ - \bar{G}) ρ} &≤ \tr{Λ + Δ} - \tr<b>{\bar{G} (Λ + Δ)} ¶\\
      &= \tr{Λ + Δ} - \tr{(Λ + Δ)^{\nicefrac{1}{2}} {Λ}^{\nicefrac{1}{2}} }
      ≤ \tr{Δ} ¶\,,
  \end{align}
  where we used $ρ ≤ Λ + Δ$ and $√{Λ + Δ} ≥ √{Λ}$. The latter inequality follows from the 
  operator monotonicity of the square root function (cf.~Section~\ref{sc:opmono}). Finally, 
  using the above bounds, the purified distance between $\rhot$ and $ρ$ is bounded by
  \begin{align}
    P(\rhot, ρ) = √{1 - F^2(\rhot, ρ) \big)} ≤ √{ 1 - \big( 1 - \tr{Δ} \big)^2 } = 
      √{2 \tr{Δ} - \tr{Δ}^2} ¶\,.
  \end{align}
  Hence, we verified that $\rhot ∈ \ball{ρ}$, which concludes the proof.
\end{proof}

In particular, this means that for a fixed $ε ∈ [0, 1)$ and $\supp{ρ} \subseteq \supp{σ}$, 
we can always find a finite $\lambda$ s.t.~Lemma~\ref{lm:aep/smooth-bound} holds. 
To see this, note that $ε(λ) = √{ 2 \tr{Δ} - \tr{Δ}^2 }$
is continuous in $λ$ with $ε(S_{\min}(ρ \,\| σ)) = 0$ and $\lim_{λ → ∞} ε(λ) = 1$.

\subsubsection{Rényi Entropies and Smooth Min-Entropy}
☼*{entropy!Rényi}

Our main tool for proving the fully quantum AEP is a family of inequalities that 
relate the smooth relative min-entropy to relative Rényi entropies for $α ∈ (1, 2]$.
This family contains the von Neumann relative entropy in the limit $α → 1$.
This can be seen as a quantum generalization of the classical inequality in~\cite{renner04}.
\begin{proposition}
  \label{pr:min-renyi}
  Let $ρ ∈ \osub{ℋ}, σ ∈ \opos{ℋ}$, $0 < ε < 1$ and $α ∈ (1, 2]$. Then, 
  \begin{align}
    S_{\min}^{ε}(ρ \,\| σ) ≥ S_{α}(ρ \,\| σ) - \frac{g(ε)}{α - 1}, \!\!\tn*{where}\!\! 
      g(ε) = \log \frac{1}{1 - √{1 - ε^2}} ¶[min-renyi-bound].
  \end{align}
\end{proposition}

\begin{proof}
  We consider two cases: (1) The $α$ entropy diverges (to $-∞$) and the inequality holds 
  trivially. (2) We have $\supp{ρ} \subseteq \supp{σ}$. In this case, we can find an isometry 
  $ℋ' → ℋ$ that maps a $σ'$ to $σ$ and $ρ'$ to $ρ$ s.t.\ $σ'$ has full support. The min- and 
  $α$-entropies are invariant under this isometry due to Proposition~\ref{pr:smooth-iso} and 
  Lemma~\ref{lm:quasi-iso}, thus, we henceforth assume that $σ$ is invertible in this proof.

  We use Lemma~\ref{lm:aep/smooth-bound} to get a first bound on $S_{\min}^{ε}$; in 
  particular, we choose $λ$ s.t.\ Lemma~\ref{lm:aep/smooth-bound} holds for $ε$.
  Next, we introduce the operator $X := ρ - 2^{-λ} σ$ with eigenbasis 
  $\{ \ket{e_i} \}_{i \in S}$. The set $S^+ \subseteq S$ contains the indices $i$ 
  corresponding to positive eigenvalues of $X$. Hence, $P^+ := ∑_{i ∈ S^+}
  \proj{e_i}$ is the projector on the positive eigenspace of $X$ and $P^+ X P^+ = Δ$ 
  as defined in Lemma~\ref{lm:aep/smooth-bound}. 
  Furthermore, let $r_i := \braket{e_i | ρ | e_i} ≥ 0$ and $s_i := \braket{e_i | σ | e_i} > 0$. 
  It follows that
  \begin{align}
    \forall\, i ∈ S^+ :\ r_i - 2^{-λ} s_i ≥ 0 \tn*{and, thus,} \frac{r_i}{s_i} 2^{λ} ≥ 1 ¶.
  \end{align}
  For any $α ∈ (1, 2]$, we bound $\tr{Δ} = 1 - √{1-ε^2}$ as follows:
  \begin{align}
    1 - √{1-ε^2} &= \tr{Δ} = ∑_{i ∈ S^+} r_i - 2^{-λ} s_i ≤ ∑_{i ∈ S^+} r_i ¶\\
    &≤ ∑_{i ∈ S^+} r_i \left( \frac{r_i}{s_i} 2^{λ} \right)^{α - 1} 
    ≤ 2^{λ (α - 1)} ∑_{i ∈ S} r_i^α\, s_i^{1 - α} ¶\,. 
  \end{align}
  Hence, taking the logarithm and dividing by $α - 1 > 0$, we get
  \begin{align}
    λ ≥ \frac{1}{1 - α} \log ∑_{i ∈ S} r_i^α\, s_i^{1 - α} - \frac{1}{α - 1} 
      \log \frac{1}{1 - √{1 - ε^2}} ¶[aep/bound1]\,.
  \end{align}

  Next, we use the monotonicity of the Rényi entropies (cf.\ Lemma~\ref{lm:renyi-mono}).
  We use the measurement TP-CPM $\sM : X ↦ ∑_{i ∈ S} \proj{e_i} X \proj{e_i}$ to obtain
  \begin{align}
     S_{α}(ρ \,\| σ) ≤ S_{α}\big( \sM(ρ) \| \sM(σ) \big) 
      = \frac{1}{1-α} \log ∑_{i ∈ S} r_i^α s_i^{1 - α} ¶.
  \end{align}
  We conclude the proof by substituting this into~§[aep/bound1] and applying the upper 
  bound on $λ$ in Lemma~\ref{lm:aep/smooth-bound}.
\end{proof}

We also note here that the term $1 - √{1 - ε^2}$ in $g(ε)$ can be bounded by simpler expressions (cf.\ Fig.~\ref{fg:smoothing-bounds}). We find $1 - √{1 - ε^2} ≥ \frac{ε^2}{2}$ using a second order Taylor expansion of the expression
around $ε = 0$ and the fact that the third derivative is non-negative. This is a very good
approximation for small $ε$.
Hence, Eq.~§[min-renyi-bound] can be simplified to
\begin{align}
    S_{\min}^{ε}(ρ \,\| σ) ≥ S_{α}(ρ \,\| σ) - \frac{1}{α - 1} 
      \log \frac{2}{ε^2} ¶[min-renyi-bound-untight]\,.
\end{align}
This form of the inequality has been reported previously~\cite{tomamichel08}. 
%However,
%for large $ε$, we prefer the bound $1 - √{1 - ε^2} ≥ 1 - √{2 (1-ε)}$, and, thus,
%\begin{align}
%    S_{\min}^{1-ε}(ρ \,\| σ) ≥ S_{α}(ρ \,\| σ) - \frac{1}{α - 1} 
%      \log \frac{1}{1 - √{2 ε}} ¶\,.
%\end{align}

\begin{figure}
  \centering %\includegraphics{gnuplot/smoothing-bounds.pdf}
   \includegraphics{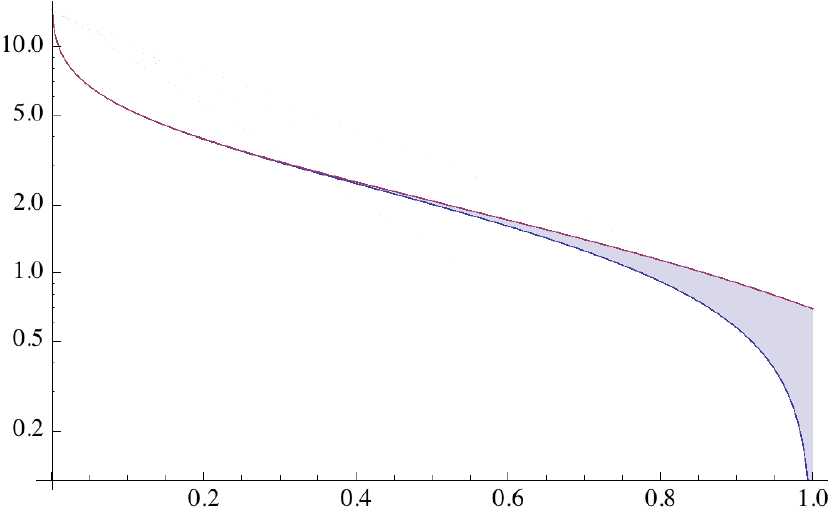}
  \caption[Bound on the Error Term $g(ε)$.]{\emph{Bound on $g(ε)$.} 
    Plot of the function $g(ε) = -\log \big( 1 - √{1 - ε^2} \big)$ and its upper bound,
    $\log \frac{2}{ε^2}$. The horizontal axis denotes $ε$. }
  \label{fg:smoothing-bounds}
\end{figure}

\subsubsection{Conditional Entropies}

Proposition~\ref{pr:min-renyi} is of particular interest when applied to the smooth conditional min-entropy.
In this case, let $ρ ∈ \osub{ℋ¬{AB}}$ and $σ$ be of the form $ⅈ¬A ⨂ ρ¬B$. Then,
for any $α ∈ (1, 2]$, we have
\begin{align}
  \hmin*{A|B}[ρ] &≥ S_{\min}^ε(ρ¬{AB} \| ⅈ¬A ⨂ ρ¬B) ¶\\
    &≥ S_{α}(ρ¬{AB} \| ⅈ¬A ⨂ ρ¬B) - \frac{g(ε)}{α-1} = \hh{α}{A|B}[ρ] - \frac{g(ε)}{α-1}
      ¶[aep/smooth-alpha-bound/min].
\end{align}
The duality relation for the smooth min- and max-entropies (cf.\ Theorem~\ref{th:smooth-dual}) and the
Rényi $α$-Entropies (cf.\ Lemma~\ref{lm:renyi-dual}) imply a corresponding dual relation 
for the max-entropy. 
For any $α ∈ [0, 1)$, we have
\begin{align}
  \hmax*{A|B}[ρ] &≤ \hh{α}{A|B}[ρ] - \frac{g(ε)}{1-α} ¶[aep/smooth-alpha-bound/max].
\end{align}

%...........................

\section{From Rényi to von Neumann Entropies}
\label{se:aep/renyi-vn}

We will use Proposition~\ref{pr:min-renyi} to get a lower bound on the smooth min-entropy in terms of $α$-entropies and, hence, it remains to find a lower bound on the $α$-entropies in terms of the von Neumann entropy. In turn, the bound on the convergence will depend on the smoothing parameter $ε$ and a contribution $Υ(ρ\,\|σ)$ that describes how fast the $α$-entropies converge to the von Neumann entropy.
\begin{definition}
  Let $ρ, σ ∈ \opos{ℋ}$, then we define the $α$-entropy ☼{convergence parameter},
  \begin{align}
    Υ(ρ\,\|σ) := 2^{-\frac{1}{2} S_{\nicefrac{3}{2}}( ρ \,\| σ )} + 
      2^{\frac{1}{2} S_{\nicefrac{1}{2}}( ρ \,\| σ )} + 1 ¶\,.
  \end{align}
\end{definition}

We can now state a bound on the $α$-entropies for $α$ close to $1$.
\begin{lemma}
  \label{lm:renyi-vn}
  Let $ρ ∈ \onorm{ℋ}, σ ∈ \opos{ℋ}$ and $1 < α < 1 + \frac{\log 3}{4 \log υ}$, where 
  $υ = Υ(ρ \,\| σ)$. Then, 
  \begin{align}
    S_{α}(ρ \,\| σ) > S(ρ \,\| σ) - 4 (α - 1) \big( \log υ \big)^2 ¶.
  \end{align}
\end{lemma}
\begin{proof}
  We assume that $σ$ is invertible in this proof. The general result then follows by the 
  arguments outlined at the beginning of the proof of Proposition~\ref{pr:min-renyi}.

  Let $\{ \ket{i} \}$ be an orthonormal basis of $ℋ$. The state $\ket{γ} := ∑_i \ket{i} 
  ⨂ \ket{i}$ is the (unnormalized) fully entangled state on $ℋ ⨂ ℋ$. We introduce a
  purification $\ket{φ} := √{ρ}\, \ket{γ}$ of $ρ$. To simplify notation, we use $β := α - 1$ 
  as well as $X := ρ ⨂ (σ\inv)^T$. 

  Let us first approximate $S_{α}$ for small $β > 0$.
  \begin{align}
    S_{α}(ρ \,\| σ) = - \frac{1}{β} \log\, \braket{φ | X^β | φ} ≥ \frac{1}{β \ln 2} 
      \big(1 - \braket{φ | X^β | φ} \big) ¶\,,
  \end{align}
  where we used $\ln x ≤ x - 1$ for all $x > 0$. We now expand the exponential 
  $t^β$ for each eigenvalue $t > 0$ of $X$ as follows: $t^β = 1 + β \ln t + r_β(t)$,
  where $r_β(t) := t^β - β \ln t - 1$. This leads to
  \begin{align}
    S_α(ρ \,\| σ) &≥ \frac{1}{β \ln 2} \big(1 - \tr{ρ} - \beta \braket{φ | \ln X | φ} - 
      \braket{φ | r_β(X) | φ} \big) ¶\\
      &≥ S(ρ \,\| σ) - \frac{1}{β \ln 2} \braket{φ | r_β(X) | φ} ¶[aep/lbound]\,.
  \end{align}
  
  To simplify this further, we note that
  \begin{align}
    r_β(t) ≤ 2 (\cosh(β \ln t) - 1) =: s_β(t) ¶\,.
  \end{align}
  It is easy to verify that $s_β$ is monotonically increasing for $t ≥ 1$ and concave in 
  $t$ for $β ≤ 1/2$ and $t ∈ [3, ∞)$. Furthermore, we have $s_β(t) = s_β(\frac{1}{t})$ and 
  $s_β(t^2) = s_{2 β}(t)$. We use this to bound\footnote{Adaptions of this step lead to 
  different bounds. Here, we are interested in a bound that can be expressed in terms of 
  $S_{\nicefrac{1}{2}}$ and $S_{\nicefrac{3}{2}}$.}
  \begin{align}
    s_β(t) &≤ s_β \Big( t + \frac{1}{t} + 2 \Big) = s_{2 β} \Big(√{t} + \frac{1}{√{t}} \Big)
     ≤ s_{2 β} \Big(√{t} + \frac{1}{√{t}} + 1 \Big) ¶[aep/s2beta]\,.
  \end{align}
  Next, we apply~§[aep/s2beta] to the matrix element in~§[aep/lbound] and use 
  the fact that the operator $√{X} + 1/√{X} + ⅈ$ has its eigenvalues in $[3, ∞)$ and $2 β 
  < \frac{\log 3}{2 \log υ} ≤ \frac{1}{2}$ together with Lemma~\ref{lm:jensen} to get
  \begin{align}
    \braket{φ | s_β(X) | φ} ≤ \braket<b>{φ | s_{2 β} \Big( √{X} + \frac{1}{√{X}} + ⅈ \Big) | φ} 
    ≤ s_{2β}(υ) ¶[aep/rbsb]\,,
  \end{align}
  where we substituted $υ = \braket{φ | √{X} + 1/√{X} + ⅈ | φ}$.

  Taylor's theorem and an expansion around $β = 0$ gives an upper bound on $s_β(t)$:
  $s_β(t) ≤ β^2 (\ln t)^2 \cosh (β \ln t)$. Hence,
  \begin{align}
    \frac{1}{β \ln 2} s_{2 β}(υ) ≤ 4 β (\log υ)^2 \ln 2 \cosh(2 β \ln υ)
      < 4 β (\log υ)^2 ¶[aep/s2bapprox]\,,
  \end{align}
  where we simplified the expression (for convenience of exposition) using 
  $\ln 2 \cosh(\ln 3 / 2) < 1$. The lemma now follows after we 
  substitute~§[aep/s2bapprox] and~§[aep/rbsb] into~§[aep/lbound].
\end{proof}

This Lemma can be extended to include sub-normalized states, $ρ ∈ \osub{ℋ}$. Let $\hat{ρ} 
= ρ/\tr{ρ}$ and $\hat{υ} = Υ(\hat{ρ} \,\| σ)$, then
\begin{align}
  S_α(ρ \,\| σ) &= S_α(\hat{ρ} \,\| σ) + \frac{α}{α - 1} \log \frac{1}{\tr{ρ}} ¶\\ 
    &> S(\hat{ρ} \,\| σ) + \frac{α}{α - 1} \log \frac{1}{\tr{ρ}} - 4 (α - 1) 
      ( \log \hat{υ} )^2 ¶\\
    &= \frac{1}{\tr{ρ}} S(ρ \,\| σ) + \frac{1}{α - 1} \log \frac{1}{\tr{ρ}} - 4 (α - 1)
      ( \log \hat{υ} )^2 ¶.
\end{align}

Now we combine Proposition~\ref{pr:min-renyi} and~\ref{lm:renyi-vn} to get the desired bound for i.i.d.\ operators. (Note that we restrict $ρ$ to normalized states because, if $\tr{ρ} < 1$ then the trace of $ρ^{⨂ n}$ drops exponentially to zero. An extension of this to sub-normalized states, while possible, thus seems uninteresting.)

\begin{theorem}
  \label{th:min-vn}
  Let $ρ ∈ \onorm{ℋ}, σ ∈ \opos{ℋ}$, $0 < ε < 1$ and $υ = Υ(ρ \,\| σ)$. Then, for any 
  $n ≥ \frac{8}{5} g(ε)$, the i.i.d.\ operators $ρ^{⨂ n}$ and $σ^{⨂ n}$ 
  satisfy
  \begin{align}
    \frac{1}{n} S_{\min}^ε(ρ^{⨂ n} \| σ^{⨂ n} ) ≥ S(ρ \,\| σ) 
      - \frac{δ(ε, υ)}{√{n}} \tn*{where} δ(ε,υ) = 4 \log υ √{g(ε)} ¶
  \end{align}
  and $g(ε) = -\log \big( 1 - √{1-ε^2} \big)$.
\end{theorem}

\begin{proof}
  We use Proposition~\ref{pr:min-renyi}, the additivity property of the 
  Rényi entropy~§[aep/renyi-add]
  and Lemma~\ref{lm:renyi-vn} to get a bound on the smooth min-entropy. Let $α := 1 + 
  \frac{1}{2 µ √{n}}$ for a parameter $µ$ (to be optimized over), then
  \begin{align}
    \frac{1}{n} S_{\min}^ε(ρ^{⨂ n} \| σ^{⨂ n} ) &≥ \frac{1}{n} S_α( ρ^{⨂ n} \| σ^{⨂ n}) 
      -\frac{1}{n (α - 1)} g(ε) ¶\\
    &= S_α(ρ \,\| σ) - \frac{2 µ}{√{n}} g(ε) ¶\\
    &≥ S(ρ \,\| σ) - \frac{2}{√{n}} \Big( µ\, g(ε) + \frac{1}{µ} 
      (\log υ)^2 \Big) ¶[aep/mubound]\,.
  \end{align}
  Clearly, we want to choose $µ$ such that it minimizes the expression 
  $µ\, g(ε) + µ\inv (\log υ)^2$.
  However, the requirement $α < 1 + \frac{\log 3}{4 \log υ}$ 
  in Lemma~\ref{lm:renyi-vn} restricts the choice of $µ$ for any fixed $n$, hence, the 
  error term is in general also a function of $n$. 
  Nonetheless, for large enough $n$ the optimum, $µ_*$, can be reached\footnote{To verify 
  this, evaluate an upper bound to $α = 1 + (2 µ_* √{n})\inv$ using the expression for $n$ 
  in~§[aep/mustar] and note that $√{5/2} < \log 3$.} and we get
  \begin{align}
    µ_* = √{\frac{(\log υ)^2}{g(ε)}} \tn*{for} n ≥ \frac{8}{5} \frac{(\log υ)^2}
    {\mu_*^{\,2}} = \frac{8}{5} g(ε) ¶[aep/mustar]\,.
  \end{align}
  Substitution of this expression into~§[aep/mubound] concludes the proof.
\end{proof}

%...........................

\section{The Asymptotic Equipartition Property}
\label{se:aep/aep}

\subsection{Direct Part}
☼*{direct bound}

In this section, we are mostly interested in the application of Theorem~\ref{th:min-vn} to
conditional min- and max-entropies. Here, for any state $ρ ∈ \onorm{ℋ¬{AB}}$, we choose
$σ¬{AB} = ⅈ¬A ⨂ ρ¬B$ and apply Theorem~\ref{th:min-vn}. This implies that
\begin{align}
    \frac{1}{n} \hmin*{A^n|B^n}[ρ] &≥ \frac{1}{n} S_{\min}^ε \big(ρ¬{AB}^{⨂n} \| 
      σ¬{AB}^{⨂n} \big) ¶\\
      &≥ S(ρ¬{AB} \| σ¬{AB}) - \frac{δ(ε, υ)}{√{n}} = \hvn{A|B}[ρ]
      - \frac{δ(ε, υ)}{√{n}} ¶.
\end{align}

This (and the dual of this relation) leads to the following corollary.
\begin{corollary}
  \label{co:aep/cond}
  Let $ρ ∈ \onorm{ℋ¬{AB}}$ and $0 < ε < 1$. Then, the smooth entropies of the
   i.i.d.\ product state $ρ¬{A^n B^n} = ρ¬{AB}^{⨂ n}$ satisfy
\begin{align}
    \frac{1}{n} \hmin*{A^n|B^n}[ρ] &≥ \hvn{A|B}[ρ]
      - \frac{δ(ε, υ)}{√{n}} \tn*{and} ¶[aep/min-vn]\\
    \frac{1}{n} \hmax*{A^n|B^n}[ρ] &≤ \hvn{A|B}[ρ]
      + \frac{δ(ε, υ)}{√{n}} ¶[aep/max-vn]\,,
\end{align}
  where $δ(ε, ν)$ is defined in Theorem~\ref{th:min-vn} and $υ = Υ(ρ¬{AB} \,\| ⅈ¬A ⨂ ρ¬B)$.
\end{corollary}

We may now use~§[aep/half-bound] and its dual relation,
$\hh{\nicefrac{3}{2}}{A|B}[ρ] ≥ \hmin{A|B}[ρ]$ (cf.\ Theorem~\ref{th:smooth-dual} and
Lemma~\ref{lm:renyi-dual}), to bound
\begin{align}
  Υ(ρ¬{AB}|ⅈ¬A ⨂ ρ¬B) ≤ √{2^{-\hmin{A|B}[ρ]}} + √{2^{\hmax{A|B}[ρ]}} + 1 ¶\,.
\end{align}
This is Result~\ref{re:aep-bound} of Section~\ref{se:aep/res}.

The following is a trivial corollary from Theorem~\ref{th:min-vn} and the above 
arguments, in particular~§[aep/min-vn] and its dual relation~§[aep/max-vn].
\begin{corollary}[AEP, direct]
  \label{co:aep-direct}
  Let $ρ ∈ \onorm{ℋ¬{AB}}$ and $0 < ε < 1$. Then,
  \begin{align}
    \lim_{n → ∞} \frac{1}{n} \hmin*{A^n|B^n}[ρ] ≥ \hvn{A|B}[ρ] ≥ 
    \lim_{n → ∞} \frac{1}{n} \hmax*{A^n|B^n}[ρ] ¶\,.
  \end{align}
\end{corollary}

\subsection{Converse Part}
☼*{converse bound}

To prove asymptotic convergence, we will also need converse bounds. For $ε = 0$,
the converse bounds are given by Proposition~\ref{pr:min-max-bounds}, i.e.\ $\hmin{A|B} ≤
\hvn{A|B} ≤ \hmax{A|B}$. For $ε > 0$, similar bounds can be derived from the
continuity of the conditional entropy in the state~\cite{alicki03}.
However, such bounds do not allow a statement of the form of
Corollary~\ref{co:aep-direct} as the deviation from the von Neumann entropy
scales as $n f(ε)$, where $f(ε) → 0$ only for $ε → 0$. (See, for
example,~\cite{tomamichel08} for such a weak converse bound.) This is not sufficient for some
applications of the asymptotic equipartition property.

Here, we prove a tighter bound, which relies on the
bound between smooth max-entropy and smooth min-entropy established in 
Proposition~\ref{pr:min-max-smooth}.
Applying this Proposition in conjunction with Eqs.~§[aep/min-vn]
and~§[aep/max-vn] establishes the converse AEP bounds. Let $0 < ε < 1$. Then,
using any smoothing parameter 
$0 < ε' < 1 - ε$, we bound
\begin{align}
  \frac{1}{n} \hmin*{A|B}[ρ] &≤ \frac{1}{n} \hmax[ε']{A|B}[ρ] + 
      \frac{1}{n} \log \frac{1}{1 - (ε + ε')^2} ¶\\
    &≤ \hvn{A|B}[ρ] + \frac{1}{n} \log \frac{1}{1 - (ε + ε')^2} + 
      \frac{δ(ε', η)}{√{n}}  ¶[aep/conv/1].
\end{align}
The corresponding statement for the smooth max-entropy follows by the dual argument.
We thus find
\begin{corollary}[AEP, converse]
  \label{co:aep-converse}
  Let $ρ ∈ \onorm{ℋ¬{AB}}$ and $0 < ε < 1$. Then,
  \begin{align}
    \lim_{n → ∞} \frac{1}{n} \hmin*{A^n|B^n}[ρ] ≤ \hvn{A|B}[ρ] ≤ 
    \lim_{n → ∞} \frac{1}{n} \hmax*{A^n|B^n}[ρ] ¶\,.
  \end{align}  
\end{corollary}
%
%\noindent Corollaries~\ref{co:aep-direct} and~\ref{co:aep-converse} together constitute
%Result~\ref{re:aep}.

These converse bounds are particularly important to bound the smooth entropies for large smoothing parameters. In this form, the AEP implies strong converse statements for many information theoretic tasks that can be characterized by smooth entropies in the one-shot setting (see, for example, Chapter~\ref{ch:app}).

%............................

\chapter{Uncertainty Relations for Smooth Entropies}
\label{ch:uc}

☼[Entropic uncertainty relations]{uncertainty!relation} (UCRs) provide lower bounds on the uncertainty of the outcomes of two incompatible measurements given side information in terms of conditional entropies. In this chapter, which is based on~\cite{tomamichel11} and~\cite{haenggi11}, we prove several uncertainty relations using ☼[smooth min- and max-entropies]{entropy!smooth} as well as the ☼[von Neumann entropy]{entropy!von Neumann} 
as measures of uncertainty. 

\section{Introduction and Related Work}
\label{sc:ucr/intro}

Uncertainty relations have inspired physicists since the early days of quantum mechanics, when Heisenberg~\cite{heisenberg27} formulated his famous uncertainty principle. In its Robertson~\cite{robertson29} form, it states that the product of the standard deviations of the outcomes of two ☼[incompatible measurements]{measurement!incompatible} on a pure state $\ket{ψ}$ is lower bounded in terms of the commutator of the ☼[observables]{observable} ($\hat{X}$ and $\hat{Y}$) of these measurements, 
\begin{align}
  Σ¬X · Σ¬Y ≥ \frac{1}{2} \abs<B>{\braket<b>{ψ | [\hat{X}, \hat{Y}] | ψ}} ¶[robertson].
\end{align}
Here, the variance of the measurement outcomes is given as $Σ¬Z^2 = \braket{ψ|\hat{Z}^2|ψ} - \braket{ψ|\hat{Z}|ψ}^2$ where $Z$ is either $X$ or $Y$. 
The commutator, $[\hat{X}, \hat{Y}] = \hat{X}\hat{Y} 
- \hat{Y}\hat{X}$, quantifies the incompatibility of the two observables. 

Note that the uncertainty relation becomes trivial if the measured state is an eigenstate of either $X$ or $Y$. More generally, the lower bound in~§[robertson] depends on the state $\ket{ψ}$ before measurement, which is often undesirable. For example, the state might be unknown or, in a cryptographic setting, prepared by an adversary. Furthermore, the standard deviation is not our preferred measure of uncertainty as it conflates two concepts: the value associated with different measurement outcomes and the uncertainty in the probability distribution of the outcomes. The latter is the uncertainty we are interested in and is often quantified by entropies.

Uncertainty relations in terms of entropies were first proposed by Hirsch\-man~\cite{hirschman57}, and Deutsch~\cite{deutsch83} for the case of a finite output alphabet.\footnote{See also~\cite{wehner09}, which offers a comprehensive review of uncertainty relations.} In the form proposed by Maassen and Uffink~\cite{maassen88}, it states that
\begin{align}
  \hvn{X} + \hvn{Y} ≥ \log \frac{1}{c}\,, \tn*{where} 
    c = \max_{x,y} \abs<b>{\braket{x|y}}^2 ¶[ucr/Maassen/proj].
\end{align}
for any state $ρ$ before measurement. Two ☼[registers]{register}, $X$ and $Y$, store the respective outcome of two different projective measurements, $\cX$ and $\cY$. 
The ☼{overlap}, $c$, is determined by these measurements and independent of the state before measurement. More specifically, the maximization in the definition of the overlap is taken over all eigenvectors $\ket{x}$ of $\hat{X}$ and $\ket{y}$ of $\hat{Y}$, where $\hat{X}$ and $\hat{Y}$ are observables corresponding to the projective measurements $\cX$ and $\cY$. 

More generally, Krishna and Parthasarathy~\cite{krishna01} considered ☼[POVMs]{POVM}, 
given by sets $\cX = \{ M_x \}$ and $\cY = \{ N_y \}$ of positive semi-definite operators.
For such measurements, the states of the registers containing the measurement outcomes 
are given by
\begin{align}
  \rho¬X = \sum_x \tr<B>{√{M_x}\, \rho √{M_x}} \proj{x} \tn*{and} 
    \rho¬Y = \sum_y \tr<B>{√{N_y}\, \rho √{N_y}} \proj{y} ¶.
\end{align}
The uncertainty relation~§[ucr/Maassen/proj] now holds for these states and
the overlap
\begin{align}
  c := \max_{x,y} \norm<B>{\sqrt{M_x} \sqrt{N_y}}[∞]^2 ¶[overlap]\,.
\end{align}
Note that the overlap reduces to~§[ucr/Maassen/proj] if the measurements are projective.

The uncertainty relation can be extended to one with classical side information (see also~\cite{hall95,cerf02}). Let the state between the system to be measured, $A$, and an observer, $O$, be of ☼[CQ]{state!CQ} form
$ρ¬{AO} = ∑_o p_o\,  \proj{o}[O] ⨂ ρ¬A^o$. 
This can be seen as the observer preparing the state $ρ¬A^o$ with probability $p_o$. What is the entropy of the observer $O$ about the measurement outcomes $X$ and $Y$? It is easy to see that
\begin{align}
  \hvn{X|O} + \hvn{Y|O} = ∑_o p_o\, \Big( \hvn{X}[ρ^o] + \hvn{Y}[ρ^o] \Big) 
    ≥ \log \frac{1}{c} ¶\,.
\end{align}
Thus, the uncertainty relation~§[ucr/Maassen/proj] still holds. Furthermore, we may model the random basis choice by a uniformly distributed bit $Θ ∈ \{\cX, \cY\}$ that is independent of the state $ρ¬{AO}$. 
Hence, we consider a joint state $ρ¬{AO Θ} = ρ¬{AO} ⨂ π¬{Θ}$ and a measurement on the $A$
system that depends on the classical register $Θ$.\footnote{This can be seen a measurement
of the joint system consisting of $A$ and $Θ$.}
The measurement outcome is denoted $Z$, which replaces $X$ and $Y$. Thus,
\begin{align}
  \hvn{Z|O Θ} = \frac{1}{2} \hvn{X|O} + \frac{1}{2} \hvn{Y|O} ≥ 
    \frac{1}{2} \log \frac{1}{c} ¶[ucr/game-cl]\,.
\end{align}
This uncertainty relation with classical side information is also visualized in 
Figure~\ref{fig:ucr/cl}.

A statement of the uncertainty relation in this form naturally leads to the question of what happens if the observer is allowed to store quantum information about the system, i.e.\ if the state $ρ¬{AO}$ is entangled. For the case of~§[ucr/game-cl], only the trivial bound
$\hvn{Z|O Θ}[ρ] ≥ 0$ holds. To see this, consider a fully entangled state between $A$ and $O$. In this case, the observer can predict the value of $Z$ perfectly using an appropriate measurement (depending on $Θ$) on his system. This is explained in Figure~\ref{fig:ucr/q}.

However, due to the monogamy of entanglement\footnote{Or, depending on the reader's preference, due to no-cloning.}, it is unclear what happens if we introduce a second observer. Note, for example, that the entropy $\hvn{A|O}$ before measurement can be negative for entangled states while the sum $\hvn{A|O_1} + \hvn{A|O_2}$ is always non-negative.

In fact, Renes and Boileau~\cite{renes09} (see also~\cite{christandl05}) conjectured the following uncertainty relation, which was later shown by Berta et al.~\cite{berta10}. For a tripartite quantum state shared between the system $A$ which is measured in the basis determined by $Θ$ and two observers, $O_1$ and $O_2$, it holds that
\begin{align}
  \hvn{Z|O_1 Θ} + \hvn{Z|O_2 Θ} ≥ \log \frac{1}{c} ¶[ucr/intro/vn]\,,
\end{align}
with $c$ defined as in~§[overlap] but only for rank-$1$ projective measurements. This proof was later extended to POVMs and simplified by Coles et al.~\cite{coles10,coles11}. Concurrently, our results also imply a simplified proof of the uncertainty relation for von Neumann entropies 
and POVMs. 

\begin{figure}
  
  \subfigure[A quantum system in the state $ρ$ is measured in one out of two bases, determined 
    by a uniformly random bit, $Θ$. Then the basis choice $Θ$ is sent to an observer, $O$, who 
    is asked to determine the measurement outcome, $X$. The observer is allowed to use classical 
    information about the state $ρ$ before measurement, including a full characterization of the 
    density matrix. One may alternatively think of the observer preparing the state $ρ$ before 
    the game. The uncertainty relation~\eqref{eq:ucr/game-cl}~ensures that the entropy the 
    observer has about the measurement outcome satisfies the lower bound $\hvn{X|O Θ} ≥ 
    \frac{1}{2} \log \frac{1}{c}$. Hence, if the overlap is nontrivial, the observer cannot 
    predict $X$ with certainty.]
  {\includegraphics[width=0.95\textwidth]{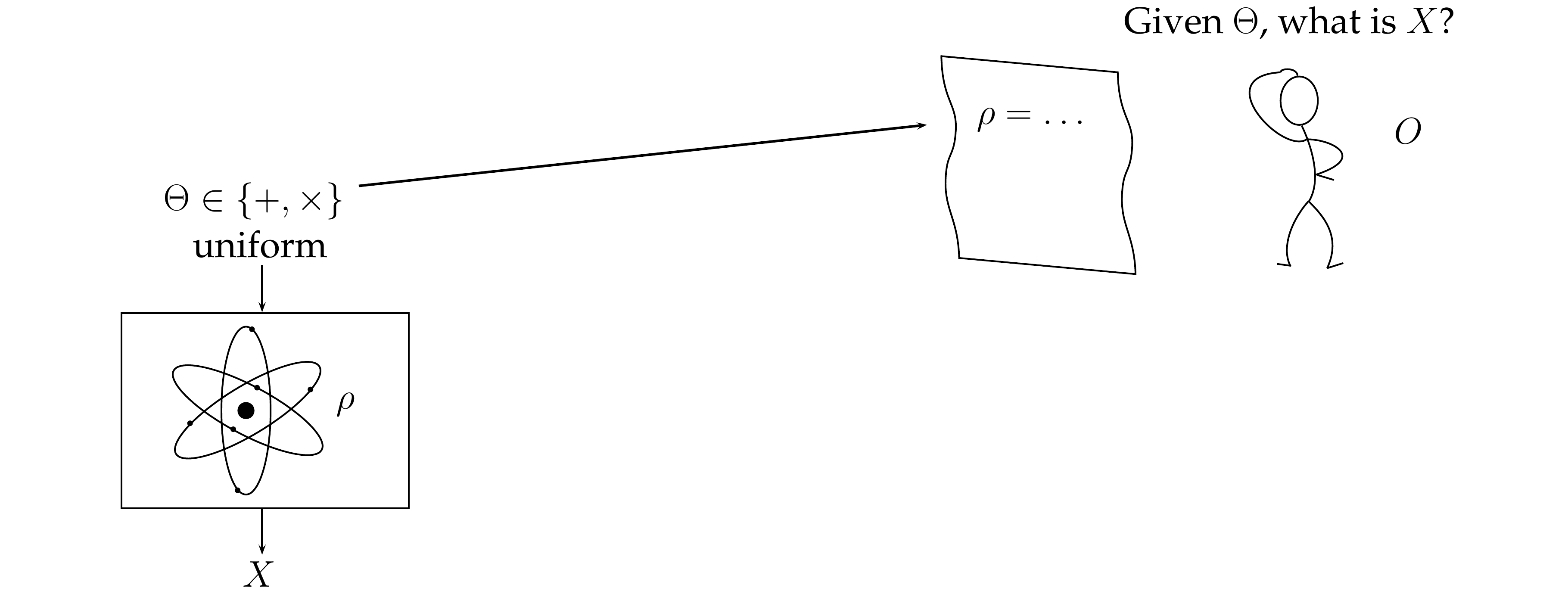}}
  
  \subfigure[This game can be trivially extended to two observers, $O_1$ and $O_2$. In this case, 
    the sum of their entropies satisfies $\hvn{X|O_1 Θ} + \hvn{X|O_2 Θ} ≥ \log \frac{1}{c}$. This scenario is mainly interesting in the quantum case.]
  {\includegraphics[width=0.95\textwidth]{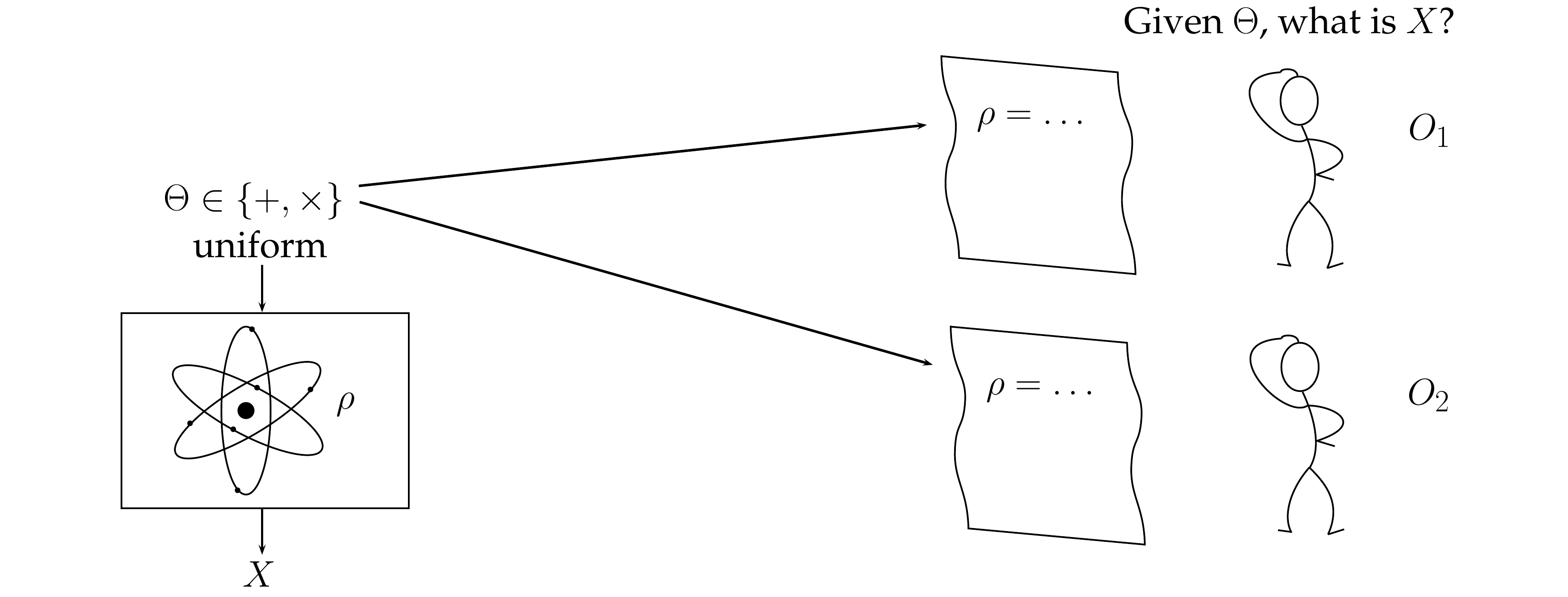}}
  
  \caption[Uncertainty Relations with Classical Side Information.]{\emph{Uncertainty Relations with Classical Side Information.}}
  \label{fig:ucr/cl}
\end{figure}

\begin{figure}
  
  \subfigure[A quantum system, $A$, is measured in one out of two bases, determined 
    by a uniformly random bit, $Θ$. Then the basis choice $Θ$ is sent to an observer, $O$, who 
    is asked to determine the measurement outcome, $X$. The observer holds a quantum system which 
    might be entangled with the measured system. One may alternatively think of the observer 
    preparing the bipartite state $\ii{ρ}{AO}$, of which the $A$ part is measured 
    and the $O$ part is stored in a quantum memory.
    For this setup, no uncertainty relation of the type~\eqref{eq:ucr/game-cl} exists. As a 
    counter-example, consider the case where $\ii{ρ}{AO}$ is fully entangled. In this case, the 
    observer can always choose a measurement (depending on $Θ$) whose outcome is perfectly 
    correlated with $X$. For example, if the state is a singlet, measuring both the $A$ and $O$ 
    parts in the same, arbitrary basis will lead to perfectly anti-correlated 
    binary variables.]
  {\includegraphics[width=0.95\textwidth]{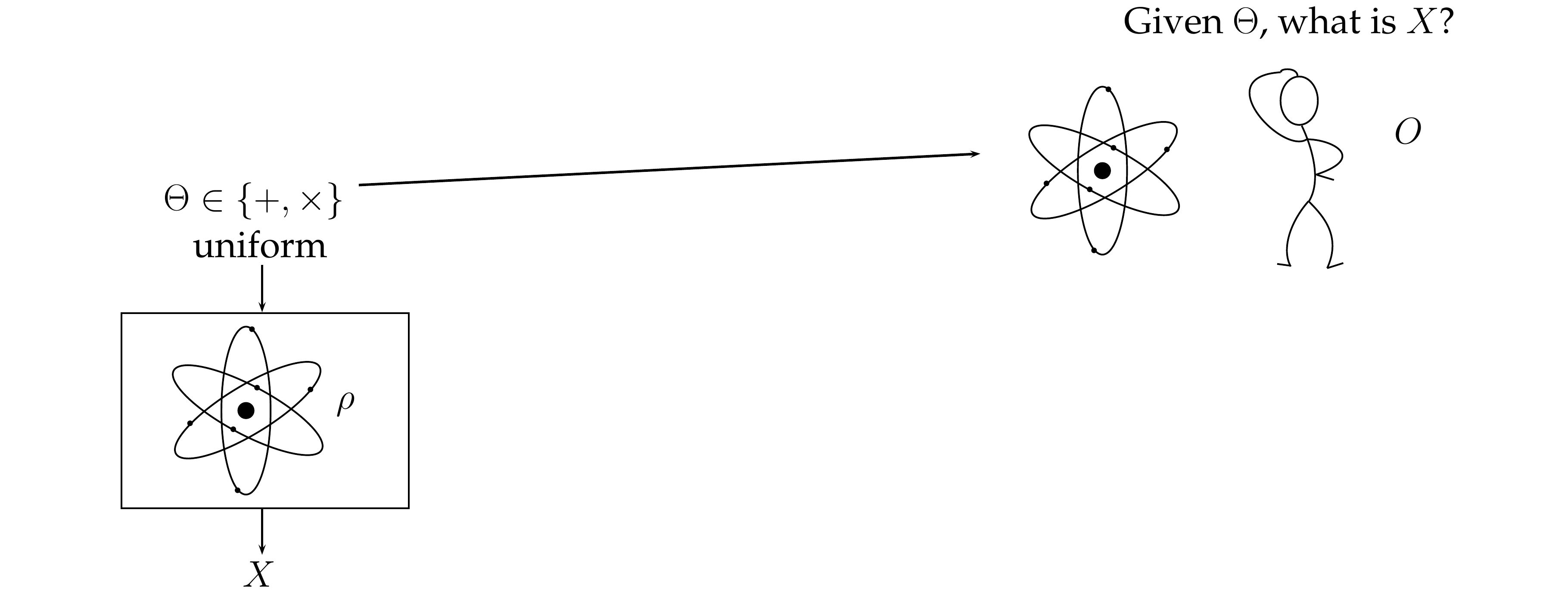}}
  
  \subfigure[However, if we instead consider two observers, $O_1$ and $O_2$, an 
    uncertainty relation is possible. Note, in particular, that the counterexample provided 
    for the case of one observer fails due to the monogamy of entanglement. As in the classical
    case, the sum of the entropies of the two observers satisfies $\hvn{X|O_1 Θ} + \hvn{X|O_2 Θ} 
    ≥ \log \frac{1}{c}$. However, in the quantum case it is possible that one of these entropies 
    is zero, which implies that the other entropy is large. This trade-off between 
    the two entropies does not exist in the classical case and can be seen as an effect of 
    no-cloning\,|\,it is generally not possible that the two observers hold copies of each 
    other's quantum information.]
  {\includegraphics[width=0.95\textwidth]{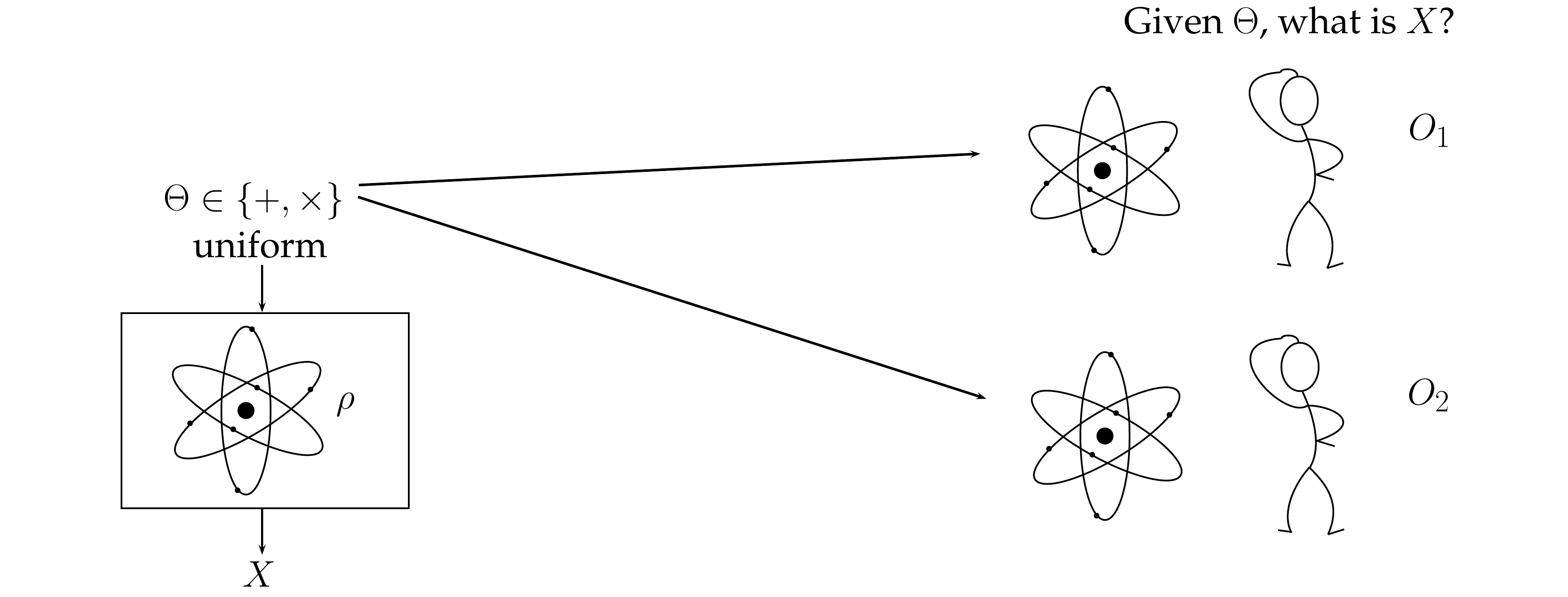}}
  
  \caption[Uncertainty Relations with Quantum Side Information.]{\emph{Uncertainty Relations with Quantum Side Information.}}
  \label{fig:ucr/q}
\end{figure}

\subsection{Main Contributions}

We have already seen in Chapter~\ref{ch:smooth} that the smooth entropies, for any tripartite state $ρ¬{ABC}$ and $ε ≥ 0$, satisfy the duality relation
\begin{align}
  	\hmin*{A|B}[ρ] + \hmax*{A|C}[ρ] ≥ 0 ¶\,.
\end{align}
Our result now shows that this lower bound increases (as is the case for von Neumann entropies) when we apply incompatible measurements on the $A$ system. This extends the entropic uncertainty relation of~§[ucr/intro/vn] to smooth entropies.

\begin{result}
  \label{res:uc}
  For $ε ≥ 0$, any tripartite state $ρ¬{ABC}$ as well as two POVMs $\{ M_x \}$ 
  and $\{ N_y \}$ on $A$, it holds that
  \begin{align}
    \hmin*{X|B} + \hmax*{Y|C} ≥ \log \frac{1}{c} , \tn*{where} 
      c = \max_{x,y} \norm<B>{√{M_x} √{N_y}}[∞]^2 ¶.
  \end{align}
\end{result}

This uncertainty relation has been used to prove security in quantum cryptography. 
In~\cite{tomamichellim11}, it was used to derive tighter key rates for finite block lengths in quantum key distribution. In~\cite{lesheridan11}, Le \emph{et al.} prove security of a reference frame independent quantum 
key distribution protocol using the above uncertainty relation.\footnote{Applications of the uncertainty relation to quantum cryptography will also be the topic of Section~\ref{se:qkd} of the next chapter.}
Furthermore, the result has been shown in the framework of general von Neumann algebras 
in~\cite{furrer11}, indicating that the UCR for min- and max-entropies
is a fundamental property of quantum physics and not a relict of the density operator formalism.

Furthermore, we explore a tighter version of the uncertainty relation with a lower bound in terms of an ☼[effective overlap]{overlap!effective}. The effective overlap\,|\,in contrast to the overlap\,|\,is a function of the marginal state prior to measurement as well as the two measurements. A preliminary version of this result appeared in~\cite{haenggi11}, where the relation is shown for von Neumann entropies. We extend these results here and show that a generalized uncertainty relation also holds for smooth min- and max-entropies, enabling its substitution into existing cryptographic security proofs~\cite{tomamichellim11,lim12}.

\begin{result}
  \label{res:ucg}
  For $ε > 0$, any tripartite state $ρ¬{ABC}$ as well as two POVMs $\cX = \{ M_x \}$ 
  and $\cY = \{ N_y \}$ on $A$, it holds that
  \begin{align}
    \hmin[3ε]{X|B} + \hmax*{Y|C} ≥ \log \frac{1}{c^*} - \log \frac{2}{ε^2} ¶,
  \end{align}
  where $c^* = \sum_k \tr{P^k ρ¬A} c_k$. Here, $\{ P^k \}$ is any projective measurement
  that commutes with both $\cX$ and $\cY$ and $c_k$ is the overlap of $\cX$ and $\cY$
  on the subspace $P^k$.
\end{result}

\subsection{Outline}

The remainder of this chapter is structured as follows. In Section~\ref{sc:ucr/trad}, 
we introduce the notion of ☼{overlap} and show the uncertainty relation for the min-
and max-entropy. The proof is very instructive and will also guide the proof of the generalized
uncertainty relation, which is given in Section~\ref{sc:ucr/g}. There, we will also
formally define the effective overlap.
In Section~\ref{sc:ucr/coll}, we discuss a variety of corollaries of the generalized
uncertainty relation. We consider an application of the uncertainty relation
to quantum key distribution. Finally, in Section~\ref{sc:ucr/bi}, we consider a bipartite uncertainty relation which might have applications in cryptographic settings where only two parties are involved.

%................................
\section{Traditional Formulation}
\label{sc:ucr/trad}

\subsection{Overlap}

Let $ρ ∈ \osub{ℋ¬{ABC}}$ be an arbitrary, tripartite quantum state. We want to bound the entropy about the result of a measurement on the $A$ subsystem given side information stored in either $B$ or $C$. Without loss of generality, such a measurement on the $A$ system can be described by a POVM.\footnote{This is true since we do not use the post-measurement state on the $A$ system and, thus, the freedom to choose a phase in the decomposition of the POVM elements $M = E†E$ is meaningless.} Here we consider two POVMs, $\cX = \{ M_x \}$ and $\cY = \{ N_y \}$, on $ℋ¬A$. They measure the state on the $A$ subsystem and store the measurement outcome in classical registers $X$ and 
$Y$, respectively. These two registers can be seen as classical random variables that are correlated with quantum side information on the $B$ and $C$ systems.
Note that, in our analysis here, we are not concerned with the state of the $A$ system after the measurement.

The relevant post-measurement states are thus given as
\begin{align}
  ρ¬{XBC} &= ∑_{x} \proj{x} ⨂ \tr[A]{√{M_x} ρ¬{ABC} √{M_x}}
    \tn*{and} ¶[ucr/pm-x]\\
  ρ¬{YBC} &= ∑_{y} \proj{y} ⨂ \tr[A]{√{N_y} ρ¬{ABC} √{N_y}} 
    ¶[ucr/pm-y]\,.
\end{align}

We also define the ☼{overlap} of these two measurements.
\begin{definition}[Overlap]
  \label{df:overlap}
  Let $\cX = \{ M_x \}$ and $\cY = \{ N_y \}$ be two POVMs. Then, we define the 
  overlap of $\cX$ and $\cY$ as
  \begin{align}
    c(\cX, \cY) := \max_{x,y} \norm<b>{√{M_x}√{N_y}}[∞]^2 ¶\,.
  \end{align}
\end{definition}
This is in accordance with~\cite{krishna01}, where an uncertainty relation (without side information) was first shown for von Neumann entropies.
If the two measurements are projective, the expression for the overlap reduces to
\begin{align}
   c = \max_{x,y} \abs<b>{\braket{x|y}}^2 ¶\,,
\end{align}
where the maximization is over all eigenvectors $\ket{x}$ of $\cX$ and $\ket{y}$ of $\cY$. The name ☼{overlap} is clearly motivated by this expression.

\subsection{Uncertainty Relation for Min- and Max-Entropies}
\label{sc:ucr/proof}

The following inequality (cf.~\cite{tomamichel11}) bounds the sum of the min- and max-entropies of the post-measurement states in terms of the overlap. In the next sections, we will formulate a generalization of this statement to smooth entropies and mixed states. Here, we simply consider the special case~§[ucr/noeps], as its proof highlights the basic concepts and techniques used for the proof of the subsequent generalized uncertainty relation.

\begin{theorem}[UCR]
  \label{th:ucr}
  Let $ρ ∈ \osub{ℋ¬{ABC}}$ be pure and $\cX$ as well as $\cY$ two POVMs on $ℋ¬A$. Then, the post   
  measurement states~§[ucr/pm-x] and~§[ucr/pm-y] satisfy
  \begin{align}
    \hmin{X|B}[ρ] + \hmax{Y|C}[ρ] ≥ \log \frac{1}{c(\cX, \cY)} ¶[ucr/noeps]\,.
  \end{align}
\end{theorem}

\begin{proof}
  It will be helpful to describe the two measurements in the Stinespring
  dilation picture~(cf.~Lemma~\ref{lm:stinespring}) 
  as isometries followed by a partial trace.  Let $U$
  be the isometry from $A$ to $A$, $X$ and $X'$ given by $U := ∑_x √{M_x} ⨂ \ket{x}
  ⨂ \ket{x}$. The isometry stores two copies of the
  measurement outcome in the registers $X$ and $X'$ and the measured 
  state in $A$. Here, $\{ \ket{x} \}$ is an orthonormal basis 
  of $ℋ¬X \iso ℋ¬{X'}$.
  Analogously, $V := ∑_y  √{N_y} ⨂ \ket{y} ⨂ \ket{y}$.  
  Furthermore, we introduce the states $ρ¬{AXX'BC} = U ρ¬{ABC} U†$ and $ρ¬{AYY'BC} =
  V ρ¬{ABC} V†$, of which the post-measurement states appearing in 
  Eq.~§[ucr/noeps], $ρ¬{XB}$ and $ρ¬{YC}$, are marginals.
  
  The duality relation (cf.~Lemma~\ref{lm:min-max/dual}) applied to $\rhot¬{AYY'BC}$ gives
  \begin{align}
    \hmax{Y|C}[ρ] + \hmin{Y|AY'B}[ρ] = 0 ¶[ucr/mo1]\,.
  \end{align}
  Comparing~§[ucr/mo1] with the statement of the theorem, it
  remains to show that $\hmin{Y|AY'B}[ρ] ≤ \hmin{X|B}[ρ] - \log \frac{1}{c}$ holds.
  More precisely, we will show that
  \begin{align}
   \hmin{Y|Y'AB}[ρ] &= \max_{σ}\, \sup \{ λ ∈ ℝ : ρ¬{AYY'B} ≤ 2^{-λ} ⅈ¬Y ⨂ σ¬{AY'B} \} ¶\\
    &≤ \sup \{ λ ∈ ℝ : ρ¬{XB} ≤ 2^{-λ} c\, ⅈ¬X ⨂ σ¬B \} ¶[ucr/mo3]\\
    &= \hmin{X|B}[ρ] - \log \frac{1}{c} ¶\,.
	\end{align}
  In order to arrive at~§[ucr/mo3], we thus need to show 
  that, for any $σ ∈ \onorm{ℋ¬{AY'B}}$, the following implication holds
  \begin{align}
    ρ¬{AYY'B} ≤ 2^{-λ} ⅈ¬Y ⨂ σ¬{AY'B} \implies 
    ρ¬{XB} ≤ 2^{-λ} c\, ⅈ¬X ⨂ σ¬B ¶[ucr/mo6]\,.
  \end{align}

  To show this, we apply the partial isometry $W := U V†$ followed by a partial
  trace over $X'$ and $A$ on both sides of the inequality on the left-hand
  side. This implies
  \begin{align}
    2^{λ} ρ¬{XB} ≤ \tr<b>[X'A]{ W ( ⅈ¬Y ⨂ σ¬{AY'B} ) W† } ¶[ucr/mo10] \, .
  \end{align}
  Moreover, substituting the definition of $W$, we find that the trace term on 
  the rhs.\ evaluates to
  \begin{align}
    \tn-{rhs.} = ∑_{x,y} \proj{x} ⨂ \braket<b>{y| \tr<b>[A]{√{M_x} √{N_y}\, 
      σ¬{AY'B} √{N_y} √{M_x} } |y} ¶[ucr/mo9]\,.
  \end{align}
  Lemma~\ref{lm:pt-norm-bound}, in particular Eq.~§[pt-norm-bound], now establishes that
  \begin{align}
    \tr<b>[A]{√{M_x} √{N_y}\, σ¬{AY'B} √{N_y} √{M_x} } 
       &≤ \norm<b>{ √{M_x} √{N_y} }[∞]^2\, σ_{Y'B} ≤ c\cdot σ¬{Y'B}¶\,.
  \end{align}
  Combining this with~§[ucr/mo9] and~§[ucr/mo10] results in the inequality
  \begin{align}
    2^{λ} ρ¬{XB} ≤ 2^{-λ} c\, ∑_{x,y} \proj{x} ⨂ \braket{y|σ¬{BY'}|y} 
      = 2^{-λ} c\, ⅈ¬X ⨂ σ¬B ¶\,.
  \end{align}
   This establishes~§[ucr/mo6] and concludes the proof.
\end{proof}

\subsection{Modeling the Measurement Basis Explicitly}
\label{sc:ucr/model}

An alternative formulation of the uncertainty principle requires an additonal random variable, $Θ$, which determines the choice of measurement on the $A$ system. Consider, for example, the setup of the previous section, where the choice is between two incompatible measurements, $\cX$ and $\cY$, which we assume both have the same number of different outcomes. In this case the random experiment of picking $θ ∈ Θ = \{ 0, 1 \}$ determines the binary choice of measurement. More specifically, say that $θ = 0$ leads to a measurement of $\cX$ and $θ = 1$ leads to a measurement of $\cY$. The measurement outcome, in either case, is stored in a classical register $Z$. 

If $Θ$ is uniform and independent of the state $ρ$ before measurement, we find
\begin{align}
  \hmin{Z|B Θ} + \hmax{Z|C Θ} ≥ \log \frac{1}{c(\cX, \cY)} ¶[ucr/alt-nonsmooth]\,.
\end{align}
To see this, note that~§[ucr/noeps] implies
\begin{align}
  2^{-\hmin{X|B}} ≤ c\, 2^{\hmax{Y|B}} \tn*{and} 2^{-\hmin{Y|B}} ≤ c\, 2^{\hmax{X|B}} ¶\,.
\end{align}
Taking the convex sum with equal weight $\frac{1}{2}$ of these two inequalities leads to~§[ucr/alt-nonsmooth], where we used the fact that the min- and max-entropies with classical side information can be expressed as averages (cf.\ Proposition~\ref{pr:classical-side-info}). 

A generalization of this type of uncertainty relation to smooth entropies will be 
discussed below in Corollary~\ref{co:ucr-basis}. 

%................................

\section{The Generalized Uncertainty Relation}
\label{sc:ucr/g}

\subsection{Effective Overlap}
☼*{effective overlap|see{overlap}}

The overlap $c(\cX, \cY)$ used in Theorem~\ref{th:ucr} is a function of the POVM elements 
of the two measurements under consideration and independent of the state prior to measurement.
This is often desirable because this state might be unknown, or, in a cryptographic setting, prepared by an adversary.
However, we will see that in some situations partial knowledge about the state before measurement can be used to improve the bound on the uncertainty.
What follows is thus a generalized uncertainty relation of the form of Theorem~\ref{th:ucr} that introduces a trade-off between information about the marginal state before measurement and tightness of the uncertainty relation. Specifically, we consider the ☼[effective overlap]{overlap!effective} of a
☼{measurement setup}, denoted $c^*$, which describes the overlap of
two measurements on a given ☼{marginal} state.

\begin{definition}[Effective Overlap]
  \label{df:overlap-eff}
  Let $ρ¬A ∈ \osub{ℋ¬A}$ be a state and let $\cX = \{ M_x \}$, $\cY = \{ N_y \}$ be two POVMs 
  on $ℋ¬A$. Then, we call the triple $\{ ρ¬A, \cX, \cY \}$ a measurement setup.
  The effective overlap of this measurement setup is defined as
  \begin{align}
    c^*(ρ¬A, \cX, \cY) := \min_{\cK} \bigg\{ \sum_k \tr{P^k ρ} 
      \max_x \norm<B>{P^k \sum_y N_y M_x N_y}[∞] \bigg\} ¶\,
  \end{align}
  where the minimum is taken over all projective measurements $\cK = \{ P^k \}$ on $ℋ¬A$ that
  commute with both $\cX$ and $\cY$.\footnote{The property that two measurements 
  $\cX$ and $\cK$ commute is equivalent to the condition $M_x P^k = P^k M_x$ for all $x$ and $k$.}
\end{definition}

In the following sections, we will show that an UCR also holds for this definition
of effective overlap.

As a first example of the usefulness of such a generalized 
UCR, consider the scenario where we apply one of two projective measurements, either in the basis 
$\{ \ket{0}, \ket{1}, \ket{\perp} \}$ or $\{
\ket{+}, \ket{-}, \ket{\perp} \}$ on a state $ρ$ which has the property that
`$\perp$' is measured with probability at most $η$.\footnote{The diagonal states $\ket{±}$ are defined as $\ket{±} := (\ket{0} ± \ket{1})/√{2}$.} 
A direct application of the state-independent uncertainty relation~(Theorem~\ref{th:ucr}) to this setup will not lead to the desired results as the overlap of the two bases is trivially $c = 1$. Still, our intuitive understanding of this situation tells us that the uncertainty about
the measurement outcome is high as long as $η$ is small. 

In fact, the effective overlap of this setup satisfies $c_* ≤ (1 - η) \frac{1}{2} + η$. This formula can be interpreted as follows: with probability $1-η$ we are in the subspace spanned by $\ket{0}$ and $\ket{1}$, where the overlap is $\frac{1}{2}$, and with probability $η$ we measure $\perp$ and have full overlap.
To prove this upper bound, simply choose the projective measurement
$\cK = \{ \proj{0} + \proj{1}, \proj{\perp} \}$ in Definition~\ref{df:overlap-eff}.
Hence, while Theorem~\ref{th:ucr} thus only provides a trivial bound for this example, an uncertainty relation in terms of the effective overlap would give the expected bound.

We get a state-independent bound on the effective 
overlap with the choice
$\cK = \{ ⅈ¬A \}$, namely
\begin{align}
  c^*(ρ¬A, \cX, \cY) ≤ \max_x \norm<B>{∑_y N_y M_x N_y}[∞] ¶\,.
\end{align}
Furthermore, note that for projective measurements $\cX$ and
$\cY$, the rhs.\ can be simplified to $\max_{x, y} \abs<b>{\braket{x|y}}^2$, in 
agreement with the the usual overlap in Definiton~\ref{df:overlap}. In general, it is conjectured that $c^* ≤ c$, i.e.
\begin{align}
  \max_x \norm<B>{ ∑_y N_y M_x N_y }[∞]\! ≤\, \max_{x,y} \norm<B>{ √{N_y} M_x √{N_y} }[∞]\!
    =\, \max_{x,y} \norm<B>{ √{M_x} √{N_y} }[∞]^2  ¶[conj/c-ceff].
\end{align}

Finally, if the two measurements on the $A$ system have binary outcomes, it is possible to
upper bound the effective overlap (and, thus, lower bound the uncertainty) 
by the maximal CHSH Bell violation~\cite{CHSH} 
that can be observed using this measurement setup on $A$ with an arbitrary second party 
(i.e.\ on an arbitrary extension of the state and using arbitrary measurements by the other party.) 
This establishes an analytic
relation between two fundamental concepts in quantum theory, Bell non-locality and uncertainty relations. We refer the interested reader to~\cite{haenggi11}, where this relation is
discussed in detail.

\subsection{The Generalized Uncertainty Relation}

We now consider a theorem that gives a very general formulation of the uncertainty principle for smooth entropies. It gives a lower bound on the uncertainty\,|\,in terms of smooth min- and max-entropies\,|\,about the outcome of two (incompatible) measurements, $X$ and $Z$, conditioned on quantum side information and the result of an additional, projective measurement, $K$, that was done on the state before measuring $X$ and $Z$.

More formally, we consider an arbitrary tripartite quantum state, $ρ¬{ABC}$, two POVMs on $A$, $\cX = \{ M_x \}$ and $\cY = \{ N_y \}$, as well as a projective measurement, $\cK = \{ P_k \}$. The post-measurement states when $X$ and $Y$ are measured after $K$ are
\begin{align}
  ρ¬{XKB} &= ∑_{x,k} \proj{x} ⨂ \proj{k} ⨂ \tr[AC]{√{M_x} P^k ρ¬{ABC} P^k √{M_x}}
    \tn*{and} ¶[ucr/pm-xk]\\
  ρ¬{YKC} &= ∑_{y,k} \proj{y} ⨂ \proj{k} ⨂ \tr[AB]{√{N_y} P^k ρ¬{ABC} P^k √{N_y}} 
    ¶[ucr/pm-yk]\,.
\end{align}

The following theorem generalizes previously known uncertainty relations for the smooth min-
and max-entropies in a tripartite setting.\footnote{The author is tempted\,|\,but will resist\,|\,to call it the mother uncertainty relation.} The UCRs 
discussed in Section~\ref{sc:ucr/coll} are corollaries of this relation.

\begin{theorem}[Generalized UCR]
  \label{th:ucr-gen}
  Let $ε ≥ 0$, $\bar{ε} > 0$ and $ρ ∈ \osub{ℋ¬{ABC}}$. Moreover, let $\cX = \{ M_x \}$, 
  $\cY = \{ N_y \}$ be POVMs on $ℋ¬A$ and $\cK = \{ P^k \}$ a projective measurement on $ℋ¬A$. 
  Then, the post-measurement states~§[ucr/pm-xk] and~§[ucr/pm-yk] satisfy
  \begin{align}
    \hmin*{X|KB}[ρ] + \hmax*{Y|KC}[ρ] &≥ \log \frac{1}{c_{\cK}}
      \tn*{and} ¶[ucr/thm1]\\
    \hmin[2ε+\bar{ε}]{X|KB}[ρ] + \hmax*{Y|KC}[ρ] &≥ \log \frac{1}{c_{\cK}^*} - 
      \log \frac{2}{\bar{ε}^2} ¶[ucr/thm2]\,.      
  \end{align}
  where the $\cK$-overlap, $c_{\cK}$, and the effective $\cK$-overlap, $c_{\cK}^*$, are given as
  \begin{align}
    c_{\cK} &:= \max_{k,x,y} \norm<b>{√{M_x} P^k √{N_y}}[∞] \tn*{and} ¶\\
    c_{\cK}^* &:= ∑_k \tr{P^k ρ} \max_{x} \norm<B>{∑_y P^k N_y P^k M_x P^k N_y P^k}[∞] ¶.
  \end{align}
\end{theorem}

\begin{proof}
  We will in the following prove the statement for pure $ρ¬{ABC}$. Its generalization to mixed states then trivially 
  follows from data processing inequalities for the smooth entropies 
  (cf.\ Theorem~\ref{th:data-proc}). More specifically, we consider a purification $ρ¬{ABCD}$ of 
  $ρ¬{ABC}$, for which~§[ucr/thm1] and~§[ucr/thm2] hold with the substitution $C → CD$ and then 
  take the partial trace over $D$. As this cannot decrease the smooth max-entropy, the 
  generalization follows.

  We consider the Stinespring dilation (cf.~Lemma~\ref{lm:stinespring}) of the joint 
  measurement of $\cX$ and $\cK$, denoted $U$, which coherently stores the measurement
  outcome of $\cX$ in registers $X$ and $X'$ and the measurement outcome of 
  $\cK$ in $K$ and $K'$. The isometry $U: ℋ¬A → ℋ¬{AXX'KK'}$ is given by
  \begin{align}
    U := ∑_{x,k} \ket{x}[X] ⨂ \ket{x}[X'] ⨂ \ket{k}[K] ⨂ \ket{k}[K'] ⨂ √{M_x}\, P^k ¶\,.
  \end{align}
  Similarly, we introduce the Stinespring dilation of the joint measurement of 
  $\cY$ and $\cK$ as
  \begin{align}
    V := ∑_{y,k} \ket{y}[Y] ⨂ \ket{y}[Y'] ⨂ \ket{k}[K] ⨂ \ket{k}[K'] ⨂ √{N_y}\, P^k ¶\,.
  \end{align}
  We will also need the partial isometry $W := U V†$, which, using 
  $P^k P^{k'} = δ_{kk'} P^k$, evaluates to
  \begin{align}
    W = ∑_{x,y,k} \ket{x}\!\bra{y} ⨂ \ket{x}\!\bra{y} ⨂ \proj{k} ⨂ \proj{k} 
      ⨂ √{M_x} P^k √{N_y} ¶[ucr/gp0]\,.
  \end{align}
  
  These isometries allow us to introduce the states $ρ¬{AXX'KK'BC} = U ρ¬{ABC} U†$
  and $ρ¬{AYY'KK'BC} = V ρ¬{ABC} V†$, whose marginals are the 
  post-measurement states $ρ¬{XKB}$ and $ρ¬{YKC}$ of~§[ucr/pm-xk] and~§[ucr/pm-yk],
  respectively.
  
  The proof now proceeds in several steps. First, we reformulate the statement of the theorem 
  in terms of smooth min-entropies using the duality relation. 
  Then, in order to show~§[ucr/thm1], we use use 
  techniques similar to the ones discussed in Section~\ref{sc:ucr/proof} 
  to extract the $\cK$-overlap.
  To show~§[ucr/thm2], we first use Lemma~\ref{lm:rel-smooth-bound} to find an upper bound 
  on one of the entropies in terms of a relative entropy that is conditioned on $ρ$. The
  properties of $ρ$ can then be used to extract the effective $\cK$-overlap.
  
% (proof of first statement)

  Due to the duality between smooth min- and max-entropy (cf.~Theorem~\ref{th:smooth-dual}), 
  the first statement of the theorem is equivalent to
  \begin{align}
    \hmin*{X|KB}[ρ] ≥ \hmin*{Y|AY'K'B}[ρ] + \log \frac{1}{c_{\cK}} ¶[ucr/gq1]\,.
  \end{align}
  Let $\rhob$ be a state that maximizes the smooth min-entropy on the rhs. Then, there 
  exists a $σ$ such that
  \begin{align}
    λ = S_{\min}(\rhob¬{AYY'K'B} \| ⅈ¬{Y} ⨂ σ¬{AY'K'B}) = \hmin*{Y|AY'K'B}[ρ] ¶[ucr/gq2]\,.
  \end{align}
  
  The state $\rhob$ can be chosen classical on $K'$ due to Proposition~\ref{pr:class-smooth}.
  Furthermore, note that the minimum distance purification of $\rhob$, since its marginal is 
  classical on $K'$, also inherits the coherence between $K$ and $K'$ from $ρ$. This follows from
  \begin{align}
    P( Π \rhob Π, ρ) = P( Π \rhob Π, Π ρ Π) ≤ P( \rhob, ρ) ¶\,,
  \end{align}
  where $Π = \sum_k \proj{k}[K] ⨂ \proj{k}[K']$ is the projector onto coherent superpositions 
  of $K$ and $K'$. We also used that $\tr[K]{Π \rhob Π} = ℳ[\rhob] = \rhob$, 
  where $ℳ$ measures the $K'$ system; thus, the projection does not change the marginal state. 
  We also need another extension of $\rhob$ to $K$, which we define via the TP-CPM $ℰ: ρ ↦ ∑_k 
  \proj{k}[K']\, ρ\, \proj{k}[K'] ⨂ \proj{k}[K]$. This map measures $K'$ and creates a
  (classical) copy of it in $K$. Since $\rhob¬{AYY'K'B}$ is classical on $K$, 
  it is easy to see that $\taub¬{AYY'KK'B} = ℰ[\rhob¬{AYY'K'B}]$ is an extension of 
  $\rhob¬{AYY'K'B}$. Furthermore, we can equivalently write 
  $\taub¬{AYY'KK'B} = Π (\rhob¬{AYY'K'B} ⨂ ⅈ¬{K}) Π$ and note that the two extensions 
  agree on the diagonals:
  \begin{align}
    \forall k, k': \braket{kk'|\taub¬{AYY'KK'B}|kk'} = \braket{kk'|\rhob¬{AYY'KK'B}|kk'} ¶[ucr/gp7]\,.
  \end{align}

  From the definition of $S_{\min}(\rhob¬{AYY'K'B} \| ⅈ¬{Y} ⨂ σ¬{AY'K'B})$, we get
  \begin{align}
    \rhob¬{AYY'K'B} ≤ 2^{-λ}\, ⅈ¬Y ⨂ σ¬{AY'K'B} \,. ¶[ucr/gq3] 
  \end{align}
  Taking the tensor product with $ⅈ¬{K}$ on both sides of~§[ucr/gq3] leads to
  \begin{align}
    \rhob¬{AYY'K'B} ⨂ ⅈ¬{K} ≤ 2^{-λ}\, ⅈ¬{YK} ⨂ σ¬{AY'K'B} ¶\,.
  \end{align}
  We conjugate this inequality with $W$ and take the partial trace over 
  $A$, $Y'$ and $K'$ to get
  \begin{align}
    \taub¬{XKB} &= \tr<b>[AX'K']{W (\rhob¬{AYY'K'B} ⨂ ⅈ¬{K}) W†} ¶\\
      &≤ 2^{-λ}\, \tr<b>[AX'K]{W (ⅈ¬{YK} ⨂ σ¬{AY'K'B}) W†} ¶[ucr/gq5]\,.
  \end{align}

  We first evaluate the trace term on the rhs.\ of~§[ucr/gq5] using the expression for 
  $W$ in~§[ucr/gp0] and Lemma~\ref{lm:pt-norm-bound} as in Eq.~§[pt-norm-bound]. We get
  \begin{align}
    &\tr<b>[AX'K']{W (ⅈ¬{YK} ⨂ σ¬{AY'K'B}) W†} ¶\\
    &\quad\!\! = ∑_{x,y,k} \proj{x} ⨂ \proj{k} ⨂ 
      \braket{yk|\tr<b>[A]{√{M_x} P^k\! √{N_y} σ¬{AY'K'B} √{N_y} P^k\! √{M_x}}|yk} ¶\\
    &\quad\!\! ≤ ⅈ¬X ⨂ ∑_k \proj{k} ⨂ \max_{x,y} \norm<B>{√{N_y} P^k\! √{M_x}}[∞]^2 
      \braket{k|σ¬{KB}|k} ¶\\
    &\quad\!\! = ⅈ¬X ⨂ \omegat¬{KB} ¶\,.
  \end{align}
  The last equality defines the operator $\omegat$.
  Note that $\tr{\omegat¬{KB}} ≤ c_{\cK}$; hence, using $ω¬{KB} = \omegat¬{KB}/c_{\cK} ∈ 
  \osub{ℋ¬{KB}}$ and~§[ucr/gq5], we find a lower bound on 
  $S_{\min}(\taub¬{XKB} \| ⅈ¬X ⨂ ω¬{KB})$ in
  terms of $λ$ and $c_{\cK}$.
  \begin{align}
    S_{\min}(\taub¬{XKB} \| ⅈ¬X ⨂ ω¬{KB}) ≥ λ + \log \frac{1}{c_{\cK}} = \hmin*{Y|AY'K'B}[ρ] + 
      \log \frac{1}{c_{\cK}} ¶[ucr/gq6],
  \end{align}
  where we substituted~§[ucr/gq2] for $λ$.
  The state $\taub¬{XKB}$ is sub-normalized.  %
  Furthermore, its purified distance from $ρ¬{XKB}$ is bounded by
  \begin{align}
    P(\taub¬{XKB}, ρ¬{XKB}) = P(\rhob¬{XKB}, ρ¬{XKB}) ≤ ε ¶\,,
  \end{align}
  where we used that $\taub¬{XKB} = \rhob¬{XKB}$. This follows from the property that the two
  states agree on the diagonals of $\ket{kk}$\,|\,in the sense of~§[ucr/gp7]\,|\,together 
  with the expression
  \begin{align}
    \taub¬{XKB} &= \tr<b>[AX'K']{W \taub¬{AYY'KK'B} W†} = ∑_{k,x,y,y'} \proj{x} ⨂ \proj{k} ⨂ ¶\\
    &\qquad \ \ \tr<B>[A]{√{M_x} P^k √{N_y} \braket<b>{yykk|\taub¬{AYY'KK'B}|y'y'kk} √{N_{y'}} 
       P^k √{M_x}} ¶,
  \end{align}
  and the respective expression for $\rhob¬{XKB}$.
  Hence, using the definition of the smooth min-entropy, we get
  \begin{align}
    \hmin*{X|KB}[ρ] ≥ S_{\min}(\taub¬{XKB} \| ⅈ¬X ⨂ ω¬{KB})¶\,,
  \end{align}
  which, substituted into~§[ucr/gq6], concludes the proof of~§[ucr/thm1].

% (proof of second statement)

  It remains to show the second statement of the theorem. Again, due to the duality between 
  smooth min- and max-entropy this is equivalent to
  \begin{align}
    \hmin[2ε+\bar{ε}]{X|KB}[ρ] ≥ \hmin*{Y|AY'K'B}[ρ] + 
      \log \frac{1}{c_{\cK}^*} - 
      \log \frac{2}{\bar{ε}^2} ¶[ucr/gp1]\,.
  \end{align}
  
  Applying Lemma~\ref{lm:rel-smooth-bound}, we define the state $\rhot ∈ 
  \ball[2ε+\bar{ε}]{ρ}$ such that the following holds:
  \begin{align}
    S_{\min}(\rhot¬{AYY'K'B} \| ⅈ¬Y ⨂ ρ¬{AY'K'B}) 
      ≥ \hmin*{Y|AY'K'B}[ρ] - \log \frac{2}{\bar{ε}^2} ¶\,.
  \end{align}
  We use the monotonicity of $S_{\min}$ under TP-CPMs to measure the 
  $K$ system, i.e.\ we apply the map $ℳ: τ ↦ ∑_k \proj{k}[K']\, τ\, \proj{k}[K']$ to both states 
  in $S_{\min}$ above. 
  This will have no effect on $ρ¬{AY'K'B}$, which is classical on $K'$ by definition. Using the 
  state $\rhob¬{AYY'K'B} = ℳ[\rhot¬{AYY'K'B}]$, we have
  \begin{align}
    λ = S_{\min}(\rhob¬{AYY'K'B}\|ⅈ¬Y ⨂ ρ¬{AY'K'B}) 
      &≥ \hmin*{Y|AY'K'B}[ρ] - \log \frac{2}{\bar{ε}^2} 
      ¶[ucr/gp2]\,.
  \end{align}
  Moreover, the purified distance satisfies 
  \begin{align}
    P(\rhob, ρ) &= P(ℳ[\rhot], ℳ[ρ]) ≤ P(\rhot, ρ) ≤ 2ε + \bar{ε} ¶\,.
  \end{align}
  From the definition of 
  $S_{\min}(\rhob¬{AYY'K'B}\|ⅈ¬Y ⨂ ρ¬{AY'K'B})$, we get
  \begin{align}
    \rhob¬{AYY'K'B} ≤ 2^{-λ}\, ⅈ¬Y ⨂ ρ¬{AY'K'B} \,, ¶[ucr/gp3] 
  \end{align}
  where we employed the marginal state
  \begin{align}
    ρ¬{AY'K'B} = \tr[YK']{V ρ¬{AB} V†} = ∑_{y,k}\! √{N_y} P^k ρ¬{AB} P^{k}\! √{N_y} ⨂ 
      \proj{k} ⨂ \proj{y} ¶[ucr/gp4].
  \end{align}
  The remainder follows the proof of the first statement closely. However,
  we will take advantage of the form of the marginal state in~§[ucr/gp4]
  to extract the effective $\cK$-overlap.
  Taking the tensor product with $ⅈ¬{K}$ on both sides of~§[ucr/gp3] leads to
  \begin{align}
    \rhob¬{AYY'K'B} ⨂ ⅈ¬{K} ≤ 2^{-λ}\, ⅈ¬{YK} ⨂ ρ¬{AY'K'B} ¶\,.
  \end{align}
  We conjugate this inequality with $W$ and take the partial trace over 
  $A$, $Y'$ and $K'$ to get
  \begin{align}
    \taub¬{XKB} &= \tr<b>[AX'K']{W (\rhob¬{AYY'K'B} ⨂ ⅈ¬{K}) W†} ¶\\
      &≤ 2^{-λ}\, \tr<b>[AX'K]{W (ⅈ¬{YK} ⨂ ρ¬{AY'K'B}) W†} ¶[ucr/gp5]\,.
  \end{align}

  We again evaluate the trace term on the rhs.\ of~§[ucr/gp5] using the expressions for 
  $W$ and $ρ¬{AY'K'B}$ in~§[ucr/gp0] and~§[ucr/gp4]. We get
  \begin{align}
    &\tr<b>[AX'K']{W (ⅈ¬{YK} ⨂ ρ¬{AY'K'B}) W†} ¶\\
    &\quad\!\! = ∑_{x,y,k} \proj{x} ⨂ \proj{k} ⨂ 
      \braket{yk|\tr<b>[A]{√{M_x} P^k\! √{N_y} ρ¬{AY'K'B} √{N_y} P^k\! √{M_x}}|yk} ¶\\
    &\quad\!\! = ∑_{x} \proj{x} ⨂ ∑_k \proj{k} ⨂ 
      \tr<B>[A]{∑_y √{M_x} P^k\! N_y P^k ρ¬{AB} P^k N_y P^k\! √{M_x}} ¶\\
    &\quad\!\! ≤ ⅈ¬X ⨂ ∑¬k \proj{k} ⨂ 
      \max_x \norm<B>{∑_y P^k N_y P^k M_x P^k N_y P^k}[∞] 
      \tr[A]{P^k ρ¬{AB}} ¶[ucr/gp25]\\
    &\quad\!\! = ⅈ¬X ⨂ \omegat¬{KB} ¶\,.
  \end{align}
  We used Lemma~\ref{lm:pt-norm-bound} to arrive at~§[ucr/gp25].
  The last equality defines the operator $\omegat$.
  Note that $\tr{\omegat¬{KB}} = c_{\cK}^*$; hence, we choose $ω¬{KB} = \omegat¬{KB}/c_{\cK}^* ∈ 
  \onorm{ℋ¬{KB}}$ and~§[ucr/gp5] and find a lower bound on $S_{\min}(\taub¬{XKB}\|ω¬{KB})$ in
  terms of $λ$ and $c_{\cK}^*$.
  \begin{align}
    S_{\min}(\taub¬{XKB}\|ω¬{KB}) ≥ λ + \log \frac{1}{c_{\cK}^*} ≥ \hmin*{Y|AY'K'B}[ρ] + 
      \log \frac{1}{c_{\cK}^*} - \log \frac{2}{\bar{ε}^2} ¶[ucr/gp6]\,,
  \end{align}
  where we substituted~§[ucr/gp2] for $λ$.
  We have $P(\taub¬{XKB}, ρ¬{XKB}) = P(\rhob¬{XKB}, ρ¬{XKB})\\ ≤ 2ε + \bar{ε}$.
  Hence, using the definition of the smooth min-entropy, we get
  \begin{align}
    \hmin[2ε+\bar{ε}]{X|KB}[ρ] ≥ S_{\min}(\taub¬{XKB}\|ⅈ¬X ⨂ σ¬{KB})¶\,,
  \end{align}
  which, substituted into~§[ucr/gp6], concludes the proof of the second statement
  and the theorem.
\end{proof}

%-------------------

\section{Miscellaneous Uncertainty Relations}
\label{sc:ucr/coll}

This section contains a collection of useful corollaries of Theorem~\ref{th:ucr-gen},
including the results discussed in the introduction of this chapter.

\subsection{Commuting Measurements}

A specialization of the generalized UCR which is of particular interest concerns the case when the two measurements~$\cX$ and $\cY$ both commute with $\cK$.
In this case, the marginal states of~§[ucr/pm-xk] and~§[ucr/pm-yk] when $K$ is traced out correspond to the post-measurement states when only $\cX$ and $\cY$ are measured, 
Eqs.~§[ucr/pm-x] and~§[ucr/pm-y]. Formally,
\begin{align}
  \tr[K]{ρ¬{XKB}} &= ∑_{x} \proj{x} ⨂ \tr[AC]{√{M_x} ρ¬{ABC} √{M_x}} = ρ¬{XB} \tn*{and} ¶\\
  \tr[K]{ρ¬{YKC}} &= ∑_{y} \proj{y} ⨂ \tr[AB]{√{N_y} ρ¬{ABC} √{N_y}} = ρ¬{YC} ¶\,.
\end{align}

They satisfy the following inequality.

\begin{corollary}
  \label{co:ucr-eff-nok}
  Let $ε > 0$, $ε' ≥ 0$, $ρ ∈ \osub{ℋ¬{ABC}}$ and $\cX = \{ M_x \}$, $\cY = \{ N_y \}$ 
  two POVMs on $ℋ¬A$. 
  Then, the post-measurement states~§[ucr/pm-x] and~§[ucr/pm-y] satisfy
  \begin{align}
    \hmin[2ε'+ε]{X|B}[ρ] + \hmax[ε']{Y|C}[ρ] ≥ \log \frac{1}{c^*(ρ¬A, \cX, \cY)} 
      -\log \frac{2}{ε^2} ¶\,,
  \end{align}
\end{corollary}

\begin{proof}
  Let $\cK$ be the measurement\,|\,commuting with $\cX$ and $\cY$\,|\,that minimizes the effective 
  overlap, $c^*(ρ¬A, \cX, \cY)$, in Definition~\ref{df:overlap-eff}. The corollary now follows 
  from Theorem~\ref{th:ucr-gen} applied to $\cK$ and the data processing inequality of the smooth 
  min- and max-entropy (cf.~Theorem~\ref{th:data-proc}) applied for the partial trace over $K$. 
\end{proof}

\subsection{Modeling the Measurement Basis Explicitly}
\label{sc:ucr/basis}

The full power of Theorem~\ref{th:ucr-gen} comes to bear when we consider the following scenario.
Let $Θ$ be a classical register storing the choice of measurement that will be performed on the system $A$. For this purpose, we decompose $ℋ¬{A'} \iso ℋ¬{Θ} ⨂ ℋ¬{A}$ and consider a family of POVMs $\{ \cZ^{θ} \}$ on $ℋ¬{A}$ that share the same output alphabet. For all $θ$, let $\cZ^{θ} = \{ L_z^θ \}$ be the POVM on $ℋ¬{A}$ that is performed if $θ$ is measured on $Θ$. This process can be equivalently modeled as a POVM $\cZ$ on $ℋ¬{Θ A}$, i.e.\
\begin{align}
  \cZ := \Big\{ ∑_{θ} \proj{θ} ⨂ L_z^θ \Big\}_z ¶[ucr/povm-z]\,.
\end{align}
If $ρ¬{Θ ABC} = ∑_θ \proj{θ}[Θ] ⨂ ρ¬{ABC}^θ$ is an arbitrary state classical on $Θ$,
the post-measurement state is
\begin{align}
  ρ¬{ΘZBC} = ∑_{θ,z} \proj{θ} ⨂ \proj{z} ⨂ \tr<B>[A]{√{L_z^θ}\, ρ¬{ABC}^{θ} √{L_z^θ}} 
    ¶[ucr/pm-theta]\,.
\end{align}

This leads to the following result, which, in contrast to the uncertainty relations discussed above, relates the smooth min- and max-entropies of the same post-measurement state.

\begin{corollary}
  \label{co:ucr-basis}
  Let $ε ≥ 0$ and $\rho¬{Θ ABC} = ∑_θ \proj{θ}[Θ] ⨂ ρ¬{ABC}^θ ∈ \osub{ℋ¬{Θ ABC}}$ be 
  classical on $Θ$ and let $f: Θ → Θ$ be a 
  bijective function such that $ρ¬{ABC}^{θ} = ρ¬{ABC}^{f(θ)}$.
  Then, the post-measurement state~§[ucr/pm-theta] after measuring the 
  POVM~§[ucr/povm-z] satisfies
  \begin{align}
    \hmin*{Z|Θ B}[ρ] + \hmax*{Z|Θ C}[ρ] ≥ \log \frac{1}{c_f} ¶\,,
  \end{align}
  where $c_f = \max_{θ} c( \cZ^θ, \cZ^{f(θ)} )$.
\end{corollary}

\begin{proof}
  We consider the POVMs $\cX = \cZ$ and $\cY = \big\{ ∑_θ \proj{θ} ⨂ L_z^{f(θ)} \big\}_z$
  as well as the projective measurement $\cK = \{ \proj{θ} \}$, $k = θ$. It is easy to verify 
  that $\cK$ commutes with $\cX$ and $\cY$ and that the $\cK$-overlap of Theorem~\ref{th:ucr-gen} 
  evaluates to $c_{\cK} = c_f$. Thus,
  \begin{align}
    \hmin*{X|Θ B}[ρ] + \hmax*{Y| Θ C}[ρ] ≥ \log \frac{1}{c_f} ¶\,.
  \end{align}
  Now, the corollary follows from the observation that the post-measurement states are equivalent 
  up to local isometries, $ρ¬{Θ XB} = ρ¬{Θ ZB}$ and $ρ¬{Θ YC} = F† ρ¬{Θ ZC} F$, where
  $F = ∑_θ \ket{f(θ)}\!\bra{θ}$ and we used the natural isometries $\ket{i} ↦ \ket{i}$ 
  between the Hilbert spaces $ℋ¬Z \iso ℋ¬X \iso ℋ¬Y$. Clearly, the smooth entropies are invariant under local 
  isometries (cf.\ Proposition~\ref{pr:smooth-iso}).
\end{proof}

Note that a similar result can be derived based on the effective overlap formulation of the uncertainty relation, resulting in an overlap of
$c_f^* = \sum_θ \tr{ρ^θ}\, c( \cZ^θ, \cZ^{f(θ)} )$.

Finally, the uncertainty relation~§[ucr/alt-nonsmooth] of Section~\ref{sc:ucr/model} 
is a special case of this corollary where 
$ρ¬{Θ ABC} = \frac{1}{2} ⅈ¬{Θ} ⨂ ρ¬{ABC}$ and $f$ is the bit flip.

\subsection{The von Neumann Limit}
\label{sc:ucr/aep}

Using the asymptotic equipartition property, we directly get an uncertainty relation for the 
von Neumann entropy as well.

\begin{corollary}
  \label{co:ucr-eff-vn}
  Let $ρ ∈ \onorm{ℋ¬{ABC}}$ and $\cX = \{ M_x \}$, $\cY = \{ N_y \}$ 
  two POVMs on $ℋ¬A$ and $\cK = \{ P^k \}$ a projective measurement on $ℋ¬A$. 
  Then, the post-measurement states~§[ucr/pm-xk] and~§[ucr/pm-yk] satisfy
  \begin{align}
    \hvn{X|KB}[ρ] + \hvn{Y|KC}[ρ] ≥ \log \frac{1}{c_{\cK}^*} ¶\,,
  \end{align}
  where $c_{\cK}$ is defined in Theorem~\ref{th:ucr-gen}.
  Furthermore, it holds that
  \begin{align}
    \hvn{X|B}[ρ] + \hvn{Y|C}[ρ] ≥ \log \frac{1}{c^*(ρ¬A, \cX, \cY)} ¶\,.
  \end{align}
\end{corollary}

\begin{proof}
  We prove the first statement, from which the second follows by the same considerations 
  that led to Corollary~\ref{co:ucr-eff-nok}.

  Consider an $n$-fold tensor product Hilbert space $ℋ¬{A^n B^n C^n} \iso ℋ¬{ABC}^{⨂ n}$ for
  an arbitrary $n ∈ ℕ$. We consider i.i.d.\ product states of $ρ¬{ABC}$ on this space, 
  i.e.\ the states $ρ¬{A^n B^n C^n} = ρ¬{ABC}^{⨂ n}$.
  Furthermore, we define i.i.d.\ product measurements $\cX^n = \cX^{⨂ n} := 
  \{ \bigotimes_i M_{x_i} \}_{x^n}$, where 
  $x^n = ( x_1, x_2, \dots, x_n )$ is a string of $n$ measurement outcomes.
  Similarly, $\cY^n = \cY^{⨂ n}$ and $\cK^n = \cK^{⨂ n}$. Clearly, 
  the post-measurement states of the $n$-fold measurement setup also have
  i.i.d.\ product form, i.e.~$ρ¬{X^n K^n B^n} = ρ¬{XKB}^{⨂ n}$ and $ρ¬{Y^n K^n C^n} = 
  ρ¬{YKC}^{⨂ n}$.
  
  The effective overlap of the $n$-fold setup can now be calculated as
  \begin{align}
    c_{\cK^n}^* &= ∑_{k^n} \tr<B>{\bigotimes_i P^{k_i} ρ} \max_{x^n} \norm<g>{\sum_{y^n} 
      \bigotimes_i P^{k^i} N_{y_i} P^{k_i} M_{x_i} P^{k_i} N_{y_i} P^{k_i} }[∞] ¶\\
      &= ∑_{k^n} \prod_i \tr{P^{k^i} ρ} \max_{x_i} \norm<B>{\sum_{y_i} 
      P^{k^i} N_{y_i} P^{k_i} M_{x_i} P^{k_i} N_{y_i} P^{k_i}}[∞] ¶\\
      &= \prod_i ∑_{k_i} \tr{P^{k^i} ρ} \max_{x_i} \norm<B>{\sum_{y_i} 
      P^{k^i} N_{y_i} P^{k_i} M_{x_i} P^{k_i} N_{y_i} P^{k_i}}[∞] ¶\\
      &= (c_{\cK}^*)^n ¶[ucr/n-fold-overlap]
  \end{align}
  
  The uncertainty relation for 
  smooth min- and max-entropies, Theorem~\ref{th:ucr-gen}, now states that
  \begin{align}
    \hmin[ε+2ε']{X^n|K^n B^n}[ρ] + \hmax[ε']{Z^n|K^n B^n}[ρ] ≥ \log \frac{1}{c_{\cK^n}} 
      - \log \frac{2}{ε^2} ¶\,.
  \end{align}
  for any $ε, ε' > 0$ such that $ε + 2ε' < 1$. We divide this inequality by $n$ and use 
  expression~§[ucr/n-fold-overlap] for the effective overlap above to get
  \begin{align}
    \frac{1}{n} \hmin[ε+2ε']{X^n|K^n B^n}[ρ] + \frac{1}{n} \hmax[ε']{Z^n|K^n B^n}[ρ] ≥ 
      \log \frac{1}{c_{\cK}} - \frac{1}{n} \log \frac{2}{ε^2} ¶\,.
  \end{align}
  Taking the limit $n → ∞$ using the asymptotic equipartition property 
  (Corollaries~\ref{co:aep-direct} and~\ref{co:aep-converse}), 
  we find the uncertainty relation for von Neumann entropies.
\end{proof}

\subsection{The Quantum Key Distribution Setup}

A specific example of the setup of Section~\ref{sc:ucr/basis} is very relevant in the 
application of the uncertainty relation to quantum key distribution. (See, for example, Section~\ref{se:qkd}.)

We use $Θ^n ∈ \{ 0, 1 \}^n$ as a basis choice for a measurement on an $n$-partite system $A^n$. The measurement we consider is a product measurement, where, on 
the $i$-th part of the system, the measurement $\cZ^{θ_i} = \{ L_z^{θ_i} \}_z$ is executed depending on the corresponding bit $θ_i$ in $θ^n$. The full POVM corresponding to this measurement process is
\begin{align}
  \cZ = \Big\{ ∑_{θ^n} \proj{θ^n} ⨂ \bigotimes_{i=1}^n L_{z^n}^{θ_i} \Big\}_{z^n} ¶[ucr/m-qkd]\,.
\end{align}
The result is an $n$-bit string $z^n$, which is stored in a classical register $Z^n$.
Hence, the post-measurement state\,|\,when $Θ^n$ is uniform and independent of $ρ$\,|\,is 
given as
\begin{align}
  ρ¬{Θ^n Z^n BC} = ∑_{θ^n} \frac{1}{2^n} \proj{θ^n} ⨂ \tr<B>[A^n]{ 
    \bigotimes_i √{L_{z_i}^{θ_i}} \, ρ¬{A^n BC}  \bigotimes_i √{L_{z_i}^{θ_i}}} ¶[ucr/pm-qkd]\,.
\end{align}

\begin{corollary}
  \label{co:ucr/qkd}
  Let $ε ≥ 0$, $\bar{ε} > 0$ and $ρ ∈ \osub{ℋ¬{A}^{⨂ n} ⨂ ℋ¬{BE}}$. Moreover, 
  let $\cZ^0$ and $\cZ^1$ be two 
  POVMs on $ℋ¬A$. Then, the post-measurement 
  state~§[ucr/pm-qkd] that is produced by measuring $\cZ$ in~§[ucr/m-qkd] in a uniform 
  and independent basis, $Θ^n$, satisfies
  \begin{align}
    \hmin*{Z^n|Θ^n B}[ρ] + \hmax*{Z^n|Θ^n C}[ρ] ≥ n \log \frac{1}{c(\cZ^0, \cZ^1)} ¶. 
  \end{align}
  Furthermore, if the marginal $ρ¬{A^n} = ρ¬A^{⨂ n}$ is of i.i.d.\ product form, we have
  \begin{align}
    \hmin[2ε+\bar{ε}]{Z^n|Θ^n B}[ρ] + \hmax*{Z^n|Θ^n C}[ρ] ≥ 
      n \log \frac{1}{c^*(ρ¬A, \cZ^0, \cZ^1)} - \log \frac{2}{\bar{ε}^2} ¶.
  \end{align}
\end{corollary}

\begin{proof}
  To prove the first statement, we start with Corollary~\ref{co:ucr-basis} and the 
  measurements described above. The function $f$ is in this case the bit-flip on the whole 
  string $Θ$ and the condition $ρ¬{A^n BC}^θ = ρ¬{A^n BC}^{f(θ)}$ is trivially satisfied since
  $Θ$ is independent of $ρ$ by assumption. Finally, it is easy to verify that the overlap $c_f$
  evaluates to $c(\cZ^0, \cZ^1)^n$.

  The second statement requires an additional, projective measurement $\cK$ on $ℋ¬A$ that
  commutes with both $\cZ^0$ and $\cZ^1$. Let $\cK$ be the measurement that minimizes the
  effective overlap $c^*(ρ¬A, \cZ^0, \cZ^1)$ in Definition~\ref{df:overlap-eff}. We measure
  $\cK$ on all $n$ subsystems together with $Θ^n$ resulting in an additional classical register
  $K^n$. Using Theorem~\ref{th:ucr-gen}, this leads to the relation
  \begin{align}
    & \hmin[2ε+\bar{ε}]{Z^n|Θ^n K^n B}[ρ] + \hmax*{Z^n|Θ^n K^n C}[ρ] ¶\\
    & \qquad ≥ n \log \frac{1}{c_{\cK}^*(ρ¬A, \cZ^0, \cZ^1)} - \log \frac{2}{\bar{ε}^2} ¶,
  \end{align}
  where we used the arguments in the proof of Corollary~\ref{co:ucr-eff-vn} to simplify 
  the expression for the effective overlap. The statement then follows after a partial trace over
  $K^n$ on both entropies.
\end{proof}

Note that the above results can be extended to the case where measurements and the marginal state are of general product but not i.i.d.\ form. In this case, the logarithm of the overlap on the rhs.\ is replaced by an average over the logarithmic overlap on all subsystems.

\section{Bipartite Uncertainty Relations}
\label{sc:ucr/bi}

We have argued that the natural generalizations of uncertainty relation without side information to uncertainty relations with quantum side information introduces two distinct observers. Here, applying the chain rules for the smooth entropies introduced in Section~\ref{sc:chain}, we show
a bipartite uncertainty relation. Note that such an uncertainty relation necessarily needs
to have a term on the rhs.\ that characterizes the entanglement between the observer and
the system prior to measurement. If entanglement is present, the bound has
to be reduced accordingly. In particular, if two projective measurements are executed on a system that is fully entangled with the observers, there is no uncertainty on the measurement outcome.\footnote{See, for example~\cite{berta08}, where this is discussed for an uncertainty relation in the von Neumann limit and for rank-$1$ projective measurements.} Hence, we see a trade-off between entanglement and uncertainty in bipartite uncertainty relations\,|\,the more entanglement there is prior to measurement, the less uncertainty is created.

Here, we show that this picture is still not complete. In order to get a tight result for arbitrary projective measurements and POVMs, we also need to consider the entanglement that is
left after the measurement. This is due to the fact that general measurements (in contrast to rank-$1$ projective measurements) do not necessarily destroy all entanglement between the two parties.

This result might have applications in two party quantum cryptography, where uncertainty relations
are used to bound the knowledge of one party about the measurement results of the other party. An example of such an application of bipartite UCRs can be found in~\cite{berta11}.

We consider a bipartite system shared between $A$ and $B$, where, as before, the system $A$ is measured using one of two measurements $\cX$ or $\cY$. The entanglement before measurement is characterized by the smooth min-entropy of $A$ given $B$ while the entanglement after measurement is given by the smooth max-entropy of the system after measurement, denoted $A'$  given 
$B$ and the measurement outcome. More precisely, we start with an arbitrary state $ρ¬{AB}$ and 
consider post-measurement states 
\begin{align}
  ρ¬{A'XB} &= ∑_{x} \proj{x} ⨂ √{M_x}\, ρ¬{AB} √{M_x}
    \tn*{and} ¶[ucr/pm-bix]\\
  ρ¬{A'YB} &= ∑_{y} \proj{y} ⨂ √{N_y}\, ρ¬{AB} √{N_y} 
    ¶[ucr/pm-biy]\,.
\end{align}

Note that, in contrast to the entropic uncertainty relations discussed earlier, the system $A'$
will appear explicitly in our statements. Since the state of this system is not unique for a given 
POVM, we simply chose the simplest measurement TP-CPM consistent with the POVM, which leads to Eqs.~§[ucr/pm-bix] and~§[ucr/pm-biy]. 
This implies the following uncertainty relation.\footnote{Other uncertainty relations of the types discussed above can be made bipartite too. Here, we are interested in a simple example.}

\begin{theorem}
  Let $ε ≥ 0$, $\bar{ε} > 0$, $\tilde{ε} ≥ 0$, $\tilde{ε}' ≥ 0$ and 
  $ρ¬{AB} ∈ \osub{ℋ¬{AB}}$.
  Moreover, let $\cX$ and $\cY$ be two 
  POVMs on $ℋ¬A$. Then, the post-measurement states~§[ucr/pm-bix] and~§[ucr/pm-biy]
  satisfy
  \begin{align}
    &\hmin[\hat{ε}]{X|B}[ρ] + \hmax[ε]{Y|B}[ρ] ≥ 
      \hmin[\tilde{ε}]{A|B}[ρ] - \hmax[\tilde{ε}']{A'|YB}[ρ] ¶\\
      & \qquad \qquad + \log \frac{1}{c^*} - 4 \log \frac{2}{\bar{ε}^2} ¶\,,
  \end{align}
  where $\hat{ε} = 7\bar{ε} + 6\tilde{ε} + 4\tilde{ε}' + 8ε$ and $c^* = c^*(ρ¬A, \cX, \cY)$. 
\end{theorem}

\begin{proof}
  We start from Corollary~\ref{co:ucr-eff-nok}, which implies
  \begin{align}
    \hmin[\hat{ε}]{X|B}[ρ] ≥ 
      \hmin[3\bar{ε} + 3\tilde{ε} + 2\tilde{ε}' + 4ε]{Y|A'Y'B}[ρ] + \log \frac{1}{c^*}       
      -\log \frac{2}{\bar{ε}^2} ¶\,.
  \end{align}
  Now, we apply the chain rule~§[cr/min-a] to the min-entropy on the rhs.\ with
  the substitutions $ε \leftarrow \bar{ε}$, $ε' \leftarrow \tilde{ε}$ 
  and $ε'' \leftarrow \bar{ε} + \tilde{ε}' + 2 ε$ as well as the systems $A \leftarrow Y$,
  $B \leftarrow A'Y'$ and $C \leftarrow B$.
  This leads to
  \begin{align}
    \hmin[\hat{ε}]{X|B}[ρ] ≥ \hmin[\tilde{ε}]{A'YY'|B}[ρ] - 
      \hmax[\bar{ε} + \tilde{ε}' + 2 ε]{A'Y|B}[ρ] 
      + \log \frac{1}{c^*} - 3 \log \frac{2}{\bar{ε}^2} ¶\,,
  \end{align}
  where we used that $Y$ and $Y'$ are interchangeable.
  Then, we apply the chain rule~§[cr/max-ab] to the max-entropy with the substitutions 
  $ε \leftarrow \bar{ε}$, $ε' \leftarrow \tilde{ε}'$ and $ε'' \leftarrow ε$ as well
  as the systems $A \leftarrow Y'$, $B \leftarrow A'$ and $C \leftarrow B$.
  This directly leads to the statement of the theorem when we note that
  \begin{align}
    \hmin[\tilde{ε}]{A'YY'|B}[ρ] = \hmin[\tilde{ε}]{A|B}[ρ]
  \end{align}
  due to the invariance of the smooth entropies under isometries (cf.\ Proposition~\ref{pr:smooth-iso}).
\end{proof}

We are also interested in the von Neumann limit of this uncertainty relation. 
Using the AEP and the same techniques as in the proof of 
Corollary~\ref{co:ucr-eff-vn}, we find
\begin{align}
  \hvn{X|B}[ρ] + \hvn{Y|B}[ρ] ≥ \log \frac{1}{c^*} + \hvn{A|B}[ρ] - \hvn{A'|YB}[ρ] ¶\,.
\end{align}
Due to the symmetry of the lhs.\ of this expression, we can replace the $Y$ on the 
rhs.\ by an $X$; thus, the lower bound is effectively a function of the minimum of 
these two entropies. In the following, we prefer to model the basis explicitly and use
that $\hvn{Z|Θ B} = \frac{1}{2} \hvn{X|B} + \frac{1}{2} \hvn{Y|B}$ 
if the basis choice is uniform. 
We get\footnote{For the following arguments, we replaced the effective overlap by
the overlap.}
\begin{align}
  \hvn{Z|Θ B}[ρ] ≥ \frac{1}{2} \log \frac{1}{c} + \frac{1}{2} 
  \Big( \hvn{A|B}[ρ] - \hvn{A'|Z Θ B}[ρ] \Big) ¶.
\end{align}

Note also that if the measurement is projective and rank-$1$, the system $A'$ will
simply contain a copy of $Z$ and the second entropy thus vanishes. This leads to
the uncertainty relation of~\cite{berta10},
\begin{align}
  \hvn{Z|Θ B}[ρ] ≥ \frac{1}{2} \Big( \log \frac{1}{c} + \hvn{A|B}[ρ] \Big) ¶.
\end{align}

\subsection{Chained Uncertainty Relations}

Consider two consecutive applications of the uncertainty relation\,|\,first measure 
$Z_1$ of a system $A_1$ in the basis $Θ_1$ and then $Z_2$ on $A_2$ in the basis $Θ_2$. 
We can now derive two different
bounds on the total entropy produced by these operations. 
First, we consider the two measurements
together to get
\begin{align}
  \hvn{Z | Θ B }[ρ] ≥ \log \frac{1}{c} + \frac{1}{2} 
    \Big( \hvn{A|B}[ρ] - \hvn{A'|Z Θ B}[ρ] \Big) ¶[ucr-bi-bound-1],
\end{align}
where $Θ = Θ_1 Θ_2$, $Z = Z_1 Z_2$, $A = A_1 A_2$ and $A' = A_1' A_2'$ contains the systems $A_1$ and $A_2$ 
after measurement. On the other hand, we may write down the uncertainty of the first measurement separately.
\begin{align}
  \hvn{Z_1 | Θ_1 B}[ρ] ≥ \frac{1}{2} \log \frac{1}{c} + 
    \frac{1}{2} \Big( \hvn{A|B}[ρ] - \hvn{A_1' A_2 |Z_1 Θ_1 B}[ρ] \Big) ¶.
\end{align}
Now, the second system is measured. Namely, we consider a measurement of $Z_2$ in the basis $Θ_2$ on the second part of the joint system $A_1' A_2$. The observers hold the information gained in the first measurement, i.e.\ $Θ_1$ and $Z_1$ in addition to $B$.
\begin{align}
  \hvn{Z_2 | Z_1 Θ B}[ρ] ≥ \frac{1}{2} \log \frac{1}{c} +
    \frac{1}{2} \Big( \hvn{A_1' A_2| Z_1 Θ_1 B}[ρ] - \hvn{A'|Z Θ B}[ρ] \Big) ¶.
\end{align}

Adding the previous two inequalities leads to
\begin{align}
  \hvn{Z_1 | Θ_1 B}[ρ] + \hvn{Z_2 | Z_1 Θ B}[ρ] 
  %&≥ \hvn{Z_1 | Θ B}[ρ] + \hvn{Z_2 | Z_1 Θ B}[ρ] =  \hvn{Z | Θ B}[ρ] ¶\\
  &≥ \log \frac{1}{c} + \frac{1}{2} 
    \Big( \hvn{A|B}[ρ] - \hvn{A'|Z Θ B}[ρ] \Big) ¶[ucr-bi-bound-2].
\end{align}
The two bounds, Eqs.~§[ucr-bi-bound-1] and~§[ucr-bi-bound-2], are equivalent since $Z_1$ is independent of $Θ_2$ and, thus, $\hvn{Z_1|Θ_1 B} = \hvn{Z_1|Θ B}$.

This shows that we can split the process of producing uncertainty into two individual steps without loosening the bound on the uncertainty. In this sense, the bipartite uncertainty relation
can be considered tight.

%............................

\chapter{Applications}
\label{ch:app}

This chapter discusses three example applications of the smooth entropy framework and uses
results from Chapters~\ref{ch:pd}-\ref{ch:uc}. 

In Section~\ref{se:strong}, 
we consider ☼{source compression}, which
has been used as an example in the introductory remarks, and show how the characterization
of this task in the ☼{one-shot} setting allows us to retrieve direct and converse bounds
on coding for finite block-lengths as well as strong converse statements. 
In Section~\ref{se:ext}, we consider ☼{randomness extraction}, a task that is very important in
cryptography and is naturally considered in the one-shot setting.
Then, in Section~\ref{se:qkd}, we combine the above results and show how, in conjunction with
an uncertainty relation, they allow us to prove security of a ☼[quantum key distribution]{QKD}
protocol.

\section{Full Characterization of Source Compression}

\label{se:strong}

This section provides an example of how the ☼{one-shot} characterization of 
an information theoretic task is sufficient to derive bounds on the resource
usage for finite block lengths and in the ☼{i.i.d.\ limit}. 

We consider source compression with quantum ☼{side information},
or, equivalently, information reconciliation.
In particular, we give a ☼[strong converse]{converse bound} to information
reconciliation. We employ many results of this thesis, including
properties of the ☼{purified distance} from Chapter~\ref{ch:pd}, the ☼{data processing}
inequalities of Chapter~\ref{ch:entropies} and the ☼[asymptotic equipartition property]{AEP} of Chapter~\ref{ch:aep}.

The bounds derived here clearly also hold in case the side information is classical or 
non-existent. Thus, we provide bounds for classical source compression tasks as well. 
From this viewpoint, the following sections prove Shannon's source coding 
theorem~\cite{shannon48}, as well as its extension to side information (Slepian-Wolf~\cite{slepian73}) and quantum side information~\cite{devetak03}.\footnote{The Slepian-Wolf setting with two separate encoders can be viewed as source compression followed by source compression with side information.}
We also prove a strong
☼[converse]{converse bound} of source compression with side information. Related to this work, a strong converse for compression of quantum information~\cite{schumacher95} was shown by Winter~\cite{winterthesis}. (See also~\cite{oohama94}, where a the error exponent of the Slepian-Wolf strong converse is investigated.)

\subsection{One-Shot Characterization}

The one-shot results are adapted from a recent paper by Renes and Renner~\cite{renesrenner10}.
Given a classical-quantum state $ρ¬{ZB}$, shared between two parties, Alice and Bob, how much information needs to be transmitted from Alice to Bob such that Bob can reconstruct $Z$ with
probability of error at most $ε$? 
To investigate this question, we consider non-interactive (one-way) information reconciliation protocols from Alice to Bob. These consist of an 
encoding function, $e: \cZ → \cM$, that creates a message $M$ to be sent from Alice to Bob. Bob then uses a decoder, a POVM $\cD = \{ D_{z'} \}$ acting on the joint state of $B$ and $M$ that produces a 
classical estimate of $Z$ in the register $Z'$ ($ℋ¬Z \iso ℋ¬{Z'}$). A protocol is thus characterized by the tuple $\cP = \{e, \cD \}$. Note that this is the most general model for a non-interactive information reconciliation protocol. In particular, since we consider a fixed state $ρ¬{ZB}$, it is sufficient to consider deterministic encoding strategies.

The initial state $ρ¬{ZB}$ is of the form
\begin{align}
  ρ¬{ZB} = ∑_z P¬Z(z)\, \proj{z}[Z] ⨂ ρ¬B^z ,¶[strong/cq-state]
\end{align}
where $P¬Z(z)$ is the distribution over the input alphabet $Z$ and $ρ¬B^z ∈ \onorm{ℋ¬B}$ are
quantum states on $B$. Applying a protocol $\{e, \cD \}$
as described above to this state results in the final state
\begin{align}
  ρ¬{ZZ'} = ∑_{z,z'} P¬Z(z) \tr<B>{D_{z'} \big(ρ¬B^z ⨂ \proj{e(z)}[M] \big) } \proj{z}[Z] 
    ⨂ \proj{z'}[Z'] ¶,
\end{align}
or, equivalently, the joint probability distribution
$P¬{ZZ'}(z, z') = P¬Z(z) \cdot \tr<b>{D_{z'} \big(ρ¬B^z ⨂ \proj{e(z)}\big) }$.
The error probability of this protocol is
\begin{align}
  p_{\tn-{err}}(\cP, ρ¬{ZB}) := 1 - ∑_{z} P(z) \tr<B>{D_{z} \big(ρ¬B^z ⨂ \proj{e(z)}[M]\big) } ¶.
\end{align}
The error probability is thus 
equivalent to the trace distance between $ρ¬{ZZ'}$ and the state $χ¬{ZZ'} = ∑_z P(z) \proj{z}[Z] ⨂ \proj{z}[Z']$ describing perfect correlation between $Z$ and $Z'$.

We characterize information reconciliation of a state $ρ¬{ZB}$ from $Z$ to $B$ with the 
minimum message length
(in bits), $m^{ε}$, required to achieve an error probability of at most $ε$.
\begin{align}
  m^{ε}(Z|B)¬{ρ} := \min \{ m : \exists\, \cP \tn{s.t.} \!\log{d¬M} = m \tn{and} 
    p_{\tn-{err}}(\cP, ρ¬{ZB}) ≤ ε \} ¶.
\end{align}

We now slightly extend a result from~\cite{renesrenner10}. The converse bound found there is not strong enough for our purposes when $ε$ gets close to $1$. However, this can be fixed easily.
\begin{theorem}
  \label{th:data-reco}
  Let $ρ¬{ZB} ∈ \osub{ℋ¬{ZB}}$ be classical on $Z$ and let $0 < ε < 1$. Then,
  \begin{align}
    \hmax[√{2ε - ε^2}]{Z|B}[ρ] ≤ m^{ε}(Z|B)¬{ρ} ≤ \hmax[ε_1]{Z|B}[ρ] + 
      2 \log \frac{1}{ε_2} + 4 ¶,
  \end{align}
  for any $ε_1, ε_2$ s.t.\ $ε_1 + ε_2 = ε$.
\end{theorem}

\begin{proof}
  The direct bound, $m^{ε}(Z|B)¬{ρ} ≤ \hmax[ε_1]{Z|B}[ρ] + 2 \log 
  \frac{1}{ε_2} + 4$, is 
  shown in~\cite{renesrenner10}.
  To get a converse bound, we also follow their argument and 
  note that $p_{\tn-{err}}(\cP, ρ¬{ZB}) ≤ ε$ is equivalent to the 
  condition $D(ρ¬{ZZ'}, χ¬{ZZ'}) ≤ ε$. This implies $P(ρ¬{ZZ'}, χ¬{ZZ'}) ≤ √{2ε - ε^2} = ε'$ 
  according to
  Proposition~\ref{pr:pd-gtd-bounds}. Hence,
  \begin{align}
    \hmax[ε']{Z|Z'}[ρ] ≤ \hmax{Z|Z'}[χ] = 0
  \end{align}  
  by definition of the smooth max-entropy.
  The data-processing inequality (Theorem~\ref{th:data-proc}) then states that
  $\hmax[ε']{Z|MB}[ρ] ≤ 0$ before measurement. And finally we use the fact that conditioning
  on classical information $M$ can at most reduce the smooth max-entropy by $\log d¬M$ (cf.\
  Proposition~\ref{pr:class/bounds-2}). This leads to the following inequality:
  \begin{align}
    0 ≥ \hmax[ε']{Z|MB}[ρ] ≥ \hmax[ε']{Z|B}[ρ] - \log d¬M ¶.
  \end{align}
  Hence, all protocols with $p_{\tn-{err}}(\cP, ρ¬{ZB}) ≤ ε$ have to satisfy the constraint
  $\log d¬M ≥ \hmax[ε']{Z|B}[ρ]$. This results in the improved converse bound of the theorem.
\end{proof}

\subsection{Finite Block Lengths and Asymptotic Limits}

The one-shot result in principle characterizes all data-reconciliation tasks. Comparing the one-shot analysis with asymptotic results, where $m^ε$ converges to the ☼[von Neumann 
entropy]{entropy!von Neumann}~\cite{devetak03}, we might be interested to know how fast the one-shot result in terms of the smooth max-entropy convergences to the von Neumann entropy.

For this purpose, we consider the task of encoding a block of length $n$ of ☼{i.i.d.}~random 
variables $Z$ about each of which Bob independently has quantum side information $B$. This is data
reconciliation for the state $ρ¬{Z^n B^n} = ρ¬{ZB}^{⨂ n}$ and Theorem~\ref{th:data-reco} applies.
The minimum message length (per round) for this task, 
$\bar{m}^ε(Z|B) = \frac{1}{n} m^ε(Z^n|B^n)$, 
thus satisfies
\begin{align}
  \frac{1}{n} \hmax[√{2ε - ε^2}]{Z^n|B^n}[ρ] ≤ \bar{m}^ε(Z|B)¬{ρ} 
    ≤ \frac{1}{n} \Big( \hmax[ε/2]{Z^n|B^n}[ρ] + 2 \log \frac{1}{ε} + 6 \Big) ¶[data-reco/n],
\end{align}
where we chose $ε_1 = ε_2 = \frac{ε}{2}$ for convenience of exposition. Next, we bound the smooth max-entropies using
the asymptotic equipartition property in its entropic form. 

\subsubsection{Direct Bounds}
☼*{direct bound}

We start with the
upper bound in~§[data-reco/n].
For sufficiently large $n$, the AEP states that (cf.\ Theorem~\ref{th:min-vn} and Corollary~\ref{co:aep/cond})
\begin{align}
  \bar{m}^ε(Z|B)¬{ρ} ≤ \hvn{Z|B}[ρ] + \frac{4 \log υ √{g(\frac{ε}{2})}}{√{n}} + \frac{2 \log \frac{1}{ε} + 6}{n} ¶[data-reco/direct],
\end{align}
where $g(ε) = - \log \big( 1 - √{1-ε^2} \big)$ and $υ = √{2^{-\hmin{Z|B}[ρ]}} + √{2^{\hmax{Z|B}[ρ]}} + 1$. Hence, in the asymptotic limit of large $n$, we have
\begin{align}
  \lim_{n → ∞} \bar{m}^ε(Z|B)¬{ρ} ≤ \hvn{Z|B}[ρ] ¶.
\end{align}
Namely, there exists an encoding strategy that transmits at most $\hvn{Z|B}[ρ]$ bits per round to Bob with an arbitrarily small positive probability of error ($ε > 0$). The converse question is whether we can do better than this. 

\subsubsection{Converse Bounds}
☼*{converse bound}

To see that this is not possible, even for large error probabilities, we apply the AEP to the lower bound in~§[data-reco/n]. In the following, we use $ε$ to denote the success probability, i.e.\ $1-ε$ is the error probability.
Using Proposition~\ref{pr:min-max-smooth} and Remark~\ref{rm:min-max-smooth}, we find
\begin{align}
  &\bar{m}^{1-ε}(Z|B)¬{ρ} ≥ \frac{1}{n} \hmax[√{1-ε^2}]{Z^n|B^n}[ρ] ¶\\
    &\quad≥ \frac{1}{n} \Bigg( \hmin[\tilde{ε}]{Z^n|B^n}[ρ] - 2 \log \frac{1}{1 - 
      \Big( ε\, \tilde{ε} + √{1-ε^2}√{1 - \tilde{ε}^2} \Big)^2} \Bigg)¶.
\end{align}
Here, the smoothing parameter of the min-entropy, $\tilde{ε}$, is restricted by the inequality $\arcsin(\tilde{ε}) + \arccos(ε) < \frac{π}{2}$. This holds if and only if $\tilde{ε} < ε$ and, for convenience of exposition, 
we choose $\tilde{ε} = \frac{ε}{2}$.
This ensures that the logarithmic
correction term is finite for any $0 < ε < 1$. For the fun of introducing a symmetry with Eq.~§[data-reco/direct] where there is none (at least not in the form of the explicit deviation terms), we further bound this as follows.\footnote{We hope the reader does not mind us skipping the details of this simple procedure, which we have
tested numerically and analytically.}
\begin{align}
  \bar{m}^{1-ε}(Z|B)¬{ρ} &≥ 
    \frac{1}{n} \Bigg( \hmin[ε/2]{Z^n|B^n}[ρ] - 2 \log \frac{1}{ε} - 6 \Bigg)¶,
\end{align}
Using the AEP, this time for the min-entropy, leads to
\begin{align}
  \bar{m}^{1-ε}(Z|B)¬{ρ} &≥ \hvn{Z|B}[ρ] - \frac{4 \log υ √{g(\frac{ε}{2})}}{√{n}} - \frac{2 \log \frac{1}{ε} + 6}{n} ¶[data-reco/converse].
\end{align}
Hence, in the asymptotic limit, for all $0 < ε < 1$, we find
\begin{align}
  \lim_{n → ∞} \bar{m}^{1-ε}(Z|B)¬{ρ} ≥ \hvn{Z|B}[ρ] ¶.
\end{align}
This is the converse of data-reconciliation. It states that, even for an arbitrarily small success probability ($ε > 0$), there does not exist a protocol that allows us to decode using less than $\hvn{Z|B}[ρ]$ transmitted bits per round.

\subsection{Strong Converse}
☼*{converse bound}

Finally, the non-asymptotic converse bound, Eq.~§[data-reco/converse] gives us the means to make this statement more precise. For example, we can ask how the maximal success probability, $ε_{\max}$, scales if we use a protocol
that encodes only $\hvn{Z|B}[ρ] - µ$ bits per round for large $n$. For this analysis, we neglect
the term scaling reciprocally in $n$, i.e.\ we consider the approximate bound
\begin{align}
  \bar{m}^{1-ε}(Z|B)¬{ρ} \gtrapprox \hvn{Z|B}[ρ] - \frac{4 \log υ √{g(\frac{ε}{2})}}{√{n}} ¶.
\end{align}
This implies that
\begin{align}
  ε_{\max}(µ) &\lessapprox \sup \bigg\{ ε > 0 : \frac{4 \log υ √{g(\frac{ε}{2})}}{√{n}} ≥ µ \bigg\} ¶\\
   &≤ √{2} \cdot 2^{-\frac{µ^2 n}{2 (4 \log υ)^2}} ¶[data-reco/strong].
\end{align}
In order to find~§[data-reco/strong], we used the bound $g(ε/2) ≤ \log \frac{8}{ε^2}$ (cf.\ Figure~\ref{fg:smoothing-bounds}). This establishes that the maximal success probability drops exponentially in $n$ and $µ^2$ if we transmit less than the Shannon limit, $\hvn{Z|B}[ρ]$ bits per round. The exponential drop of the success probability in $n$ is called the strong converse.

\section{Randomness Extraction}

\label{se:ext}

This section discusses ☼{randomness extraction}, the art of extracting uniform randomness 
from a biased source. More precisely, we consider a source that outputs a register
$Z$ about which there exists ☼{side information} $E$\,|\,potentially quantum\,|\,and ask how
much uniform randomness $S$, independent of $E$, can be extracted from $Z$. This primitive
is of crucial importance in many cryptographic tasks, for example in quantum 
cryptography. There, we are interested to distill a secret key from a 
raw key that is partially correlated with a quantum eavesdropper.

The quality of the extracted randomness is measured using the trace distance; we consider the distance
\begin{align}
  Δ(S|E)_ρ := \min_{σ¬E} D( ρ¬{SE} , π¬S ⨂ σ¬E ), \tn*{where} π¬S = \frac{1}{d¬S} ⅈ¬S ¶
\end{align}
and the optimization is over all $σ¬E ∈ \onorm{ℋ¬E}$. Due to the operational interpretation
of the trace distance as a ☼{distinguishing advantage}~§[dist-adv], a small $Δ$ implies that
the extracted random variable cannot be distinguished from a uniform and independent random variable with probability more than $\frac{1}{2} (1 + Δ)$. This viewpoint is at the root
of the universally composable security framework~\cite{canetti01,pfitzmann01} in cryptography, 
which was extended to the quantum setting in~\cite{unruh10} based on 
earlier work, for example in~\cite{renner05}.

We allow probabilistic protocols to extract this uniform randomness. These can be modeled by introducing an additional independent random seed, stored in a register $F$ 
and then applying a (deterministic) function $f ∈ F$ on $X$ to get $S$. 
Namely, we consider a joint state $ρ¬{ZEF}$ of the form
\begin{align}
  ρ¬{ZEF} = ∑_{z} P¬Z(z)\, \proj{z}[Z] ⨂ ρ¬{E}^z ⨂ ∑_{f} P¬F(f)\, \proj{f}[F] , 
    \quad ρ¬E^z ∈ \osub{ℋ¬E} ¶.
\end{align}
and require that $∑ P¬F(f)\, Δ(S|E)_{ρ^f}$ is small, where $ρ¬{SE}^f$ is the
state produced when $f$ is applied to the $Z$ register of $ρ¬{ZE}$.\footnote{This
is satisfied if if $Δ(S|EF)_{ρ}$ is small, i.e.\ if the extracted
randomness is independent of the seed $F$ as well as $E$~\cite{tomamichel10}. 
This is also called the strong extractor regime in classical cryptography.}
A randomness extraction protocol, $\cP$,
thus consists of a ☼[probability distribution]{probability} $P¬F(f)$ on a ☼{register} 
$F$ and a family of functions $f ∈ F$ from $Z$ to $S$.

The maximal number of bits of uniform randomness, $\ell^ε$, that can be extracted from 
a state $ρ¬{ZE}$ is defined as
\begin{align}
  ℓ^{ε}(Z|E)_{ρ} := \max \big\{ ℓ : \exists\, \cP \tn{s.t.} 
    \log d¬S = ℓ \tn{and} ∑ P¬F(f)\, Δ(S|E)_{ρ^f} ≤ ε  \big\} ¶.
\end{align}

A protocol that satisfies $∑ P¬F(f)\, Δ(S|E)_{ρ^f} ≤ ε$ is called $ε$-good.

\subsection{Direct Bound}
☼*{direct bound}

A particular class of protocols that can be used to extract uniform randomness are based on two-universal hashing (see Carter and Wegmann~\cite{carter79}). The classical Leftover Hash
Lemma~\cite{mcinnes87,impagliazzo89,zuckerman89} states the 
amount of extractable randomness is at least the min-entropy of $Z$ given $E$. 
This construction was then extended to the quantum setting in~\cite{maurer05,rennerkoenig05,renner05}. In fact, since hashing is an entirely
classical process, one may expect that the physical nature
of the side information is irrelevant and that a purely classical
treatment is sufficient. This is, however, not necessarily the case. 
For example, the output of certain extractor functions may be
partially known if side information about their input is stored in a
quantum device of a certain size, while the same output is
almost uniform conditioned on any side information stored in a
classical system of the same size.\footnote{See~\cite{gavinsky07} for a
concrete example and~\cite{koenig07} for a more general
discussion).}
For protocols based on two-universal hashing, the following ☼{direct bound} holds~\cite{tomamichel10}.
\begin{align}
  ℓ^{ε}(Z|E)_{ρ} ≥ \hmin[ε_1]{Z|E}[ρ] - 2 \log \frac{1}{ε_2}  + 1, \tn*{where} ε = ε_1 + ε_2 
  ¶[ext/direct].
\end{align}
Other protocol families that extract the min-entropy against quantum adversaries\footnote{These families are considered mainly because they need a smaller seed or can be implemented more efficiently than two-universal hashing.} are based on almost two-universal hashing~\cite{tomamichel10} or Trevisan extractors~\cite{portmann09}.

Note that the protocol families discussed above work on any state $ρ¬{ZE}$ with sufficiently high min-entropy, i.e.\ they do not take into account other properties of the state.
Next, we will see that these protocols are essentially optimal.

\subsection{Converse Bound}
☼*{converse bound}

We prove a ☼{converse bound} by contradiction. This formalizes an intuitive argument 
given in~\cite{tomamichel10}. Assume for the sake of the argument that we have 
an $ε$-good protocol that extracts $ℓ > \hmin[ε']{Z|E}[ρ]$ bits of randomness, where 
$ε' = √{2ε - ε^2}$. Then, due to Proposition~\ref{pr:func} we know that applying a function on 
$Z$ cannot increase the smooth min-entropy, thus
\begin{align}
  \forall\ f ∈ F : \hmin[ε']{S|E}[ρ^f] < ℓ \tn*{and, thus,} Δ(S|E)_{ρ^f} > ε ¶[ext/converse-1].
\end{align}
The second statement of~§[ext/converse-1] follows from the following observation. The inequality 
$\hmin[ε']{S|E}[ρ^f] < ℓ$ implies that
all states $\rhot$ with $D(\rhot¬{SE}, ρ¬{SE}^f) ≤ ε$ (cf.\ Proposition~\ref{pr:pd-gtd-bounds}) must 
necessarily satisfy $\hmin{S|E}[\rhot] < ℓ$. In particular, these close states can thus 
not be of the form $π¬S ⨂ σ¬E$, because 
such states have min-entropy $ℓ$. Hence, $Δ(S|E)_{ρ^f} > ε$.

However, §[ext/converse-1] implies that $∑ P¬F(f)\, Δ(S|E)¬{ρ^f} > ε$, contradicting our
initial assumption that the protocol was $ε$-good. This implies the following
converse bound
\begin{align}
  ℓ^{ε}(Z|E)_{ρ} ≤ \hmin[√{2ε - ε^2}]{Z|E}[ρ] ¶[ext/converse].
\end{align}

Collecting §[ext/direct] and~§[ext/converse], we have the following theorem
\begin{theorem}
  \label{th:rand-ext}
  Let $ρ¬{ZE} ∈ \osub{ℋ¬{ZE}}$ be classical on $Z$ and let $0 < ε < 1$. Then,
  \begin{align}
     \hmin[ε_1]{Z|E}[ρ] - 2 \log \frac{1}{ε_2}  + 1 ≤ ℓ^{ε}(Z|E)_{ρ} 
       ≤ \hmin[√{2ε - ε^2}]{Z|E}[ρ] ¶,
  \end{align}
  for any $ε_1, ε_2$ s.t.\ $ε_1 + ε_2 = ε$.
\end{theorem}

We have established that the extractable uniform and independent randomness is characterized by the smooth min-entropy. A similar analysis for finite block lengths and the ☼{i.i.d.\ limit} as in Section~\ref{se:strong} is thus possible. However, we omit it here since most applications consider the task of randomness extraction only
in the ☼{one-shot} setting. An example of such an application, quantum key distribution, follows in the next section.

\section{Security in Quantum Key Distribution}

\label{se:qkd}

The smooth entropy formalism was first introduced in the quantum setting by Renner~\cite{renner05} 
in order to prove information theoretic security of ☼[quantum key distribution]{QKD} (QKD) 
protocols~\cite{bb84,ekert91} against adversaries restricted only by the laws of physics.
Prior to this work, the security of QKD protocols was mostly analyzed in the
limit of long keys and using questionable security definitions~\cite{koenigrenner07}.
☼*{quantum!cryptography}

In this section, we show that the uncertainty relations for smooth entropies in Chapter~\ref{ch:uc} can be employed to provide a very concise and intuitive security proof for QKD. This is based on two recent papers~\cite{tomamichel11} and~\cite{tomamichellim11}. The latter paper
contains a thorough analysis of the statistical tests required to assert security of the final key that goes beyond the discussion here.

\subsection{The Protocol}

We consider the original BB84 quantum key distribution protocol~\cite{bb84}. Here, two legitimate parties, Alice and Bob, try to distill a shared secret key that is independent of all wiretapped information. To do this, Alice and Bob share a public quantum channel and a public, authenticated classical channel. In the entanglement-based version of the BB84 protocol~\cite{bennett92}, 
Alice now prepares $n$ fully entangled qubit pairs of the form
\begin{align}
  \ket{ψ}[A_i B_i] = \frac{1}{√{2}} \Big( \ket{0}[A_i] ⨂ \ket{0}[B_i] + 
    \ket{1}[A_i] ⨂ \ket{1}[B_i] \Big) ,
\end{align}
where $i ∈ \{ 1, \dots, n \}$. She then sends the systems $B^n = B_1 \dots B_n$ 
over the public quantum channel to
Bob. 

Next, Alice chooses a basis for each qubit, $Θ_i ∈ \{ 0, 1 \}$, uniformly at random. She communicates the basis choices, $Θ^n = Θ_1 \dots Θ_n$,
to Bob over the classical channel and then measures her
qubits $A_i$ in the basis $\cX = \{ \ket{0}, \ket{1} \}$ if $Θ_i = 0$ or 
$\cY = \{ \ket{+}, \ket{-} \}$ if $Θ_i = 1$.\footnote{The diagonal basis is given by $\ket{±} = 1/\sqrt{2} \big( \ket{0} ± \ket{1} \big)$.} The collected measurement outcomes, 
called raw key, are stored
in a string $Z^n = Z_1 \dots Z_n$. Bob does a measurement on each of his systems
depending on the value of $Θ_i$ to produce a raw key $\bar{Z}^n$. This already concludes
the quantum part of the protocol.

Next, Alice and Bob employ a one-way data reconciliation protocol 
that allows Bob to create an estimate $\hat{Z}^n$ of $Z^n$ 
using his raw key $\bar{Z}^n$. 
Finally, Alice invokes a randomness extraction protocol on her string $Z^n$ to distill
a key $S$ and broadcasts the seed, $F$, over the classical channel in order 
to allow Bob to extract an estimate $\hat{S}$ of $S$. The latter step is usually called
privacy amplification~\cite{bennett88,bennett95}.

We note that, if the eavesdropper does not interfere and Bob simply measures his part of the entangled pair using the same bases as Alice, the resulting raw keys, $Z^n$ and $\bar{Z}^n$, 
will be perfectly correlated. In this case, security of the key simply follows from the monogamy of entanglement (i.e.\ Bob is the preferred observer of Alice's system) and no information reconciliation is necessary. In the next section, we analyze this protocol in the presence of 
an eavesdropper and noise. (The effects of noise and wiretapping can in general not be distinguished; hence, we consequently assume that correlations are degraded due to wiretapping
on the quantum channel.)

\subsection{Security in the Finite Key Regime}

Under what conditions will the key extracted by the above protocol be both secret (i.e.\ uniform and independent of the eavesdropper's information) and correct (i.e.\ $S = \hat{S}$)? We will see
that quantum mechanics allows us to ascertain both secrecy and correctness from the correlations between Alice's measurement outcomes and Bob's side information about them. This is in contrast to
classical theory, where, without further assumptions, we could only hope to determine correctness from such correlations.
Quantum mechanics enables this due to the asymmetry between different quantum mechanical observers, as discussed in the introduction of this thesis. If it can be established that one observer, Bob, is a preferred observer of a quantum system Alice holds, then this ensures that an eavesdropper's knowledge about the outcomes of measurements performed on that system is limited. Here, we show
how this asymmetry can be verified using the entropic uncertainty relation of Chapter~\ref{ch:uc}.

Since the eavesdropper is allowed to interact arbitrarily with the quantum communication sent from Alice to Bob, we do not make assumptions about the form of the resulting state $ρ¬{A^n B^n E}$. In 
particular, the systems $B_i$ that arrive at Bob's lab do not need to be qubits and $E$ is an 
arbitrarily large system, held by the eavesdropper, that may be correlated with the systems $A^n = A_1 \dots A_n$ and $B^n$.

The uncertainty relation for smooth min- and max-entropies, in particular Corollary~\ref{co:ucr/qkd}, can be applied to this setup. It states that
\begin{align}
  \hmin*{Z^n|Θ^n E} + \hmax*{Z^n|Θ^n B} ≥ n \log \frac{1}{c} ¶[qkd/ucr], 
\end{align}
where $c$ is the overlap of Alice's measurements (on one qubit). In the case of 
entanglement-based BB84 as described above, we have $c = \frac{1}{2}$ and, thus, 
the left-hand side 
of~§[qkd/ucr] evaluates to $n$. This means that the smooth min-entropy of the eavesdropper's information about $Z^n$ right after Alice's measurements is lower bounded by $n$ minus the max-entropy of Bob's correlations with $Z^n$. After Bob's measurement, using the data-processing inequality (cf.\ Theorem~\ref{th:data-proc}), we thus have
\begin{align}
  \hmin*{Z^n|Θ^n B} ≥ n - \hmax*{Z^n|\bar{Z}^n} ¶.
\end{align}

During information reconciliation, using a protocol as discussed in 
the direct part of Theorem~\ref{th:data-proc}, Alice sends
a message $M$ to Bob which satisfies
\begin{align}
  \log d¬M ≤ \hmax[ε_c/2]{Z^n|\bar{Z}^n} + 2 \log \frac{1}{ε_c} + 6 ¶,
\end{align}
where $ε_c$ is the required correctness. This ensures that the probability that Bob cannot correctly estimate Alice's string is at most $ε_c$.
Note that this message might be wiretapped by the eavesdropper and stored in a 
register $M'$. However, using Proposition~\ref{pr:class/bounds-2}, we bound
\begin{align}
  &\hmin*{Z^n|Θ^n E M} ≥ \hmin*{Z^n|Θ^n E} - \log d¬M ¶\\
    &\quad ≥ n - \hmax*{Z^n|\bar{Z}^n} - \hmax[ε_c/2]{Z^n|\bar{Z}^n} 
    - 2 \log \frac{1}{ε_c} - 6 ¶[qkd/b1].
\end{align}
The direct part of Theorem~\ref{th:rand-ext} now states 
that we can extract at least $ℓ$ bits of $ε_s$-secure key, where
\begin{align}
  ℓ ≥ \hmin[ε_s/2]{Z^n|Θ^n E M} - 2 \log \frac{1}{ε_s} - 1 ¶.
\end{align}
By the term $ε_s$-secure we mean that $Δ(S|Θ^n E M F)_{ρ} ≤ ε_s$. This implies that
the secret key is in particular independent of the seed $F$, which Alice needs to
send to Bob and which could be wiretapped by the eavesdropper.

Using~§[qkd/b1], we can express this bound entirely in terms of correlations between
Alice and Bob, i.e.\
\begin{align}
  ℓ ≥ n - \hmax[ε_s/2]{Z^n|\bar{Z}^n} - \hmax[ε_c/2]{Z^n|\bar{Z}^n}
    - 2 \log \frac{1}{ε_c} - 2 \log \frac{1}{ε_s} - 7 ¶.
\end{align}
This means that if the classical correlations between Alice's and Bob's measurement outcomes are sufficiently good\,|\,namely if the corresponding smooth max-entropies are small\,|\,we can safely extract a secret key using this protocol. 

It thus remains to find a statistical test to determine upper bounds on the smooth max-entropies. This is usually done in the following way. Alice and Bob, after measuring, compare a random subset of their raw keys using the classical channel. Then, if the frequency of errors found in this sample is smaller than an agreed threshold value, 
they will proceed with the protocol and extract a key. 
In this case, it can be shown that the smooth max-entropy is indeed small, and the key thus secure. Otherwise, they abort and do not produce a key.
(A detailed analysis of the statistical tests and the precise security statements that follow is beyond the scope of this thesis. Such an analysis can be found in~\cite{tomamichellim11}, for example.)

%\subsubsection{Devetak-Winter Bound}
%
%If we assume that the marginal state $ρ¬{A^n B^n}$ is of i.d.d.\ 
%form, i.e.\ $ρ¬{A^n B^n} = τ¬{AB}^{⨂ n}$, we can apply the AEP to this problem. (Note that this 
%assumption corresponds to the restriction to collective attacks by the eavesdropper.)

%............................

\chapter{Conclusions and Outlook}
\label{ch:outlook}

The goal of this thesis was to consolidate the smooth entropy framework for non-asymptotic information theory and to introduce important additions to the framework, including the entropic
asymptotic equipartition property and various uncertainty relations. 
My sincere hope is that this work will provide a reference for researchers interested in non-asymptotic quantum information theory and quantum cryptography. 

\section{Applications of the Framework}

The smooth entropy framework has already found a wide range of applications since its inception.

In cryptography, particularly quantum key distribution, the smooth entropy framework has become a standard tool to analyze security for finite keys~\cite{renner05,scarani08p}. This analysis
was simplified for some protocols thanks to the entropic uncertainty relation~\cite{tomamichel10,tomamichellim11}. The smooth entropy formalism allows to investigate entropically secure encryption~\cite{fehr08,dupuis07}, and, together with results from randomness extraction, we use it to show that bit commitments cannot be expanded~\cite{winkler11} in a quantum world.
Furthermore, composable security in the bounded storage~\cite{wehner07} and noisy storage~\cite{koenig09} models is analyzed using smooth entropies.

Decoupling of quantum systems~\cite{dupuis10} can be viewed as a fully quantum generalization of randomness extraction and is characterized by smooth entropies. The decoupling approach\,|\,a quantum generalization of random coding\,|\,leads to direct bounds for many information theoretic tasks. (Some of them are listed in~\cite{dupuis09} and~\cite{dupuis10}.)
We have also shown that decoupling is possible using 
approximate two-designs~\cite{szehr11}, which suggests that decoupling can be 
achieved efficiently in nature.

The smooth entropy formalism has also been used to investigate various quantum channel 
capacities and converses~\cite{datta09,datta11a,datta11b}. In particular, it leads to a conceptually simple proof of the quantum Reverse Shannon Theorem~\cite{bertachristandl11}.

In thermodynamics, smooth entropies have been used to quantify work extraction~\cite{dahlsten11} and the work cost of erasure~\cite{delrio11} in small systems. The smooth entropies have also been
used to investigate thermalization~\cite{hutter11}.

Finally, it has been shown~\cite{dattarenner08} that the smooth entropy framework entails the information spectrum method by Han and Vérdu~\cite{verdu93,han02} and its quantum generalization due to Hayashi, Nagaoka, and Ogawa~\cite{hayashi03,nagaoka07,ogawa00}.

\section{Outlook and Open Questions}

Some of the technical results in this thesis are new and applications remain unexplored. 
The entropic asymptotic equipartition property presented here has been improved from earlier 
work~\cite{tomamichel08} and now provides a converse bound for finite smoothing. This may
help in the quest to prove strong converse statements for various tasks in quantum information 
theory, including channel capacities. These arguments apply to classical theory as well
and it remains to be seen how these converse bounds compare
to the literature.

In Chapter~\ref{ch:uc}, Eq.~§[conj/c-ceff], it was conjectured that the effective overlap is always smaller than the overlap. It remains an open question to show that this is true for 
general POVMs.

The entropic uncertainty relation has been generalized from previous work~\cite{tomamichel11}.
The lower bound on the uncertainty is now expressed using an effective overlap. In turn,
this effective overlap can be bounded in terms of the maximal CHSH value that can be 
reached\,|\,using the same measurement setup\,|\,with an arbitrary second party. The CHSH value
is measure of the non-locality of classical correlations produced by two parties. We thus believe that the generalized uncertainty relation
provides a new avenue into device-independent quantum cryptography~\cite{haenggi11,lim12}.
More generally, the applications of the uncertainty relation for smooth entropies in quantum cryptography are not fully explored yet.
Another interesting extension is the bipartite uncertainty relation. It has potential 
applications in two-party quantum cryptography, which are unexplored.

Finally, the recent addition of a complete set of chain rules for the smooth entropies provides
an important missing link and will lead to an even larger range of applicability of the framework.

% bibliography
\cleardoublepage
\phantomsection
\addcontentsline{toc}{chapter}{Bibliography}

\bibliographystyle{alphaarxiv}
\bibliography{library}

% list of figures
%\cleardoublepage
%\phantomsection
%\addcontentsline{toc}{chapter}{List of Figures}
%
%\listoffigures
%
%\cleardoublepage
%\phantomsection
%\addcontentsline{toc}{chapter}{List of Tables}
%
%\listoftables

%............................ 

\appendix

\chapter{Various Lemmas}
\label{ap:lemmas}

\section{Two Lemmas for Tensor Spaces}

The following Lemma gives an operator inequality
that relates the ☼{marginal} states on one Hilbert space before and after a
☼{CPM} is applied on the other space.
\begin{lemma}
  \label{lm:pt-norm-bound}
  Let $M ∈ \opos{ℋ ⨂ ℋ'}$ and let $ℰ$ be a CPM from $ℋ'$ to $ℋ''$. Then,
  \begin{align}
    \tr<b>[ℋ'']{ℰ[M]} ≤ \norm<b>{ℰ†[ⅈ]}[∞]\, \tr[ℋ']{M} ¶\,.
  \end{align}
\end{lemma}

\begin{proof}
  We write $ℰ$ in its ☼[Kraus representation]{representation!Kraus}, 
  i.e.\ $ℰ[M] = \sum_k E_k M\, E_k†$, where
  $E_k ∈ \olin{ℋ', ℋ''}$. Due to the linearity and cyclicity of the partial trace, we
  have
  \begin{align}
    \tr<b>[ℋ'']{ℰ[M]} = \tr<B>[ℋ']{ ∑_k E_k† E_k\, M} = \tr<b>[ℋ']{ℰ†[ⅈ]\, M} ¶
  \end{align}
  We introduce the operator 
  $R = ⅈ \norm{ℰ†[ⅈ]}[∞] - ℰ†[ⅈ] ≥ 0$ and note that $\tr[ℋ']{√{R} M √{R}} ≥ 0$ and, thus,
  \begin{align}
    \tr<b>[ℋ']{ℰ†[ⅈ]\, M} &≤ \tr<b>[ℋ']{(ℰ†[ⅈ]\, + R) M} = \norm<b>{ℰ†[ⅈ]}[∞]\, \tr[ℋ']{M} ¶.
      \qedhere
  \end{align}
\end{proof}

In particular, this lemma implies that ☼[TP-CPMs]{TP-CPM} do not affect the marginal state 
on another space. To see this, note that the bound evaluates to $\tr[ℋ'']{ℰ[M]} ≤ \tr[ℋ']{M}$. Hence, the operator
$\tr[ℋ']{M} - \tr[ℋ'']{ℰ[M]}$ is positive and has vanishing trace, implying that the two marginal states are in fact equal. Moreover, if the map $ℰ$ is trace non-increasing, we
still find $\tr[ℋ'']{ℰ[M]} ≤ \tr[ℋ']{M}$ as its adjoint is sub-unital.

As a further example, consider the CPM $\sL : M ↦ L M L†$ that conjugates $M$
with an arbitrary linear operator $L ∈ \olin{ℋ', ℋ''}$. In this case, $\sL†[ⅈ]$ evaluates
to $L†L = \abs{L}^2$ and we find 
\begin{align}
  \tr<b>[ℋ'']{\sL[M]} ≤ \norm<b>{ \abs{L}^2}[∞] \tr[ℋ']{M} 
    = \norm{L}[∞]^2\, \tr[ℋ']{M} ¶[pt-norm-bound].
\end{align}
To simplify the expression, we used that $\norm{N^2}[∞] = \norm{N}[∞]^2$ for $N ∈ \opos{H}$ and that $\norm{\abs{L}} = \norm{L}$ due to the ☼[unitary invariance]{unitarily invariant} of the norm applied to the ☼[polar decomposition]{decomposition!polar} of $L$.

Moreover, the following operator inequality holds.

\begin{lemma}
  \label{lm:op-bound}
  Let $M ∈ \opos{ℋ ⨂ ℋ'}$. Then, $M ≤ \dim{ℋ'}\, \tr[ℋ']{M} ⨂ ⅈ¬{ℋ'}$.
\end{lemma}

\begin{proof}
  Since $M$ has an eigenvalue decomposition 
  with positive eigenvalues, it is sufficient
  to prove the property for an operator $M$ of the form $M = \proj{φ}$, where $φ ∈ ℋ ⨂ ℋ'$. 
  The general result then follows by linearity.
  
  Let $X = \tr[ℋ']{M}$ and $Λ = X^{-½} \proj{φ} X^{-½}$, where we take the 
  ☼[generalized inverse]{inverse!generalized} of $X$. 
  Since $\tr{Λ} = \rank{X} ≤ \dim{ℋ'}$ due 
  to the Schmidt decomposition, we have $Λ ≤ \dim{ℋ'}\, ⅈ$. The
  statement of the lemma than follows by conjugating this inequality 
  with $X^{½}$.
\end{proof}

\section{Entropies of Coherent Classical States}
\label{se:lemmas/cc}

The following two technical lemmas are useful when ☼{coherent classical} states are discussed.
\begin{lemma}
  \label{lm:smooth/class/class-entangled-min}
  Let $τ ∈ \osub{ℋ¬{XX'AB}}$. Then, the corresponding coherent classical 
  state $ρ = Π¬{XX'} τ\, Π¬{XX'}$ satisfies
  \begin{align}
    \hmin{XA|X'B}[τ] ≤ \hmin{XA|X'B}[ρ] ≤ \hmin{A|XX'B}[ρ] ¶\,.
  \end{align}
\end{lemma}
\begin{proof}
  By the definition of the min-entropy, there exists a 
  state $σ ∈ \onorm{ℋ¬{X'B}}$ such that
  \begin{align}
    τ¬{XX'AB} ≤ 2^{-\hmin{XA|X'B}[τ]} ⅈ¬{XA} ⨂ σ¬{X'B} ¶\,.
  \end{align}
  Therefore, by conjugation with $Π¬{XX'}$, we have
  \begin{align}
    ρ¬{XX'AB} &≤ 2^{-\hmin{XA|X'B}[τ]} ⅈ¬{A} ⨂ ∑_x \proj{x}[X] ⨂ 
      \proj{x}[X'] ⨂ \braket{x|σ¬{X'B}|x}[X'] ¶\\
    &≤ 2^{-\hmin{XA|X'B}[τ]} ⅈ¬{AX} ⨂ \sM¬{X'}[σ¬{X'B}] ¶\,.
  \end{align}
  Since the measurement~$\sM¬{X'}$ is trace-preserving, 
  $\sM¬{X'}[σ]$ is a candidate for the optimization of $\hmin{XA|X'B}[ρ]$ and the 
  first inequality of the lemma follows. The second inequality follows by a similar 
  argument, where $σ¬{X'B}$ is chosen to maximize $\hmin{XA|X'B}[ρ]$ and the state
  \begin{align}
    σ¬{XX'B} = ∑_x \proj{x}[X] ⨂ \proj{x}[X'] ⨂ \braket{x|σ¬{X'B}|x}[X'] ¶
  \end{align}
  is a candidate for the optimization of $\hmin{A|XX'B}$.
\end{proof}
\begin{lemma}
  \label{lm:smmoth/class/class-entangled-max}
  Let $ρ ∈ \osub{ℋ¬{XX'AB}}$ be coherent classical between $X$ and $X'$. Then,
  \begin{align}
    \hmax{XA|X'B}[ρ] ≤ \hmax{A|XX'B}[ρ] ¶\,.
  \end{align}
\end{lemma}
\begin{proof}
  Since the coherent classical state $ρ$ commutes with $Π¬{XX'}$, there exists a state $σ ∈ \onorm{ℋ¬{X'B}}$ such that
  \begin{align}
    2^{\hmax{XA|X'B}} &= \fidb{ ρ¬{XX'AB} }{ ⅈ¬{XA} ⨂ σ¬{X'B} } ¶\\
      &= \fidb{ ρ¬{XX'AB} }{ ⅈ¬A ⨂ Π¬{XX'} (ⅈ¬{X} ⨂ σ¬{X'B}) Π¬{XX'} } ¶\\
      &≤ 2^{\hmax{A|XX'B}} ¶\,.
  \end{align}
  The last inequality follows from the fact that 
  $\tr<b>{Π¬{XX'} (ⅈ¬{X} ⨂ σ¬{X'B}) Π¬{XX'}} = \tr<b>{\sM¬{X'}[σ]} = 1$, 
  which makes this state a candidate for the optimization in $\hmax{A|XX'B}$.
\end{proof}

\section{Selected Relations between Entropies}
\label{se:lemmas/rel}

The first lemma appeared in~\cite{tomamichel10} and relates the ☼[min-entropy]{entropy!min-entropy}
and the relative entropy where the conditioning is done on $ρ$.
\begin{lemma}
  \label{lm:rel-min-bound}
  Let $ε > 0$ and $ρ ∈ \osub{ℋ¬{ABC}}$ be pure. Then, there exists a 
  projector $Π¬{AC}$ on $ℋ¬{AC}$ and a state $\rhot = Π ρ Π$ 
  such that $\rhot ∈ \ball*{ρ}$ and
  \begin{align}
    S_{\min}(\rhot¬{AB} \| ⅈ¬A ⨂ ρ¬B) ≥ \hmin{A|B}[ρ] - \log \frac{2}{ε^2} ¶\,.
  \end{align}
\end{lemma}
\begin{proof}
  The proof is structured as follows: First, we give a lower bound on
  the entropy $S_{\min}(\rhot¬{AB}\| ⅈ¬A ⨂ ρ¬B)$ in terms of $\hmin{A|B}[ρ]$ 
  and a projector $Π¬B$ that is a dual projector of $Π¬{AC}$ in the sense described below. 
  We then find a lower bound on the purified distance between $ρ$ and $\rhot$ in terms of $Π¬B$ 
  and define $Π¬B$ (and, thus, $Π¬{AC}$) such that this distance does not
  exceed $ε$.

  Let $λ$ and $σ$ be the pair that optimizes the min-entropy
  $\hmin{A|B}$, i.e.\ $\hmin{A|B}[ρ] = S_{\min}(ρ¬{AB}\| ⅈ¬A ⨂ σ¬B) = λ$.
  We have $\rhot¬B ≤ ρ¬B$ by definition of $\rhot$. Hence,
  $S_{\min}(\rhot¬{AB} \| ⅈ¬A ⨂ ρ¬B)$ is finite and can be written as
  \begin{align}
    2^{-S_{\min}(\rhot¬{AB} \| ⅈ¬A ⨂ ρ¬B)} &= \norm<b>{
     ρ¬B^{-½} \rhot¬{AB} ρ¬B^{-½}}[∞] ¶[ucr/rel-min1]\,.
  \end{align}
  We bound this expression using the dual projector of $Π¬{AC}$, $Π¬B$, 
  i.e.\ the projector that satisfies (cf.\ Lemma~\ref{lm:mirror})
  \begin{align}
    Π¬B \ket{Γ} = Π¬{AC} \ket{Γ}, \tn*{where} 
      \ket{Γ} = ρ¬B^{-½} \ket{ρ} = ρ¬{AC}^{-½} \ket{ρ} = \veciso (Π^{ρ¬B}) ¶\,
  \end{align} 
  is the unnormalized fully entangled state between $\supp{ρ¬B}$ and $\supp{ρ¬{AC}}$.
  We also use that $ρ¬{AB} ≤ 2^{-λ} ⅈ¬A ⨂σ¬B$ by definition of the min-entropy.
  Thus, we bound the rhs.\ of~§[ucr/rel-min1] as
  \begin{align}
    \tn-{rhs.} &= \norm<b>{\tr<b>[C]{(Π¬{AC} ⨂ ρ¬B^{-½})\,
    ρ¬{ABC}\, (Π¬{AC} ⨂ ρ¬B^{-½})}}[∞] ¶\\
      &= \norm<b>{Π¬B\, ρ¬B^{-½} ρ¬{AB}\,ρ¬B^{-½} Π¬B}[∞] ¶\\
      &≤ 2^{-λ} \norm<b>{ⅈ¬A ⨂ (Π¬B\, ρ¬B^{-½} σ¬B\,ρ¬B^{-½} Π¬B)}[∞] ¶\\
      &= 2^{-λ} \norm{Π¬B Λ¬B Π¬B}[∞] ¶\,,
  \end{align}
  where, in the last step, we introduced the Hermitian operator
  $Λ¬B := ρ¬B^{-½} σ¬B\,ρ¬B^{-½}$. 
  Taking the logarithm on both sides leads to
  \begin{align}
    S_{\min}(\rhot¬{AB} \| ⅈ¬A ⨂ ρ¬B) ≥ \hmin{A|B}[ρ] - \log \norm{Π¬B Λ¬B Π¬B}[∞] 
      ¶[ucr/rel-min2]\,.
  \end{align}
  We use Lemma~\ref{lm:pd-proj} to bound the distance between
  $ρ¬{ABC}$ and $\rhot¬{ABC}$, namely
  \begin{align}
    P(ρ, \rhot) ≤ √{2 \tr{\Pip¬{AC} ρ¬{ABC}}} = √{2 \tr{\Pip¬B ρ¬B}} \,,
  \end{align}
  where the $\Pip¬{AC}$ and $\Pip¬B$ are the orthogonal complements of $Π¬{AC}$ and 
  $Π¬B$, respectively.
  
  Clearly, the optimal choice of $Π¬B$ will cut off the largest eigenvalues of 
  $Λ$ in~§[ucr/rel-min2] while keeping $ρ$ and $\rhot$ close. We thus define $Π¬B$ 
  to be the minimum rank projector onto the smallest eigenvalues of $Λ$ such that
  $\tr{Π¬B ρ¬B} ≥ \tr{ρ} - ε^2/2$ or, equivalently, $\tr{\Pip¬B ρ¬B} ≤ ε^2/2$. 
  This definition immediately implies that $ρ \ecl \rhot$ and it remains to find 
  an upper bound on $\norm{Π¬B Λ¬B Π¬B}[∞]$.
  
  Let $Π¬B'$ be the projector onto the largest remaining eigenvalue in the operator 
  $Π¬B Λ¬B Π¬B$ and note that $Π¬B'$ and $\Pip¬B$ commute with $Λ¬B$. Then,
  \begin{align}
    \norm{Π¬B Λ¬B Π¬B}[∞] = \tr{Π¬B' Λ¬B} = \min_{µ} \frac{\tr{µ¬B (\Pip¬B + Π¬B')
      Λ¬B}}{\tr{µ}} ¶\,,
  \end{align}
  where $µ$ is minimized over all positive operators in the support
  of $\Pip¬B + Π¬B'$. Fixing instead $µ¬B = (\Pip¬B + Π¬B') ρ¬B (\Pip¬B + Π¬B')$, 
  we find
  \begin{align}
    \norm{Π¬B Λ¬B Π¬B}[∞] ≤ \frac{\tr{Γ¬B^{½} ρ¬B Λ¬B^{½} (\Pip¬B + Π¬B')}}
      {\tr{(\Pip¬B + Π¬B') ρ¬B}} ≤ \frac{\tr{Γ¬B^{½} ρ¬B Λ¬B^{½}}}{\tr{(\Pip¬B + Π¬B') ρ¬B}} 
      ≤ \frac{2}{ε^2} ¶\,.
  \end{align}
  In the last step we used that 
  \begin{align}
    \tr{Λ¬B^{½} ρ¬B Λ¬B^{½}} = \tr{σ¬B} = 1 \tn*{and}
      \tr{(\Pip¬B + Π¬B') ρ¬B} ≥ \frac{ε^2}{2} ¶
  \end{align}
  by definition of $\Pip¬B$. This concludes the proof.
\end{proof}

In addition, we need the following extension of this result to the smooth min-entropy.

\begin{lemma}
  \label{lm:rel-smooth-bound}
  Let $ε > 0, ε' ≥ 0$ and $ρ ∈ \osub{ℋ¬{AB}}$. Then, there exists a 
  state $\rhob ∈ \ball[ε + 2ε']{ρ}$ such that
  \begin{align}
    S_{\min}(\rhob¬{AB} \| ⅈ¬A ⨂ ρ¬B) ≥ \hmin[ε']{A|B}[ρ] 
      - \log \frac{2}{ε^2} ¶\,.
  \end{align}  
\end{lemma}

\begin{proof}
  Let $ρ ∈ \osub{ℋ¬{ABC}}$ and $\rhoh ∈ \ball[ε']{ρ¬{ABC}}$ be pure states such that 
  $\hmin[ε']{A|B}[ρ] = \hmin{A|B}[\rhoh]$. We apply 
  Lemma~\ref{lm:rel-min-bound} to this state to get
  \begin{align}
    S_{\min}(\rhot¬{AB} \| ⅈ¬A ⨂ \rhoh¬B) 
      ≥ \hmin[ε']{A|B}[ρ] - \log \frac{2}{ε^2} \,, \tn*{where} 
      \ket{\rhot} = Π¬{AC} \ket{\rhoh}, \rhot \ecl \rhoh ¶\,.
  \end{align}
  Using Lemma~\ref{lm:pd-ext-constr}, we define the operator $F¬B$ with the property 
  $F¬B \rhoh¬B F¬B† = \rho¬B$; hence $F¬B \rhoh¬{ABC} F¬B† \ecl[ε'] \rhoh¬{ABC}$. Applying this 
  to the defining operator inequality of the relative entropy above leads to
  \begin{align}
    \rhot¬{AB} ≤ 2^{-λ} ⅈ¬A ⨂ \rhoh¬B \implies \rhob¬{AB} := F¬B 
      \rhot¬{AB} F¬B† ≤ 2^{-λ} ⅈ¬A ⨂ ρ¬B ¶
  \end{align}
  and, thus, $S_{\min}(\rhot¬{AB} \| ⅈ¬A ⨂ \rhoh¬B) ≤ S_{\min}(\rhob¬{AB} \| ⅈ¬A ⨂ ρ¬B)$. Furthermore, $\rhob ∈ \osub{ℋ¬{ABC}}$ since $\tr{\rhob} = \tr{F¬B \rhot¬B F¬B†} ≤ \tr{F¬B \rhoh¬B F¬B†} = \tr{ρ¬B} ≤ 1$.
  Hence, it remains to bound $P(\rhob, ρ) ≤ P(\rhob, \rhot) + P(\rhot, \rhoh) 
  + P(\rhoh, ρ) ≤ P(\rhob, \rhot) + ε + ε'$ and
  \begin{align} 
    P(\rhob, \rhot) &= P \big( (F¬B ⨂ Π¬{AC}) \,\rhoh¬{ABC} (F¬B ⨂ Π¬{AC}), Π¬{AC} 
    \,\rhoh¬{ABC} 
      Π¬{AC} \big) ¶\\
      &≤ P(F¬B \rhoh¬{ABC} F¬B†, \rhoh¬{ABC}) ≤ ε' ¶,
  \end{align}
  where we used the monotonicity of the purified distance (cf.\ Theorem~\ref{th:pd-mono}) 
  under projections. This concludes the proof.
\end{proof}

\chapter{Properties of Quasi-Entropies}
\label{ap:renyi}

This appendix discusses properties of the quasi-entropies and 
relative Rényi Entropies introduced in Chapter~\ref{ch:aep}.

\section{Properties of Quasi-Entropies}

The quasi entropies (cf.\ Def.~\ref{df:quasi}) are well-defined in the sense that they are ☼{covariant} under isometries on $A$ and $B$.

\begin{lemma}
  \label{lm:quasi-iso}
  Let $U: ℋ → ℋ'$ be an isometry. Then, for all $f$-quasi entropies and 
  for all $A, B ∈ \opos{ℋ}$, we have
  \begin{align}
    S_f(A \,|\, B) = S_f(U A U† \,\|\, U B U†) ¶\,.
  \end{align}
\end{lemma}
\begin{proof}
  Let $A = ∑_i λ_i \proj{e_i}$ and $B = ∑_j μ_j \proj{f_j}$
  with eigenvalues $λ_i ≥ 0$, $μ_j ≥ 0$.
  Now, we write (see also~\cite{ogawa00})
  \begin{align}
    S_f(A \,\|\, B) = \lim_{ξ → 0} ∑_{i,j} \big( μ_j + ξ \big) \, f \left( \frac{λ_i}{μ_j +
      ξ} \right) \abs<b>{ \braket{e_i | f_j} }^2 ¶\, .
  \end{align}
  The isometry $U$ keeps eigenvalues and the scalar product
  $\braket{e_i|f_j}$ invariant. Furthermore, any zero
  eigenvalues introduced do not contribute since they lie in a space
  orthogonal to the image of $U$, where the summands vanish since $\lim_{ξ → 0}\, 
  ξ f(0) = 0$.\footnote{In Def.~\ref{df:quasi}, we require that $f(0) ∈ ℝ$ is finite.}
\end{proof}

Furthermore, it turns out that quasi-entropies have nice properties if the function $f$ is
chosen operator ☼{concave}. In particular, the following property
holds.\footnote{This was essentially already shown in~\cite{petz86} for the partial trace (see also~\cite{hayashi06}) up to the continuity arguments above that allow us to define the quasi-entropies for non-invertible $B$. This extension (cf.~Lemma~\ref{lm:quasi-iso}) is crucial, since even if $B$ is invertible, $ℰ(B)$ is generally not if $ℰ$ is an isometry. Hence, our contribution is to extend the monotonicity argument to general TP-CPMs.} 
\begin{lemma}
  \label{lm:quasi-mono}
  Let $ℰ$ be a TP-CPM. Then, for all $f$ quasi-entropies with $f$ operator concave and for all
  $A, B ∈ \opos{ℋ}$, we have
  \begin{align}
    S_f(A \,\| B) ≤ S_f\big( ℰ(A) \,\|\, ℰ(B) \big) ¶\,.
  \end{align}
\end{lemma}
\begin{proof}
  Any TP-CPM can be expressed as an isometry followed by a partial trace (cf.~Lemma~\ref
  {lm:stinespring}), hence, in conjunction with Lemma~\ref{lm:quasi-iso}, it remains to show the 
  property for the partial trace operation. We will show this under the assumption that $B$ is 
  invertible and the result for  general $B$ will follow from the continuity (by definition) of 
  $S_f$ when $ξ → 0$.

  To show monotonicity under partial trace, we let $ℋ = ℋ_1 ⨂ ℋ_2$ with bases $\{ \ket{i}_1 \}$   
  and $\{ \ket{j}_2 \}$, respectively. We use the (unnormalized) fully entangled state $\ket{γ} = 
  \sum_{i,j} \big( \ket{i}_1 ⨂ \ket{j}_2 \big) ⨂ \big( \ket{i}_1 ⨂ \ket{j}_2 \big)$ in the 
  product basis and its marginal $\ket{γ}_1 = \sum_i \ket{i}_1 ⨂ \ket{i}_1$. 
  It remains to show that $S_f(A \,\| B) ≤ S_f(A_1 \| B_1)$, where $A_1 = \tr[2]{A}$ and 
  $B_1 = \tr[2]{B}$. 

  Let us define a linear map $V: ℋ_1 ⨂ ℋ_1' → ℋ ⨂ ℋ'$ by
  \begin{align}
    V := ∑_i \left( √{B} \left( √{B_1}\inv ⨂ \ket{i}_2 \right) \right) ⨂ ⅈ_1 ⨂ \ket{i}_2 ¶\,.
  \end{align}
  The map $V$ is an isometry, i.e.\ $V† V = ⅈ_{11}$ and satisfies 
  \begin{align}
    V \big( √{B_1} ⨂ ⅈ_1 \ket{γ}_1 \big)  = √{B} ⨂ ⅈ_{12} \ket{γ} ¶[quasi-mono/iso-action]\,.
  \end{align}
  We have $\tr[2]{Aᵀ} = A_1ᵀ$, since the transpose is taken in the product basis. 
  Hence, $V† (B\inv ⨂ Aᵀ) V = B_1\inv ⨂ A_1ᵀ$.
  Next, we apply the Operator Jensen Inequality (Lemma~\ref{lm:opjensen}) to get
  \begin{align}
    V† f(B\inv ⨂ Aᵀ) V ≤ f \big( V† (B\inv ⨂ Aᵀ) V \big) = f(B_1\inv ⨂ A_1ᵀ) ¶\,. 
  \end{align}
  Finally, using §[quasi-mono/iso-action], we recover $S_f(A \,\| B) ≤ S_f(A_1 \| B_1)$
  by taking the matrix element for $(√{B_1} ⨂ ⅈ_1 ) \ket{γ}_1$ on both sides of the inequality.
\end{proof}

\section{Properties of the Rényi Entropy}

Here, we prove some general properties of the relative ☼[Rényi entropies]{entropy!Rényi}, 
introduced in Chapter~\ref{ch:aep} as
\begin{align}
  S_{α} (A \| B) = \frac{1}{1-α} \log \tr<b>{A^α B^{1-α}} ¶.
\end{align}
Note that a similar quantity appears in quantum hypothesis testing~\cite{ogawa00,audenaert07} and in~\cite{hayashi06,ohya93,mosonyidatta08}, where alternative proofs of some of the following properties can be found.

Unlike their classical counterparts, the quantum relative (and conditional) min-
and max-entropies cannot be recovered as special cases of
$α$-entropies. However, it can be shown~\cite{koenig08} that
\begin{align}
  S_{½}(A \| B) = 2 \log \tr{√{A} √{B}} 
    ≤ 2 \log \trace \big| √{A} √{B} \big| = S_{\max}(A \| B) ¶[aep/half-bound]\,.
\end{align}
Furthermore, using the eigenvalue decompositions $A = ∑_i λ_i \proj{e_i}$ and
$B = ∑_j µ_j \proj{f_j}$, we have
\begin{align}
  S_{∞}(A \,\| B) = \lim_{ξ → 0} -\log \max_{\genfrac{}{}{0pt}{}{i,j}{\braket{e_i|f_j} 
    \neq 0}} \ \frac{λ_i}{µ_j + ξ} ≤ S_{\min}(A \,\| B) ¶\,.
\end{align}

The entropies are additive, e.g.\ evaluation for an i.i.d.\ operator $A^{⨂ n}$
relative to another i.i.d.\ operator $B^{⨂ n}$ results in
\begin{align}
  S_{α} \big(A^{⨂ n} \| B^{⨂ n} \big) = n S_{α}(A \| B) ¶[aep/renyi-add]\, .
\end{align}

The relative Rényi entropies decrease monotonically in $α$.

\begin{lemma}
  Let $α ≥ β ≥ 0$ and let $ρ ∈ \onorm{ℋ}, σ ∈ \opos{ℋ}$ with $\supp{ρ} \subseteq
  \supp{σ}$. Then, $S_{α}(ρ \,\| σ) ≤ S_{β}(ρ \,\| σ)$.
\end{lemma}

\begin{proof}
  We first show that $S_{α}$ decreases monotonically with increasing $α$ 
  by showing that its derivative is negative for all $α ∈ [0, 1) \cup (1, ∞)$.

  Using the (unnormalized) fully entangled state $\ket{γ}$, we define a purification $\ket{φ} := 
  (√{ρ} ⨂ ⅈ) \ket{γ}$ of $A$. Furthermore, we set $\bar{α} = α - 1$ and 
  $X = ( ρ ⨂ σ\inv )ᵀ$. It is easy to verify that, for $f: t ↦ t \log t$,
  \begin{align}
    S_{α}(A \,\| B) &= -\frac{1}{\bar{α}} \log\, \braket{φ|X^{\bar{α}}|φ}\, \tn*{and} ¶\\
    \frac{∂}{∂ α} S_{α}(A\,\|B) &= \frac{1}{\bar{α}^2} \log\, \braket{φ|X^{\bar{α}}|φ} - 
    \frac{1}{\bar{α}} \frac{\braket{φ|X^{\bar{α}} \log X|φ}}{\braket{φ|X^{\bar{α}}|φ}} ¶\\ 
    &= \frac{ f(\braket{φ|X^{\bar{α}}|φ}) - \braket{φ|f(X^{\bar{α}})|φ}}{\bar{α}^2\,
    \braket{φ|X^{\bar{α}}|φ}} ≤ 0¶\,.
  \end{align}
  The monotonicity follows from the convexity of $f$ together with Jensen's 
  inequality (cf.~Lemma~\ref{lm:jensen}). 
\end{proof}

The Rényi entropies cannot decrease under simultaneous application of a CPM on both arguments. This is equivalent to a data-processing inequality for the conditional version of the 
entropies.\footnote{Conversely, it is easy to see that the data-processing inequality for the min- and max-entropy (cf.~Theorem~\ref{th:data-proc}) imply monotonicity for $S_{\max}$ and $S_{\min}$.}
\begin{lemma}[Monotonicity of Rényi Entropy]
  \label{lm:renyi-mono}
  Let $α ∈ [0, 2]$, let $A, B ∈ \opos{ℋ}$ and let $ℰ$ be a TP-CPM. Then,
  $S_{α}( A \,\| B ) ≤ S_{α} \big( ℰ(A) \,\| ℰ(B) \big)$.
\end{lemma}
\begin{proof}
  This follows directly from Lemma~\ref{lm:quasi-mono} and the fact that $g_{α} : t ↦ t^{α}$ 
  is operator concave for $α ∈ (0, 1)$ and operator convex for $α ∈ (1, 2]$. (Also note that 
  the pre-factor conveniently changes sign between these two domains.) In the limits $α → 0$ and 
  $α → 1$, the property follows by continuity.\footnote{For the von Neumann entropy this property 
  (and strong sub-additivity) also follows from the operator concavity of $h: t ↦ - t \log t$, 
  as noted in~\cite{nielsenpetz04}.}
\end{proof}

We define conditional versions of these entropies as follows.
For a state $ρ ∈ \osub{ℋ¬{AB}}$, the conditional Rényi $α$-entropy of $A$ given $B$ is
\begin{align}
  \hh{α}{A|B}[ρ] := S_{α}( ρ¬{AB} \| ⅈ ⨂ ρ¬{B} ) ¶\,.
\end{align}
This definition allows a ☼{duality} relation for ☼[pure]{state!pure} tri-partite states.
\begin{lemma}
  \label{lm:renyi-dual}
  Let $ρ ∈ \osub{ℋ¬{ABC}}$ be pure and $α ∈ [0, 1) \cup (1, 2]$. Then
  \begin{align}
    \hh{α}{A|B}[ρ] + \hh{2-α}{A|C}[ρ] = 0 ¶\,.
  \end{align}
\end{lemma}
\begin{proof}
  We write $ρ = \proj{θ}$ and note that the marginal states $ρ¬{AB}$ and $ρ¬C$ satisfy
  $ρ¬{AB} \ket{θ} = ρ¬{C} \ket{θ}$ and $ρ¬B \ket{θ} = ρ¬{AC} \ket{θ}$. Thus,
  \begin{align}
    (1 - α) \hh{α}{A|B}[ρ] &= \log \tr<b>{ρ¬{AB}^α ρ¬B^{1 - α}} 
     = \log\, \braket{θ | ρ¬{AB}^{α - 1} ρ¬B^{1 - α} | θ} ¶\\
    &= \log\, \braket{θ | ρ¬C^{α - 1} ρ¬{AC}^{1 - α} | θ} = (α - 1) \hh{2-α}{A|B}[ρ]¶.
  \end{align}
  The last equality follows from $α - 1 = 1 - (2 - α)$.
\end{proof}

% index
%\cleardoublepage
%\phantomsection
%\addcontentsline{toc}{chapter}{Index}

%\printindex

% biography

%\chapter*{Biography}
%\addcontentsline{toc}{chapter}{Biography}

%\input{bio}

\end{document}